\documentclass{lmcs}
\usepackage{versions}
\includeversion{full}
\excludeversion{slim}

\PassOptionsToPackage{hyphens}{url}
\usepackage{amsmath}
\usepackage{amssymb}
\usepackage{amsthm}
\usepackage{cite}
\usepackage{centernot}
\usepackage{xcolor}
\usepackage{stmaryrd}
\usepackage{url}
\usepackage{mathtools}
\usepackage{microtype}
\usepackage{prftree}
\usepackage{tikz-cd}
\usepackage{tikz}
\usetikzlibrary{calc}
\usepackage{enumitem}
\setlist{noitemsep}

\usepackage[mathscr]{eucal}
\usepackage{csquotes}
\usepackage{todonotes}
\usepackage{fontawesome}	%

\newcommand\slimfull[2]{\processifversion{slim}{#1}\processifversion{full}{#2}}

\newcommand\REF[1]{#1} %

\makeatletter
\g@addto@macro{\UrlBreaks}{\UrlOrds\do\=\do\_}
\makeatother

\usepackage{xr}
\usepackage[
   a4paper,
   pdftitle={Optimistic Higher-Order Superposition},
   pdfauthor={},
   pdfkeywords={},
   pdfborder={0 0 0},
   draft=false,
   bookmarksnumbered,
   bookmarks,
   bookmarksdepth=3,
   bookmarksopenlevel=2,
   bookmarksopen]{hyperref}

\makeatletter
\def\@seccntformat#1{%
  \protect\textup{\protect\@secnumfont
    \ifnum\pdfstrcmp{subsection}{#1}=0 \bfseries\fi%
    \ifnum\pdfstrcmp{subsubsection}{#1}=0 \bfseries\fi%
    \csname the#1\endcsname
    \protect\@secnumpunct
  }%
}  
\makeatother

\DeclareFontFamily{OT1}{pzc}{}
\DeclareFontShape{OT1}{pzc}{m}{it}{<-> s * [1.10] pzcmi7t}{}
\DeclareMathAlphabet{\mathcalx}{OT1}{pzc}{m}{it}
\newcommand\cal[1]{\mathcalx{#1}}

\def\negvthinspace{\kern-0.083333em}
\def\vthinspace{\kern+0.083333em}
\def\vvthinspace{\kern+0.0416667em}
\def\negvvthinspace{\kern-0.0416667em}
\def\hypsep{\hskip1em}

\newcommand\omicron{o}

\newcommand\medrightarrow{\mathrel{{{\color{black}\relbar}\kern-0.9ex\rlap{\color{white}\ensuremath{\blacksquare}}\kern-0.9ex}\joinrel{\color{black}\rightarrow}}}
\newcommand\medleftarrow{\mathrel{{\color{black}\leftarrow}\kern-0.9ex\rlap{\color{white}\ensuremath{\blacksquare}}\kern-0.9ex\joinrel{{\color{black}\relbar}}}}
\newcommand\medleftrightarrow{\mathrel{\leftarrow\kern-1.685ex\rightarrow}}

\newcommand{\rewrite}{\rightarrow}

\newcommand{\leftrightrewrite}{\leftrightarrow}

\newcommand\Section{Sect.}

\definecolor{light-gray}{gray}{0.875}
\definecolor{darker-gray}{gray}{0.45}

\newcommand\smalllam{\raisebox{\depth}{\scalebox{1}[-1.05]{\ensuremath{_\mathsf{y}}}}\negvthinspace}

\newcommand\lkbg{\mathsf{g\smalllam kbo}}

\newcommand\lkb{\mathsf{\smalllam kbo}}

\newcommand\llpg{\mathsf{g\smalllam lpo}}

\newcommand\llp{\mathsf{\smalllam lpo}}

\newcommand{\Sigmaty}{\Sigma_\mathsf{ty}}
\newcommand{\VV}{\mathscr{V}}
\newcommand{\Vty}{\VV_\mathsf{ty}}
\newcommand{\III}{\mathscr{I}}
\newcommand{\II}{\mathscr{J}}
\newcommand{\IIty}{\II_\mathsf{ty}}
\newcommand{\IIIty}{\III_\mathsf{ty}}
\newcommand{\DD}{\mathscr{D}}

\newcommand{\RfN}{R}

\newcommand{\UU}{\mathscr{U}}

\newcommand{\EE}{\mathscr{E}}
\newcommand{\dho}{\DD}

\newcommand{\LL}{\mathscr{L}}
\newcommand\ty{\mathsf{ty}}

\newcommand{\Rbasic}{R}
\newcommand{\Deltabasic}{\Delta}

\newcommand{\infname}[1]{\textsc{#1}}

\newcommand{\Sup}{\infname{Sup}}
\newcommand{\FluidSup}{\infname{Fluid\-Sup}}
\newcommand{\EqRes}{\infname{EqRes}}
\newcommand{\EqFact}{\infname{EqFact}}
\newcommand{\Clausify}{\infname{Clausify}}
\newcommand{\BoolHoist}{\infname{Bool\-Hoist}}
\newcommand{\LoobHoist}{\infname{Loob\-Hoist}}
\newcommand{\FluidBoolHoist}{\infname{Fluid\-Bool\-Hoist}}
\newcommand{\FluidLoobHoist}{\infname{Fluid\-Loob\-Hoist}}
\newcommand{\FalseElim}{\infname{False\-Elim}}
\newcommand{\ArgCong}{\infname{Arg\-Cong}}
\newcommand{\Ext}{\infname{Ext}}  
\newcommand{\FluidExt}{\infname{FluidExt}}  
\newcommand{\Diff}{\infname{Diff}}

\newcommand{\IPGSup}{\infname{IPGSup}}
\newcommand{\IPGEqRes}{\infname{IPGEqRes}}
\newcommand{\IPGEqFact}{\infname{IPGEqFact}}

\newcommand{\IPGBoolHoist}{\infname{IPGBoolHoist}}
\newcommand{\IPGLoobHoist}{\infname{IPGLoobHoist}}
\newcommand{\IPGClausify}{\infname{IPGClausify}}
\newcommand{\IPGFalseElim}{\infname{IPGFalseElim}}
\newcommand{\IPGExt}{\infname{IPGExt}}
\newcommand{\IPGDiff}{\infname{IPGDiff}}

\newcommand{\IGSup}{\infname{IGSup}}
\newcommand{\IGEqRes}{\infname{IGEqRes}}
\newcommand{\IGEqFact}{\infname{IGEqFact}}

\newcommand{\IGBoolHoist}{\infname{IGBoolHoist}}
\newcommand{\IGLoobHoist}{\infname{IGLoobHoist}}
\newcommand{\IGClausify}{\infname{IGClausify}}
\newcommand{\IGFalseElim}{\infname{IGFalseElim}}
\newcommand{\IGExt}{\infname{IGExt}}
\newcommand{\IGDiff}{\infname{IGDiff}}

\newcommand{\PGSup}{\infname{PGSup}}
\newcommand{\PGEqRes}{\infname{PGEqRes}}
\newcommand{\PGEqFact}{\infname{PGEqFact}}
\newcommand{\PGArgCong}{\infname{PGArgCong}}
\newcommand{\PGBoolHoist}{\infname{PGBoolHoist}}
\newcommand{\PGLoobHoist}{\infname{PGLoobHoist}}
\newcommand{\PGClausify}{\infname{PGClausify}}
\newcommand{\PGFalseElim}{\infname{PGFalseElim}}
\newcommand{\PGExt}{\infname{PGExt}}
\newcommand{\PGDiff}{\infname{PGDiff}}

\newcommand{\GSup}{\infname{GSup}}
\newcommand{\GEqRes}{\infname{GEqRes}}
\newcommand{\GEqFact}{\infname{GEqFact}}
\newcommand{\GArgCong}{\infname{GArgCong}}
\newcommand{\GBoolHoist}{\infname{GBoolHoist}}
\newcommand{\GLoobHoist}{\infname{GLoobHoist}}
\newcommand{\GClausify}{\infname{GClausify}}
\newcommand{\GFalseElim}{\infname{GFalseElim}}
\newcommand{\GExt}{\infname{GExt}}
\newcommand{\GDiff}{\infname{GDiff}}

\newcommand{\PFSup}{\infname{PFSup}}
\newcommand{\PFExt}{\infname{PFExt}}
\newcommand{\PFEqRes}{\infname{PFEqRes}}
\newcommand{\PFEqFact}{\infname{PFEqFact}}
\newcommand{\PFBoolHoist}{\infname{PFBoolHoist}}
\newcommand{\PFLoobHoist}{\infname{PFLoobHoist}}
\newcommand{\PFFalseElim}{\infname{PFFalseElim}}
\newcommand{\PFClausify}{\infname{PFClausify}}
\newcommand{\PFArgCong}{\infname{PFArgCong}}
\newcommand{\PFDiff}{\infname{PFDiff}}

\newcommand{\FSup}{\infname{FSup}}
\newcommand{\FEqRes}{\infname{FEqRes}}
\newcommand{\FEqFact}{\infname{FEqFact}}
\newcommand{\FBoolHoist}{\infname{FBoolHoist}}
\newcommand{\FLoobHoist}{\infname{FLoobHoist}}
\newcommand{\FFalseElim}{\infname{FFalseElim}}
\newcommand{\FClausify}{\infname{FClausify}}
\newcommand{\FDiff}{\infname{FDiff}}
\newcommand{\FExt}{\infname{FExt}}
\newcommand{\FArgCong}{\infname{FArgCong}}

\newcommand{\leftsubterm}{[}
\newcommand{\rightsubterm}{]}
\newcommand{\subterm}[2]{#1\leftsubterm#2\rightsubterm}

\newcommand{\lang}{\begin{picture}(5,7)
\put(1.1,2.5){\rotatebox{45}{\line(1,0){6.0}}}
\put(1.1,2.5){\rotatebox{315}{\line(1,0){6.0}}}
\end{picture}}
\newcommand{\rang}{\begin{picture}(5,7)
\put(0,2.5){\rotatebox{135}{\line(1,0){6.0}}}
\put(0,2.5){\rotatebox{225}{\line(1,0){6.0}}}
\end{picture}}

\newcommand{\leftgreensubterm}{\lang\,}
\newcommand{\rightgreensubterm}{\rang}
\newcommand{\greensubterm}[2]{\ensuremath{#1\leftgreensubterm #2\rightgreensubterm}}

\newcommand{\leftorangesubterm}{\lang\!\!\leftgreensubterm}
\newcommand{\rightorangesubterm}{\rightgreensubterm\!\!\rang}
\newcommand{\orangesubterm}[2]{#1\leftorangesubterm #2\rightorangesubterm}

\newcommand{\leftinterpret}{\llbracket}
\newcommand{\rightinterpret}{\rrbracket}
\newcommand{\interpret}[3]{\smash{\leftinterpret #1\rightinterpret_{#2}^{#3}}}
\newcommand{\interpretR}[1]{\interpret{#1}{R}{}}
\newcommand{\interpretRmapF}[1]{\interpretR{\mapF{#1}}}

\newcommand{\interpretaxi}[1]{\interpret{#1}{\III}{\xi}}
\newcommand{\interpretfo}[2]{\interpret{#1}{\RfN}{#2}}

\newcommand{\xity}{\xi_\mathsf{ty}}
\newcommand{\xite}{\xi_\mathsf{te}}

\renewcommand{\doteq}{\mathrel{\dot\approx}}
\newcommand{\ceqneq}{\mathrel{\dot{\approx}}}
\newcommand{\eq}{\approx}
\newcommand{\noteq}{\not\eq}
\newcommand{\ceq}{\approx}
\newcommand{\cneq}{\not\approx}

\newcommand{\eqR}[2]{#1\sim#2}

\newcommand{\namedinference}[3]{\prftree[r]{\relax{\infname{#1}}}{\strut#2}{\strut#3}}
\newcommand{\inference}[2]{\namedinference{}{\strut#1}{\strut#2}}

\newcommand{\namedsimp}[3]{\prftree[d][r]{\relax{\infname{#1}}}{\strut#2}{\strut#3}}

\DeclareMathOperator{\csu}{CSU} %
\newcommand{\csuupto}{\csu^{\mathrm{upto}}}

\newcommand\UNIF{\mathrel{\smash{\stackrel{\lower.1ex\hbox{\ensuremath{\scriptscriptstyle ?}}}{=}}}}

\newcommand{\tuple}[1]{\bar{#1}}

\newcommand{\cst}[1]{{\mathsf{#1}}}

\newcommand\defeq{=}

\newcommand\fun{\rightarrow}

\newcommand\foralltynospace[1]{\mathsf{\Pi}#1.}
\newcommand\forallty[1]{\foralltynospace{#1}\;}

\newcommand\fofun{\Rightarrow}

\newcommand{\typeargs}[1]{{\langle#1\rangle\negvthinspace}}

\newcommand{\params}[1]{{(#1)}}

\newcommand\oftype{:}
\newcommand\oftypedecl{:}

\newcommand{\diff}{\cst{diff}}

\newcommand{\db}{\cst{db}}
\newcommand\DB[1]{#1}

\newcommand{\llor}{\mathrel\lor}

\newcommand\benf[1]{#1\vthinspace{\downarrow}_{\beta\eta}}
\newcommand\bnf[1]{#1\vthinspace{\downarrow}_{\beta}}

\newcommand{\gnd}{{\mathsf{Gnd}}}
\newcommand{\concl}{\mathit{concl}}
\newcommand{\prems}{\mathit{prems}}
\newcommand{\mprem}{\mathit{mprem}}

\newcommand{\HRedC}{{\mathit{HRed}}_{\mathrm{C}}}

\newcommand{\HRedI}{{\mathit{HRed}}_{\mathrm{I}}}
\newcommand{\PGRedI}{\mathit{\slimfull{}{P}GRed}_{\mathrm{I}}}
\newcommand{\IPGRedI}{\mathit{I\slimfull{}{P}GRed}_{\mathrm{I}}}
\newcommand{\GRedI}{\mathit{GRed}_{\mathrm{I}}}
\newcommand{\IGRedI}{\mathit{IGRed}_{\mathrm{I}}}
\newcommand{\PFRedI}{\mathit{\slimfull{}{P}FRed}_{\mathrm{I}}}
\newcommand{\FRedI}{\mathit{FRed}_{\mathrm{I}}}

\newcommand{\HInf}{\mathit{HInf}}
\newcommand{\PGInf}{\mathit{\slimfull{}{P}GInf}}
\newcommand{\IGInf}{\mathit{IGInf}}
\newcommand{\IPGInf}{\mathit{I\slimfull{}{P}GInf}}
\newcommand{\PFInf}{\mathit{\slimfull{}{P}FInf}}
\newcommand{\FInf}{\mathit{FInf}}

\newcommand{\wrt}{\hbox{w.r.t.}}
\newcommand{\soundmodels}{\mathrel{|\kern-.1ex}\joinrel\approx}

\newcommand{\modelsolam}{\models_\mathrm{o\lambda}}
\newcommand{\modelsfol}{\models_\mathrm{fol}}

\newcommand{\constraint}[1]{[\![#1]\!]}
\newcommand{\closure}{\cdot}

\newcommand{\irred}{\mathrm{irred}}

\newcommand{\TT}{\mathcalx{T}}

\newcommand{\CC}{\mathcalx{C}}

\newcommand{\SigmaI}{\Sigma_\mathrm{I}}
\newcommand{\SigmaH}{\Sigma_\mathrm{H}}

\newcommand{\VVPG}{\VV_\mathrm{PG}}
\newcommand{\VVH}{\VV_\mathrm{H}}

\newcommand{\levelH}{\mathrm{H}}
\newcommand{\levelG}{\mathrm{G}}
\newcommand{\levelPG}{\mathrm{\slimfull{}{P}G}}
\newcommand{\levelIPG}{\mathrm{I\slimfull{}{P}G}}
\newcommand{\levelPF}{\mathrm{\slimfull{}{P}F}}
\newcommand{\levelF}{\mathrm{F}}

\newcommand{\termsPG}{\TT_\levelPG}
\newcommand{\termsIPG}{\TT_\levelIPG}
\newcommand{\termsPF}{\TT_\levelPF}
\newcommand{\termsF}{\TT_\levelF}

\newcommand{\clausesH}{\CC_\levelH}
\newcommand{\clausesG}{\CC_\levelG}

\newcommand{\clausesPG}{\CC_\levelPG}
\newcommand{\clausesIPG}{\CC_\levelIPG}
\newcommand{\clausesPF}{\CC_\levelPF}
\newcommand{\clausesF}{\CC_\levelF}

\newcommand{\mapGonly}{\mathcalx{G}}
\newcommand{\mapPGonly}{\mathcalx{PG}}
\newcommand{\mapPonly}{\mathcalx{P}}
\newcommand{\mapIonly}{\mathcalx{J}} %
\newcommand{\mapFonly}{\mathcalx{F}}
\newcommand{\mapTonly}{\mathcalx{T}}

\newcommand{\mapG}[1]{\mapGonly(#1)}
\newcommand{\mapPG}[1]{\mapPGonly(#1)}
\newcommand{\mapP}[1]{\mapPonly(#1)}
\newcommand{\mapI}[1]{\mapIonly(#1)}
\newcommand{\mapF}[1]{\mapFonly(#1)}
\newcommand{\mapT}[1]{\mapTonly(#1)}

\newcommand{\fipg}[1]{\mapFonly\mapIonly\mapPonly\mapGonly(#1)}
\newcommand{\fip}[1]{\mapFonly\mapIonly\mapPonly(#1)}
\newcommand{\tfip}[1]{\mapTonly\mapFonly\mapIonly\mapPonly(#1)}

\newcommand\mapponly{\mathfrak{p}}
\newcommand\mapqonly{\mathfrak{q}}
\newcommand\mapp[1]{\mapponly(#1)}
\newcommand\mapq[1]{\mapqonly(#1)}

\newcommand{\flooronly}{\mathcalx{F}}

\newcommand{\floor}[1]{\flooronly\!(#1)}

\newdimen\carpetH
\newdimen\carpetV
\def\carpet#1{\setbox0=\hbox{\ensuremath{#1}}%
  \kern+2\carpetH
    \raise+2\carpetV\copy0\kern-\wd0
    \raise+\carpetV\copy0\kern-\wd0
    \copy0\kern-\wd0
    \raise-\carpetV\copy0\kern-\wd0
    \raise-2\carpetV\copy0\kern-\wd0
  \kern-\carpetH
    \raise+2\carpetV\copy0\kern-\wd0
    \raise+\carpetV\copy0\kern-\wd0
    \copy0\kern-\wd0
    \raise-\carpetV\copy0\kern-\wd0
    \raise-2\carpetV\copy0\kern-\wd0
  \kern-\carpetH
    \raise+2\carpetV\copy0\kern-\wd0
    \raise+\carpetV\copy0\kern-\wd0
    \copy0\kern-\wd0
    \raise-\carpetV\copy0\kern-\wd0
    \raise-2\carpetV\copy0\kern-\wd0
  \kern-\carpetH
    \raise+2\carpetV\copy0\kern-\wd0
    \raise+\carpetV\copy0\kern-\wd0
    \copy0\kern-\wd0
    \raise-\carpetV\copy0\kern-\wd0
    \raise-2\carpetV\copy0\kern-\wd0
  \kern-\carpetH
    \raise+2\carpetV\copy0\kern-\wd0
    \raise+\carpetV\copy0\kern-\wd0
    \copy0\kern-\wd0
    \raise-\carpetV\copy0\kern-\wd0
    \raise-2\carpetV\copy0
  \kern2\carpetH
}
\newcommand\heavy[1]{\carpetH=.02ex\carpetV=.02ex\carpet{#1}}
\newcommand\light[1]{\carpetH=.01ex\carpetV=.01ex\carpet{\scriptstyle#1}}

\newcommand{\ifalse}{\heavy\bot}
\newcommand{\itrue}{\heavy\top}
\newcommand{\inot}{\heavy{\neg}}
\newcommand{\iand}{\mathbin{\heavy\land}}
\newcommand{\ior}{\mathbin{\heavy\lor}}
\newcommand{\iimplies}{\mathrel{\heavy\rightarrow}}

\newcommand{\ieq}{\mathrel{\heavy\approx}}
\newcommand{\ineq}{\mathrel{\heavy{\not\approx}}}

\newcommand{\inotlight}{\light{\neg}}
\newcommand{\ifalselight}{\light\bot}
\newcommand{\itruelight}{\light\top}
\newcommand{\iandlight}{\mathbin{\light\land}}
\newcommand{\iorlight}{\mathbin{\light\lor}}

\newcommand{\ineqlight}{\mathrel{\light{\not\approx}}}

\newcommand\mapdbl[5]{#1\{\DB{0}\mapsto #2_#3, \ldots, \DB{#5}\mapsto #2_#4\}}
\newcommand\mapdbtuple[3]{#1\{(\DB{0},\dots,\DB{#2})\mapsto #3\}}
\newcommand\mapdb[3]{\mapdbtuple{#1}{#3-1}{\tuple{#2}_#3}}
\newcommand{\dbsubst}[1]{\{\DB{0}\mapsto#1\}}

\begin{document}

\begin{full}

\title{Optimistic Lambda-Superposition}

\author[A.~Bentkamp]{Alexander Bentkamp\lmcsorcid{0000-0002-7158-3595}}[a]
\author[J.~Blanchette]{Jasmin Blanchette\lmcsorcid{0000-0002-8367-0936}}[a]
\author[M.~Hetzenberger]{Matthias Hetzenberger\lmcsorcid{0000-0002-2052-8772}}[b]
\author[U.~Waldmann]{Uwe Waldmann\lmcsorcid{0000-0002-0676-7195}}[c]

\address{Ludwig-Maximilians-Universität München, Geschwister-Scholl-Platz 1,
80539 München, Germany}
\email{a.bentkamp@ifi.lmu.de,jasmin.blanchette@ifi.lmu.de}

\address{TU Wien Informatics,
Favoritenstraße 9--11,
1040 Vienna, Austria}
\email{matthias.hetzenberger@tuwien.ac.at}

\address{Max Planck Institute for Informatics, Campus E1 4, 66123 Saarbrücken, Germany}
\email{uwe@mpi-inf.mpg.de}

\begin{abstract}
\noindent
The $\lambda$-superposition calculus is a successful approach to proving
higher-order formulas. However, some parts of the calculus are extremely
explosive, notably due to the higher-order unifier enumeration and the
functional extensionality axiom.
In the present work, we introduce an ``optimistic'' version of
$\lambda$-superposition that addresses these two issues.
Specifically, our new calculus delays explosive unification problems
using constraints stored along with the clauses, and it applies functional extensionality in a more targeted way.
The calculus is sound and refutationally complete with respect to a Henkin
semantics. We have yet to implement it in a prover, but examples
suggest that it will outperform, or at least usefully complement, the original
$\lambda$-superposition calculus.
\end{abstract}

\maketitle

\end{full}

\begin{slim}
\title{Simplified Variant of Optimistic Lambda-Superposition}
\maketitle
This document describes a simplified variant of the optimistic $\lambda$-superposition calculus.
The main difference is that the present variant does not annotate clauses with constraints.
This simplifies especially the completness proof because we can use ground clauses instead of
ground closures in the first-order part of the proof.
It also strengthens and simplifies the redundancy criterion.
However, we are forced to introduce superposition inferences into variables
when those variables also have occurrences inside parameters.
\end{slim}

\begin{full}

\section{Introduction}

The (Boolean) $\lambda$-superposition calculus \cite{bentkamp-et-al-2023-hosup-journal},
which generalizes Bachmair and Ganzinger's superposition calculus \cite{bachmair-ganzinger-1994},
has shown itself to be 
a powerful automated reasoning method for classical higher-order logic
with function and Boolean extensionality.
The calculus is sound and refutationally complete with respect to a Henkin semantics.
It is implemented in the Zipperposition prover \cite{vukmirovic-et-al-2022-making-hosup-work},
and a refutationally incomplete, pragmatic version of it 
drives the E prover's higher-order mode \cite{vukmirovic-et-al-2023-extending}.

These implementations of $\lambda$-superposition achieve remarkable empirical results,
but to do so, they must deprioritize or---in incomplete variants---disable
specific features of the calculus that would otherwise cause combinatorial explosion.
Among these features, the most problematic are the following:
\begin{itemize}
  \item 
  the hugely expensive computation of unifiers of flex--flex pairs,
  which the calculus requires instead of allowing Huet's preunification procedure;
  \item the functional extensionality axiom and its orientation according to the
    term order, which enforces a lot of wasteful extensionality reasoning unrelated to the actual proof goal; and
  \item the so-called fluid superposition rule, which simulates rewriting below applied variables
  and which causes lots of inferences that rarely lead to a successful proof.
\end{itemize}

In this article, we introduce the \emph{optimistic $\lambda$-superposition}
calculus (Section~\ref{sec:calculus}), which addresses the first two issues:
\begin{itemize}
\item
  For unification, our new calculus delays explosive unification problems
  using constraints stored along with the clauses.
\item For functional extensionality, it introduces a
  targeted inference rule that works in tandem with tailored term orders,
  described in a companion article \cite{bentkamp-et-al-optimistic-orders}.
  The new rule works by first assuming that two functions are equal
  and delays the proof that they are equal on all
  arguments until the assumption is found to be useful.
\end{itemize}
Both of these new features delay some work and can be considered
``optimistic,'' hence the calculus's name.

As a pleasant side effect of the new functional extensionality rule, we can
strengthen the redundancy criterion used to simplify clauses. Some inference
rules of the original $\lambda$-superposition calculus are now simplification
rules in our new calculus.

\begin{exa}
As an illustration of the
stronger redundancy criterion, consider a derivation starting from the following clauses:
\begin{align*}
&(1)\ (\lambda x.\> x + 1) \cneq (\lambda x.\> 1 + x)\\
&(2)\ y + z \ceq z + y
\end{align*}
A negative extensionality inference from (1) yields the clause
\[
(3)\ \diff(\lambda x.\> x + 1,\>\lambda x.\> 1 + x) + 1 \cneq 1 + \diff(\lambda x.\> x + 1,\>\lambda x.\> 1 + x)
\]
which eventually leads to a derivation of the empty clause using (2).
The original $\lambda$-superposition calculus required us to keep
clause (1) and perform further inferences with it,
whereas our new calculus can immediately discard (1) when generating (3).
\end{exa}

We prove our calculus sound (Section~\ref{sec:soundness})
and refutationally complete (Section~\ref{sec:refutational-completeness}).
The completeness proof is structured as
six levels, starting from superposition for a ground first-order logic
and culminating with nonground higher-order logic with functional extensionality.
The parts of the proof concerned with the constraints attached to clauses
and with the new functional extensionality rule
are inspired by basic superposition \cite{nieuwenhuis-rubio-1992,bachmair-et-al-1992}.
The two new features make the proof rather complicated,
but the calculus is simpler than the original $\lambda$-superposition calculus
in many respects:
\begin{itemize}
  \item The intricate notions of ``deep occurrences'' and ``fluid terms'' play no role in our calculus.
  \item We removed the support for inner clausification, which does not perform well in practice and complicates the original $\lambda$-superposition calculus.
  As an additional benefit, this enables us to select literals of the form $t \ceq \ifalse$
  (a claim made for the original $\lambda$-superposition calculus as well \cite{bentkamp-et-al-2023-hosup-journal}
  but corrected later \cite{bentkamp-et-al-2023-hosup-errata, nummelin-et-al-2021-boolsup-errata}).
  \item Our calculus does not require the presence of Hilbert's choice axiom.
\end{itemize}
In principle, these simplifications could be applied to the original $\lambda$-superposition calculus as well
without adding unification constraints or the new functional extensionality rule.

Our calculus's two main features are inspired by Vampire's higher-order mode \cite{bhayat-suda-2024},
which is currently the best performing higher-order prover in CASC \cite{sutcliffe-casc-j12,sutcliffe-desharnais-casc-29}.
Like our calculus, Vampire delays unification problems and functional extensionality proofs.
The mechanisms are slightly different, however, because
Vampire stores delayed unification problems in negative literals instead of constraints,
allowing inference rules to be applied to them,
and it uses unification with abstraction for
functional extensionality (which is also implemented in E
\cite[p.~13]{vukmirovic-et-al-2023-extending}) instead of an inference rule.
The similarities to performant provers, along with example problems we have
studied, suggest that our calculus not only is refutationally complete but will
also perform well empirically.
For further related work, we refer to Bentkamp et al.~\cite{bentkamp-et-al-2023-hosup-journal}.

\end{full}

\section{Logic}
\label{sec:logic}

Our formalism is higher-order logic
with functional and Boolean extensionality, rank-1 polymorphism,
but without choice and the axiom of infinity.
The logic closely resembles Gordon and Melham's HOL \cite{gordon-melham-1993},
the TPTP TH1 standard \cite{kaliszyk-et-al-2016},
and the logic underlying $\lambda$-superposition by Bentkamp et
al.~\cite{bentkamp-et-al-2023-hosup-journal}.

Departing from Bentkamp et al.,
in the present work, quantifiers are not supported and must always be encoded as
$(\lambda x.\> t) \ieq (\lambda x.\> \itrue)$
and $(\lambda x.\> t) \ineq (\lambda x.\> \ifalse)$.
This is necessary because quantifiers would
prevent us from constructing a suitable term order
for the extensionality behavior
that we want to achieve.
Moreover, we do not include the axiom of choice.

To make the positive literal of the extensionality axiom maximal, 
we introduce a special type of argument to constants into our syntax, the \emph{parameters}.
A constant that takes parameters cannot occur without them;
partial application is not allowed for parameters.
Moreover, parameters cannot contain variables bound by $\lambda$-abstractions.

As our semantics, we use Henkin semantics.
True statements in these semantics correspond exactly to
provable statements in the HOL systems. Since Henkin semantics
is not subject to G\"odel's first incompleteness theorem,
it allows us to prove refutational completeness.

\subsection{Syntax}

We use the notation $\tuple{a}_n$ or $\tuple{a}$ to denote
a tuple $(a_1, \ldots, a_n)$. If $f$ is a unary function,
we write $f(\tuple{a}_n)$ for the elementwise application
$(f(a_1), \ldots, f(a_n))$.

\subsubsection{Types}

To define our logic's types, we fix an infinite set $\Vty$
of type variables. A set $\Sigmaty$ of type constructors,
each associated with an arity, 
is a \emph{type signature}
if it contains at least one nullary type constructor $\omicron$ of Booleans
and a binary type constructor $\fun$ of functions.
A \emph{type} is either a type variable $\alpha \in \Vty$ or an applied type
constructor $\kappa(\tuple{\tau}_n)$ for some $n$-ary $\kappa \in \Sigmaty$
and types~$\tuple{\tau}_n$.
To indicate that an expression $e$ has type $\tau$,
we write $e \oftype \tau$.

\subsubsection{Lambda-Preterms and Lambda-Terms}

To define our logic's terms, for a given type signature $\Sigmaty$,
we fix a set $\VV$
of variables with associated types.
We write $x\langle\tau\rangle$ for a variable named $x$ with associated type $\tau$.
We require that
$\VV$ contains infinitely many variables of any type.

A \emph{term signature} $\Sigma$ is
a set of constants. Each constant is associated with
a type declaration of the form $\Pi\tuple{\alpha}_m.\>\tuple{\tau}_n\fofun\upsilon$,
where
$\bar\tau_n$ and $\upsilon$ are types
and $\bar\alpha_m$ is a tuple of distinct variables
that contains all type variables from $\bar\tau_n$ and $\upsilon$.
The types $\bar\tau_n$ are the types of the parameters of the constant,
and $\upsilon$ may be a function type if the constant takes
nonparameter arguments.
We require that $\Sigma$ contains the logical symbols
${\itrue,\ifalse}\oftype\omicron$;
$\inot\oftype\omicron\fun\omicron$;
${\iand,\ior,\iimplies}\oftype\omicron\fun\omicron\fun\omicron$;
and ${\ieq},{\ineq} \oftype \Pi\alpha.\>\alpha\fun\alpha\fun\omicron$.
A type signature and a term signature form a \emph{signature}.

Our syntax makes use of a locally nameless notation \cite{chargueraud-2012}
using De Bruijn indices \cite{de-bruijn-1972}. We distinguish between 
$\lambda$-preterms, $\lambda$-terms, preterms, and terms.
Roughly, $\lambda$-preterms are raw syntactic expressions,
$\lambda$-terms are the subset of locally closed $\lambda$-preterms,
preterms are $\beta\eta$-equivalence classes of $\lambda$-preterms,
and terms are $\beta\eta$-equivalence classes of $\lambda$-terms.
More precisely, we define these notions as follows.

The set of $\lambda$-preterms is built from the following expressions:
\begin{itemize}
\item a variable $x\langle\tau\rangle \oftype \tau$ for $x\langle\tau\rangle \in \VV$;
\item a symbol $\cst{f}\typeargs{\bar\upsilon_m}\params{\bar u_n} \oftype \tau$
for a constant $\cst{f}\in\Sigma$ with type declaration $\Pi\tuple{\alpha}_m.\>\tuple{\tau}_n\fofun\tau$, types $\tuple{\upsilon}_m$,
and $\lambda$-preterms $\bar u \oftype \tuple{\tau}_n$ such that all De Bruijn indices in $\bar u$ are bound;
\item a De Bruijn index $\DB{n}\langle\tau\rangle \oftype \tau$ for a natural number $n\geq 0$ and a type $\tau$, where $\tau$
  represents the type of the bound variable;
\item a $\lambda$-expression $\lambda\langle\tau\rangle\> t \oftype \tau\to\upsilon$
  for a type $\tau$ and a $\lambda$-preterm $t\oftype\upsilon$
  such that all De Bruijn indices bound by the new $\lambda\langle\tau\rangle$
  have type $\tau$;
\item an application $s\>t \oftype \upsilon$ for
  $\lambda$-preterms $s\oftype\tau\to\upsilon$ and $t\oftype\tau$.
\end{itemize}
The type arguments $\langle\bar\tau\rangle$ carry enough information to enable
typing of any $\lambda$-preterm without any context. We often leave them
implicit, when they are irrelevant or can be inferred.
In $\cst{f}\langle\bar\upsilon_m\rangle(\bar u_n) : \tau$,
we call $\bar u_n$ the parameters.
We omit $()$ when a symbol has no parameters.
Notice that it is possible for a term to contain multiple occurrences of the same free De Bruijn index
with different types. In contrast, the types of bound De Bruijn indices always match.

The set of \emph{$\lambda$-terms} is the subset
$\lambda$-preterms without free De Bruijn indices, i.e, the subset of locally closed $\lambda$-preterms.
We write
$\TT^\lambda(\Sigma,\VV)$ for the set of all $\lambda$-terms and
$\TT^{\lambda\mathrm{pre}}(\Sigma,\VV)$ for the set of all $\lambda$-preterms,
sometimes omitting the set $\VV$ when it is clear from the context.

A $\lambda$-preterm is called \emph{functional} if
its type is of the form $\tau \to \upsilon$ for some types $\tau$ and $\upsilon$.
It is called \emph{nonfunctional} otherwise.

Given a $\lambda$-preterm $t$ and $\lambda$-terms $s_0, \ldots, s_n$,
we write $\mapdbl{t}{s}{0}{n}{n}$ for the $\lambda$-preterm
resulting from substituting $s_i$ for each De Bruijn index $\DB{i + j}$ 
enclosed into exactly $j$ $\lambda$-abstractions in $t$.
For example, $(\cst{f}\>\DB{0}\>\DB{1}\>(\lambda\>\cst{g}\>\DB{1}\>\DB{2}))\{\DB{0}\mapsto\cst{a},\DB{1}\mapsto\cst{b}\}
=\cst{f}\>\cst{a}\>\cst{b}\>(\lambda\>\cst{g}\>\cst{a}\>\cst{b})$.
Given a $\lambda$-preterm $t$ and 
a tuple $\tuple{s}_n$ of $\lambda$-terms,
we abbreviate $\mapdbl{t}{s}{1}{n}{(n-1)}$ as $\mapdbtuple{t}{n-1}{\tuple{s}_n}$.

We write $\bnf{t}$ for the $\beta$-normal form of a $\lambda$-preterm $t$.

A $\lambda$-preterm $s$ is a \emph{subterm} of a $\lambda$-preterm $t$,
written $t = t[s]$,
if $t = s$,
if $t = \cst{f}\typeargs{\bar\tau_m}\params{\bar u}\>v$ with $u_i = u_i[s]$ or $v = v[s]$,
if $t = \lambda\>u[s]$,
if $t = (u[s])\>v$, or
if $t = u\>(v[s])$.
A subterm is \emph{proper} if it is distinct from the $\lambda$-preterm itself.

A $\lambda$-preterm is \emph{ground} if it contains no type variables and no term variables, i.e., if it is closed and monomorphic.
We write $\smash{\TT^{\lambda\mathrm{pre}}_\mathrm{ground}(\Sigma)}$ for the set of ground $\lambda$-preterms
and $\TT^\lambda_\mathrm{ground}(\Sigma)$ for the set of ground $\lambda$-terms.

\subsubsection{Preterms and Terms}

The set of (\emph{pre})\emph{terms} consists of the
$\beta\eta$-equivalence classes of $\lambda$-(pre)terms.
For a given set of variables $\VV$ and signature $\Sigma$, we write
$\TT(\Sigma,\VV)$ for the set of all terms and
$\TT^\mathrm{pre}(\Sigma,\VV)$ for the set of all preterms,
sometimes omitting the set $\VV$ when it is clear from the context.
We write $\TT_\mathrm{ground}(\Sigma)$ for the set of ground terms.

When referring to properties of a preterm
that depend on the representative of its equivalence class modulo $\beta$
(e.g., when checking whether a preterm is ground or whether a preterm contains a given variable $x$),
we use a $\beta$-normal representative
as the default representative of the $\beta\eta$-equivalence class.
When referring to properties of a preterm that depend on the choice of
representative modulo $\eta$, we state the intended representative explicitly.

Clearly, any preterm in $\beta$-normal form
has one of the following
four mutually exclusive forms:
\begin{itemize}
\item $x\langle\tau\rangle\> \bar{t}$ for a variable $x\langle\tau\rangle$ and terms $\bar{t}$;
\item $\cst{f}\langle\bar\tau\rangle(\bar u)\> \bar{t}$
for a symbol $f$, types $\tuple{\tau}$, and terms $\bar{u}$, $\bar{t}$;
\item $\DB{n}\langle\tau\rangle\> \bar{t}$
for a De Bruijn index $\DB{n}\langle\tau\rangle$ and terms $\bar{t}$;
\item $\lambda\langle\tau\rangle\> t$ for a term $t$.
\end{itemize}

\subsubsection{Substitutions}

A substitution is a mapping $\rho$ from type variables  $\alpha \in \Vty$ to types $\alpha\rho$
and from term variables $x \langle\tau\rangle \in \VV$ to ($\lambda$-)terms $x\rho \oftype \tau\rho$.
A substitution $\rho$ applied to a ($\lambda$-)term $t$ yields a ($\lambda$-)term $t\rho$
in which each variable $x$ is replaced by $x\rho$.
Similarly, subsitutions can be applied to types.
The notation $\{\tuple{\alpha} \mapsto \tuple{\tau}, \tuple{x} \mapsto \tuple{t}\}$
denotes a substitution that maps each $\alpha_i$ to $\tau_i$ and each $x_i$ to $t_i$,
and all other type and term variables to themselves.
The composition $\rho\sigma$ of two substitutions applies first $\rho$ and then $\sigma$:
$t\rho\sigma = (t\rho)\sigma$.
A \emph{grounding} substitution maps all variables to ground types and ground ($\lambda$-)terms.
The notation $\sigma[\tuple{x}\mapsto \tuple{t}]$
denotes the substitution that maps each $x_i$ to $t_i$ and otherwise coincides with $\sigma$.

\subsubsection{Clauses}

Finally,
we define the higher-order clauses on which
our calculus operates.
A \emph{literal} is an unordered pair of two terms $s$ and $t$
associated with a positive or negative sign.
We write positive literals as $s \ceq t$ 
and negative literals as $s \cneq t$.
The notation $s \doteq t$ stands for either $s \ceq t$ or $s \cneq t$.
Nonequational literals are not supported and must be
encoded as $s \ceq \itrue$ or $s \ceq \ifalse$.
A clause $L_1 \llor \cdots \llor L_n$ is a finite multiset of literals.
The empty clause is written as $\bot$.
\begin{slim}
Finally, we define a grounding function $\mapGonly$ on clauses as
$\mapG{C} =\{C\theta\mid \theta\text{ is a grounding substitution}\}$
\end{slim}

\begin{full}

\subsubsection{Constraints}
A constraint is a term pair, written as
$s \equiv t$.
A set of constraints 
$s_1 \equiv t_1,\> \dots,\> s_n \equiv t_n$ is \emph{true} if $s_i$ and $t_i$ are syntactically equal for all $i$.
A set of constraints $S$ is \emph{satisfiable} if there exists a substitution
such that $S\theta$ is true.
A \emph{constrained clause} $C\constraint{S}$ is a pair of a clause $C$ and a finite set of constraints $S$.
We write $\clausesH$ for the set of all constrained clauses.
Similarly, a \emph{constrained term} $t\constraint{S}$ is a pair of a term $t$ and a finite set of constraints $S$.
Terms and clauses are special cases of constrained terms and constrained clauses
where the set of constraints is empty.
Given $C\constraint{S} \in \clausesH$ and a grounding substitution $\theta$
such that $S\theta$ is true,
we call $C\theta$ a
\emph{ground instance} of $C\constraint{S}$.
We write $\gnd(C\constraint{S})$ for the set of
all ground instances of a constrained clause $C\constraint{S}$.
  
\end{full}

\subsection{Semantics}
The semantics is essentially the same as in Bentkamp et al. \cite{bentkamp-et-al-2023-hosup-journal},
adapted to the modified syntax.

A \emph{type interpretation} $\IIIty = (\UU, \IIty)$ is defined as follows.
The \emph{universe} $\UU$ is a collection of nonempty sets, called
\emph{domains}. We require that $\{0,1\}\in\UU$.
The function $\IIty$ associates a function
$\IIty(\kappa) : \UU^n \rightarrow \UU$
with each $n$-ary type constructor~$\kappa$,
such that $\IIty(\omicron) = \{0,1\}$
and for all domains $\DD_1,\DD_2\in\UU$, the set $\IIty(\fun)(\DD_1,\DD_2)$
is a subset of the function space from $\DD_1$ to $\DD_2$.
The semantics is \emph{standard} if
$\IIty({\fun})(\DD_1,\DD_2)$ is the entire function space for all $\DD_1,\DD_2$.
A \emph{type valuation} $\xity$ is a function that maps every type variable to a domain.
The \emph{denotation} of a type for a type interpretation $\IIIty$
and a type valuation $\xity$ is recursively defined by
$\interpret{\alpha}{\IIIty}{\xity}=\xity(\alpha)$ and
$\interpret{\kappa(\tuple{\tau})}{\IIIty}{\xity}=
\IIty(\kappa)(\interpret{\tuple{\tau}}{\IIIty}{\xity})$.

Given a type interpretation $\IIIty$ and a type valuation $\xity$, a \emph{term valuation} $\xite$
assigns an element $\xite(x)\in\interpret{\tau}{\IIIty}{\xity}$ to each variable $x \oftype \tau$.
A valuation $\xi = (\xity, \xite)$ is a pair of a type valuation $\xity$ and a term valuation $\xite$.

An \emph{interpretation function} $\II$ for a type interpretation $\IIIty$ associates with each symbol
$\cst{f}\oftypedecl\forallty{\tuple{\alpha}_m}\>\tuple{\tau}\fofun\upsilon$,
a domain tuple $\tuple{\DD}_m\in\UU^m$, and
values $\tuple{a} \in
\interpret{\tuple{\tau}}{\IIIty}{\xity}$
a value
$\II(\cst{f},\tuple{\DD}_m,\tuple{a}) \in
\interpret{\upsilon}{\IIIty}{\xity}$,
where $\xity$ is a type valuation that maps each $\alpha_i$ to $\DD_i$.
We require that

\medskip

\begin{minipage}{.35\textwidth}
\begin{enumerate}[label=(I\arabic*)]
	\item \label{item:interpretation:true}
	$\II(\itrue) = 1 $
	\item \label{item:interpretation:false}
	$\II(\ifalse) = 0$
	\item \label{item:interpretation:and}
	$\II(\iand)(a,b) = \min\,\{a,b\}$
	\item \label{item:interpretation:or}
	$\II(\ior)(a,b) = \max\,\{a,b\}$
\end{enumerate}
\end{minipage}
\begin{minipage}{.65\textwidth}
\begin{enumerate}[label=(I\arabic*), start=5]
	\item \label{item:interpretation:not}
	$\II(\inot)(a) = 1 - a $
	\item \label{item:interpretation:implies}
	$\II(\iimplies)(a,b) = \max\,\{1-a,b\}$
	\item \label{item:interpretation:eq}
	$\II(\ieq,\DD)(c,d) = 1$ if $c = d$ and $0$ otherwise
	\item \label{item:interpretation:neq}
	$\II(\ineq,\DD)(c,d) = 0$ if $c = d$ and $1$ otherwise
\end{enumerate}
\end{minipage}

\medskip

\noindent
for all $a,b\in\{0,1\}$, $\DD\in\UU$, and $c,d\in \DD$.

The comprehension principle states that every function designated by
a $\lambda$-expression is contained in the corresponding domain.
Loosely following Fitting~\cite[\Section~2.4]{fitting-2002}, we initially allow
$\lambda$-expressions to designate arbitrary elements of the domain, to be
able to define the denotation of a $\lambda$-term. We impose restrictions afterward
using the notion of a proper interpretation, enforcing comprehension.

A \emph{$\lambda$-designation function} $\LL$
for a type interpretation $\IIIty$ is a function that maps
a valuation $\xi$ and a $\lambda$-expression of type $\tau$ to elements
of $\interpret{\tau}{\IIIty}{\xity}$.
We require that the value $\LL(\xi,t)$ depends only
on values of $\xi$ at type and term variables that actually occur in $t$.
A type interpretation, an interpretation function, and a $\lambda$-designation function form an
\emph{interpretation} $\III = (\IIIty,\II,\LL)$.

For an interpretation~$\III$ and a valuation~$\xi$, the \relax{denotation of a $\lambda$-term} is defined
as
$\interpretaxi{x} \defeq \xite(x)$,
$\interpretaxi{\cst{f}\typeargs{\tuple{\tau}}\params{\tuple{s}}} \defeq
\II(\cst{f},\interpret{\tuple{\tau}}{\IIIty}{\xity},\interpretaxi{\tuple{s}})$,
$\interpretaxi{s\>t} \defeq \interpretaxi{s} (\interpretaxi{t})$, and
$\interpretaxi{\lambda\langle\tau\rangle\> t} \defeq \LL(\xi,\lambda\langle\tau\rangle\> t)$.
For ground $\lambda$-terms $t$, the denotation does not depend on the choice of the valuation $\xi$,
which is why we sometimes write $\interpret{t}{\III}{}$ for $\interpretaxi{t}$.

An interpretation $\III$ is \emph{proper} if
$\interpret{\lambda\langle\tau\rangle\>t}{\III}{(\xity,\xite)}(a) = \interpret{t\dbsubst{x}}{\III}{(\xity,\xite[x\mapsto a])}$ for all
$\lambda$-expressions $\lambda\langle\tau\rangle\>t$ and all valuations $\xi$, where $x$ is a fresh variable.
Given an interpretation $\III$ and a valuation $\xi$, a
positive literal $s\ceq t$ (resp.\ negative literal $s\cneq t$) is
\relax{true} if %
$\interpretaxi{s}$ and $\interpretaxi{t}$ are equal (resp.\ different).
A clause is \relax{true} if at least one of its literals is true.
\begin{full}
A constrained clause 
$C \constraint{s_1\equiv t_1,\ldots,s_n\equiv t_n}$ is \emph{true}
if $C \llor s_1 \cneq t_1 \llor \dots \llor s_n \cneq t_n$ is true.
\end{full}
A set of \begin{full}constrained\end{full} clauses is \relax{true} if all its elements are true.
A proper interpretation $\III$ is a \emph{model} of a set $N$ of \begin{full}constrained\end{full} clauses,
written $\III \models N$, if $N$ is true in $\III$ for all valuations $\xi$.
Given two sets $M, N$ of \begin{full}constrained\end{full} clauses,
we say that $M$ \emph{entails} $N$, written $M \models N$,
if every model of $M$ is also a model of $N$.

\subsection{The Extensionality Skolem Constant} %
Any given signature can be extended with a distinguished constant
$\cst{diff}\oftypedecl\Pi\alpha,\beta.\>(\alpha\fun\beta, \alpha\fun\beta) \fofun \alpha$,
which we require for our calculus. Interpretations as defined above
can interpret the constant $\cst{diff}$ arbitrarily.
The intended interpretation of $\cst{diff}$ is as follows:
\begin{defi}\label{def:H:diff-aware}
We call a proper interpretation $\III$
\emph{$\diff$-aware} if $\III$ is a model of the extensionality axiom---i.e.,
\[\III\models z\>(\diff\typeargs{\alpha,\beta}(z,y))\cneq y\>(\diff\typeargs{\alpha,\beta}(z,y)) \llor z\ceq y\]
Given two sets $M, N$ of \begin{full}constrained\end{full} clauses,
we write $M \soundmodels N$
if every $\diff$-aware interpretation that is a model of $M$ is also a model of $N$.
\end{defi}
Our calculus is sound and refutationally complete \wrt\ $\soundmodels$ but unsound \wrt\ $\models$.

\section{Calculus}
\label{sec:calculus}

\begin{full}

The optimistic $\lambda$-superposition calculus is designed to process
an unsatisfiable set of higher-order clauses
that have no constraints and do not contain constants with parameters,
to enrich this clause set with clauses that may have constraints and may contain the constant $\diff$,
and to eventually derive an empty clause with satisfiable constraints.

Central notions used to define the calculus are \emph{green subterms} (Section~\ref{ssec:green-subterms}), which many of the calculus rules are restricted to,
and \emph{complete sets of unifiers up to constraints} (Section~\ref{ssec:complete-sets-of-unifiers-up-to-constraints}), which replace the first-order notion of a most general unifier.
Existing unification algorithms must be adapted to cope with terms containing parameters (Section~\ref{ssec:unification}).
The calculus is parameterized by
a term order and a selection function,
which must fulfill certain requirements (Section~\ref{ssec:term-orders-and-selection-functions}).
The concrete term orders defined in a companion article fulfill the requirements (Section~\ref{ssec:concrete-term-orders}).
Our core inference rules describe how the calculus derives new clauses (Section~\ref{ssec:the-core-inference-rules}),
and
a redundancy criterion defines abstractly under which circumstances clauses may be deleted and when
inferences may be omitted (Section~\ref{ssec:redundancy}).
The abstract redundancy criterion supports a wide collection of concrete
simplification rules (Section~\ref{ssec:simplification-rules}).
Examples illustrate the calculus's strengths and limitations (Section~\ref{ssec:examples}).

\end{full}

\subsection{Orange, Yellow, and Green Subterms}\label{ssec:green-subterms}
As in the original $\lambda$-superposition calculus,
a central notion of our calculus is the notion of green subterms.
These are the subterms that we consider for superposition inferences.
For example, in the clause $\cst{f}\>\cst{a} \cneq \cst{b}$,
a superposition inference at $\cst{a}$ or $\cst{f}\;\cst{a}$ is possible,
but not at $\cst{f}$.
Our definition here deviates from Bentkamp et al.~\cite{bentkamp-et-al-2023-hosup-journal}
in that functional terms never have nontrivial green subterms.

In addition to green subterms,
we define yellow subterms, which extend green subterms with
subterms inside $\lambda$-expressions,
and orange subterms,
which extend yellow subterms with
subterms containing free De Bruijn indices.
Orange subterms are the subterms that our redundancy criterion allows
simplification rules to rewrite at.
For example, the clauses $\lambda\>\cst{c} \cneq \cst{b}$
and $\cst{f}\>x\>x \ceq \cst{c}$ can make $\lambda\>\cst{f}\>\DB{0}\>\DB{0} \cneq \cst{b}$
redundant (assuming a suitable clause order),
but $\cst{g}\>\cst{a} \cneq \cst{b}$
and $\cst{g} \ceq \cst{f}$
cannot make $\cst{f}\>\cst{a} \cneq \cst{b}$ redundant.
It is convenient to define orange subterms first,
then derive yellow and green subterms based on orange subterms.

Orange subterms depend on the choice of $\beta\eta$-normal form:
\begin{defi}[$\beta\eta$-Normalizer]
  Given a preterm $t$, let $t\downarrow_{\beta\eta\mathrm{long}}$ be its $\beta$-normal $\eta$-long form
  and
  let $t\downarrow_{\beta\eta\mathrm{short}}$ be its $\beta$-normal $\eta$-short form.
  A $\beta\eta$-normalizer is a function $\benf{} \in \{\downarrow_{\beta\eta\mathrm{long}}, \downarrow_{\beta\eta\mathrm{short}}\}$.
\end{defi}

\begin{defi}[Orange Subterms]
  \label{def:orange-subterms}
  We start by defining orange positions and orange subterms on $\lambda$-preterms.

  Given a list of natural numbers $p$ and $s,t \in \TT^{\lambda\mathrm{pre}}(\Sigma)$, we 
  say that $p$ is an \emph{orange position} of $t$,
  and $s$ is an \emph{orange subterm} of $t$ at $p$, written $t|_p = s$,
  if this can be derived inductively from the following rules:
  \begin{enumerate}[label=\arabic*.,ref=\arabic*]
  \item $u|_\varepsilon = u$ for all $u \in \TT^{\lambda\mathrm{pre}}(\Sigma)$, where $\varepsilon$ is the empty list.
  \item If $u_i|_p = v$, then $(\cst{f}\typeargs{\tuple{\tau}}(\bar s)\> \bar u_n)|_{i.p} = v$ for all $\cst{f} \in \Sigma$, types $\tuple{\tau}$, $\lambda$-preterms $\bar s, \bar{u}_n, v \in \TT^{\lambda\mathrm{pre}}(\Sigma)$, and $1\leq i \leq n$.
  \item If $u_i|_p = v$, then $(\DB{m}\langle\tau\rangle\> \bar u_n)|_{i.p} = v$ for all De Bruijn indices $m$, types $\tau$, $\lambda$-preterms $\bar{u}_n, v \in \TT^{\lambda\mathrm{pre}}(\Sigma)$, and $1\leq i \leq n$.
  \item If $u|_p = v$, then $(\lambda\langle\tau\rangle\> u)|_{1.p} = v$ for all types $\tau$ and $\lambda$-preterms $u, v \in \TT^{\lambda\mathrm{pre}}(\Sigma)$.
  \end{enumerate}
  We extend these notions to preterms as follows.
  Given a $\beta\eta$-normalizer $\benf{}$, a list of natural numbers $p$ and $s,t \in \TT^{\mathrm{pre}}(\Sigma)$, we 
  say that $p$ is an \emph{orange position} of $t$,
  and $s$ is an \emph{orange subterm} of $t$ at $p$ \wrt\ $\benf{}$, written $t|_p = s$,
  if $(\benf{t})|_p = \benf{s}$.

  The context $u[\phantom{i}]$ surrounding an orange subterm $s$ of $u[s]$ is
  called an \emph{orange context}. The notation $\orangesubterm{u}{s}_p$ or $\orangesubterm{u}{s}$ indicates
  that $s$ is an orange subterm in $u[s]$ at position $p$, and $\orangesubterm{u}{\phantom{i}}$
  indicates that $u[\phantom{i}]$ is an orange context.
\end{defi}

\begin{exa}
  Whether a preterm is an orange subterm of another preterm depends on the chosen $\beta\eta$-normal form $\benf{}$.
  For example, the preterms $\cst{f}\>\DB{0}$ and $\DB{0}$ are orange subterms of $\lambda\>\cst{f}\>\DB{0}$ in $\eta$-long form,
  but they are not orange subterms of the $\eta$-short form $\cst{f}$ of the same term.
\end{exa}

\begin{rem}
  The possible reasons for a subterm not to be orange are the following:
  \begin{itemize}
  \item It is applied to arguments.
  \item It occurs inside a parameter.
  \item It occurs inside an argument of an applied variable.
  \end{itemize}
\end{rem}

\begin{defi}[Yellow Subterms] \label{def:yellow-subterms}
  Let $\benf{}$ be a $\beta\eta$-normalizer.
  A \emph{yellow subterm} \wrt\ $\benf{}$ is an orange subterm that does not contain free De Bruijn indices.
  A \emph{yellow position} \wrt\ $\benf{}$ is an orange position that identifies a yellow subterm.
  The context surrounding a yellow subterm is called a \emph{yellow context}.
\end{defi}

\begin{lem}
Whether a preterm is a yellow subterm of another preterm is independent of $\benf{}$.
(On the other hand, its yellow position may differ.)
\end{lem}
\begin{proof}
It suffices to show that a single $\eta$-expansion or $\eta$-contraction
from a $\beta$-reduced $\lambda$-preterm $s$ into another $\beta$-reduced $\lambda$-preterm
cannot remove yellow subterms.
This suffices because only such $\eta$-conversations are needed to
transform a $\beta$-normal $\eta$-long form into a $\beta$-normal $\eta$-short form and vice versa.

Assume $s$ has a yellow subterm at yellow position $p$.
Consider the possible forms that a $\beta$-reduced $\lambda$-preterm $s$ can have:
\begin{itemize}
  \item $x\langle\tau\rangle\> \bar{t}$ for a variable $x\langle\tau\rangle$ and $\lambda$-preterms $\bar{t}$;
  \item $\cst{f}\langle\bar\tau\rangle(\bar u)\> \bar{t}$
  for a symbol $\cst{f}$, types $\tuple{\tau}$, and $\lambda$-preterms $\bar{u}$, $\bar{t}$;
  \item $\DB{n}\langle\tau\rangle\> \bar{t}$
  for a De Bruijn index $\DB{n}\langle\tau\rangle$ and $\lambda$-preterms $\bar{t}$;
  \item $\lambda\langle\tau\rangle\> t$ for a $\lambda$-preterm $t$.
\end{itemize}
Consider where an $\eta$-conversion could happen:
If an $\eta$-expansion takes place at the left-hand side of an application,
the result is not $\beta$-reduced.
If an $\eta$-reduction takes place at the left-hand side of an application,
the original $\lambda$-preterm is not $\beta$-reduced.
If the yellow subterm at $p$ does not overlap with the place of $\eta$-conversion,
the $\eta$-conversion has no effect on the yellow subterm.
This excludes the case where the $\eta$-conversion takes place in an argument of
an applied variable or in a parameter.
So the only relevant subterms for $\eta$-conversions
are 
(a) the entire $\lambda$-preterm $s$,
(b) a subterm of $\tuple{t}$ in $\cst{f}\langle\bar\tau\rangle(\bar u)\> \bar{t}$,
(c) a subterm of $\tuple{t}$ in $\DB{n}\langle\tau\rangle\> \bar{t}$,
or (d) a subterm of $t$ in $\lambda\langle\tau\rangle\> t$.

Next, we consider the possible positions $p$.
If the $\eta$-conversion takes place inside of the yellow subterm,
it certainly remains orange
because orange subterms only depend on the outer structure of the $\lambda$-preterm.
It also remains yellow because $\eta$-conversion does not introduce free De Bruijn indices.
This covers in particular the case where $p$ is the empty list.
Otherwise,
the yellow subterm at $p$ is also
(i) a yellow subterm of $\tuple{t}$ in $\cst{f}\langle\bar\tau\rangle(\bar u)\> \bar{t}$,
(ii) a yellow subterm of $\tuple{t}$ in $\DB{n}\langle\tau\rangle\> \bar{t}$,
or (iii) a yellow subterm of $t$ in $\lambda\langle\tau\rangle\> t$.
In cases (b), (c), and (d),
we can apply the induction hypothesis to $\tuple{t}$ or $t$
and conclude that the yellow subterm of $\tuple{t}$ or $t$ remains yellow
and thus the yellow subterm of $s$ at $p$ remains yellow as well.
In case (a), we distinguish between the cases (i) to (iii) described above:
\begin{enumerate}
  \item[(i)] Then the only option is an $\eta$-expansion of $\cst{f}\langle\bar\tau\rangle(\bar u)\> \bar{t}$
  to $\lambda\>\cst{f}\langle\bar\tau\rangle(\bar u)\> \bar{t}\>\DB{0}$.
  Clearly, the yellow subterm in $\tuple{t}$ remains yellow, although its yellow position changes.
  \item[(ii)] Analogous to (i).
  \item[(iii)] Here, one option is an $\eta$-expansion of $\lambda\> t$ to
  $\lambda\>\lambda\> t\>\DB{0}$, which can be treated analogously to~(i).

  The other option is an $\eta$-reduction of
  $\lambda\> t$ to $t'$, where $t = t'\> \DB{0}$.
  We must show that a yellow subterm of $t$ is also a yellow subterm of $t'$.
  Since a yellow subterm of $t$ cannot contain the free De Bruijn index $\DB{0}$,
  the $\lambda$-preterm $t'$ must be of the form $v\>\tuple{w}$, where the preterm $v$ is a symbol or a De Bruijn index
  and the yellow subterm of $t = v\>\tuple{w}\>\DB{0}$ must be a yellow subterm of 
  one of the arguments $\tuple{w}$.
  Then it is also a yellow subterm of $v\>\tuple{w}=t'$.\qedhere
\end{enumerate}
\end{proof}

\begin{defi}[Green Subterms] \label{def:green-subterms}
  A \emph{green position} is an orange position $p$
  such that each orange subterm at a proper prefix of $p$
  is nonfunctional.
  \emph{Green subterms} are orange subterms at green positions.
  The context surrounding a green subterm $s$ of $u[s]$ is
  called a \emph{green context}. The notation $\greensubterm{u}{s}_p$ or $\greensubterm{u}{s}$ indicates
  that $s$ is a green subterm in $u[s]$ at position $p$, and $\greensubterm{u}{\phantom{i}}$
  indicates that $u[\phantom{i}]$ is a green context.
\end{defi}

Clearly, green subterms can equivalently be described as follows:
Every term is a green subterm of itself.
If $u$ is nonfunctional,
then every green subterm of one of its arguments $s_i$ is
a green subterm of $u = \cst{f}(\bar t)\> \bar s$
and of $u = \DB{n}\>\bar t$.
Moreover, since $\eta$-conversions can occur only at functional subterms,
both green subterms and green positions do not depend on the choice of
a $\beta\eta$-normalizer $\benf{}$.

\begin{exa}
  Let $\iota$ be a type constructor.
  Let $\alpha$ be a type variable.
  Let $x \oftype \iota\fun\iota$ be a variable.
  Let $\cst{a} \oftype \iota$, 
  $\cst{f} \oftypedecl \Pi \alpha.\> \iota \fofun (\iota \fun \iota) \fun \alpha$,
  and $\cst{g} \oftypedecl \iota \fun \iota \fun \iota$ be constants.
  Consider the term 
  $\cst{f}\typeargs{\alpha}(\cst{a})\>(\lambda\>\cst{g}\>(x\>\cst{a})\>\DB{0})$.
  Its green subterms
  are the entire term (at position $\varepsilon$)
  and $\lambda\>\cst{g}\>(x\>\cst{a})\>\DB{0}$ (at position $1$).
  Its yellow subterms are the green subterms and
  $x\>\cst{a}$ (at position $1.1.1$ \wrt\ $\downarrow_{\beta\eta\mathrm{long}}$
  or at position $1.1$ \wrt\ $\downarrow_{\beta\eta\mathrm{short}}$).
  Its orange subterms \wrt\ $\downarrow_{\beta\eta\mathrm{long}}$ are the yellow subterms
  and $\cst{g}\>(x\>\cst{a})\>\DB{0}$ (at position $1.1$) and $\DB{0}$
  (at position $1.1.2$).
  Using $\downarrow_{\beta\eta\mathrm{short}}$,
  the orange subterms of this term are exactly the yellow subterms.
\end{exa}

For positions in clauses, natural numbers are not appropriate because clauses
and literals are unordered.
A solution is the following definition:
\begin{defi}[Orange, Yellow, and Green Positions and Subterms in Clauses]
  Let $C$ be a clause, let $L = s \ceqneq t$ be a literal in $C$,
  and let $p$ be an orange position of $s$.
  Then we call the expression $L.s.p$ an \emph{orange position} in $C$,
  and the \emph{orange subterm} of $C$ at position $L.s.p$ is the
  orange subterm of $s$ at position $p$.
  \emph{Yellow positions/subterms} and \emph{green positions/subterms} of clauses are defined analogously.
\end{defi}

\begin{exa}
The clause $C = K \llor L$ with $K= \cst{f} \>\cst{a} \cneq \cst{b}$ and $L= \cst{c} \ceq \cst{f} \>\cst{a}$
contains the orange subterm $\cst{a}$ twice,
once at orange position $L.(\cst{f} \>\cst{a}).1$ and once at orange position $K.(\cst{f} \>\cst{a}).1$.
\end{exa}

\begin{full}

\subsection{Complete Sets of Unifiers up to Constraints}\label{ssec:complete-sets-of-unifiers-up-to-constraints}

Most of our calculus rules can be used in conjunction with Huet-style preunification, full unification, and various variants thereof.
Only some rules require full unification.
To formulate the calculus in full generality,
we introduce the notion of a complete set of unifiers up to constraints.
The definition closely resembles the definition of a complete set of unifiers,
but allows us to unify only partially and specify the remainder in form of constraints.

\begin{defi}\label{def:csu-upto}
  Given a set of
  constraints $S$ and a set~$X$ of variables,
  where $X$ contains at least the variables occurring in $S$,
  a \emph{complete set of unifiers up to constraints} is a
  set~$P$ whose elements are pairs, each containing a substitution
  and a set of constraints, with the following properties:
  \begin{itemize}
    \item Soundness: For every $(\sigma,T)  \in P$ and unifier~$\rho$ of $T$,
    $\sigma\rho$ is a unifier of $S.$
    \item Completeness: For every unifier $\theta$ of $S$,
    there exists a pair $(\sigma,T)  \in P$ and a unifier~$\rho$ of $T$
    such that
    $x\sigma\rho = x\theta$ for all $x \in X.$
  \end{itemize}

  Given a set of
  constraints $S$ and a set~$X$ of variables,
  we let $\csuupto_X(S)$ denote an arbitrary
  complete set of unifiers upto constraints
  with the following properties.
  First,
  to avoid ill-typed terms, we require that for the substittions in $\csuupto_X(S)$ unify the types of equated terms in $S$.
  In practice, this is not a severe restriction because type unification always terminates.
  Second, we require that the substitutions $\sigma$ in $\csuupto_X(S)$ are idempotent on $X$---i.e.,
  $x\sigma\sigma = x\sigma$ for all $x \in X,$
  which can always be achieved by renaming variables.

  The set~$X$ will
  consist of the free variables of the clauses that the constraints $S$ originate from and
  will be left implicit.
\end{defi}

\begin{exa}\label{ex:csu-upto}
For the constraint
$y\>\cst{a} \equiv \cst{f}\>(z\>\cst{b})$ and $X = \{y,z\}$,
the set
$\{(\sigma, \{w\>\cst{a} \equiv z\>\cst{b}\})\}$
with $\sigma = \{y \mapsto \lambda\>\cst{f}\>(w\>0)\}$
is a complete set of unifiers up to constraints.
It is sound 
because for every unifier $\rho$ of $w\>\cst{a} \equiv z\>\cst{b}$,
the subsitution $\sigma\rho$ is a unifier of
$y\>\cst{a} \equiv \cst{f}\>(z\>\cst{b})$ since
$(y\>\cst{a} \equiv \cst{f}\>(z\>\cst{b}))\sigma = (\cst{f}\>(w\>\cst{a})\equiv \cst{f}\>(z\>\cst{b}))$.
It is complete because, for every unifier $\theta$ of $y\>\cst{a} \equiv \cst{f}\>(z\>\cst{b})$,
the term $y\theta$ must be of the form $\lambda\>\cst{f}\>t$ for some preterm $t$,
and then
the substitution $\rho = \{w \mapsto \lambda\>t, z \mapsto z\theta\}$ is a unifier of $w\>\cst{a} \equiv z\>\cst{b}$ and
fulfills $x\sigma\rho = x\theta$ for $x \in \{y,z\}$.
\end{exa}

\end{full}
\begin{slim}
\subsection{Complete Sets of Unifiers}\phantom{.}
\end{slim}

\begin{defi}\label{def:csu}
  Given a set of
  constraints $S$ and a set~$X$ of variables,
  where $X$ contains at least the variables occurring in $S$,
  a \emph{complete set of unifiers}
  is a set $P$ of unifiers of $S$ such that
  for each unifier $\theta$ of $S$,
  there exists a substitution $\sigma \in P$
  and a substitution $\rho$ such that
  $x\sigma\rho = x\theta$ for all $x \in X.$

  Given a set of
  constraints $S$ and a set~$X$ of variables,
  we write $\csu_X(S)$ or $\csu(S)$ for an arbitrary complete set of unifiers.
  Again, we require that all elements of $\csu(S)$ unify at least the types of
  the terms pairs in $S$ and that all elements of $\csu(S)$ are idempotent.
\end{defi}

\begin{full}

Equivalently, we could define a
complete set of unifiers
as a set $P$ of substitutions such that
$\{(\sigma, \emptyset)\mid \sigma\in P\}$
is a complete set of unifiers up to constraints.

The definitions above require $x\sigma\rho = x\theta$ only for variables $x \in X,$
not for other variables,
because the substitutions should be allowed to use
auxiliary variables.
For instance,
in Example~\ref{ex:csu-upto} above,
for most unifiers $\theta$, it is impossible to find a suitable $\rho$
that fulfills $x\sigma\rho = x\theta$ for all variables $x$, including $x = w$.

When choosing a strategy to compute complete sets of unifiers up to constraints,
there is a trade-off between how much computation time is spent
and how precisely the resulting substitutions instantiate variables.
At one extreme, we can compute a complete set of unifiers,
which instantiates variables as much as possible.
At the other extreme,
the set containing only the identity substitution and the original set of constraints is
always a complete set of unifiers up to constraints,
demonstrating that there exist terminating procedures that compute
complete sets of unifiers up to constraints.
In between these extremes lies Huet's preunification procedure~\cite{huet-1975, dowek-2001}.
A good compromise in practice may be to run Huet's preunification procedure
and to abort after a fixed number of steps, as described in the following subsection.

\subsection{A Concrete Unification Strategy}
\label{ssec:unification}

As a strategy to compute $\csuupto$, we suggest the following
procedure, which is a bounded variant of Huet's preunification procedure, adapted to
cope with polymorphism and parameters.
This approach avoids coping with infinite streams of unifiers
(except for rules that must use $\csu$ instead of $\csuupto$)
and resembles Vampire's strategy~\cite{bhayat-suda-2024}.

Analogously to what we describe below, for $\csu$, one can extend procedures for
the computation of complete sets of unifiers,
such as Vukmirovi\'c et al.'s procedure~\cite{vukmirovic-et-al-2021-unif},
to cope with parameters.

\begin{defi}[Flex-Flex, Flex-Rigid, Rigid-Rigid]
\label{def:flex-rigid}
Let $s \equiv t$ be a constraint.
We write $s$ and $t$ in $\beta$-normal $\eta$-long form as
$s = \lambda \cdots \lambda\> a\>u_1\>\cdots u_p$ and
$t = \lambda \cdots \lambda \> b\>v_1\>\cdots v_q$,
where $a$ and $b$ are
variables, De Bruijn indices, or symbols (possibly with type arguments and parameters)
and $u_i$ and $v_i$ are preterms.
If $a$ and $b$ are both variables, we say that $s \equiv t$ is
a flex-flex constraint.
If only one of them is a variable, we say that $s \equiv t$ is
a flex-rigid constraint.
If neither $a$ nor $b$ is a variable, we say that $s \equiv t$ is
a rigid-rigid constraint.
\end{defi}

The Huet preunification procedure computes a substitution
that unifies a set of constraints up to flex-flex pairs. 
It works as follows.
Given a finite set of constraints $S_0$,
we construct a search tree whose nodes are
either failure nodes \faTimes\ or
pairs $(\sigma, S)$ of a substitution $\sigma$ and a set $S$ of constraints.
The root node is the pair $(\{\}, S_0)$.
Any node $(\sigma, S)$ where $S$ contains only flex-flex constraints
is a successful leaf node.
All failure nodes \faTimes\ are also leaf nodes.
To construct the children of any other node $(\sigma, S)$,
we pick one of the constraints $s \equiv t \in S$ that is not a flex-flex constraint
and apply the following rules:
\begin{itemize}
  \item Type unificaiton:
  We attempt to unify the types of $s$ and $t$,
  which can be done using a first-order unification procedure.
  If the types are unifiable with a most general type unifier $\rho$, we add a child node $(\sigma\rho, S\rho)$.
  Otherwise, we add a child node \faTimes.
  \item If the types of $s$ and $t$ are equal,
  we write
  $s = \lambda \cdots \lambda\> a\>u_1\>\cdots u_p$ and
  $t = \lambda \cdots \lambda \> b\>v_1\>\cdots v_q$ as in Definition~\ref{def:flex-rigid}
  and apply the following rules:
  \begin{itemize}
    \item Rigid-rigid cases: Let $S' = S \setminus \{s \equiv t\}$.
    \begin{itemize}
      \item If $a$ and $b$ are different De Bruijn indices, or if one of them is De Bruijn index and the other a symbol, we add a child node \faTimes.
      \item If $a$ and $b$ are identical De Bruijn indices, we add a child node
      $(\sigma, S' \cup \{u_1 \equiv v_1, \ldots, u_p \equiv v_p\})$.
      \item If $a = \cst{f}\typeargs{\tuple{\tau}}(s_1,\ldots,s_k)$ and $b = \cst{g}\typeargs{\tuple{\tau}}(t_1,\ldots,t_l)$ with $\cst{f} \ne \cst{g}$, we add a child node \faTimes.
      \item If $a = \cst{f}\typeargs{\tuple{\tau}}(s_1,\ldots,s_k)$ and $b = \cst{f}\typeargs{\tuple{\upsilon}}(t_1,\ldots,t_k)$ where $\tuple{\tau}$ and $\tuple{\upsilon}$ are not unifiable, we add a child node \faTimes.
      \item If $a = \cst{f}\typeargs{\tuple{\tau}}(s_1,\ldots,s_k)$ and $b = \cst{f}\typeargs{\tuple{\upsilon}}(t_1,\ldots,t_k)$ where $\tuple{\tau}$ and $\tuple{\upsilon}$ are unifiable with a most general type unifier $\rho$,
      we add a child node $(\sigma\rho, (S' \cup \{s_1 \equiv t_1, \ldots, s_k \equiv t_k, u_1 \equiv v_1, \ldots, u_p \equiv v_p\})\rho)$.
      \end{itemize}
      \item Flex-rigid cases: Let $\tau_1 \fun \cdots \fun \tau_p \fun \tau$ be the type of $a$ and $\upsilon_1 \fun \cdots \fun \upsilon_q \fun \tau$ be the type of $b$.
      \begin{itemize}
        \item Imitation: If $a$ is a variable $x$ and
        $b$ is either a De Bruijn index or a symbol (possibly with type arguments and parameters),
        we add a child node $(\sigma\rho, S\rho)$ 
        with $\rho = \{x \mapsto \lambda\langle\tau_1\rangle \cdots \lambda\langle\tau_p\rangle\>b\>(y_1\>\DB{(p-1)}\>\cdots\>0)\>\cdots\>(y_q\>\DB{(p-1)}\>\cdots\>0)\}$,
        where $y_1, \ldots, y_q$ are fresh variables with $y_i$ of type $\tau_1 \fun \cdots \fun \tau_{p} \fun \upsilon_i$ for each $i$.
        \item Projection: If $a$ is a variable $x$ and $b$ is either a De Bruijn index or a symbol (possibly with type arguments and parameters), then
        for each $0 \leq i < p$ where $\tau_i = \tau'_1 \fun \cdots \fun \tau'_k \fun \tau$
        for some $\tau'_1, \ldots, \tau'_k$,
        we add a child node $(\sigma\rho, S\rho)$
        with $\rho = \{x \mapsto \lambda\langle\tau_1\rangle \cdots \lambda\langle\tau_p\rangle\>\DB{i}\>(y_1\>\DB{(p-1)}\>\cdots\>0)\>\cdots\>(y_k\>\DB{(p-1)}\>\cdots\>0)\}$,
        where $y_1, \ldots, y_k$ are fresh variables with $y_i$ of type $\tau_1 \fun \cdots \fun \tau_{p} \fun \tau'_j$ for each $j$.
        \item The same applies with the roles of $a$ and $b$ swapped.
      \end{itemize}
    \end{itemize}
  \end{itemize}
Ultimately, the tree's leaf nodes are either failure nodes \faTimes{} or success nodes $(\sigma, S)$,
where $S$ contains only flex-flex constraints and $\sigma$ is the corresponding preunifier.
Collecting all the preunifiers in the leaves yields the result of the standard, i.e., unbounded,
Huet preunification procedure.

We propose to use a bounded variant instead to ensure that unification always terminates.
In the bounded version, we construct the tree only up to a predetermined depth.
Collecting all unifiers and their associated constraints in the leaves of the resulting tree also yields a complete
set of unifiers up to constraints, which we can use in the role of the $\csuupto$ function of our core inference rules.

In addition,
following Vukmirovi\'c et al.~\cite{vukmirovic-et-al-2021-unif},
we propose to extend this procedure with algorithms for decidable fragments such
as pattern unification \cite{miller-1992}, fixpoint unification \cite{huet-1975}, and solid unification \cite{vukmirovic-et-al-2021-unif}.
When one of these fragments applies to one of the constraints $s \equiv t$ of a node $(\sigma, S)$,
the most general unifier $\rho$ for this constraint can be determined in finite time,
and we can add a single child node $(\sigma\rho, (S \setminus \{s \equiv t\})\rho)$
instead of the child nodes that would be added by the standard procedure.

\begin{lem}\label{lem:huet-csuupto}
The above procedure yields a complete set of unifiers up to constraints
(Definition~\ref{def:csu-upto}).
\end{lem}
\begin{proof}
  Let $S_0$ be a set of constraints.
  Consider a search tree constructed by the above procedure.
  We must show that the successful leaves $P$ of the tree form a complete set of unifiers up to constraints, i.e.,
  we must show:
  \begin{itemize}
    \item Soundness: For every $(\sigma,T)  \in P$ and unifier~$\rho$ of $T$,
    $\sigma\rho$ is a unifier of $S_0.$
    \item Completeness: For every unifier $\theta$ of $S_0$,
    there exists a pair $(\sigma,T)  \in P$ and a unifier~$\rho$ of $T$
    such that
    $x\sigma\rho = x\theta$ for all $x \in X.$
  \end{itemize}
  For soundness, we prove the following more general property:
  For every node $(\sigma, T)$ of the tree and every unifier $\rho$ of $T$,
  the substitution $\sigma\rho$ is a unifier of $S_0$.
  It is easy to check that 
  the initial node has this property and that for each of the rules above,
  the constucted child node has the property if the parent node has it.
  Thus, soundness follows by induction on the structure of the tree.

  For completeness, we prove the following more general property:
  Given a node $(\sigma_0, U_0)$ of the tree
  and a unifier $\theta_0$ of $U_0$,
  there exists a pair $(\omega,V)  \in P$ and a unifier~$\pi$ of $V$
  such that
  $x\omega\pi = x\sigma_0\theta_0$ for all $x \in X.$

  Since the search tree is clearly finite, we can apply structural induction on the tree.
  So we may assume that the property holds for all child nodes of $(\sigma_0, U_0)$.
  We proceed by a case distinction analogous to the cases describing the procedure above.

  If no decidable fragment applies to $U_0$ and $U_0$ contains only flex-flex pairs or
  if the depth limit has been reached, then $(\sigma_0, U_0)$ is a leaf node.
  Then $(\sigma_0, U_0)\in P$
  and the property holds with $\pi = \theta_0$.
  
  Otherwise, if a decidable fragment applies to a constraint $s \equiv t \in U_0$
  and provides a most general unifier $\rho$, then 
  we have a child node $(\sigma_0\rho, (U_0 \setminus \{s \equiv t\})\rho)$.
  Since $\rho$ is a most general unifier, there exists a substitution $\theta_1$
  such that $y\theta_0 = y\rho\theta_1$
  for all variables $y$ in $x\sigma_0$ with $x \in X$ and
  for all variables $y$ in $U_0$.
  So $\theta_1$ is a unifier of $(U_0 \setminus \{s \equiv t\})\rho$
  and by the induction hypothesis,
  there exists a pair $(\omega,V)  \in P$ and a unifier~$\pi$ of $V$
  such that
  $x\omega\pi = x\sigma_0\rho\theta_1$ for all $x \in X.$
  Thus, $x\omega\pi = x\sigma_0\rho\theta_1 = x\sigma_0\theta_0$ for all $x \in X,$ as required.

  Otherwise, no decidable fragment applies to $U_0$ and $U_0$ contains a pair that is not flex-flex.
  Then our procedure picks one such pair $s \equiv t \in U_0$.

  If the types of $s$ and $t$ are not equal, then they must be unifiable because $\theta_0$ is a unifier.
  So there exists a child node $(\sigma_0\rho, U_0\rho)$, where $\rho$ is the most general type unifier of $s$ and $t$.
  We can then proceed as in the decidable fragment case above.

  Otherwise, the types of $s$ and $t$ are equal.
  Let 
  $s = \lambda \cdots \lambda\> a\>u_1\>\cdots u_p$ and
  $t = \lambda \cdots \lambda \> b\>v_1\>\cdots v_q$ as in Definition~\ref{def:flex-rigid}.

  If $s \equiv t$ is a rigid-rigid pair,
  then $a$ and $b$ must be unifiable because $\theta_0$ is a unifier.
  So $a$ and $b$ are either identical De Bruijn indices or unifiable symbols.
  In both cases, we can then proceed analogously to the decidable fragment case above.

  If $s \equiv t$ is a flex-rigid pair,
  we assume without loss of generality that $a$ is a variable $x$.
  Since $\theta_0$ is a unifier of $s$ and $t$ and parameters cannot contain free De Bruijn indices, the term $a\theta_0$ must be 
  either of the form $\lambda\>\cdots\>\lambda\>b\>\tuple{s}$ for some terms $\tuple{s}$
  or of the form $\lambda\>\cdots\>\lambda\>i\>\tuple{s}$ for some De Bruijn index $i$ and terms $\tuple{s}$.
  In the first case, we apply the induction hypothesis to the child node produced by the imitation rule,
  and in the second case, we apply the induction hypothesis to the child node produced by the projection rule.
  In both cases, given the substitution $\rho$ used by the rule, it is easy to construct a substitution $\theta_1$ such that
  $y\theta_0 = y\rho\theta_1$ for all relevant variables $y$.
  Then we can proceed as in the decidable fragment case above.
\end{proof}
For an efficient implementation, it is important to $\beta\eta$-normalize terms and apply substitutions lazily,
similarly to the approach of Vukmirovi\'c et al.~\cite{vukmirovic-et-al-2021-unif}.
\end{full}

\subsection{Term Orders and Selection Functions}
\label{ssec:term-orders-and-selection-functions}

Our calculus is parameterized 
by a relation $\succ$ on \begin{full}constrained\end{full} terms, \begin{full}constrained\end{full} literals, and \begin{full}constrained\end{full} clauses.
We call $\succ$ the term order,
but it need not formally be a partial order.
Moreover, our calculus is parameterized by
a literal selection function.

The original $\lambda$-superposition calculus
also used 
a nonstrict term order $\succsim$ to compare
terms that may become equal when instatiated,
such as $x\>\cst{b} \succsim x\>\cst{a}$, where $\cst{b} \succ \cst{a}$.
However, contrary to the claims made for 
the original $\lambda$-superposition calculus,
employing the nonstrict term order can lead to incompleteness\cite{bentkamp-et-al-2021-lamsup-journal-errata},
which is why we do not use it in our calculus.

Moreover, the original $\lambda$-superposition calculus used
a Boolean selection function to restrict
inferences on clauses containing Boolean subterms.
For simplicity, we omit this feature in our calculus
because an evaluation did not reveal any practical benefit~\cite{nummelin-et-al-2021}.

\begin{defi}[Admissible Term Order]\label{def:admissible-term-order}
A relation $\succ$ on \begin{full}constrained\end{full} terms and on \begin{full}constrained\end{full} clauses is
an \emph{admissible term order}
if it fulfills the following criteria, where $\succeq$ denotes the
reflexive closure of $\succ$:
\begin{enumerate}[label=(O\arabic*), leftmargin=3em, ref=(O\arabic*)]
  \item the relation $\succ$ on ground terms is a well-founded total order\label{cond:order:total};
  \item ground compatibility with yellow contexts:\enskip $s' \succ s$ implies
  $\orangesubterm{t}{s'} \succ \orangesubterm{t}{s}$
  for ground terms $s$, $s'$, and $t$;\label{cond:order:comp-with-contexts}
  \item ground yellow subterm property:\enskip $\orangesubterm{t}{s} \succeq s$
  for ground terms $s$ and $t$;\label{cond:order:subterm}
  \item $u \succ \ifalse \succ \itrue$ for all ground terms $u \notin \{\itrue, \ifalse\}$;\label{cond:order:t-f-minimal}
  \item $u \succ u\>\diff\typeargs{\tau,\upsilon}(s,t)$ for all ground types $\tau, \upsilon$ and ground terms $s,t,u \oftype \tau \fun \upsilon$;\label{cond:order:ext}
  \item the relation $\succ$ on ground clauses is the standard extension of $\succ$ on ground terms via multisets \cite[\Section~2.4]{bachmair-ganzinger-1994};\label{cond:order:clause-extension}
  \item stability under grounding substitutions for terms:\enskip
  \begin{full}$t\constraint{T} \succ s\constraint{S}$\end{full}
  \begin{slim}$t\succ s$\end{slim}
  implies $t\theta \succ s\theta$ for all grounding substitutions~$\theta$%
  \begin{full} such that $T\theta$ and $S\theta$ are true\end{full};
  \label{cond:order:stability-terms}
  \item stability under grounding substitutions for clauses:\enskip
  \begin{full}$D\constraint{T} \succ C\constraint{S}$\end{full}
  \begin{slim}$D\succ C$\end{slim}
  implies $D\theta \succ C\theta$ for all grounding substitutions~$\theta$%
  \begin{full} such that $T\theta$ and $S\theta$ are true\end{full};
  \label{cond:order:stability-clauses}
  \item transitivity on \begin{full}constrained\end{full} literals: the relation $\succ$ on \begin{full}constrained\end{full} literals is transitive;\label{cond:order:transitive}
  \begin{full}
  \item for all terms $t$ and $s$ such that $t \succ s$
  and all substitutions $\theta$ such that
  for all type variables $\alpha$,
  the type $\alpha\theta$ is ground
  and such that
  for all variables $x$, all variables in $x\theta$
  are nonfunctional,
  if $s\theta$ contains a variable
  outside of parameters,
  then $t\theta$ must also contain that variable outside of parameters.\label{cond:order:variable}
  \end{full}
  \end{enumerate}
  \newcounter{ordercounter}
  \setcounter{ordercounter}{\value{enumi}}
\end{defi}

\begin{defi}[Maximality]\label{def:maximality}
  Given a term order $\succ$,
  a literal $K$ of a \begin{full}constrained\end{full} clause $C\begin{full}\constraint{S}\end{full}$ is \emph{maximal}
  if for all $L \in C$ such that $L\begin{full}\constraint{S}\end{full} \succeq K\begin{full}\constraint{S}\end{full}$,
  we have $L\begin{full}\constraint{S}\end{full} \preceq K\begin{full}\constraint{S}\end{full}$.
  It is \emph{strictly} maximal if it is maximal and occurs only once in $C$.    
\end{defi}

In addition to the term order, our calculus is parameterized by a selection function:

\begin{defi}[Literal Selection Function] \label{def:lit-sel}
  A literal selection function is a mapping from 
  each \begin{full}constrained\end{full} clause to a subset of its literals.
  The literals in this subset are called \emph{selected}.
  Only negative literals and
  literals of the form $t \ceq \ifalse$ may be selected.
\end{defi}

Based on the term order and the selection function, we define \emph{eligibility} as follows:
\begin{defi}[Eligibility] \label{def:ho-eligible}
  \,A literal $L$ is (\emph{strictly}) \emph{eligible} \wrt\ a substitution $\sigma$ in $C\begin{full}\constraint{S}\end{full}$ if it is
  selected in $C\begin{full}\constraint{S}\end{full}$ or there are no selected literals in $C\begin{full}\constraint{S}\end{full}$ 
  and $L\sigma$ is (strictly) maximal in 
  \begin{slim}$C\sigma$\end{slim}%
  \begin{full}$(C\constraint{S})\sigma$\end{full}.

  A green position
  $L.s.p$
  of a clause $C\begin{full}\constraint{S}\end{full}$ is \emph{eligible}  \wrt\ a substitution $\sigma$
  if the literal $L$ is either
  negative and
  eligible
  or positive and strictly eligible (\wrt\ $\sigma$ in $C\begin{full}\constraint{S}\end{full}$); and
  $L$ is of the form $s \ceqneq t \in C$ such that
  \begin{full}$(s\constraint{S})\sigma \not\preceq (t\constraint{S})\sigma$\end{full}%
  \begin{slim}$s\sigma \not\preceq t\sigma$\end{slim}.
  \begin{slim}
  
  When we do not specify a substitution, we mean eligibility \wrt\ the identity substitution.
  \end{slim}
\end{defi}

\subsection{Concrete Term Orders}
\label{ssec:concrete-term-orders}

A companion article \cite{bentkamp-et-al-optimistic-orders} defines two concrete term orders
fulfilling the criteria of Definition~\ref{def:admissible-term-order}:
$\lambda$KBO, inspired by the Knuth--Bendix order, and $\lambda$LPO,
inspired by the lexicographic path order.
Since the companion article defines the orders only on terms,
we extend $\succ_\lkb$ and $\succ_\llp$ to literals and clauses
via the standard extension using multisets \cite[\Section~2.4]{bachmair-ganzinger-1994}.
\begin{full}We extend the orders to constrained terms, constrained literals, and constrained clauses by ignoring the constraints.\end{full}

\begin{thm}
  Let $\succ_\lkb$ denote the strict variant of
  $\lambda$KBO as defined in the companion article.
  The order is parameterized by a precedence relation $>$ on symbols,
  a function $\cal w$ assigning weights to symbols,
  a constant $\cal w_\db$ defining the weight of De Bruijn indices,
  and a function $\cal k$ assigning argument coefficients to symbols.
  Assume that these parameters fulfill $\cal w(\itrue) = \cal w(\ifalse) = 1$,
   $\cal w_\db \ge \cal w(\cst{diff})$,
  $\cst{f} > \ifalse > \itrue$ for all symbols $\cst{f} \notin\{\itrue, \ifalse\}$,
  and $\cal k(\cst{diff}, i) = 1$ for every $i$.
  Using the extension defined above,
  $\succ_\lkb$ is an admissible term order.
\end{thm}
\begin{proof}
For most of the criteria, we use that
by 
Theorems~\ref{new-orders:thm:lkbm-coincide-ground}
and~\ref{new-orders:thm:lkb-coincide-monomorphic} of the companion article,
$\succ_\lkbg$ is the restriction of
$\succ_\lkb$ to ground terms.
\begin{enumerate}
\item[\ref{cond:order:total}]
By 
Theorems~\ref{new-orders:thm:lkbg-strict-partial-order} and \ref{new-orders:thm:lkbg-total}  of the companion article,
$\succ_\lkbg$ is a
total order.
By
Theorem~\ref{new-orders:thm:lkbg-well-founded} of the companion article,
it is well founded. 
\item[\ref{cond:order:comp-with-contexts}]
By
Theorem~\ref{new-orders:thm:lkbg-compat-orange-contexts} of the companion article,
$\succ_\lkbg$ is compatible with orange contexts and thus also with yellow contexts.
\item[\ref{cond:order:subterm}] By Theorem~\ref{new-orders:thm:lkbg-orange-subterm-property} of the companion article,
$\succ_\lkbg$ enjoys the orange subterm property and thus also the yellow subterm property.
\item[\ref{cond:order:t-f-minimal}] 
By Theorem~\ref{new-orders:thm:lkbg-top-bot-smallest} of the companion article,
$u \succ_\lkbg \ifalse \succ_\lkbg \itrue$ for all ground terms $u \notin \{\itrue, \ifalse\}$,
using our assumptions about the weight and precedence of $\itrue$ and $\ifalse$.
\item[\ref{cond:order:ext}] By Theorem~\ref{new-orders:thm:lkbg-diff} of the companion article,
$u \succ_\lkbg u\>\cst{diff}\typeargs{\tau,\upsilon}(s,t)$ for all ground types $\tau, \upsilon$ and ground terms $s,t,u \oftype \tau \fun \upsilon$,
using our assumptions about the weight and argument coefficients of $\cst{diff}$.
\item[\ref{cond:order:clause-extension}] By definition of our extension of $\succ_\lkb$ to clauses.
\item[\ref{cond:order:stability-terms}] By Theorems~\ref{new-orders:thm:lkbm-grounding-subst-stable}
and~\ref{new-orders:thm:lkb-monomorphizing-subst-stable} of the companion article.
\begin{full}Since we ignore the constraints in the order, we also have stability under substitutions for constrained terms.\end{full}
\item[\ref{cond:order:stability-clauses}] Using the Dershowitz--Manna definition \cite{dershowitz-manna-1979} of a multiset, it is easy to see that
stability under substitutions for terms implies stability under substitutions for clauses.
\begin{full}Since we ignore the constraints in the order, we also have stability under substitutions for constrained clauses.\end{full}
\item[\ref{cond:order:transitive}] By Theorem~\ref{new-orders:thm:lkb-transitive} of the companion article,
$\succ_\lkb$ is transitive on terms. Since the multiset extension preserves transitivity,
it is also transitive on literals. \begin{full}Since we ignore the constraints in the order,
it is also transitive on constrained literals.\end{full}
\begin{full}\item[\ref{cond:order:variable}] By Theorem~\ref{new-orders:thm:lkb-variable-guarantee} of the companion article.\end{full}
\end{enumerate}
\end{proof}

\begin{thm}
Let $\succ_\llp$ denote the strict variant of
$\lambda$LPO as defined in the companion article.
The order is parameterized by a precedence relation $>$ on symbols
and a watershed symbol $\cst{ws}$.
Assume that $\cst{f} > \ifalse > \itrue$ for all symbols $\cst{f}\notin\{\itrue, \ifalse\}$,
that $\ifalse \leq \cst{ws}$, and that $\cst{diff} \leq \cst{ws}$.
Using the extension defined above,
$\succ_\llp$ is an admissible term order.
\end{thm}
\begin{proof}
For most of the criteria, we use that
by 
Theorems~\ref{new-orders:thm:llpm-coincide-ground}
and~\ref{new-orders:thm:llp-coincide-monomorphic} of the companion article,
$\succ_\llpg$ is the restriction of
$\succ_\llp$ to ground terms.
\begin{enumerate}
\item[\ref{cond:order:total}]
By 
Theorems~\ref{new-orders:thm:llpg-strict-partial-order} and \ref{new-orders:thm:llpg-total}  of the companion article,
$\succ_\llpg$ is a total order.
By
Theorem~\ref{new-orders:thm:llpg-well-founded} of the companion article,
it is well founded. 
\item[\ref{cond:order:comp-with-contexts}]
By
Theorem~\ref{new-orders:thm:llpg-compat-orange-contexts} of the companion article,
$\succ_\llpg$ is compatible with orange contexts and thus also with yellow contexts.
\item[\ref{cond:order:subterm}] By Theorem~\ref{new-orders:thm:llpg-orange-subterm-property} of the companion article,
$\succ_\llpg$ enjoys the orange subterm property and thus also the yellow subterm property.
\item[\ref{cond:order:t-f-minimal}] 
By Theorem~\ref{new-orders:thm:llpg-top-bot-smallest} of the companion article,
$u \succ_\llpg \ifalse \succ_\llpg \itrue$ for all ground terms $u \notin \{\itrue, \ifalse\}$,
using our assumptions about the precedence of $\itrue$ and $\ifalse$.
\item[\ref{cond:order:ext}] By Theorem~\ref{new-orders:thm:llpg-diff} of the companion article,
$u \succ_\llpg u\>\cst{diff}\typeargs{\tau,\upsilon}(s,t)$ for all ground types $\tau, \upsilon$ and ground terms $s,t,u \oftype \tau \fun \upsilon$,
using our assumption about the precedence of $\cst{diff}$.
\item[\ref{cond:order:clause-extension}] By definition of our extension of $\succ_\llp$ to clauses.
\item[\ref{cond:order:stability-terms}] By Theorems~\ref{new-orders:thm:llpm-grounding-subst-stable}
and~\ref{new-orders:thm:llp-nsllp-monomorphizing-subst-stable} of the companion article.
\begin{full}Since we ignore the constraints in the order, we also have stability under substitutions for constrained terms.\end{full}
\item[\ref{cond:order:stability-clauses}] Using the Dershowitz--Manna definition \cite{dershowitz-manna-1979} of a multiset, it is easy to see that
stability under substitutions for terms implies stability under substitutions for clauses.
\begin{full}Since we ignore the constraints in the order, we also have stability under substitutions for constrained clauses.\end{full}
\item[\ref{cond:order:transitive}] By Theorem~\ref{new-orders:thm:llp-transitive} of the companion article,
$\succ_\llp$ is transitive on terms. Since the multiset extension preserves transitivity,
it is also transitive on literals. \begin{full}Since we ignore the constraints in the order,
it is also transitive on constrained literals.\end{full}
\begin{full}\item[\ref{cond:order:variable}] By Theorem~\ref{new-orders:thm:llp-variable-guarantee} of the companion article.\end{full}
\end{enumerate}
\end{proof}

\subsection{The Core Inference Rules}
\label{ssec:the-core-inference-rules}

\begin{full}
The optimistic $\lambda$-superposition calculus consists of the following core
inference rules, which a priori must be performed to guarantee refutational completeness.
\end{full}
The calculus is parameterized by an admissible term order $\succ$
and a selection function $\mathit{hsel}$.
We denote this calculus as $\HInf^{\succ,\mathit{hsel}}$ or just $\HInf$.

Each of our inference rules describes a collection of inferences, which we formally define as follows:
\begin{defi}\label{def:inference}
An \emph{inference} $\iota$ is a tuple $(C_1, C_2, \ldots, C_{n+1})$ of \begin{full}constrained\end{full} clauses, written
\[\inference{C_1\quad C_2\quad \cdots\quad C_n}{C_{n+1}}\]
The \begin{full}constrained\end{full} clauses $C_1, C_2, \ldots, C_n$ are called \emph{premises},
denoted by $\prems(\iota)$, and $C_{n+1}$ is called \emph{conclusion}, denoted by $\concl(\iota)$.
The clause $C_n$ is called the \emph{main premise} of $\iota$, denoted by $\mprem(\iota)$.
We assume that the premisses of an inference do not have any variables in common,
which can be achieved by renaming them apart when necessary.
\end{defi}

Our variant of the superposition rule, originating from the standard superposition calculus,
is stated as follows:
\[\namedinference{\Sup}
{\overbrace{D' \llor { t \ceq t'}}^{\vphantom{\cdot}\smash{D}} \slimfull{}{\>\constraint{T}} \hypsep
 \greensubterm{C}{u}\slimfull{}{\>\constraint{S}}}
{(D' \llor \greensubterm{C}{t'})\sigma\slimfull{}{\>\constraint{U}}}\]
\begin{enumerate}[label=\arabic*.,ref=\arabic*]
  \item\label{sup:one} 
  \begin{full}$(\sigma, U) \in \csuupto(T,S,t\equiv u)$;\end{full}
  \begin{slim}$\sigma\in\csu(t \equiv u)$;\end{slim}
  \item\label{sup:two} 
  \begin{full}$u$ is not a variable;\end{full}
  \begin{slim}$u$ is not a variable, unless there exists another occurrence of that variable inside of a parameter in $C$;\end{slim}
  \item\label{sup:three} $u\sigma$ is nonfunctional;
  \item\label{sup:four}
  \begin{full}$(t\constraint{T})\sigma \not\preceq (t'\constraint{T})\sigma$;\end{full}
  \begin{slim}$t\sigma \not\preceq t'\sigma$;\end{slim}
  \item\label{sup:five} the position of $u$ is eligible in $C\begin{full}\constraint{S}\end{full}$ \wrt\ $\sigma$;
  \item\label{sup:six} $t \ceq t'$ is strictly maximal in $D\begin{full}\constraint{T}\end{full}$ \wrt\ $\sigma$;
  \item\label{sup:seven} there are no selected literals in $D\begin{full}\constraint{T}\end{full}$.
\end{enumerate}

The rule \FluidSup{} simulates superposition below applied variables:
\[\namedinference{\FluidSup}
{\overbrace{D' \llor t \ceq t'}^{\vphantom{\cdot}\smash{D}}\slimfull{}{\>\constraint{T}} \hypsep
\greensubterm{C}{u}\slimfull{}{\>\constraint{S}}}
{(D' \llor \greensubterm{C}{z\>t'}\slimfull{}{\>\constraint{T,S}}) \sigma}\]
with the following side conditions, in addition to \Sup's conditions
3~to~7:

\begin{enumerate}
  \item[1.] $\sigma\in\csu(z\>t\equiv u)$;
  \item[2.] 
  \begin{full}$u$ is not a variable but is variable-headed;\end{full}
  \begin{slim}$u$ is variable-headed, and if $u$ is a variable,
    then there exists another occurrence of that variable inside of a parameter in $C$;\end{slim}
  \item[8.] $z$ is a fresh variable;
  \item[9.] $(z\>t)\sigma \not= (z\>t')\sigma$;
  \item[10.] $z\sigma \not= \lambda\>\DB{0}$.
\end{enumerate}

The equality resolution rule \EqRes{} and the equality factoring rule \EqFact{} also originate from the standard superposition calculus:
\begin{align*}
 &\namedinference{\EqRes}
{\overbrace{C' \llor {u \cneq u'}}^{\vphantom{\cdot}\smash{C}}\slimfull{}{\>\constraint{S}}}
{C'\sigma\slimfull{}{\>\constraint{U}}}
&&\namedinference{\EqFact}
{\overbrace{C' \llor {u'} \ceq v' \llor {u} \ceq v}^{\vphantom{\cdot}\smash{C}}\slimfull{}{\>\constraint{S}}}
{(C' \llor v \cneq v' \llor u \ceq v')\sigma\slimfull{}{\>\constraint{U}}}
\end{align*}

\noindent
Side conditions for \EqRes{}:

\begin{enumerate}[label=\arabic*.,ref=eqres-\arabic*]
\item
\begin{full}$(\sigma,U)\in\csuupto(S,u\equiv u')$;\end{full}
\begin{slim}$\sigma\in\csu(u\equiv u')$;\end{slim}
\item $u \cneq u'$ is eligible in $C\slimfull{}{\>\constraint{S}}$ \wrt\ $\sigma$.
\end{enumerate}

\noindent
Side conditions for~\EqFact{}:

\begin{enumerate}[label=\arabic*.,ref=eqfact-\arabic*]
\item
\begin{full}$(\sigma, U)\in\csuupto(S,u\equiv u')$;\end{full}
\begin{slim}$\sigma\in\csu(u\equiv u')$;\end{slim}
\item $u \ceq v$ is eligible in $C\slimfull{}{\>\constraint{S}}$ \wrt\ $\sigma$;
\item there are no selected literals in $C\slimfull{}{\>\constraint{S}}$;
\item $\slimfull{u}{(u{\>\constraint{S}})}\sigma \not\preceq \slimfull{v}{(v{\>\constraint{S}})}\sigma$.
\end{enumerate}

The following rules \Clausify{}, \BoolHoist{}, \LoobHoist{}, and \FalseElim{}
are responsible for converting Boolean terms into clausal form.
The rules \BoolHoist{} and \LoobHoist{} each
come with an analogue, respectively called \FluidBoolHoist{} and \FluidLoobHoist{},
which simulates their application below applied variables.
\begin{align*}
  \namedinference{\Clausify}{C' \llor s \ceq t \slimfull{}{\>\constraint{S}}}
  {(C' \llor D \slimfull{}{\>\constraint{S}})\sigma}
\end{align*}
with the following side conditions:
\begin{enumerate}[label=\arabic*.,ref=clausify-\arabic*]
\item $\sigma\in\csu(s\equiv s', t\equiv t')$;%
\item $s \ceq t$ is strictly eligible in $C\slimfull{}{\>\constraint{S}}$ \wrt\ $\sigma$;%
\item $s$ is not a variable%
\begin{slim}, unless there exists an occurrence of that variable inside of a parameter in $C$\end{slim};%
\item the triple $(s', t', D)$ is one of the following,
where $\alpha$ is a fresh type variable and $x$ and $y$ are fresh term variables:
\begin{align*}
  &(x \iand y,\ \itrue,\ x \ceq \itrue)&
  &(x \iand y,\ \itrue,\ y \ceq \itrue)&
  &(x \iand y,\ \ifalse,\ x \ceq \ifalse \llor y \ceq \ifalse)\\
  &(x \ior y,\ \itrue,\ x \ceq \itrue \llor y \ceq \itrue)&
  &(x \ior y,\ \ifalse,\ x \ceq \ifalse)&
  &(x \ior y,\ \ifalse,\ y \ceq \ifalse)\\
  &(x \iimplies y,\ \itrue,\ x \ceq \ifalse \llor y \ceq \itrue)&
  &(x \iimplies y,\ \ifalse,\ x \ceq \itrue)&
  &(x \iimplies y,\ \ifalse,\ y \ceq \ifalse)\\
  &(x \mathrel{{\ieq}\typeargs{\alpha}} y,\ \itrue,\ x \ceq y)&
  &(x \mathrel{{\ieq}\typeargs{\alpha}} y,\ \ifalse,\ x \cneq y)\\
  &(x \mathrel{{\ineq}\typeargs{\alpha}} y,\ \itrue,\ x \cneq y)&
  &(x \mathrel{{\ineq}\typeargs{\alpha}} y,\ \ifalse,\ x \ceq y)\\
  &(\inot x,\ \itrue,\ x \ceq \ifalse)&
  &(\inot x,\ \ifalse,\ x \ceq \itrue)
\end{align*}
\end{enumerate}

\begin{align*}
 &\namedinference{\BoolHoist}{\greensubterm{C}{u} \slimfull{}{\>\constraint{S}}}
    {(\greensubterm{C}{\ifalse} \llor u \ceq \itrue\slimfull{}{\>\constraint{S}})\sigma}
&&\namedinference{\LoobHoist}{\greensubterm{C}{u} \slimfull{}{\>\constraint{S}}}
    {(\greensubterm{C}{\itrue} \llor u \ceq \ifalse\slimfull{}{\>\constraint{S}})\sigma}
\end{align*}
each with the following side conditions:
\begin{enumerate}[label=\arabic*.,ref=\arabic*]
  \item $\sigma$ is the most general type substitution such that $u\sigma$ is of Boolean type
  (i.e., the identity if $u$ is of Boolean type or $\{\alpha \mapsto \omicron\}$
  if $u$ is
  of type $\alpha$ for some type variable~$\alpha$);%
  \item 
  \begin{full}$u$ is not a variable and is neither $\itrue$ nor $\ifalse$;\end{full}
  \begin{slim}$u$ is neither $\itrue$ nor $\ifalse$, and if $u$ is a variable,
     there exists another occurrence of that variable inside of a parameter in $C$;\end{slim}
  \item the position of $u$ is eligible in $C\slimfull{}{\>\constraint{S}}$ \wrt\ $\sigma$;%
  \item the occurrence of $u$ is not in a literal of the form $u \ceq \ifalse$ or $u \ceq \itrue$.
\end{enumerate}

\begin{align*}
  \namedinference{\FluidBoolHoist}
  {\greensubterm{C}{u}\slimfull{}{\>\constraint{S}}}
  {(\greensubterm{C}{z\>\ifalse} \llor x \ceq \itrue\slimfull{}{\>\constraint{S}}) \sigma}
  \end{align*}
  \begin{enumerate}[label=\arabic*.,ref=\arabic*]
    \item 
    \begin{full}$u$ is not a variable but is variable-headed;\end{full}
    \begin{slim}$u$ is variable-headed, and if $u$ is a variable,
      there exists another occurrence of that variable inside of a parameter in $C$;\end{slim}
    \item $u\sigma$ is nonfunctional;
    \item $x$ is a fresh variable of Boolean type, and $z$ is a fresh variable of function type from Boolean to the type of $u$;
    \item $\sigma\in\csu(z\>x\equiv u)$;
    \item\label{fluidboolhoist:five} $(z\>\ifalse)\sigma \not= (z\>x)\sigma$;
    \item $z\sigma \not= \lambda\>\DB{0}$;
    \item $x\sigma \not= \itrue$ and $x\sigma \not= \ifalse$;
    \item the position of $u$ is eligible in $C\slimfull{}{\>\constraint{S}}$ \wrt\ $\sigma$.
  \end{enumerate}
  \begin{align*}
  \namedinference{\FluidLoobHoist}
  {\greensubterm{C}{u}\slimfull{}{\>\constraint{S}}}
  {(\greensubterm{C}{z\>\itrue} \llor x \ceq \ifalse\slimfull{}{\>\constraint{S}}) \sigma}
  \end{align*}
with the same side conditions as \FluidBoolHoist, but where
$\ifalse$ is replaced by $\itrue$ in condition \ref{fluidboolhoist:five}.

\begin{align*}
  \namedinference{\FalseElim}{\overbrace{C' \llor s \ceq t}^C \slimfull{}{\>\constraint{S}}}
  {C'\sigma\slimfull{}{\>\constraint{U}}}
\end{align*}
with the following side conditions:
\begin{enumerate}[label=\arabic*.,ref=\arabic*]
  \item
  \begin{full}$(\sigma, U)\in\csuupto(S,s\equiv \ifalse,t\equiv \itrue)$;\end{full}
  \begin{slim}$\sigma\in\csu(s\equiv \ifalse,t\equiv \itrue)$;\end{slim}
  \item $s\ceq t$ is strictly eligible in $C\slimfull{}{\>\constraint{S}}$ \wrt\ $\sigma$.
\end{enumerate}

The argument congruence rule \ArgCong{} and the extensionality rule \Ext{} convert functional terms
into nonfunctional terms.
The rule \Ext{} also comes with an analogue \FluidExt{}, which simulates its application below applied variables.
\[\namedinference{ArgCong}
{\overbrace{C' \llor s \eq s'}^C \slimfull{}{\>\constraint{S}}}
{C'\sigma \llor s\sigma\>x \eq s'\sigma\>x\slimfull{}{\>\constraint{S\sigma}}}\]
with the following side conditions:
\begin{enumerate}[label=\arabic*.,ref=\arabic*]
  \item $\sigma$ is the most general type substitution such that $s\sigma$ is functional
  (i.e., the identity if $s$ is functional or $\{\alpha \mapsto (\beta \fun \gamma)\}$
  for fresh $\beta$ and $\gamma$ if $s$ is
  of type $\alpha$ for some type variable~$\alpha$);%
  \item $s \eq s'$ is strictly eligible in $C\slimfull{}{\>\constraint{S}}$ \wrt\ $\sigma$;
  \item $x$ is a fresh variable.
\end{enumerate}

\begin{align*}
\namedinference{\Ext}
{\greensubterm{C}{u}\slimfull{}{\>\constraint{S}}}
{\greensubterm{C\sigma}{y} \llor u\sigma\>(\diff\typeargs{\tau,\upsilon}(u\sigma,y)) \cneq y\>(\diff\typeargs{\tau,\upsilon}(u\sigma,y))\slimfull{}{\>\constraint{S\sigma}}}
\end{align*}
with the following side conditions:
\begin{enumerate}[label=\arabic*.,ref=\arabic*]
  \item $\sigma$ is the most general type substitution such that $u\sigma$ is of type $\tau \to \upsilon$ for some $\tau$ and $\upsilon$;
  \item $y$ is a fresh variable of the same type as $u\sigma$;
  \item the position of $u$ is eligible in $C\slimfull{}{\>\constraint{S}}$ \wrt\ $\sigma$.
\end{enumerate}

\begin{align*}
  \namedinference{\FluidExt}
  {\greensubterm{C}{u}\slimfull{}{\>\constraint{S}}}
  {(\greensubterm{C}{z\>y} \llor x\>(\diff\typeargs{\alpha,\beta}(x,y))\> \cneq y\>(\diff\typeargs{\alpha,\beta}(x,y))\slimfull{}{\>\constraint{S}}) \sigma}
  \end{align*}
with the following side conditions:
\begin{enumerate}[label=\arabic*.,ref=\arabic*]
  \item\label{fluidext:var} 
  \begin{full}$u$ is not a variable but is variable-headed;\end{full}
  \begin{slim}$u$ is variable-headed;\end{slim}
  \item\label{fluidext:nonfunctional} $u\sigma$ is nonfunctional;
  \item\label{fluidext:fresh} $x$ and $y$ are fresh variables of type $\alpha \to \beta$,
   and $z$ is a fresh variable of function type from $\alpha \to \beta$ to the type of $u$;
  \item\label{fluidext:csu} $\sigma\in\csu(S,z\>x \equiv u)$;
  \item\label{fluidext:csu-not-id} $(z\>x)\sigma \not= (z\>y)\sigma$;
  \item\label{fluidext:csu-not-proj} $z\sigma \not= \lambda\>\DB{0}$;
  \item\label{fluidext:eligible} the position of $u$ is eligible in $C\slimfull{}{\>\constraint{S}}$ \wrt\ $\sigma$.
\end{enumerate}

Our calculus also includes the following axiom (i.e., nullary inference rule), which
establishes the interpretation of the extensionality Skolem constant $\diff$.
\begin{gather*}
  \namedinference{\Diff}{}
   {y\>(\diff\typeargs{\alpha,\beta}(y,z))\cneq z\>(\diff\typeargs{\alpha,\beta}(y,z)) \llor y\>x\ceq z\>x}
\end{gather*}

\subsection{Redundancy}
\label{ssec:redundancy}

Our calculus includes a redundancy criterion that can be used to delete certain
clauses and avoid certain inferences deemed redundant. The criterion is based
on a translation to ground monomorphic first-order logic.

Let $\Sigma$ be a higher-order signature.
We require $\Sigma$ to contain a symbol $\cst{diff}\oftypedecl\Pi\alpha,\beta.\>(\alpha\fun\beta,\allowbreak \alpha\fun\beta) \fofun \alpha$.
Based on this higher-order signature,
we construct a first-order signature $\mapF{\Sigma}$
as follows.
The type constructors are the same, but $\fun$ is an uninterpreted symbol in the first-order logic.
For each ground higher-order term of the form 
$\cst{f}\langle\bar\tau\rangle(\tuple u)\oftype\tau_1\fun\cdots\fun\tau_m\fun\tau$,
with $m\geq 0$, we introduce a first-order symbol $\cst{f}^{\tuple{\tau}}_{\tuple{u}}\oftype
\tau_1 \times \cdots \times \tau_m \fofun \tau$.%
\begin{full}
Moreover, we introduce a first-order symbol $\cst{fun}_{t} \oftype 
\tau_1\times\dots\times\tau_n\fofun (\tau \fun\upsilon)$
for each expression $t$ 
obtained by replacing each
outermost proper yellow subterm in a
higher-order term of type $\tau \fun \upsilon$
by a placeholder symbol $\square$,
where $\tau_1,\dots,\tau_n$ are the types of the replaced subterms in order
of occurrence.
\end{full}
\begin{slim}
Moreover, we introduce a first-order symbol $\cst{fun}_{t} \oftype \tau \to \upsilon$
for each higher-order term $t$ of type $\tau \fun \upsilon$.
\end{slim}

We define an encoding $\mapFonly$ from higher-order
ground terms
to first-order terms:
\begin{defi}\label{def:fol} 
  For ground terms $t$, we define $\mapFonly$ recursively as follows:
  If $t$ is functional,
  then%
  \begin{full}
  let $t'$ be the expression
  obtained by replacing each
  outermost proper yellow subterm in $t$
  by the placeholder symbol $\square$,
  and let $\mapF{t} = \cst{fun}_{t'}(\mapF{\tuple{s}_{n}})$, where $\tuple{s}_{n}$
  are the replaced subterms in order of occurrence.
  \end{full}%
  \begin{slim}
  let $\mapF{t} = \cst{fun}_{t}$.
  \end{slim}%
  Otherwise, $t$ is of the form $\cst{f}\langle\bar\tau\rangle(\tuple u)\> \tuple{t}_m$,
  and we define $\mapF{t}
  = \cst{f}^{\tuple{\tau}}_{\tuple{u}} (\mapF{\tuple{t}_1},\dots,\mapF{\tuple{t}_m})$.  
  
  For clauses, we apply $\mapFonly$ on each side of each literal individually.
\end{defi}

\begin{full}
\begin{exa}
$\mapF{\lambda\;(\cst{f}\;(\lambda\;\DB{1})\;(\lambda\;(\lambda\;\DB{0})))}
  =\cst{fun}_{\lambda\;(\cst{f}\;(\lambda\;\DB{1})\;\square)}(\cst{fun}_{\lambda\;\square}(\cst{fun}_{\lambda\;\DB{0}}))$.
\end{exa}

\begin{rem}
A simpler yet equivalent formulation of the redundancy criterion
can be obtained by defining $\mapF{t} = \cst{fun}_{t}$ for functional terms $t$,
without using the $\square$ symbol.
The completeness proof, however, would become more complicated.
\end{rem}
\end{full}

\begin{lem}\label{lem:fol-bijection}
  The map $\mapFonly$ is a bijection between higher-order ground terms and first-order ground terms.
\end{lem}
\begin{proof}
  We can see that $\mapF{s} = \mapF{t}$ implies $s = t$ for all ground $s$ and $t$
  by structural induction on $\mapF{s}$.
  Moreover, we can show that
  for each first-order ground term $t$, there exists an $s$ such that $\mapF{s} = t$
  by structural induction on $t$.
  Injectivity and surjectivity imply bijectivity.
\end{proof}

We consider two different semantics for our first-order logic:
$\modelsfol$ and $\modelsolam$.
The semantics $\modelsfol$ is the standard semantics of first-order logic.
The semantics $\modelsolam$ restricts $\modelsfol$ to interpretations $\III$
with the following properties:
\begin{itemize}
  \item Interpreted Booleans:
  The domain of the Boolean type has exactly two elements,
  $\interpret{\itrue}{\III}{}$ and $\interpret{\ifalse}{\III}{}$,
  and the symbols $\inot$, $\iand$, $\ior$, $\iimplies$, $\ieq^\tau$, $\ineq^\tau$ are interpreted as the corresponding logical operations.
  \item Extensionality \wrt\ $\diff$: For all ground $u,w \oftype \tau \fun \upsilon$,
  if $\III \modelsfol \mapF{u\>\diff\typeargs{\tau,\upsilon}(s,t)} \ceq \mapF{w\>\diff\typeargs{\tau,\upsilon}(s,t)}$ for all ground $s,t \oftype \tau \to \upsilon$,
  then $\III \modelsfol \mapF{u} \ceq \mapF{w}$.
  \item Argument congruence \wrt\ $\diff$: For all ground $u,w,s,t \oftype \tau \fun \upsilon$,
  if $\III \modelsfol \mapF{u} \ceq \mapF{w}$,
  then $\III \modelsfol \mapF{u\>\diff\typeargs{\tau,\upsilon}(s,t)} \ceq \mapF{w\>\diff\typeargs{\tau,\upsilon}(s,t)}$.
\end{itemize}

\begin{full}
As another building block of our redundancy criterion,
we introduce the notion of trust.
As a motivating example, consider the clauses
$\cst{b} \cneq \cst{a}$ and
$\cst{b} \ceq \cst{a}$, where $\cst{b} \succ \cst{a}$.
Clearly, the empty clause can be derived via a $\Sup$ and a $\EqRes$ inference.
If we replace the clause $\cst{b} \ceq \cst{a}$ with 
the logically equivalent clause $x \cneq \cst{a}\>\constraint{x \equiv \cst{b}}$, however,
an empty clause with satisfiable constraints cannot be derived because $\Sup$ does not apply at variables.
In this sense, the clause $\cst{b} \cneq \cst{a}$
is more powerful than $x \cneq \cst{a}\>\constraint{x \equiv \cst{b}}$.
Technically,
the reason for this is that
the calculus is only guaranteed to derive contradictions entailed by
so-called variable-irreducible instances
of clauses,
and the instance of $x \cneq \cst{a}\>\constraint{x \equiv \cst{b}}$ that maps
$x$ to $\cst{b}$ is not variable-irreducible.
Since variable-irreducibility cannot be computed in general, 
when we replace a clause with another,
we use the notion of trust to ensure that
for every variable-irreducible instance of the replaced clause,
there exists a corresponding variable-irreducible instance of the replacing clause.
Concretely,
for any variable that occurs in a constraint or in a parameter of the replacing clause,
there must exist a variable in a similar context in the replaced clause.
The formal definition is as follows.

\begin{defi}[Trust]\label{def:H:trust}
Let $C \theta$ be a ground instance of $C\constraint{S}\in \clausesH$
and $D\rho$ be a ground instance of $D \constraint{T} \in \clausesH$.
We say that the $\theta$-instance of $C\constraint{S}$
\emph{trusts} the $\rho$-instance of $D \constraint{T}$
if for each variable $x$ in $D$,
\begin{enumerate}[label=(\roman*)]
  \item\label{cond:H:trust:corresponding-var} for every literal $L \in D$
  containing $x$ outside of parameters,
  there exists a literal $K \in C$ and a substitution $\sigma$ such that
  $z\theta = z\sigma\rho$ for all variables $z$ in $C$
  and $L \preceq K\sigma$; or
  \item \label{cond:H:trust:unconstrained}  $x$ neither occurs in parameters in $D$ nor appears in $T$.
\end{enumerate}
\end{defi}

The most general form of redundancy criteria for constrained superposition calculi
are notoriously difficult to apply to concrete simplification rules.
In the spirit of Nieuwenhuis and Rubio \cite{nieuwenhuis-rubio-1992}, we therefore
introduce a simpler, less general notion of redundancy that suffices for most simplification rules.
We provide a simple criterion for clauses and one for inference rules.
\end{full}
\begin{full}\subsubsection{Simple Clause Redundancy}\mbox{}\end{full}%
\begin{slim}\subsubsection{Clause Redundancy}\mbox{}\end{slim}
\label{ssec:simple-clause-redundancy}
Our redundancy criterion for clauses provides two conditions that can make a clause redundant.
The first condition applies when the ground instances of a clause are entailed
by smaller ground instances of other clauses.
It generalizes the standard superposition redundancy criterion
to higher-order clauses\begin{full} with constraints\end{full}.
The second condition applies when there are other clauses with the same ground instances
and can be used to justify subsumption.
For this second condition, we fix a well-founded partial order
$\sqsupset$ on $\clausesH$,
which prevents infinite chains of clauses where each clause is made redundant by the next one.
\begin{slim}
For example, following Bentkamp et al.\cite[\Section~3.4]{bentkamp-et-al-2021-lamsup-journal},
a sensible choice is to define $C \sqsupset D$ 
if either $C$ is larger than $D$ in syntactic size (i.e., number of variables, constants, and De Bruijn indices),
or if $C$ and $D$ have the same syntactic size and $C$ contains fewer distinct variables than $D$.
\end{slim}

\begin{slim}
\begin{defi}\label{def:order-mapF}
Since $\mapFonly$ is bijective on ground terms by Lemma~\ref{lem:fol-bijection},
we can convert a term order $\succ$ on higher-order terms
into a relation $\succ_\mapFonly$ on ground first-order terms as follows.
For two ground first-order terms $s$ and $t$, let $s \mathrel{\succ_\mapFonly} t$ if $\mapFonly^{-1}(s) \succ \mapFonly^{-1}(t)$.
\end{defi}

\begin{defi}[Clause Redundancy]\label{def:H:RedC}
  Given a clause $C$ and a clause set $N$, let $C\in\HRedC(N)$
  if for each grounding substitution $\theta$ at least one of the following two conditions holds:
  \begin{enumerate}[label=\arabic*.,ref=\arabic*]
    \item \label{cond:H:RedC:entailment}
      $\{E \in \mapF{\mapG{N}} \mid E \prec_\mapFonly \mapF{C\theta}\}\modelsolam \mapF{C\theta}$; or
    \item \label{cond:H:RedC:subsumed} there exists a clause $D \in N$ and a grounding substitution $\rho$
    such that $C \sqsupset D$ and $D\rho = C\theta$.
  \end{enumerate}
\end{defi}
\end{slim}
\begin{full}
\begin{defi}[Simple Clause Redundancy]\label{def:H:simple-redundancy}
  Let $N \subseteq \clausesH$ and $C\constraint{S} \in \clausesH$.
  We call $C\constraint{S}$ \emph{simply redundant} \wrt{} $N$, written $C\constraint{S} \in \HRedC^\star(N)$,
  if for every $C\theta \in \gnd(C\constraint S)$
  at least one of the following two conditions holds:
\begin{enumerate}[label=\arabic*.,ref=\arabic*]
  \item\label{cond:red:entailed-by-smaller} There exist
  an indexing set $I$ and for each $i\in I$
  a ground instance $D_i \rho_i$ of a clause
  $D_i \constraint{T_i} \in N$,
  such that
  \begin{enumerate}
    \item \label{cond:red:entailment} $\mapF{\{D_i\rho_i\mid i\in I\}} \modelsolam \mapF{C\theta}$; %
    \item \label{cond:red:order} for all $i \in I$, $D_i \rho_i \prec C \theta$; and
    \item \label{cond:red:trust} for all $i \in I$, the $\theta$-instance of $C\constraint S$ trusts the $\rho_i$-instance
    of $D_i\constraint {T_i}$.
  \end{enumerate}
  \item\label{cond:red:subsumed} There exists a
  ground instance $D\rho$ of some
  $D\constraint {T}\in N$
  such that
  \begin{enumerate}
    \item \label{cond:red:subsumed:eq} $D\rho = C\theta$; %
    \item \label{cond:red:subsumed:sqsupset} $C \constraint S \sqsupset D\constraint {T}$; and
    \item \label{cond:red:subsumed:trust} the $\theta$-instance of $C\constraint S$ trusts the $\rho$-instance
    of $D\constraint {T}$.
  \end{enumerate}
\end{enumerate}
\end{defi}

\begin{rem}\label{rem:sqsupset}
Although the calculus is refutationally complete for any choice of $\sqsupset$,
we propose the following definition for $\sqsupset$.
Given a clause $C\constraint{S}$ with nonempty $S$ and a clause $D$ with no constraints,
we define $C\constraint{S}\sqsupset D$.
For two clauses $C$ and $D$ with no constraints, following Bentkamp et al.\cite[\Section~3.4]{bentkamp-et-al-2021-lamsup-journal}, we propose to define $C \sqsupset D$ 
if either $C$ is larger than $D$ in syntactic size (i.e., number of variables, constants, and De Bruijn indices),
or if $C$ and $D$ have the same syntactic size and $C$ contains fewer distinct variables than $D$.
\end{rem}
\end{full}

\begin{full}\subsubsection{Simple Inference Redundancy}\mbox{}\end{full}%
\begin{slim}\subsubsection{Inference Redundancy}\mbox{}\end{slim}
To define inference redundancy, 
we first define a calculus $\FInf$
on ground first-order logic with Booleans.
It is parameterized by a relation $\succ$ on
ground first-order terms\slimfull{ and a selection function on ground first-order clauses}{}.
\begin{full}
For simplicity, there is no selection function,
but our notion
of eligibility is adapted to overapproximate any possible selection function as follows.
A literal $L\in C$ is (\emph{strictly}) \emph{eligible} in $C$ if $L$ is
negative
or if $L$ is of the form $t \ceq \ifalse$
or if $L$ is (strictly) maximal in $C$.
A position
$L.s.p$
of a clause $C$ is \emph{eligible}
if the literal $L$ is
of the form $s \doteq t$ with $s \succ t$
and $L$ is either
negative and
eligible
or positive and strictly eligible.

We define green subterms on first-order terms as follows.
Every term is a green subterm of itself.
Every direct subterm of a nonfunctional green subterm is also a green subterm.
In keeping with our notation for higher-order terms, we write $\greensubterm{t}{u}$ for
a term $t$ containing a green subterm $u$.

\begin{gather*}
  \begin{aligned}
    &
    \namedinference{\FSup}
      {\overbrace{D' \llor t \ceq t'}^{D}
      \quad
      \greensubterm{C}{t}}
      {D' \llor \greensubterm{C}{t'}}
    &&\quad
    \namedinference{\FEqRes}
      {\overbrace{C' \llor u \cneq u}^{C}}
      {C'}
  \end{aligned}\\[\jot]
  \begin{aligned}
    &
    \namedinference{\FEqFact}
      {\overbrace{C' \llor u \ceq v' \llor u \ceq v}^{C}}
      {C' \llor v \cneq v' \llor u \ceq v'}
    \;\,
    ~~\quad
    \namedinference{\FClausify}
      {C' \llor s \ceq t }
      {C' \llor D}
  \end{aligned}\\[\jot]
  \begin{aligned}
    \namedinference{\FBoolHoist}
      {\greensubterm{C}{u}}
      {\greensubterm{C}{\ifalse} \llor u \ceq \itrue}
      \;\,
      ~~\quad
      \namedinference{\FLoobHoist}
        {\greensubterm{C}{u}}
        {\greensubterm{C}{\itrue} \llor u \ceq \ifalse}
  \end{aligned}\\[\jot]
  \begin{aligned}
    \namedinference{\FFalseElim}
      {\overbrace{(C' \llor \ifalse \ceq \itrue)}^{C}}
      {C'}
  \end{aligned}\\[\jot]
  \begin{aligned}
    \namedinference{\FArgCong}
      {\overbrace{(C' \llor \mapF{s} \ceq \mapF{s'})}^{C}}
      {C'\llor \mapF{s\>\diff\typeargs{\tau,\upsilon}(u,w)} \ceq \mapF{s'\>\diff\typeargs{\tau,\upsilon}(u,w)}}
      \;\,
  \end{aligned}\\[\jot]
  \begin{aligned}
      \namedinference{\FExt}
        {\greensubterm{C}{\mapF{u}}}
        {\greensubterm{C}{\mapF{w}}\llor \mapF{u\>\diff\typeargs{\tau,\upsilon}(u,w)} \cneq \mapF{w\>\diff\typeargs{\tau,\upsilon}(u,w)}}
  \end{aligned}\\[\jot]
  \begin{aligned}
      \namedinference{\FDiff}
        {}
        {\mapF{u\>\diff\typeargs{\tau,\upsilon}(u,w)} \cneq \mapF{w\>\diff\typeargs{\tau,\upsilon}(u,w)} \llor \mapF{u\>s} \ceq \mapF{w\>s}}
  \end{aligned}
  \end{gather*}

\noindent
Side conditions for $\FSup$:
\begin{enumerate}[label=\arabic*.,ref=\arabic*]
\item $t \succ t'$;
\item $D \prec C[t]$;
\item $t$ is nonfunctional;
\item the position of $t$ is eligible in $C[t]$;
\item $t \ceq t'$ is strictly eligible in $D$;
\item if $t'$ is Boolean, then $t' = \itrue$.
\end{enumerate}

\noindent
No side conditions for $\FEqRes$.
Side conditions for $\FEqFact$:
\begin{enumerate}[label=\arabic*.,ref=\arabic*]
\item $u\ceq v$ is maximal in $C$;
\item $u \succ v$.
\end{enumerate}
  
\noindent
Side conditions for $\FClausify$:
\begin{enumerate}[label=\arabic*.,ref=\arabic*]
\item $s \ceq t$ is strictly eligible in $C' \llor s \ceq t$;
\item the triple ($s, t, D)$ has one of the following forms, where $\tau$ is an arbitrary type and $u$, $v$ are arbitrary terms:%
    \begin{align*}
    &(u \iand v,\ \itrue,\ u \ceq \itrue)&
    &(u \iand v,\ \itrue,\ v \ceq \itrue)&
    &(u \iand v,\ \ifalse,\ u \ceq \ifalse \llor v \ceq \ifalse)\\
    &(u \ior v,\ \itrue,\ u \ceq \itrue \llor v \ceq \itrue)&
    &(u \ior v,\ \ifalse,\ u \ceq \ifalse)&
    &(u \ior v,\ \ifalse,\ v \ceq \ifalse)\\
    &(u \iimplies v,\ \itrue,\ u \ceq \ifalse \llor v \ceq \itrue)&
    &(u \iimplies v,\ \ifalse,\ u \ceq \itrue)&
    &(u \iimplies v,\ \ifalse,\ v \ceq \ifalse)\\
    &(u \ieq^\tau v,\ \itrue,\ u \ceq v)&
    &(u \ieq^\tau v,\ \ifalse,\ u \cneq v)\\
    &(u \ineq^\tau v,\ \itrue,\ u \cneq v)&
    &(u \ineq^\tau v,\ \ifalse,\ u \ceq v)\\
    &(\inot u,\ \itrue,\ u \ceq \ifalse)&
    &(\inot u,\ \ifalse,\ u \ceq \itrue)
    \end{align*}
\end{enumerate}

\noindent
Side conditions for $\FBoolHoist$ and $\FLoobHoist$:
\begin{enumerate}[label=\arabic*.,ref=\arabic*]
\item $u$ is of Boolean type;
\item $u \ne \ifalse$ and $u\ne\itrue$;
\item the position of $u$ is eligible in $C$;
\item the occurrence of $u$ is not in a literal of the form $u \ceq \ifalse$ or $u \ceq \itrue$.
\end{enumerate}

\noindent
Side condition for $\FFalseElim$:
\begin{enumerate}[label=\arabic*.,ref=\arabic*]
\item $\ifalse\ceq\itrue$ is strictly eligible in $C$.
\end{enumerate}

\noindent
Side conditions for $\FArgCong$:
\begin{enumerate}[label=\arabic*.,ref=\arabic*]
\item $s$ is of type $\tau \fun \upsilon$;
\item $u,w$ are ground terms of type $\tau \fun \upsilon$;
\item $\mapF{s}\ceq\mapF{s'}$ is eligible in $C$.
\end{enumerate}

\noindent
Side conditions for $\FExt$:
\begin{enumerate}[label=\arabic*.,ref=\arabic*]
\item the position of $\mapF{u}$ is eligible in $C$;
\item the type of $u$ is $\tau \fun \upsilon$;
\item $w$ is a ground term of type $\tau \fun \upsilon$;
\item $u \succ w$.
\end{enumerate}

\noindent
Side conditions for $\FDiff$:
\begin{enumerate}[label=\arabic*.,ref=\arabic*]
\item $\tau$ and $\upsilon$ are ground types;
\item $u$, $w$, and $s$ are ground terms.
\end{enumerate}
\end{full}%
\begin{slim}
\begin{gather*}
  \begin{aligned}
    &
    \namedinference{\FSup}
      {\overbrace{D' \llor t \ceq t'}^{D}
      \quad
      \subterm{C}{t}}
      {D' \llor \subterm{C}{t'}}\quad
    &&\quad
    \namedinference{\FEqRes}
      {\overbrace{C' \llor u \cneq u}^{C}}
      {C'}
  \end{aligned}\\[\jot]
  \begin{aligned}
    &
    \namedinference{\FEqFact}
      {\overbrace{C' \llor u \ceq v' \llor u \ceq v}^{C}}
      {C' \llor v \cneq v' \llor u \ceq v}
    \;\,
    ~~\quad
    \namedinference{\FClausify}
      {C' \llor s \ceq t}
      {C' \llor D}
  \end{aligned}\\[\jot]
  \begin{aligned}
    \namedinference{\FBoolHoist}
      {\subterm{C}{u}}
      {(\subterm{C}{\ifalse} \llor u \ceq \itrue)}
      \;\,
      ~~\quad
      \namedinference{\FLoobHoist}
        {\subterm{C}{u}}
        {\subterm{C}{\itrue} \llor u \ceq \ifalse}
  \end{aligned}\\[\jot]
  \begin{aligned}
    \namedinference{\FFalseElim}
      {\overbrace{C' \llor \ifalse \ceq \itrue}^{C}}
      {C'}
  \end{aligned}\\[\jot]
  \begin{aligned}
    \namedinference{\FArgCong}
      {\overbrace{C' \llor \mapF{s} \ceq \mapF{s'}}^{C}}
      {C'\llor \mapF{s\>\diff\typeargs{\tau,\upsilon}(u,w)} \ceq \mapF{s'\>\diff\typeargs{\tau,\upsilon}(u,w)}}
      \;\,
  \end{aligned}\\[\jot]
  \begin{aligned}
      \namedinference{\FExt}
        {\subterm{C}{\mapF{u}}}
        {\subterm{C}{\mapF{w}}\llor \mapF{u\>\diff\typeargs{\tau,\upsilon}(u,w)} \cneq \mapF{w\>\diff\typeargs{\tau,\upsilon}(u,w)}}
  \end{aligned}\\[\jot]
  \begin{aligned}
      \namedinference{\FDiff}
        {}
        {\mapF{u\>\diff\typeargs{\tau,\upsilon}(u,w)} \cneq \mapF{w\>\diff\typeargs{\tau,\upsilon}(u,w)} \llor \mapF{u\>s} \ceq \mapF{w\>s}}
  \end{aligned}
  \end{gather*}
  Side conditions for \FSup:
    \begin{enumerate}[label=\arabic*.,ref=\arabic*]
    \item $t$ is nonfunctional;
    \item $t \succ t'$;
    \item $D \prec C[t]$;
    \item the position of $t$ is eligible in $C$;
    \item $t \ceq t'$ is strictly eligible in $D$;
    \item if $t'$ is Boolean, then $t' = \itrue$.
    \end{enumerate}
  Side conditions for \FEqRes:
    \begin{enumerate}[label=\arabic*.,ref=\arabic*]
    \item $u \cneq u$ is eligible in $C$.
    \end{enumerate}
  Side conditions for \FEqFact:
    \begin{enumerate}[label=\arabic*.,ref=\arabic*]
    \item $u\ceq v$ is maximal in $C$;
    \item there are no selected literals in $C$;
    \item $u \succ v$,
    \end{enumerate}
  Side conditions for \FClausify:
    \begin{enumerate}[label=\arabic*.,ref=\arabic*]
    \item $s \ceq t$ is strictly eligible in $C' \llor s \ceq t$;
    \item The triple ($s, t, D)$ has one of the following forms, where $\tau$ is an arbitrary type and $u$, $v$ are arbitrary terms:%
    \begin{align*}
    &(u \iand v,\ \itrue,\ u \ceq \itrue)&
    &(u \iand v,\ \itrue,\ v \ceq \itrue)&
    &(u \iand v,\ \ifalse,\ u \ceq \ifalse \llor v \ceq \ifalse)\\
    &(u \ior v,\ \itrue,\ u \ceq \itrue \llor v \ceq \itrue)&
    &(u \ior v,\ \ifalse,\ u \ceq \ifalse)&
    &(u \ior v,\ \ifalse,\ v \ceq \ifalse)\\
    &(u \iimplies v,\ \itrue,\ u \ceq \ifalse \llor v \ceq \itrue)&
    &(u \iimplies v,\ \ifalse,\ u \ceq \itrue)&
    &(u \iimplies v,\ \ifalse,\ v \ceq \ifalse)\\
    &(u \ieq^\tau v,\ \itrue,\ u \ceq v)&
    &(u \ieq^\tau v,\ \ifalse,\ u \cneq v)\\
    &(u \ineq^\tau v,\ \itrue,\ u \cneq v)&
    &(u \ineq^\tau v,\ \ifalse,\ u \ceq v)\\
    &(\inot u,\ \itrue,\ u \ceq \ifalse)&
    &(\inot u,\ \ifalse,\ u \ceq \itrue)
    \end{align*}
    \end{enumerate}
  Side conditions for \FBoolHoist and \FLoobHoist:
    \begin{enumerate}[label=\arabic*.,ref=\arabic*]
    \item $u$ is of Boolean type
    \item $u$ is neither $\itrue$ nor $\ifalse$;
    \item the position of $u$ is eligible in $C$;
    \item the occurrence of $u$ is not in a literal $L$ with $L = u \ceq \ifalse$ or $L = u \ceq \itrue$.
    \end{enumerate}
  Side conditions for \FFalseElim:
    \begin{enumerate}[label=\arabic*.,ref=\arabic*]
    \item $\ifalse \ceq \itrue$ is strictly eligible in $C$.
    \end{enumerate}
  Side conditions for \FArgCong:
    \begin{enumerate}[label=\arabic*.,ref=\arabic*]
    \item $s$ is of type $\tau \fun \upsilon$;
    \item $u,w$ are ground terms of type $\tau \fun \upsilon$;
    \item $\mapF{s} \ceq \mapF{s'}$ is strictly eligible in $C$.
    \end{enumerate}
  Side conditions for \infname{FExt}:
    \begin{enumerate}[label=\arabic*.,ref=\arabic*]
    \item the position of $\mapF{u}$ is eligible in $C$;
    \item the type of $u$ is $\tau \fun \upsilon$;
    \item $w$ is a ground term of type $\tau \fun \upsilon$;
    \item $u \succ w$.
    \end{enumerate}
  Side conditions for \FDiff:
    \begin{enumerate}[label=\arabic*.,ref=\arabic*]
    \item $\tau$ and $\upsilon$ are ground types;
  \item $u,w,s$ are ground terms.
    \end{enumerate}
\end{slim}

\begin{full}
\begin{defi}
Since $\mapFonly$ is bijective on ground terms by Lemma~\ref{lem:fol-bijection},
we can convert a term order $\succ$ on higher-order terms
into a relation $\succ_\mapFonly$ on ground first-order terms as follows.
For two ground first-order terms $s$ and $t$, let $s \mathrel{\succ_\mapFonly} t$ if $\mapFonly^{-1}(s) \succ \mapFonly^{-1}(t)$.
\end{defi}
\end{full}

\begin{slim}
\begin{defi}\label{def:sel-mapF}
  We convert a selection function $\mathit{hsel}$ on higher-order clauses
  into a selection function $\mapF{\mathit{hsel}}$ on ground first-order clauses as follows:
  Let a literal $L$ of a first-order ground clause $C$ be selected
  if $\mapFonly^{-1}(L)$ is selected in $\mapFonly^{-1}(C)$.
\end{defi}
\end{slim}

\begin{defi}\label{def:fol-inferences}
  Let $\iota \in \HInf^{\succ,\mathit{hsel}}$
  for a term order $\succ$ and a selection function $\mathit{hsel}$.
  Let $C_1\slimfull{}{\constraint{S_1}}$, \dots, $C_m\slimfull{}{\constraint{S_m}}$ be its premises and
   $C_{m+1}\slimfull{}{\constraint{S_{m+1}}}$ its conclusion.
  Let $(\theta_1$, \dots, $\theta_{m+1})$ be a tuple of grounding substitutions.
  We say that $\iota$ is \emph{rooted} in $\FInf{}$ for $(\theta_1$, \dots, $\theta_{m+1})$
  if and only if
\begin{slim}
  \[\inference
  {\mapF{C_1\theta_1}\quad \cdots\quad \mapF{C_m\theta_m}}
  {\mapF{C_{m+1}\theta_{m+1}}}\]
  is a valid $\FInf^{\succ_\mapFonly, \mapF{\mathit{hsel}}}$ inference $\iota'$
  such that the rule names of $\iota$ and  $\iota'$ correspond up to the prefixes $\infname{F}$ and $\infname{Fluid}$.
\end{slim}
\begin{full}
\begin{itemize}
  \item$S_1\theta_1, \dots, S_{m+1}\theta_{m+1}$ are true
  and
  \item
  \[\inference
  {\mapF{C_1\theta_1}\quad \cdots\quad \mapF{C_m\theta_m}}
  {\mapF{C_{m+1}\theta_{m+1}}}\]
  is a valid $\FInf^{\succ_\mapFonly}$ inference $\iota'$
  such that the rule names of $\iota$ and  $\iota'$ correspond up to the prefixes $\infname{F}$ and $\infname{Fluid}$.
\end{itemize}
\end{full}
\end{defi}

\begin{full}
\begin{defi}[Simple Inference Redundancy]\label{def:H:simple-inference-redundancy}
  Let $N \subseteq \clausesH$.
  Let $\iota \in \HInf$
  an inference with premises $C_1\constraint{S_1}$, \dots, $C_m\constraint{S_m}$ and
   conclusion $C_{m+1}\constraint{S_{m+1}}$.
  We call $\iota$ \emph{simply redundant} \wrt{} $N$, written $\iota \in \HRedI^\star(N)$, 
  if for every tuple of substitutions $(\theta_1$, \dots, $\theta_{m+1})$
  for which $\iota$ is rooted in $\FInf$ (Definition~\ref{def:fol-inferences}),
  there exists an index set $I$ and 
  for each $i\in I$ a ground instance $D_i \rho_i$ of a clause
  $D_i \constraint{T_i} \in N$
  such that
  \begin{enumerate}[label=\arabic*.,ref=\arabic*]
    \item \label{cond:HRedI:entailment} $\mapF{\{D_i\rho_i\mid i\in I\}} \modelsolam \mapF{C_{m+1}\theta_{m+1}}$; %
    \item \label{cond:HRedI:order} $\iota$ is a $\Diff$ inference or for all $i \in I$, $D_i  \rho_i \prec C_m\theta_m$; and
    \item \label{cond:HRedI:trust} for all $i \in I$, the $\theta_{m+1}$-instance of $C_{m+1}\constraint{S_{m+1}}$ trusts the $\rho_i$-instance
    of $D_i\constraint {T_i}$.
  \end{enumerate}
\end{defi}
\end{full}

\begin{slim}
\begin{defi}[Inference Redundancy]\label{def:H:RedI}
    Let $N \subseteq \clausesH$.
    Let $\iota \in \HInf$
    an inference with premises $C_1$, \dots, $C_m$ and
     conclusion $C_{m+1}$.
    We define $\HRedI$ so that $\iota \in \HRedI(N)$
    if for all substitutions $(\theta_1, \dots, \theta_{m+1})$
    for which $\iota$ is rooted in $\FInf{}$ (Definition~\ref{def:fol-inferences}), we have
    \begin{itemize}
      \item $\iota$ is a $\Diff$ inference and $\mapF{\mapG{N}}\modelsolam \mapF{C_{m+1}\theta_{m+1}}$; or
      \item $\iota$ is some other inference and $\{E \in \mapF{\mapG{N}} \mid E \prec_{\mapFonly} \mapF{C_m\theta_m}\}\modelsolam \mapF{C_{m+1}\theta_{m+1}}$.
    \end{itemize}
\end{defi}
\end{slim}

\subsection{Simplification Rules}
\label{ssec:simplification-rules}

\subsubsection{Analogues of First-Order Simplification Rules}

Our notion of \slimfull{}{simple} clause redundancy (Definition~\slimfull{\ref{def:H:RedC}}{\ref{def:H:simple-redundancy}})
can justify most analogues of the
simplification rules implemented in Schulz's E prover \cite[Sections
2.3.1 and 2.3.2]{schulz-2002-brainiac}. Deletion of duplicated literals,
deletion of resolved literals, \begin{full}and\end{full} syntactic tautology deletion\begin{slim}, positive simplify-reflect, and negative simplify-reflect\end{slim}
adhere to our redundancy criterion\begin{full}, even when the involved clauses carry constraints\end{full}.
Semantic tautology deletion can be applied as
well, \begin{full}even on constrained clauses,\end{full} but we must use the entailment relation $\modelsolam$ under the encoding $\mapFonly$.
\begin{full}
Positive and negative simplify-reflect can be applied as well, even with constraints,
as long as the substitution makes each constraint of the unit clause true
or translates it into a constraint already present on the other clause.
\end{full}

Our analogue of clause subsumption is the following.
\begin{full}The subsumed clause can have constraints, but the subsuming clause cannot.\end{full}
\[\namedsimp{\infname{Subsumption}}
{C \quad C\sigma \llor D\slimfull{}{\>\constraint{S}}}
{C}\]
with the following side condition\slimfull{}{s}:
\begin{enumerate}[label=\arabic*.,ref=\arabic*]
  \item \label{cond:subsumption:sqsupset} $D \ne \bot$ or $C\sigma \slimfull{}{\constraint{S}} \sqsupset C$\slimfull{.}{;}
  \begin{full}
  \item \label{cond:subsumption:parameters} $C$ does not contain a variable occurring both inside and outside of parameters.
  \end{full}
\end{enumerate}

\begin{lem}
  \infname{Subsumption} can be justified by \slimfull{}{simple} clause redundancy.
\end{lem}
\begin{slim}
\begin{proof}
Let $\theta$ be a grounding substitution.
If $D$ is nonempty, we apply condition~\ref{cond:H:RedC:entailment} of Definition~\ref{def:H:RedC},
which holds because
the clause $C\sigma\theta$ is a proper subclause of 
$(C\sigma \llor D)\theta$
and therefore
$\mapF{C\sigma\theta} \modelsolam \mapF{(C\sigma \llor D)\theta}$
and $C\sigma\theta \prec (C\sigma \llor D)\theta$.
If $D = \bot$, we apply condition~\ref{cond:H:RedC:subsumed} of Definition~\ref{def:H:RedC},
which holds by condition \ref{cond:subsumption:sqsupset} of $\infname{Subsumption}$.
\end{proof}
\end{slim}
\begin{full}
\begin{proof}
  Let $(C\sigma \llor D)\theta\in \gnd(C\sigma \llor D\constraint{S})$.

  If $D$ is nonempty, we apply condition~\ref{cond:red:entailed-by-smaller}
  of Definition~\ref{def:H:simple-redundancy},
  using $I = \{*\}$, $D_* = C$ and $\rho_* = \sigma\theta$.
  The clause $C\sigma\theta$ is a proper subclause of 
  $(C\sigma \llor D)\theta$
  and therefore
  $\mapF{C\sigma\theta} \modelsolam \mapF{(C\sigma \llor D)\theta}$
  (condition~\ref{cond:red:entailment})
  and $C\sigma\theta \prec (C\sigma \llor D)\theta$
  (condition~\ref{cond:red:order}).
  For condition~\ref{cond:red:trust},
  let $x$ be a variable in $C$.
  By condition~\ref{cond:subsumption:parameters} of \infname{Subsumption},
  $x$ occurs only inside parameters or only outside parameters.
  If it occurs only inside, we apply condition~\ref{cond:H:trust:corresponding-var} 
  of Definition~\ref{def:H:trust};
  if it occurs only outside, we apply condition~\ref{cond:H:trust:unconstrained} 
  of Definition~\ref{def:H:trust}.

  If $D$ is $\bot$, we apply condition~\ref{cond:red:subsumed} of Definition~\ref{def:H:simple-redundancy},
  using $C$ for $D$ and $\sigma\theta$ for $\rho$.
  Condition~\ref{cond:red:subsumed:eq}
  clearly holds.
  Condition~\ref{cond:red:subsumed:sqsupset} holds by condition~\ref{cond:subsumption:sqsupset}
  of \infname{Subsumption}.
  Condition~\ref{cond:red:subsumed:trust} follows from condition~\ref{cond:subsumption:parameters}
  of \infname{Subsumption} as above.
\end{proof}
\end{full}

For rewriting of positive and negative literals (demodulation) and equality subsumption,
we need to establish the following properties of orange subterms first:

\begin{lem}\label{lem:orange-green-chain}
Let $\benf{}$ be a $\beta\eta$-normalizer.
An orange subterm relation $\orangesubterm{u}{s}_p$ \wrt{} $\benf{}$ can be disassembled into a
sequence $s_1 \dots s_k$ as follows:
$s_1$ is a green subterm of $u$;
$s_k = s$; and
for each $i<k$,
$s_i = \lambda\>s'_i$
and $s_{i+1}$ is a green subterm of $s'_i$.
\end{lem}
\begin{proof}
By induction on the size of $u$ in $\eta$-long form.

If each orange subterm at a proper prefix of $p$ is nonfunctional,
then $p$ is green, and we are done with $k = 1$ and $s_1 = s$.

Otherwise, let $p = q.r$ such that $q$ is the shortest prefix with nonempty $r$,
where the orange subterm $s_1$ at $q$ is functional.
Then $s_1$ is a green subterm of $u$ at $q$ because
there does not exist a shorter prefix with a functional orange subterm.
Moreover, since $s_1$ is functional, modulo $\eta$-conversion,
$s_1 = \lambda\>s'_1$ for some $s'_1$.
Since $r$ is nonempty and $s$ is the orange subterm of $s_1$ at $r$,
there exists $r'$ at most as long as $r$ such that $s$ is the orange subterm of $s'_1$ at $r'$.
Specifically, if $\benf{s_1}$ is a $\lambda$-abstraction,
we use $1.r' = r$ and otherwise $r' = r$.
By the induction hypothesis, since $s$ is an orange subterm of $s'_1$,
there exist $s_2, \dots, s_k$ with $s_k = s$
such that 
$s_i = \lambda\>s'_i$
and $s_{i+1}$ is a green subterm of $s'_i$
for each $i < k$.
\end{proof}

\begin{lem}\label{lem:orange-ext}
Let $\benf{}$ be a $\beta\eta$-normalizer.
Let $u$ be a ground term, and let $p$ be an orange position of $u$ \wrt\ $\benf{}$.
Let $v, v'$ be ground preterms such that
$\orangesubterm{u}{v}_p$ and $\orangesubterm{u}{v'}_p$ are terms.
Let $k$ be a number large enough such that
$\mapdb{v}{t}{k}$ and $\mapdb{v'}{t}{k}$
do not contain free De Bruijn indices for all tuples of terms $\tuple{t}_k$.
Then
\begin{gather*}
\{\mapF{\mapdb{v}{t}{k} \ceq \mapdb{v'}{t}{k}}
\mid \text{\upshape each } t_i \text{\upshape\ of the form } \diff\typeargs{\_,\_}(\_,\_)\}
\\\modelsolam 
\mapF{\orangesubterm{u}{v}_p \ceq \orangesubterm{u}{v'}_p}
\end{gather*}
\end{lem}
\begin{proof}
Let $\III$ be a $\modelsolam$-interpretation with 
\[\III\modelsolam \mapF{\mapdb{v}{t}{k}\ceq \mapdb{v'}{t}{k}}\]
for all tuples of terms $\tuple{t}_k$, where each $t_i$ is of the form $\diff\typeargs{\_,\_}(\_,\_)$
for arbitrary values of `$\_$'.
By Lemma~\ref{lem:orange-green-chain},
we have $\orangesubterm{u}{v}_p
= \greensubterm{u}{\lambda\>\greensubterm{w_1}{\lambda\>\greensubterm{w_2}{\cdots\greensubterm{w_n}{v}\cdots}}}$.

\medskip\noindent
\textsc{Step 1.}\enskip
Since $v$ is a green subterm of $\greensubterm{w_n}{v}$
and the terms $\tuple{t}_k$ have a form that does not trigger $\beta$-reductions
when substituting them for De Bruijn indices,
$\mapdb{v}{t}{k}$ is a green subterm of $\mapdb{\greensubterm{w_n}{v}}{t}{k}$
and thus
\[\III\modelsolam  \mapF{\mapdb{\greensubterm{w_n}{v}}{t}{k}\ceq \mapdb{\greensubterm{w_n}{v'}}{t}{k}}\]

\medskip\noindent
\textsc{Step 2.}\enskip
Using the property of extensionality \wrt\ $\diff$ of $\modelsolam$-interpretations
and using the fact that we have shown the above for all $t_1$ of the form $\diff\typeargs{\_,\_}(\_,\_)$,
we obtain
\[\III \modelsolam \mapF{\mapdbl{(\lambda\>\greensubterm{w_n}{v})}{t}{2}{k}{(k-2)}\ceq \mapdbl{(\lambda\>\greensubterm{w_n}{v'})}{t}{2}{k}{(k-2)}}\]
Iterating steps 1 and 2 over $w_n, \dots, w_1, u$, we obtain
\[\III \modelsolam \mapF{\orangesubterm{u}{v}_p \ceq \orangesubterm{u}{v'}_p}\qedhere\]
\end{proof}

Our variant of rewriting of positive and negative literals (demodulation)
is the following.
\begin{full}The rewritten clause can have constraints, but the rewriting clause cannot.\end{full}
\[
\namedsimp{\infname{Demod}}
{t\eq t' \quad \orangesubterm{C}{v} \slimfull{}{\constraint{S}}}
{t\eq t' \quad \orangesubterm{C}{v'} \slimfull{}{\constraint{S}}}
\]
with the following side conditions:
\begin{enumerate}[label=\arabic*.,ref=\arabic*]
  \item\label{cond:demod:mapdb} $t\sigma = \mapdb{v}{x}{k}$ and $t'\sigma = \mapdb{v'}{x}{k}$ for some fresh variables $\tuple{x}_k$ and a substitution $\sigma$.
  \item\label{cond:demod:term-order} $\orangesubterm{C}{v}\slimfull{}{\constraint{S}} \succ \orangesubterm{C}{v'}\slimfull{}{\constraint{S}}$;
  \item\label{cond:demod:clause-order} for each tuple $\tuple{t}_k$, where each $t_i$ is of the form $\diff\typeargs{\_,\_}(\_,\_)$,
  we have $\orangesubterm{C}{v} \succ \mapdb{v}{t}{k} \ceq \mapdb{v'}{t}{k}$;
  \begin{full}\item\label{cond:demod:params} $t \eq t'$ does not contain a variable occurring both inside and outside of parameters.\end{full}
\end{enumerate}
\begin{rem}\label{rem:demod:clause-order}
In general, it is unclear how to compute condition~\ref{cond:demod:clause-order} of \infname{Demod}.
For $\lambda$KBO and $\lambda$LPO described in Section~\ref{ssec:concrete-term-orders},
however, the condition can easily be overapproximated by
$\orangesubterm{C}{v} \succ v \ceq v'$,
using the fact that the orders are also defined on preterms.

To prove that this is a valid overapproximation, it suffices to show the following:
Let $u$ and $s$ be preterms with $u \succ s$ (resp. $u \succsim s$). Let $s'$ be the result of replacing some De Bruijn indices in $s$ by
terms of the form $\diff\typeargs{\_,\_}(\_,\_)$. Then  $u \succ s'$ (resp. $u \succsim s'$).

\medskip

\noindent
\textsc{Proof for $\lambda$KBO:}\enskip
By induction on the rule deriving $u \succ s$ or $u \succsim s$.
Since we assume in Section~\ref{ssec:concrete-term-orders} that
$\cal w_\db \ge \cal w(\cst{diff})$ and $\cal k(\cst{diff}, i) = 1$ for every $i$,
we have $\cal W(s) \geq \cal W(s')$.
It is easy to check that there is always a corresponding rule deriving $u \succ s'$ or $u \succsim s'$,
in some cases using the induction hypothesis.

\medskip

\noindent
\textsc{Proof for $\lambda$LPO:}\enskip
By induction on the rule deriving $u \succ s$ or $u \succsim s$.
Considering that we assume in Section~\ref{ssec:concrete-term-orders} that
$\cst{ws} \geq \cst{diff}$,
it is easy to check that there is always a corresponding rule deriving $u \succ s'$ or $u \succsim s'$,
in some cases using the induction hypothesis.
\end{rem}

Since \infname{Demod} makes use of orange subterms,
it depends on the choice of $\beta\eta$-normalizer.
Both $\downarrow_{\beta\eta\mathrm{long}}$ and $\downarrow_{\beta\eta\mathrm{short}}$
yield a valid simplification rule:

\begin{lem}\label{lem:demod}
  \infname{Demod} can be justified by \slimfull{}{simple} clause redundancy,
  regardless of the choice of $\beta\eta$-normalizer.
\end{lem}
\begin{slim}
\begin{proof}
Let $\theta$ be a grounding substitutiion.
We apply condition~\ref{cond:H:RedC:entailment} of Definition~\ref{def:H:RedC},
using $\orangesubterm{C}{v}$ for $C$.
Let $T = \{\tuple{t}_k \mid \text{each } t_i \text{ is a ground term of the form } \diff\typeargs{\_,\_}(\_,\_)\}$
and let
$\rho_{\tuple{t}_k} = \sigma\{\tuple{x}_k \mapsto \tuple{t}_k\}\theta$
for each $\tuple{t}_k \in T$.
By condition~\ref{cond:demod:mapdb} of $\infname{Demod}$,
\begin{align*}
(t \ceq t')\rho_{\tuple{t}_k} &= \mapdb{v}{t}{k}\theta \ceq \mapdb{v'}{t}{k}\theta\\
&= \mapdb{v\theta}{t}{k} \ceq \mapdb{v'\theta}{t}{k}
\end{align*}
for each tuple $\tuple{t}_k\in T$.

By Lemma~\ref{lem:orange-ext},
$\mapF{\{(t \ceq t')\rho_{\tuple{t}_k}\mid \tuple{t}_k\in T\}} \modelsolam \mapF{\orangesubterm{u\theta}{v\theta} \ceq \orangesubterm{u\theta}{v'\theta}}$,
where $u$ is a side of a literal in $\orangesubterm{C}{v}$ containing the orange subterm $v$.
Thus
\[\mapF{\{(t \ceq t')\rho_{\tuple{t}_k}\mid \tuple{t}_k\in T\}} \cup \{\mapF{\orangesubterm{C}{v'}\theta}\} \modelsolam \mapF{\orangesubterm{C}{v}\theta}\]
Condition \ref{cond:demod:term-order} of \infname{Demod} 
implies $\orangesubterm{C}{v'}\theta \prec \orangesubterm{C}{v}\theta$.
Condition~\ref{cond:demod:clause-order} of \infname{Demod} implies
$(t \ceq t') \rho_{\tuple{t}_k} \prec \orangesubterm{C}{v} \theta$ for all $\tuple{t}_k\in T$.
Thus condition~\ref{cond:H:RedC:entailment} of Definition~\ref{def:H:RedC} applies.
\end{proof}
\end{slim}%
\begin{full}
\begin{proof}
We apply condition~\ref{cond:red:entailed-by-smaller} of Definition~\ref{def:H:simple-redundancy},
using $\orangesubterm{C}{v}\constraint{S}$ for $C\constraint{S}$.
Let $\orangesubterm{C}{v} \theta \in \gnd(\orangesubterm{C}{v}\constraint S)$.
Let $*$ be a placeholder we use to extend a set of terms by an additional element.
Then we set
\begin{align*}
I &= \{\tuple{t}_k \mid \text{each } t_i \text{ is a ground term of the form } \diff\typeargs{\_,\_}(\_,\_)\} \cup \{*\} \\
D_{\tuple{t}_k}\constraint{T_{\tuple{t}_k}} &= t \eq t'\\
\rho_{\tuple{t}_k} &= \sigma\{\tuple{x}_k \mapsto \tuple{t}_k\}\theta\\
D_*\constraint{T_*} &= \orangesubterm{C}{v'}\constraint{S}\\
\rho_* &= \theta
\end{align*}
By condition~\ref{cond:demod:mapdb} of $\infname{Demod}$,
\begin{align*}
D_{\tuple{t}_k}\rho_{\tuple{t}_k} &= \mapdb{v}{t}{k}\theta \ceq \mapdb{v'}{t}{k}\theta\\
&= \mapdb{v\theta}{t}{k} \ceq \mapdb{v'\theta}{t}{k}
\end{align*}
for each tuple $\tuple{t}_k$, where each $t_i$ is of the form
$\diff\typeargs{\_,\_}(\_,\_)$.

By Lemma~\ref{lem:orange-ext},
$\mapF{\{D_i\rho_i\mid i\in I\setminus\{*\}\}} \modelsolam \mapF{\orangesubterm{u\theta}{v\theta} \ceq \orangesubterm{u\theta}{v'\theta}}$,
where $u$ is a side of a literal in $\orangesubterm{C}{v}$ containing the orange subterm $v$.
Thus
$\mapF{\{D_i\rho_i\mid i\in I\}} \modelsolam \mapF{\orangesubterm{C}{v}\theta}$
(condition~\ref{cond:red:entailment} of simple redundancy).
Condition \ref{cond:demod:term-order} and \ref{cond:demod:clause-order} of \infname{Demod} imply
$D_i \rho_i \prec \orangesubterm{C}{v} \theta$ for all $i\in I$
(condition~\ref{cond:red:order} of simple redundancy).

For condition \ref{cond:red:trust} of simple redundancy, 
consider first a variable in $D_{\tuple{t}_k} = t \eq t'$.
By condition~\ref{cond:demod:params} of \infname{Demod}, 
either condition \ref{cond:H:trust:unconstrained} or (if the variable occurs only inside parameters) \ref{cond:H:trust:corresponding-var} 
of trust is fulfilled.
Second, consider a variable $x$ in $\orangesubterm{C}{v'}$.
Then we apply condition \ref{cond:H:trust:corresponding-var} of trust.
For every literal $L \in \orangesubterm{C}{v'}$ that contains $x$ outside of parameters,
we use the corresponding literal $K \in \orangesubterm{C}{v}$
and the identity substitution for the $\sigma$ of condition \ref{cond:H:trust:corresponding-var}.
By condition~\ref{cond:demod:term-order},
$L \preceq K$.
\end{proof}
\end{full}

Our variant of equality subsumption
is the following:
\[
\namedsimp{\infname{EqualitySubsumption}}
{t\eq t' \quad C' \llor \orangesubterm{s}{v} \ceq \orangesubterm{s'}{v'} \slimfull{}{\constraint{S}}} 
{t\eq t'}
\]
with the following side conditions:
\begin{enumerate}[label=\arabic*.,ref=\arabic*]
  \item $t\sigma = \mapdb{v}{x}{k}$ and $t'\sigma = \mapdb{v'}{x}{k}$ for some fresh variables $\tuple{x}_k$ and a substitution $\sigma$;
  \item\label{cond:eqsubsump:clause-order} for each tuple $\tuple{t}$, where each $t_i$ is of the form $\diff\typeargs{\_,\_}(\_,\_)$,
  we have $\orangesubterm{C}{v} \succ \mapdb{v}{t}{k} \ceq \mapdb{v'}{t}{k}$;
  \begin{full}
  \item\label{cond:eqsubsump:params} $t \eq t'$ does not contain a variable occurring both inside and outside of parameters.
  \end{full}
\end{enumerate}
To compute condition \ref{cond:eqsubsump:clause-order}, we can exploit Remark~\ref{rem:demod:clause-order}.
\begin{lem}
  \infname{EqualitySubsumption} can be justified by simple clause redundancy,
  regardless of the choice of $\beta\eta$-normalizer.
\end{lem}
\begin{proof}
Analogous to Lemma~\ref{lem:demod}.
\end{proof}

\subsubsection{Additional Simplification Rules}

The core inference rules $\ArgCong$, $\Clausify$, $\FalseElim$, $\LoobHoist$, and $\BoolHoist$
described in Section~\ref{ssec:the-core-inference-rules}
can under certain conditions be applied as simplification rules.

\begin{lem}\label{lem:simp:arg-cong}
  $\ArgCong$ can be justified as a simplification rule by \slimfull{}{simple} clause redundancy
  when $\sigma$ is the identity.
  Moreover, it can even be applied when its eligibility condition does not hold.
\end{lem}
\begin{proof}
\begin{slim}
Let $\theta$ be a grounding substitution.
We apply condition~\ref{cond:H:RedC:entailment} of Definition~\ref{def:H:RedC}.
Let $\tau \fun \upsilon$ be the type of $s\theta$ and $s'\theta$.
By the extensionality property of $\modelsolam$,
we have 
\[\{\mapF{(C' \llor s\>x \eq s'\>x)\theta[x \mapsto \diff\typeargs{\tau,\upsilon}(u,w)]} \mid u,w \oftype \tau \fun \upsilon \text{ ground} \} \modelsolam \mapF{(C' \llor s \eq s')\theta}\]
By \ref{cond:order:ext},
we have $(C' \llor s'\>x \eq s\>x)\theta[x \mapsto \diff\typeargs{\tau,\upsilon}(u,w)] \prec (C' \llor s \eq s')\theta$ for all such $\tau$, $\upsilon$, $u$, and $w$.
Thus condition~\ref{cond:H:RedC:entailment} of Definition~\ref{def:H:RedC} applies.
\end{slim}%
\begin{full}
Let $C\theta$ be a ground instance of $C \constraint{S}$.
Let $\tau \fun \upsilon$ be the type of $s\theta$ and $s'\theta$.
We apply condition~\ref{cond:red:entailed-by-smaller} of Definition~\ref{def:H:simple-redundancy},
using $I = \{(u,w) \mid u,w \oftype \tau \fun \upsilon \text{ ground}\}$, $D_{(u,w)} = C' \llor s\>x \eq s'\>x$, $T_{(u,w)} = S$, and 
$\rho_{(u,w)} = \theta[x \mapsto \diff\typeargs{\tau,\upsilon}(u,w)]$.
Condition~\ref{cond:red:entailment} follows from the extensionality property of $\modelsolam$.
Condition~\ref{cond:red:order} follows from \ref{cond:order:ext}.

For condition~\ref{cond:red:trust},
first consider the fresh variable $x$.
Since $x$ is fresh
and parameters cannot contain free De Bruijn indices,
$x$ cannot occur in parameters in $C' \llor s\>x \eq s'\>x$, and thus
condition~\ref{cond:H:trust:unconstrained} of Definition~\ref{def:H:trust}
applies.

Now consider any other variable $y$ in $D_{(u,w)}$.
Such a variable must occur in $C$.
We apply condition~\ref{cond:H:trust:corresponding-var} of Definition~\ref{def:H:trust},
using the identity substitution for $\sigma$ and---if $y$ occurs in $s$ or $s'$---using
\ref{cond:order:ext}.
\end{full}
\end{proof}

\begin{lem}\label{lem:simp:clausify}
  $\Clausify$ can be justified as a simplification rule by \slimfull{}{simple} clause redundancy
  when $\sigma$ is the identity for all variables other than $x$ and $y$.
  Moreover, it can even be applied when its eligibility condition does not hold.
\end{lem}
\begin{proof}
\begin{full}By condition~\ref{cond:red:entailed-by-smaller} of Definition~\ref{def:H:simple-redundancy},\end{full}%
\begin{slim}By condition~\ref{cond:H:RedC:entailment} of Definition~\ref{def:H:RedC},\end{slim}
using the fact that $\modelsolam$ interprets Booleans.
\end{proof}

\begin{lem}\label{lem:simp:falseelim}
  $\FalseElim$ can be justified as a simplification rule by \slimfull{}{simple} clause redundancy
  when $\sigma$ is the identity.
  Moreover, it can even be applied when its eligibility condition does not hold.
\end{lem}
\begin{proof}
\begin{full}By condition~\ref{cond:red:entailed-by-smaller} of Definition~\ref{def:H:simple-redundancy},\end{full}%
\begin{slim}By condition~\ref{cond:H:RedC:entailment} of Definition~\ref{def:H:RedC},\end{slim}
  using the fact that $\modelsolam$ interprets Booleans.
\end{proof}

\begin{lem}\label{lem:simp:boolhoist}
  $\BoolHoist$ and $\LoobHoist$ can be justified to be applied together as a simplification rule by \slimfull{}{simple} clause redundancy
  when $\sigma$ is the identity.
  Moreover, they can even be applied when their eligibility condition does not hold.
\end{lem}
\begin{proof}
\begin{full}By condition~\ref{cond:red:entailed-by-smaller} of Definition~\ref{def:H:simple-redundancy},\end{full}%
\begin{slim}By condition~\ref{cond:H:RedC:entailment} of Definition~\ref{def:H:RedC},\end{slim}
  using the fact that $\modelsolam$ interprets Booleans.
\end{proof}

The following two rules normalize negative literals with $\itrue$ and $\ifalse$ into positive literals.
\begin{align*}
 &\namedsimp{\infname{NotTrue}}
{C' \llor s \cneq \itrue\slimfull{}{\ \constraint{S}}}
{C' \llor s \ceq \ifalse\slimfull{}{\ \constraint{S}}}
&&\namedsimp{\infname{NotFalse}}
{C' \llor s \cneq \ifalse\slimfull{}{\ \constraint{S}}}
{C' \llor s \ceq \itrue\slimfull{}{\ \constraint{S}}}
\end{align*}

\begin{lem}\label{lem:simp:nottrue-notfalse}
  $\infname{NotTrue}$ and $\infname{NotFalse}$ can be justified as simplification rules by simple clause redundancy.
\end{lem}
\begin{proof}
\begin{full}By condition~\ref{cond:red:entailed-by-smaller} of Definition~\ref{def:H:simple-redundancy},\end{full}%
\begin{slim}By condition~\ref{cond:H:RedC:entailment} of Definition~\ref{def:H:RedC},\end{slim}
  using the fact that $\modelsolam$ interprets Booleans.
\end{proof}

\begin{full}
The following simplification rule, $\infname{Unif}$,
allows us to run a unification procedure to remove the constraints
of a clause.
\[\namedsimp{\infname{Unif}}
{C \constraint{S}}
{C\sigma_1 \quad \cdots \quad C\sigma_n}\]
with the following side conditions:
\begin{enumerate}[label=\arabic*.,ref=\arabic*]
  \item \label{cond:unif:csu} $\{\sigma_1, \dots, \sigma_n\}$ is a complete set of unifiers for $S$;
  \item \label{cond:unif:sqsupset} $C\constraint{S} \sqsupset C\sigma_i$ for all $i$.
\end{enumerate}

\begin{lem}
  \infname{Unif} can be justified by simple clause redundancy.
\end{lem}
\begin{proof}
  Let $C\theta\in \gnd(C \constraint{S})$.
  By condition~\ref{cond:unif:csu} of \infname{Unif}
  and Definition~\ref{def:csu},
  there must exist an index $i$ and a substitution $\rho$ such that
  $z\sigma_i\rho = z\theta$ for all $z$ in $C\constraint{S}$.
  We apply condition \ref{cond:red:subsumed} of Definition~\ref{def:H:simple-redundancy}.
  We use $C\sigma_i$ for $D$
  and $\rho$ for $\rho$.
  Condition~\ref{cond:red:subsumed:eq} follows from 
  the fact that $z\sigma_i\rho = z\theta$ for all $z$ in $C\constraint{S}$.
  Condition~\ref{cond:red:subsumed:sqsupset} follows from condition \ref{cond:unif:sqsupset}
  of \infname{Unif}.
  
  For condition~\ref{cond:red:subsumed:trust}, we must show that
  the $\theta$-instance of $C\constraint{S}$ trusts the
  $\rho$-instance of $C\sigma_i$.
  We will use condition~\ref{cond:H:trust:corresponding-var} of trust.
  Let $L \in C\sigma_i$.
  Let $K$ be a literal in $C$ such that $K\sigma_i = L$.
  We use $\sigma_i$ for $\sigma$. Then we have
  $z\theta = z\sigma_i\rho = z\sigma\rho$ for all variables $z$ in $C$.
  Moreover, $K\sigma = K\sigma_i = L$ implies $L \preceq K\sigma$.
\end{proof}
\end{full}

The following rule is inspired by one of Leo-II's extensionality rules \cite{benzmuller-2015-leo2}:
\[\namedsimp{\infname{NegExt}}
{\overbrace{C' \llor s \cneq s'}^C\slimfull{}{\>\constraint{S}}}
{C' \llor s\>\diff\typeargs{\tau,\upsilon}(s,s')\>\cneq s'\>\diff\typeargs{\tau,\upsilon}(s,s')\slimfull{}{\>\constraint{S}}}\]

\begin{lem}
  \infname{NegExt} can be justified by simple clause redundancy.
\end{lem}
\begin{proof}
\begin{slim}
Let $\theta$ be a grounding substitution.
We apply condition~\ref{cond:H:RedC:entailment} of Definition~\ref{def:H:RedC}.
By the argument congruence property of $\modelsolam$,
we have 
\[\mapF{(C' \llor s\>\diff\typeargs{\tau,\upsilon}(s,s')\>\cneq s'\>\diff\typeargs{\tau,\upsilon}(s,s'))\theta} \modelsolam \mapF{(C' \llor s \cneq s')\theta}\]
By \ref{cond:order:ext},
we have $(C' \llor s'\>\diff\typeargs{\tau,\upsilon}(s,s')\>\cneq s\>\diff\typeargs{\tau,\upsilon}(s,s'))\theta \prec (C' \llor s \cneq s')\theta$.
Thus condition~\ref{cond:H:RedC:entailment} of Definition~\ref{def:H:RedC} applies.
\end{slim}%
\begin{full}
Let $C\theta$ be a ground instance of $C \constraint{S}$.
We apply condition~\ref{cond:red:entailed-by-smaller} of Definition~\ref{def:H:simple-redundancy},
using $I = \{*\}$, $D_* = C' \llor s\>\diff\typeargs{\tau,\upsilon}(s,s')\>\cneq s'\>\diff\typeargs{\tau,\upsilon}(s,s')$, $T_* = S$, and 
$\rho_* = \theta$.
Condition~\ref{cond:red:entailment} follows from the argument congruence property of $\modelsolam$.
Condition~\ref{cond:red:order} follows from \ref{cond:order:ext}.
For condition~\ref{cond:red:trust},
consider a variable $y$ in $C$.
We apply condition~\ref{cond:H:trust:corresponding-var} of Definition~\ref{def:H:trust},
using the identity substitution for $\sigma$ and
\ref{cond:order:ext}.
\end{full}
\end{proof}

\subsection{Examples}
\label{ssec:examples}

In this subsection, we illustrate the various rules of our calculus on concrete
examples. For better readability, we use nominal $\lambda$~notation.

\begin{exa}[Selection of Negated Predicates]
This example demonstrates the value of allowing selection of literals of the form $t \ceq \ifalse$.
Although the original $\lambda$-superposition calculus was claimed
to support selection of such literals,
its completeness proof was flawed in this respect
\cite{bentkamp-et-al-2023-hosup-errata, nummelin-et-al-2021-boolsup-errata}.

Consider the following clause set:
\begin{align*}
&(1)\ \cst{p}\>\cst{a} \ceq \itrue\\
&(2)\ \cst{q}\>\cst{b} \ceq \itrue\\
&(3)\ \cst{r}\>\cst{c} \ceq \itrue\\
&(4)\ \cst{p}\>x \ceq \ifalse \llor \cst{q}\>y \ceq \ifalse \llor \cst{r}\>z \ceq \ifalse
\end{align*}
Let us first explore what happens without literal selection.
Due to the variables in (4), all of the literals in (4) are incomparable \wrt\ any term order.
So, since none of the literals is selected, there are three possible $\Sup$ inferences:
(1) into (4), (2) into (4), and (3) into (4). After applying $\FalseElim$ to their conclusions,
we obtain:
\begin{align*}
&(5)\ \cst{q}\>y \ceq \ifalse \llor \cst{r}\>z \ceq \ifalse\\
&(6)\ \cst{p}\>x \ceq \ifalse \llor \cst{r}\>z \ceq \ifalse\\
&(7)\ \cst{p}\>x \ceq \ifalse \llor \cst{q}\>y \ceq \ifalse
\end{align*}
For each of these clauses, we can again apply a $\Sup$ inference using (1), (2), or (3), in two different ways each.
After applying $\FalseElim$ to their conclusions, we obtain three more clauses: 
$\cst{p}\>x \ceq \ifalse$,  $\cst{q}\>y \ceq \ifalse$ and $\cst{r}\>z \ceq \ifalse$.
From each of these clauses, we can then derive the empty clause by another $\Sup$ and $\FalseElim$ inference.
So, without literal selection, depending on the prover's heuristics, a prover might in the worst case need to perform $3 + 3 \cdot 2 + 1 = 10$ $\Sup$ inferences
to derive the empty clause.

Now, let us consider the same initial clause set but we select exactly one literal whenever possible.
In (4), we can select one of the literals, say the first one. Then there is only one possible $\Sup$ inference:
(1) into (4), yielding (5) after applying \FalseElim. In (5), we can again select the first literal.
Again, only one $\Sup$ inference is possible, yielding
$\cst{r}\>z \ceq \ifalse$ after applying \FalseElim.
Another $\Sup$ and another $\FalseElim$ inference yield the empty clause.
Overall, there is a unique derivation of the empty clause, consisting of only three $\Sup$ inferences.
\end{exa}

\begin{exa}[Simplification of Functional Literals]
Consider the following clauses, where $\cst{f}$ and $\cst{g}$ are constants of type $\iota \to \iota$.
\begin{align*}
&(1)\ \cst{f} \ceq \cst{g}\\
&(2)\ \cst{f} \cneq \cst{g}
\end{align*}
A $\Sup$ inference from (1) into (2) is not possible because the terms are functional.
Instead, we can apply \infname{ArgCong} and \infname{NegExt} to derive the following clauses:
\begin{align*}
&(3)\ \cst{f}\>x \ceq \cst{g}\>x\text{\quad (by $\ArgCong$ from (1))}\\
&(4)\ \cst{f}\>\diff(\cst{f},\cst{g}) \cneq \cst{g}\>\diff(\cst{f},\cst{g})\text{\quad (by $\infname{NegExt}$ from (2))}
\end{align*}
Both \infname{ArgCong} and \infname{NegExt} are simplification rules,
so we can delete (1) and (2) after deriving (3) and (4).
Now, a $\Sup$ inference from (3) into (4) and a $\EqRes$ inference yield the empty clause.

In contrast, the original superposition calculus requires both the $\Sup$ inference from (1) into (2)
and also a derivation similar to the one above. Moreover, its redundancy criterion
does not allow us to delete (1) and (2). This amounts to doubling the number of clauses and
inferences---even more if $\cst{f}$ and $\cst{g}$ had more than one argument.
\end{exa}

\begin{exa}[Extensionality Reasoning]
Consider the following clauses:
\begin{align*}
&(1)\ \cst{map}\>(\lambda u.\>\cst{sqrt}\>(\cst{add}\>u\>1))\>x \cneq \cst{map}\>(\lambda u.\>\cst{sqrt}\>(\cst{add}\>1\>u))\>x \\
&(2)\ \cst{add}\>u\>v \ceq \cst{add}\>v\>u
\end{align*}
For better readability, we omit type arguments and use subscripts for the parameters of $\diff$.
Using our calculus, we derive the following clauses:
\begin{align*}
&(3)\ \cst{sqrt}\>(\cst{add}\>(\diff_{\lambda u.\>\cst{sqrt}\>(\cst{add}\>u\>1),z})\>1)\cneq 
         z\>(\diff_{\lambda u.\>\cst{sqrt}\>(\cst{add}\>u\>1),z}) \llor \\[-\jot]
&\phantom{(0)\ }\cst{map}\>z\>x \cneq \cst{map}\>(\lambda u.\>\cst{sqrt}\>(\cst{add}\>1\>u))\>x\text{\quad(by $\Ext$ from (1))}\\
&(4)\ \cst{sqrt}\>(\cst{add}\>\diff_{\lambda u.\>\cst{sqrt}\>(\cst{add}\>u\>1),\lambda u.\>\cst{sqrt}\>(\cst{add}\>1\>u)}\>1)\cneq \\[-\jot]
&\phantom{(0)\ }\cst{sqrt}\>(\cst{add}\>1\>\diff_{\lambda u.\>\cst{sqrt}\>(\cst{add}\>u\>1),\lambda u.\>\cst{sqrt}\>(\cst{add}\>1\>u)})\text{\quad (by $\EqRes$ from (3))}\\
&(5)\ \cst{sqrt}\>(\cst{add}\>1\>\diff_{\lambda u.\>\cst{sqrt}\>(\cst{add}\>u\>1),\lambda u.\>\cst{sqrt}\>(\cst{add}\>1\>u)}) \cneq \\[-\jot]
&\phantom{(0)\ }\cst{sqrt}\>(\cst{add}\>1\>\diff_{\lambda u.\>\cst{sqrt}\>(\cst{add}\>u\>1),\lambda u.\>\cst{sqrt}\>(\cst{add}\>1\>u)})
\text{\quad (by $\Sup$ from (2), (4))}\\
&(6)\ \bot\text{\quad (by $\EqRes$ from (5))}
\end{align*}
While such a derivation is also possible in the original $\lambda$-superposition calculus,
the term orders of the original calculus were not able to compare
the literals of the extensionality axiom
\[y\>\diff_{y,z}\cneq z\>\diff_{y,z} \llor y\ceq z\]
As a result, the extensionality axiom leads to an explosion of inferences.
Our calculus avoids this problem by ensuring that the positive literal of
the extensionality axiom is maximal, via the ordering property \ref{cond:order:ext}.
By replacing the extensionality axiom with the \infname{Ext} rule,
we avoid in addition $\Sup$ inferences into functional terms,
and it strengthens our redundancy criterion.
\end{exa}

\begin{full}
\begin{exa}[Delaying Unification Using Constraints]
Consider the following clause set:
\begin{align*}
&(1)\ \cst{map}\>(\lambda u.\>y\>(\cst{s}\>u))\>\cst{a} \cneq \cst{map}\>(\lambda u.\> z\>(\cst{s}\>(y\>u)))\>\cst{a} \llor \cst{lt}\>(y\>\cst{zero})\>(\cst{s}\>(\cst{s}\>(\cst{s}\>\cst{zero})))\ceq\itrue\\
&(2)\ \cst{lt}\>x\>x \ceq \ifalse
\end{align*}
We assume that the first literal of (1) is selected.
Using a $\csuupto$ function that implements Huet's preunification procedure, we can derive
the following clauses:
\begin{align*}
&(3)\ \cst{lt}\>(y\>\cst{zero})\>(\cst{s}\>(\cst{s}\>(\cst{s}\>\cst{zero})))\ceq\itrue\ 
\constraint{\lambda u.\>y\>(\cst{s}\>u) \equiv \lambda u.\> z\>(\cst{s}\>(y\>u))}
\text{\quad (by $\EqRes$ from (1))}\\
&(4)\ \ifalse \ceq \itrue \ 
\text{\quad (by $\Sup$ from (3),(2))}\\
&(5)\ \bot\ 
\text{\quad (by $\Clausify$ from (4))}
\end{align*}
If our calculus did not support constraints, we would have to solve the unification problem in (3) first,
which yields an infinite number of solutions among which the simplest ones are dead ends.
\end{exa}
\end{full}

\begin{exa}[Universal Quantification]
Consider the following clause set:
\begin{align*}
&(1)\ (\lambda x.\>\cst{p}\>x) \ceq (\lambda x.\> \itrue)\\
&(2)\ \cst{p}\>\cst{a} \ceq \ifalse
\end{align*}
Here, clause (1) encodes the universal quantification $\forall x.\> \cst{p}\>x$.
We can derive a contradiction as follows:
\begin{align*}
&(3)\ p\>x \ceq \itrue \text{\quad (by \infname{ArgCong} from (1))}\\
&(4)\ \itrue \ceq \ifalse \text{\quad (by \infname{Sup} from (2), (3))}\\
&(5)\ \bot \text{\quad (by \infname{FalseElim} from (4))}
\end{align*}
Since the $\ArgCong$ inference creating clause (3) can be used as a simplification
rule by Lemma~\ref{lem:simp:arg-cong},
clause (1) can be deleted when creating clause (3).
So we do not need to apply any $\Ext$ inferences into clause (1).
Except for inferences into (1) and except for a $\Diff$ inference, the
inferences required in the derivation above are the only ones possible.
In this sense, the encoding of the universal quantifier using $\lambda$-abstractions
has no overhead.
\end{exa}

\begin{exa}[Existential Quantification]
Negated universal quantification or existential quantification
can be dealt with similarly. Consider the following clause set:
\begin{align*}
&(1)\ (\lambda x.\>\cst{p}\>x) \cneq (\lambda x. \itrue)\\
&(2)\ \cst{p}\>x \ceq \itrue
\end{align*}
We can derive a contradiction as follows:
\begin{align*}
&(3)\ \cst{p}\>\diff\typeargs{\iota,\omicron}(\lambda x.\>\cst{p}\>x, \lambda x. \itrue) \cneq \itrue \text{\quad(by $\infname{NegExt}$ from (1))}\\
&(4)\ \cst{p}\>\diff\typeargs{\iota,\omicron}(\lambda x.\>\cst{p}\>x, \lambda x. \itrue) \ceq \ifalse \text{\quad(by $\infname{NotTrue}$ from (3))}\\
&(5)\ \itrue \ceq \ifalse \text{\quad (by \infname{Sup} from (2), (4))}\\
&(6)\ \bot \text{\quad (by \infname{FalseElim} from (5))}
\end{align*}
Again, we can delete (1) when creating (3), preventing any $\Ext$ inferences from (1).
Moreover, we can delete (3) when creating (4).
As a result, encoding existential quantification using $\lambda$-abstraction does not have overhead either.
\end{exa}

\begin{exa}
This example illustrates why 
\begin{full}%
condition \ref{cond:H:trust:unconstrained} of
our definition of trust (Definition~\ref{def:H:trust}) must require the variable not to occur in parameters.
\end{full}%
\begin{slim}%
condition \ref{sup:two} of $\Sup$ allows $u$ to be a variable if it has another occurrence inside of a parameter.
\end{slim}
Consider the following clause set:
\begin{align*}
&(1)\ \cst{b} \approx \cst{a}\\
&(2)\ (\lambda x.\> (\inot \cst{p}\>x\>y) \iand (\cst{p}\>x\>y \ior y \ineq \cst{a}))\cneq (\lambda x.\> \ifalse)\\
&(3)\ (\lambda x.\> (\inot \cst{p}\>x\>\cst{b}) \iand (\cst{p}\>x\>\cst{b} \ior \cst{b} \ineq \cst{a}))\cneq (\lambda x.\> \ifalse)
\end{align*}
Note that the clauses $(\lambda x.\>\dots) \cneq (\lambda x.\> \ifalse)$
can be read as $\exists x. \dots$ and that (3) is an instance of (2).
Clauses (1) and (3) alone are unsatisfiable
because (1) ensures that the right side of the disjunction $\cst{p}\>x\>\cst{b} \ior \cst{b} \ineq \cst{a}$ in (3) is false,
and since $(\inot \cst{p}\>x\>\cst{b}) \iand (\cst{p}\>x\>\cst{b})$ is clearly false,
clause (3) is false.

For the following derivation, we assume $\cst{b} \succ \cst{a}$.
Applying $\infname{NegExt}$ to (2) and (3) followed by $\infname{NotFalse}$ yields
\begin{align*}
&(4)\ \inot \cst{p}\>\diff_{\lambda x.\> (\inotlight \cst{p}\>x\>y) \iandlight (\cst{p}\>x\>y \iorlight y \ineqlight \cst{a}),\lambda x.\> \ifalselight}\>y \,\iand\, \cst{p}\>\diff_{\lambda x.\> (\inotlight \cst{p}\>x\>y) \iandlight (\cst{p}\>x\>y \iorlight y \ineqlight \cst{a}),\lambda x.\> \ifalselight}\>y \ior y \ineq \cst{a}\ceq \itrue\\
&(5)\ \inot \cst{p}\>\diff_{\lambda x.\> (\inotlight \cst{p}\>x\>\cst{b}) \iandlight (\cst{p}\>x\>\cst{b} \iorlight \cst{b} \ineqlight \cst{a}),\lambda x.\> \ifalselight}\>\cst{b} \,\iand\, \cst{p}\>\diff_{\lambda x.\> (\inotlight \cst{p}\>x\>\cst{b}) \iandlight (\cst{p}\>x\>\cst{b} \iorlight \cst{b} \ineqlight \cst{a}),\lambda x.\> \ifalselight}\>\cst{b} \ior \cst{b} \ineq \cst{a}\ceq \itrue
\end{align*}
For better readability, we omit the type arguments and write the parameters of $\diff$ as subscripts.
Applying $\infname{Clausify}$ several times yields
\begin{align*}
&(6)\ \cst{p}\>\diff_{\lambda x.\> (\inotlight \cst{p}\>x\>y) \iandlight (\cst{p}\>x\>y \iorlight y \ineqlight \cst{a}),\lambda x.\> \ifalselight}\>y \ceq \ifalse\\
&(7)\ \cst{p}\>\diff_{\lambda x.\> (\inotlight \cst{p}\>x\>y) \iandlight (\cst{p}\>x\>y \iorlight y \ineqlight \cst{a}),\lambda x.\> \ifalselight}\>y \ceq \itrue \llor y \cneq \cst{a}\\
&(8)\ \cst{p}\>\diff_{\lambda x.\> (\inotlight \cst{p}\>x\>\cst{b}) \iandlight (\cst{p}\>x\>\cst{b} \iorlight \cst{b} \ineqlight \cst{a}),\lambda x.\> \ifalselight}\>\cst{b} \ceq \ifalse\\
&(9)\ \cst{p}\>\diff_{\lambda x.\> (\inotlight \cst{p}\>x\>\cst{b}) \iandlight (\cst{p}\>x\>\cst{b} \iorlight \cst{b} \ineqlight \cst{a}),\lambda x.\> \ifalselight}\>\cst{b} \ceq \itrue \llor \cst{b} \cneq \cst{a}
\end{align*}
By positive simplify-reflect on (9), followed by $\infname{Demod}$ from (1) into the resulting clause, we obtain the clause
\begin{align*}
&(10)\  \cst{p}\>\diff_{\lambda x.\> (\inotlight \cst{p}\>x\>\cst{b}) \iandlight (\cst{p}\>x\>\cst{b} \iorlight \cst{b} \ineqlight \cst{a}),\lambda x.\> \ifalselight}\>\cst{a} \ceq \itrue
\end{align*}
In this derivation, (2), (3), (4), (5), and (9) can be deleted because $\infname{NegExt}$, $\infname{NotFalse}$, $\infname{Clausify}$, $\infname{Demod}$, and positive simplify-reflect can be applied as simplification rules.
\begin{slim}Moreover, (8) can be deleted by \infname{Subsumption} using (6) and a suitable relation $\sqsupset$.\end{slim}%
\begin{full}

To illustrate why condition~\ref{cond:H:trust:unconstrained} 
does not apply to variables that occur in parameters,
we also remove (8), which is against the redundancy criterion
but would be justified by \infname{Subsumption} of (8) by (6) if condition~\ref{cond:H:trust:unconstrained} ignored parameters.
\end{full}%
The following clauses remain:
\begin{align*}
  &(1)\ \cst{b} \approx \cst{a}\\
  &(6)\ \cst{p}\>\diff_{\lambda x.\> (\inotlight \cst{p}\>x\>y) \iandlight (\cst{p}\>x\>y \iorlight y \ineqlight \cst{a}),\lambda x.\> \ifalselight}\>y \ceq \ifalse\\
  &(7)\ \cst{p}\>\diff_{\lambda x.\> (\inotlight \cst{p}\>x\>y) \iandlight (\cst{p}\>x\>y \iorlight y \ineqlight \cst{a}),\lambda x.\> \ifalselight}\>y \ceq \itrue \llor y \cneq \cst{a}\\
  &(10)\  \cst{p}\>\diff_{\lambda x.\> (\inotlight \cst{p}\>x\>\cst{b}) \iandlight (\cst{p}\>x\>\cst{b} \iorlight \cst{b} \ineqlight \cst{a}),\lambda x.\> \ifalselight}\>\cst{a} \ceq \itrue
\end{align*}
\begin{slim}
Assuming that the negative literal in (7) is selected
and that $\cst{b} \succ \cst{a}$,
we now need a $\Sup$ inference from (1) into the variable $y$ in the second literal of (7),
which is possible because $y$ also appears in a parameter in (7).
This $\Sup$ inference yields:
\begin{align*}
&(11)\ \cst{p}\>\diff_{\lambda x.\> (\inotlight \cst{p}\>x\>\cst{b}) \iandlight (\cst{p}\>x\>\cst{b} \iorlight \cst{b} \ineqlight \cst{a}),\lambda x.\> \ifalselight}\>\cst{b} \ceq \itrue \llor \cst{a} \cneq \cst{a}
\end{align*}
The empty clause can then be derived using $\EqRes$, a $\Sup$ inference with (6), and $\FalseElim$.
\end{slim}%
\begin{full}
Assuming that the negative literal in (7) is selected
and that $\cst{b} \succ \cst{a}$,
no core inference rule other than $\Diff$ applies.
Due to the explosive nature of $\Diff$,
it is difficult to predict whether $\Diff$ inferences lead anywhere,
but we conjecture that this is indeed a counterexample
to a redundancy criterion that ignores parameters.

An alternative approach with a stronger redundancy criterion
that does not need to treat parameters specially
may be to enforce superposition inferences into
variables that have other occurrences inside parameters.
In the example above, this would entail a superposition inference
from (1) into the variable $y$ in the second literal of (7),
which would indeed lead to a refutation.
\end{full}
\end{exa}

\section{Soundness}
\label{sec:soundness}

To prove our calculus sound, we need a substitution lemma for terms and clauses,
which our logic fulfills:
\begin{lem}[Substitution Lemma]\label{lem:subst-lemma}
  Let $\theta$ be a substitution, and let $t$ be a term of type $\tau$.
  For any proper interpretation $\III = (\IIIty, \II, \LL)$ and any valuation $\xi$,
  \[
    \interpret{t\theta}{\III}{\xi} = \interpret{t}{\III}{\xi'}
  \]
  where the modified valuation $\xi'$ is defined
  by $\xity'(\alpha) = \interpret{\alpha\theta}{\IIIty}{\xity}$
  for type variables $\alpha$
  and $\xite'(x) = \interpret{x\theta}{\III}{\xi}$ for term variables $x$.
\end{lem}

\begin{proof}
  By induction on the size of the term $t$.

  \medskip
  \noindent
  \textsc{Case} $t = x\langle \tau \rangle$:
      \begin{align*}
        \interpret{t\theta}{\III}{\xi} &= \interpret{x\theta}{\III}{\xi} \\
        &= \xi'(x) \quad \text{(by the definition of interpretation)} \\
        &= \interpret{x}{\III}{\xi'} \quad \text{(since $x$ is mapped to $\interpret{x\theta}{\III}{\xi}$)} \\
        &= \interpret{t}{\III}{\xi'}
      \end{align*}

  \medskip
  \noindent
  \textsc{Case} $t = \cst{f}\typeargs{\bar{\tau}}\params{\bar{u}}$:
      \begin{align*}
        \interpret{t\theta}{\III}{\xi} &= \interpret{\cst{f}\typeargs{\bar{\tau}\theta}\params{\bar{u}\theta}}{\III}{\xi} \\
        &= \II\left(\cst{f}, \interpret{\bar{\tau}\theta}{\IIIty}{\xity}, \interpretaxi{\bar{u}\theta}\right) \quad \text{(by definition)} \\
        &= \II\left(\cst{f}, \interpret{\bar{\tau}}{\IIIty}{\xity'}, \interpret{\bar{u}}{\III}{\xi'}\right) \quad \text{(by induction hypothesis)} \\
        &= \interpret{\cst{f}\typeargs{\bar{\tau}}\params{\bar{u}}}{\III}{\xi'} \quad \text{(by definition)} \\
        &= \interpret{t}{\III}{\xi'}
      \end{align*}

  \medskip
  \noindent
  \textsc{Case} $t = s\>v$:
      \begin{align*}
        \interpret{t\theta}{\III}{\xi} &= \interpret{s\theta \> v\theta}{\III}{\xi} \\
        &= \interpret{s\theta}{\III}{\xi} \left( \interpret{v\theta}{\III}{\xi} \right) \quad \text{(by definition)} \\
        &= \interpret{s}{\III}{\xi'} \left( \interpret{v}{\III}{\xi'} \right) \quad \text{(by induction hypothesis)} \\
        &= \interpret{s \> v}{\III}{\xi'} \quad \text{(by definition)} \\
        &= \interpret{t}{\III}{\xi'}
      \end{align*}

  \medskip
  \noindent
  \textsc{Case} $t = \lambda\langle \tau \rangle\> u$:
      \begin{align*}
        \interpret{t\theta}{\III}{\xi}(a) 
        &= \interpret{\lambda\langle \tau\theta \rangle\> u\theta}{\III}{\xi}(a) \\
        &= \interpret{u\theta \dbsubst{x}}{\III}{(\xity,\xite[x\mapsto a])} \quad\text{(since $\III$ is proper; for some fresh variable $x$)}\\
        &= \interpret{u\dbsubst{x}\theta }{\III}{(\xity,\xite[x\mapsto a])}\\
        &= \interpret{u\dbsubst{x} }{\III}{(\xity',\xite'[x\mapsto a])} \quad\text{(by induction hypothesis)}\\
        &= \interpret{\lambda\langle \tau \rangle\> u}{\III}{\xi'}(a) \quad \text{(since $\III$ is proper)} \\
        &= \interpret{t}{\III}{\xi'}(a)
      \end{align*}
\end{proof}

\begin{lem}[Substitution Lemma for Clauses]\label{lem:subst-lemma-clause}
  Let $\theta$ be a substitution, and let $C$ be a clause.
  For any proper interpretation $\III = (\IIIty, \II, \LL)$ and any valuation $\xi$,
  $C\theta$ is true \wrt\ $\III$ and $\xi$ if and only if
  $C$ is true \wrt\ $\III$ and $\xi'$,
  where the modified valuation $\xi'$ is defined
  by $\xity'(\alpha) = \interpret{\alpha\theta}{\IIIty}{\xity}$
  for type variables $\alpha$
  and $\xite'(x) = \interpret{x\theta}{\III}{\xi}$ for term variables $x$.
\end{lem}
\begin{proof}
By definition of the semantics of clauses,
$C\theta$ is true \wrt\ $\III$ and $\xi$ if and only if
one of its literals is true \wrt\ $\III$ and $\xi$.
By definition of the semantics of literals,
a positive literal $s\theta\ceq t\theta$ (resp.\ negative literal $s\theta\cneq t\theta$) 
of $C\theta$ is
true \wrt\ $\III$ and $\xi$ if and only if
$\interpretaxi{s\theta}$ and $\interpretaxi{t\theta}$ are equal (resp.\ different).
By Lemma~\ref{lem:subst-lemma},
$\interpretaxi{s\theta}$ and $\interpretaxi{t\theta}$ are equal (resp.\ different)
if and only if $\interpret{s}{\III}{\xi'}$ and $\interpret{t}{\III}{\xi'}$ are equal (resp.\ different)---i.e.,
if and only if a literal $s\ceq t$ (resp.\ $s\cneq t$) in $C$ is true \wrt\ $\III$ and $\xi'$.
This holds if and only if $C$ is true \wrt\ $\III$ and $\xi'$.
\end{proof}

\begin{thm}
  All core inference rules are sound \wrt\ $\soundmodels$ (Definition~\ref{def:H:diff-aware}).
  All core inference rules except for $\Ext$, $\FluidExt$, and $\Diff$ are also sound \wrt\ $\models$.
  This holds even when ignoring order, selection, and eligibility conditions.
\end{thm}
\begin{proof}
  We fix an inference and an interpretation $\III$ that is a model of the premises.
  For $\Ext$, $\FluidExt$, and $\Diff$ inferences, we assume that $\III$ is $\diff$-aware.
  We need to show that it is also a model of the conclusion.
  By Lemma~\ref{lem:subst-lemma-clause}, $\III$ is a model of the $\sigma$-instances of the premises as well, where
  $\sigma$ is the substitution used for the inference.
  From the semantics of our logic, it is easy to see that congruence
  holds at green positions and at the left subterm of an application.
  To show that $\III$ is a model of the conclusion, it suffices to show that the conclusion
  is true under $\III,\xi$ for all valuations $\xi$.
  
  For most rules, it suffices to make
  distinctions on the truth under $\III,\xi$ of the literals of the $\sigma$-instances of the premises,
  to consider the conditions that $\sigma$ is a unifier where applicable,
  and to apply congruence.
  For $\BoolHoist$, $\LoobHoist$, $\FalseElim$, $\Clausify$, $\FluidBoolHoist$, $\FluidLoobHoist$,
  we also use the fact that $\III$ interprets logical symbols correctly.
  For $\Ext$, $\FluidExt$, and $\Diff$,
  we also use the assumption that $\III$ is $\diff$-aware.
\end{proof}

\section{Refutational Completeness}
\label{sec:refutational-completeness}

\begin{full}

Superposition is a saturation-based calculus.
Provers that implement it start from an initial clause set $N_0$
and 
repeatedly add new clauses by performing inferences
or remove clauses by determining them to be redundant.
In the limit, this process results in a (possibly infinite)
set $N_\infty$ of persistent clauses.
Assume that inferences are performed in a fair fashion; i.e., no nonredundant
inference is postponed forever.
Then the set $N_\infty$ is saturated, meaning that all
inferences are redundant (for example because their conclusion is in the set).
Refutational completeness is the property that
if $N_\infty$ does not contain the empty clause, $N_0$ has a model.
Since refutational completeness is the only kind of completeness that interests
us in this article, we will also refer to it as ``completeness.''

Due the role of constraints and parameters in our calculus,
our completeness result, stated in Corollary~\ref{cor:H:completeness}, makes two additional assumptions:
It assumes that
the clauses in $N_0$
have no constraints and do not contain constants with parameters.
And, instead of the usual assumption that $N_\infty$ does not contain the empty clause,
we assume that $N_\infty$ does not contain an empty clause with satisfiable constraints.

\subsection{Proof Outline}

The idea of superposition completeness proofs in general is the following:
We assume that $N_\infty$ does not contain the empty clause.
We construct a term rewrite system derived from the ground instances
of $N_\infty$.
We view this system as an interpretation $\III$ and show that it is
a model of the ground instances and thus of $N_\infty$ itself.
Since only redundant clauses are removed during saturation,
$\III$ must also be a model of $N_0$.

Completeness proofs of \emph{constrained} superposition calculi,
including our the completeness proof of our calculus,
must proceed differently.
The constraints prevent us from showing $\III$ to be a model
of all ground instances of $N_\infty$.
Instead, we restrict ourselves to proving that $\III$ is a model
of the variable-irreducible ground instances of $N_\infty$.
Roughly speaking, a variable-irreducible ground instance is one where
the terms used to instantiate variables
are irreducible \wrt\ the constructed term rewrite system.
The notion of redundancy must be based on
the notion of variable-irreducible ground instances as well,
so that if $\III$ is
a model
of the variable-irreducible ground instances of $N_\infty$,
it is also a model of the
variable-irreducible ground instances of $N_0$.
Assuming that the initial clauses $N_0$ do
not have constraints, $\III$
is then a model of all ground instances of $N_0$
because every ground instance has a corresponding
variable-irreducible instance with the same truth value in $\III$.
It follows that 
$\III$ is a model of $N_0$.

\begin{figure}[tb]
  \begin{tikzpicture}[every node/.style={outer sep=1mm}]
    \baselineskip=0pt
    \node (H)  at (0,10) [align=center] {$\levelH$\\higher-order\\constrained clauses};
    \node (G)  at (0,8) [align=center] {$\levelG$\\higher-order closures};
    \node (PG) at (0,6) [align=center] {$\levelPG$\\partly substituted\\higher-order closures};
    \node (IPG) at (0,4) [align=center] {$\levelIPG$\\indexed partly substituted\\higher-order closures};
    \node (PF)  at (0,2) [align=center] {$\levelPF$\\partly substituted\\ground first-order closures};
    \node (F)  at (0,0) [align=center] {$\levelF$\\ground first-order clauses};
    \draw[->] (H) -- (G)
      node[midway,right] {$\mapGonly$};
    \draw[->] (G) -- (PG)
    node[midway,right] {$\mapPonly$};
    \draw[->] (PG) -- (IPG)
    node[midway,right] {$\mapIonly$};
    \draw[->] (IPG) -- (PF)
      node[midway,right] {$\mapFonly$};
    \draw[->] (PF) -- (F)
      node[midway,right] {$\mapTonly$};
    \node (Hex)  at (7,10) [align=center] 
      {$x \>(\diff\typeargs{\alpha,\alpha}(\cst{g}\typeargs{\alpha},\cst{h}\typeargs{\alpha})\> z \ceq \cst{c} $\\
      $\constraint{x\>(\cst{g}\typeargs{\alpha}\>y) \equiv x\>(\cst{g}\typeargs{\alpha}\>\cst{a}\typeargs{\alpha}),\ z\equiv \lambda\typeargs{\alpha}\>\cst{k}\typeargs{\alpha}\>\DB{0}\>y}$};
    \node (Gex)  at (7,8) [align=center]
      {$x \>(\diff\typeargs{\iota,\iota} (\cst{g}\typeargs{\alpha},\cst{h}\typeargs{\alpha})) z \ceq \cst{c}  $\\
      $\closure\ \{\alpha\mapsto\iota,\ x \mapsto \cst{f}\typeargs{\iota},\ y \mapsto \cst{a}\typeargs{\iota},\ z \mapsto \lambda\typeargs{\iota}\>\cst{k}\typeargs{\iota}\>\DB{0}\>\cst{a}\typeargs{\iota}\}$};
    \node (PGex) at (7,6) [align=center] 
      {$\cst{f}\typeargs{\iota} \>(\diff\typeargs{\iota,\iota} (\cst{g}\typeargs{\iota},\cst{h}\typeargs{\iota})) \>(\lambda\typeargs{\iota}\>\cst{k}\typeargs{\iota}\>\DB{0}\>z_{1.2}) \ceq \cst{c}$\\
      $\closure\ \{y_1 \mapsto \cst{a}\typeargs{\iota},\ z_{1.2} \mapsto \cst{a}\typeargs{\iota}\}$};
    \node (IPGex) at (7,4) [align=center]
      {$\cst{f}\typeargs{\iota} \>(\diff^{\iota,\iota}_{(\cst{g}\typeargs{\iota},\cst{h}\typeargs{\iota})}) \>(\lambda\typeargs{\iota}\>\cst{k}\typeargs{\iota}\>\DB{0}\>z_{1.2})  \ceq \cst{c}$\\
      $\closure\ \{y_1 \mapsto \cst{a}\typeargs{\iota},\ z_{1.2} \mapsto \cst{a}\typeargs{\iota}\}$};
    \node (PFex)  at (7,2) [align=center]
      {$\cst{f}^\iota \>(\diff^{\iota,\iota}_{(\cst{g}\typeargs{\iota},\cst{h}\typeargs{\iota})},\>\cst{fun}_{\lambda\typeargs{\iota}\>\cst{k}\typeargs{\iota}\>\DB{0}\>\square}(z_{1.2})) \ceq \cst{c}$\\
      $\closure\ \{y_1 \mapsto \cst{a}^\iota_0,\  z_{1.2} \mapsto \cst{a}^\iota_0\}$};
    \node (Fex)  at (7,0) [align=center] {$\cst{f}^\iota_1 \>(\diff^{\iota,\iota}_{(\cst{g}\typeargs{\iota},\cst{h}\typeargs{\iota}),0},\> \cst{fun}_{\lambda\typeargs{\iota}\>\cst{k}\typeargs{\iota}\>\DB{0}\>\square}(\cst{a}^\iota_0))\ceq \cst{c}$};
    \draw[->] (Hex) -- (Gex)
      node[midway,right] {$\mapGonly$};
    \draw[->] (Gex) -- (PGex)
    node[midway,right] {$\mapPonly$};
    \draw[->] (PGex) -- (IPGex)
    node[midway,right] {$\mapIonly$};
    \draw[->] (IPGex) -- (PFex)
      node[midway,right] {$\mapFonly$};
    \draw[->] (PFex) -- (Fex)
      node[midway,right] {$\mapTonly$};
  \end{tikzpicture}
  \caption{Overview of the levels\label{fig:overview-levels}}
\end{figure}
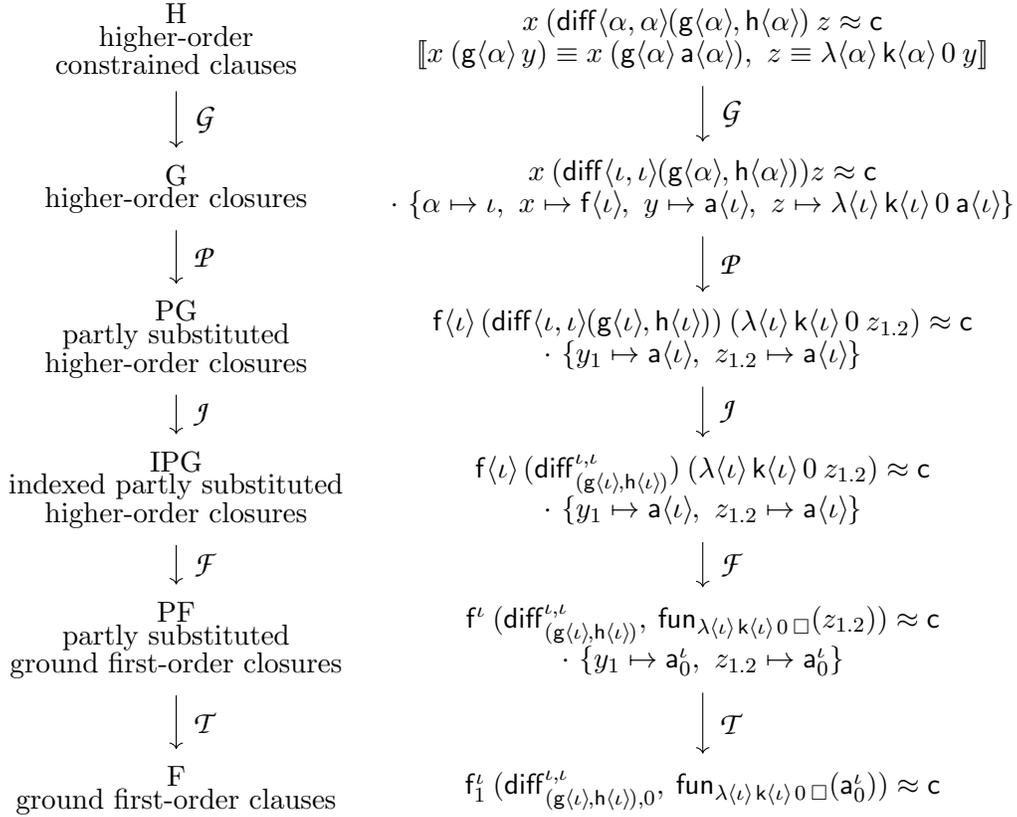
  
\smallskip
To separate concerns, our proof is structured as a sequence of six
levels, most of which have their own
logic, calculus, redundancy criterion, and completeness property.
The levels are called
$\levelH$, $\levelG$, $\levelPG$, $\levelIPG$, $\levelPF$, and $\levelF$.
They are connected by functions encoding clauses from one level to the next.

The level $\levelH$ is the level of higher-order constrained clauses, 
using the logic described
in Section~\ref{sec:logic} and the calculus described in Section~\ref{sec:calculus}.
Our ultimate goal is to prove completeness on this level.

The level $\levelG$ is the level of higher-order closures,
where a closure $C\closure\theta$ is a pair consisting of a clause $C$ and
a grounding substitution $\theta$.
The function $\mapGonly$
maps each clause from level $\levelH$ to a corresponding set of closures on
level $\levelG$ 
using all possible grounding substitutions.

The level $\levelPG$ is the level of \emph{partly substituted}
higher-order closures. It is the fragment of $\levelG$
that contains no type variables and no functional variables.
The map $\mapPonly$ encodes closures from $\levelG$
into level $\levelPG$ by applying a carefully crafted
substitution to functional variables.

The level $\levelIPG$ is the level of \emph{indexed}
partly substituted closures. It modifies the signature of the previous levels
by replacing each symbol with parameters
$\cst{f}\oftypedecl\forallty{\tuple{\alpha}_m}\>\tuple{\tau}_n\fofun\tau$
by a collection of symbols
$\cst{f}^{\tuple{\upsilon}_m}_{\tuple{u}_n}\oftype\tau$
for each tuple of
types $\tuple{\upsilon}_m$ and
each tuple of ground terms $\tuple{u}_n\oftype\tuple{\tau}_n$.
The map $\mapIonly$ encodes closures from $\levelPG$
into $\levelIPG$ by moving 
type arguments into the superscript indices $\tuple{\upsilon}_m$
parameters into the subscript
indices $\tuple{u}_n$.

The level $\levelPF$ is the level of partly substituted
ground first-order closures.
Its logic is the one described in Section~\ref{ssec:redundancy}
except with variables and closures.
We extend the encoding $\mapFonly$ (Definition~\ref{def:fol})
with variables, yielding an encoding from $\levelIPG$
into $\levelPF$.

The level $\levelF$ is the level of first-order clauses.
It uses the same logic as $\levelPF$ but uses ground clauses
instead of closures. The map $\mapTonly$ connects the two by
mapping a closure $C\closure\theta$ to the clause $C\theta$.

Figure~\ref{fig:overview-levels} gives
an overview of the hierarchy of levels
and an example of a clause instance across the levels.

\end{full}

\subsection{Logics and Encodings}

In our completeness proof, we use two higher-order signatures and one first-order signature.

Let $\SigmaH$ be the higher-order signature
used by the calculus described in Section~\ref{sec:calculus}.
It is required to contain a symbol
$\cst{diff}\oftypedecl\Pi\alpha,\beta.\>(\alpha\fun\beta, \alpha\fun\beta) \fofun \alpha$

Let $\SigmaI$ be the signature obtained from
$\SigmaH$ in the following way:
We replace each constant with parameters
${\cst{f}: \forallty{\bar{\alpha}_m} \bar{\tau}_n \Rightarrow \tau}$ in $\SigmaH$
with a family of constants
$\cst{f}^{\tuple{\upsilon}_m}_{\tuple{t}_n}: \tau$, indexed by all possible ground types
$\tuple{\upsilon}_m$ and ground terms
$\tuple{t}_n \in \TT_\mathrm{ground}(\SigmaH)$ of type $\tuple{\tau}_n\{\bar{\alpha}_m\mapsto\tuple{\upsilon}_m\}$.
Constants without parameters (even those with type arguments) are left as they are.

In some contexts,
it is more convenient to
use terms from $\TT_\mathrm{ground}(\SigmaI)$ instead of $\TT_\mathrm{ground}(\SigmaH)$
in the subscripts $t_i$ of the constants $\cst{f}^{\tuple{\upsilon}_m}_{\tuple{t}_n}$.
We follow this convention:
\begin{conv}\label{convention:indexing-subscripts}
  In the subscripts $t_i$ of constants $\cst{f}^{\tuple{\upsilon}_m}_{\tuple{t}_n} \in \SigmaI$,
  we identify each term of the form
  $\cst{f}\typeargs{\tuple{\upsilon}_m}(\tuple{t}_n) \in \TT_\mathrm{ground}(\SigmaH)$
  with the term
  $\cst{f}^{\tuple{\upsilon}_m}_{\tuple{t}_n} \in \TT_\mathrm{ground}(\SigmaI)$, 
  whenever $n > 0$.
\end{conv}

Similarly, the first-order signatures $\mapF{\SigmaI}$ and $\mapF{\SigmaH}$
as defined in Section~\ref{ssec:redundancy}
are almost identical, the only difference being that
the subscripts $t$ of the symbols $\cst{fun}_t \in \mapF{\SigmaH}$
may contain symbols with parameters,
whereas the subscripts $t$ of the symbols $\cst{fun}_t \in \mapF{\SigmaI}$
may not. To repair this mismatch, we adopt the following convention
using the obvious correspondence between the symbols in $\mapF{\SigmaH}$ and $\mapF{\SigmaI}$:
\begin{conv}\label{convention:first-order-signatures}
In the subscripts of constants $\cst{fun}_t$ in $\mapF{\SigmaH}$ and $\mapF{\SigmaI}$,
we identify each term of the form
$\cst{f}\typeargs{\tuple{\upsilon}_m}(\tuple{t}_n) \in \TT_\mathrm{ground}(\SigmaH)$
with the term
$\cst{f}^{\tuple{\upsilon}_m}_{\tuple{t}_n} \in \TT_\mathrm{ground}(\SigmaI)$, 
whenever $n > 0$.
Using this identification,
we can consider the first-order signatures $\mapF{\SigmaH}$ and $\mapF{\SigmaI}$
to be identical.
\end{conv}

\begin{full}
Our completeness proof uses two sets of variables.
Let  $\VVH$ be the set of variables 
used by the calculus described in Section~\ref{sec:calculus}.
Based on $\VVH$, we define the variables $\VVPG$ of the
$\levelPG$ level as
\[\VVPG = \VVH \cup \{y_p\langle\tau\rangle \mid y \in \VVH,\ p \text{ a list of natural numbers},\ \tau\text{ a nonfunctional type} \}\]
\end{full}

The table below summarizes our completeness proof's \slimfull{four}{six} levels, each with a set
of terms and a set of clauses. We write $\TT_\mathrm{X}$ for the set of terms
and $\CC_\mathrm{X}$ for the set of clauses of a given level $X$:\strut
\begin{full}
\begin{center}
\begin{tabular}{p{1cm}p{8cm}p{5cm}}
 Level & Terms & Clauses \\
 \hline
$\levelF$ &
ground first-order terms over $\mapF{\SigmaI}$
 & clauses over $\termsF$ \\
$\levelPF$ & 
first-order terms over $\mapF{\SigmaI}$ and $\VVPG$
that do not contain variables whose type is of the form 
$\tau \fun \upsilon$
 & closures over $\termsPF$ \\
$\levelIPG$ & $\{t \in \TT(\SigmaI, \VVPG)\mid$
 $t$ contains neither type variables nor functional variables$\}$ & closures over $\TT_\levelIPG$  \\
$\levelPG$ & $\{t \in \TT(\SigmaH, \VVPG)\mid$
 $t$ contains neither type variables nor functional variables$\}$ & closures over $\TT_\levelPG$ \\
$\levelG$ & $\TT(\SigmaH, \VVH)$ & closures over $\TT_\levelG$ \\
$\levelH$ & $\TT(\SigmaH, \VVH)$ & constrained clauses over $\TT_\levelH$
\end{tabular}
\end{center}
\end{full}
\begin{slim}
  \begin{center}
    \begin{tabular}{p{1cm}p{8cm}p{5cm}}
     Level & Terms & Clauses \\
     \hline
    $\levelF$ &
    ground first-order terms over $\mapF{\SigmaI}$
     & clauses over $\termsF$ \\
    $\levelIPG$ & $\TT_\mathrm{ground}(\SigmaI)$ & clauses over $\TT_\levelIPG$  \\
    $\levelG$ & $\TT_\mathrm{ground}(\SigmaH)$ & clauses over $\TT_\levelG$ \\
    $\levelH$ & $\TT(\SigmaH)$ & clauses over $\TT_\levelH$
    \end{tabular}
    \end{center}
\end{slim}

\begin{slim}
\subsubsection{First-Order Encoding}
We use the map $\mapFonly$ defined in Definition~\ref{def:fol}
both as an encoding from $\termsIPG$/$\clausesIPG$ to $\termsPF$/$\clausesPF$
and as an encoding from $\termsPG$/$\clausesPG$ to $\termsPF$/$\clausesPF$.
Potential for confusion is minimal because the two encodings coincide on the values that are in the domain of both.
\end{slim}
\begin{full}
\subsubsection{First-Order Encodings}
The transformation $\mapTonly$ from $\clausesPF$ to $\clausesF$ is simply defined as $\mapT{C\closure\theta} = C\theta$.
We also define a bijective encoding from 
$\termsIPG$ into
$\termsPF$
and from 
$\clausesIPG$ into 
$\clausesPF$.
It is very similar to the encoding $\mapFonly : \TT_\mathrm{ground}(\SigmaH) \to \termsF$
defined in Definition~\ref{def:fol},
but also encodes variables and does not encode parameters.
We reuse the name $\mapFonly$ for this new encoding.
Potential for confusion is minimal because the two encodings coincide on the values that are in the domain of both.

\begin{defi}[First-Order Encoding $\mapFonly$]\label{def:IPG:mapF}
We define $\mapFonly : \termsIPG \to \termsPF$ recursively as follows:
If $t$ is functional,
then
let $t'$ be the expression
obtained by replacing each
outermost proper yellow subterm in $t$
by the placeholder symbol $\square$,
and let $\mapF{t} = \cst{fun}_{t'}(\mapF{\tuple{s}_{n}})$, where $\tuple{s}_{n}$
are the replaced subterms in order of occurrence.
If $t$ is a variable $x$, we define $\mapF{t} = x$.
Otherwise, $t$ is of the form $\cst{f}\langle\bar\tau\rangle\> \tuple{t}_m$
and we define $\mapF{t}
= \cst{f}^{\tuple{\tau}} (\mapF{\tuple{t}_1},\dots,\mapF{\tuple{t}_m})$. 

Applied to a closure $C\closure\theta \in\clausesIPG$, the function $\flooronly$
is defined by $\floor{C\closure\theta} = \floor{C}\closure\floor{\theta}$,
where $\mapFonly$ maps each side of each literal
and each term in a substitution individually.
\end{defi}

\begin{lem}\label{lem:IPG:F-bijection}
The map $\mapFonly$ is a bijection between $\termsIPG$ and $\termsPF$
and between $\clausesIPG$ and $\clausesPF$.
\end{lem}
\begin{proof}
Injectivity of $\mapFonly$ can be shown by structural induction.
For surjectivity, let $t \in \termsPF$.
We must show that there exists some $s \in \termsIPG$ such that $\mapF{s} = t$.
We proceed by induction on $t$.

If $t$ is of the form $\cst{fun}_{t'}(\tuple{t}_n)$,
we use the induction hypothesis to derive the existence of
some $\tuple{s}_n \in \termsIPG$
such that $\mapF{\tuple{s}_n} = \tuple{t}_n$.
Let $s$ be the term resulting from replacing the placeholder symbols $\square$ in $t'$ by $\tuple{s}_n$ in order of occurrence.
Then $\mapF{s} = t$.

If $t$ is a variable $x$, by definition of $\termsPF$, $t$'s type is not of the form $\tau \fun \upsilon$.
So, we can set $s = x \in \termsIPG$. Then $\mapF{s} = t$.

If $t = \cst{f}^{\tuple{\tau}}(\tuple{t}_n)$, where $\cst{f}^{\tuple{\tau}}$ is not a $\cst{fun}$ symbol,
by the induction hypothesis
there exist $\tuple{s}_n$ such that $\mapF{\tuple{s}_n} = \tuple{t}_n$ and set $s = \cst{f}\typeargs{\tuple{\tau}}\> \tuple{s}_n$. Then $\mapF{s} = t$.

It follows that $\mapFonly$ is also a bijection between $\clausesIPG$ and $\clausesPF$.
\end{proof}

\begin{lem}\label{lem:IPG:mapF-subst}
For all terms $t \in \termsIPG$,
all clauses over $\termsIPG$,
and all grounding substitutions $\theta$, we have
$\mapF{t}\mapF{\theta} = \mapF{t\theta}$
and 
$\mapF{C}\mapF{\theta} = \mapF{C\theta}$.
\end{lem}
\begin{proof}
Since $\mapFonly$ maps each side of each literal individually,
it suffices to show that $\mapF{t}\mapF{\theta} = \mapF{t\theta}$
for all $t \in \termsIPG$.
We proceed by structural induction on $t$.

If $t$ is a variable, the claim is trivial.

If $t$ is nonfunctional and headed by a symbol, the claim follows from the induction hypothesis.

Finally, we consider the case where $t$ is functional.
Let $t'$ be the expression
obtained by replacing each
outermost proper yellow subterm in $t$
by the placeholder symbol $\square$,
and let $\tuple{s}_{n}$
be the replaced subterms in order of occurrence.
Since all variables in $t$ are nonfunctional,
they must be located in a proper yellow subterm of $t$,
and thus replacing
the outermost proper yellow subterm in $t\theta$
by the placeholder symbol $\square$
will result in $t'$ as well.
So, using the induction hypothesis,
$\mapF{t}\mapF{\theta}
= \cst{fun}_{t'}(\mapF{\tuple{s}_{n}}\mapF{\theta})
= \cst{fun}_{t'}(\mapF{\tuple{s}_{n}\theta})
= \mapF{t\theta}$.
\end{proof}

\begin{lem}\label{lem:IPG:mapF-subterms}
A term $s\in \termsIPG$ is a yellow subterm of $t \in \termsIPG$ if and only if $\mapF{s}$
is a subterm of $\mapF{t}$.
\end{lem}
\begin{proof}
By induction using the definition of $\mapFonly$.
\end{proof}

\end{full}

\begin{lem}\label{lem:IPG:mapF-green-subterms}
  A term $s\in \termsIPG$ is a green subterm of $t \in \termsIPG$ if and only if $\mapF{s}$
  is a \slimfull{}{green} subterm \slimfull{}{(as defined in Section~\ref{ssec:redundancy})} of $\mapF{t}$.
\end{lem}
\begin{proof}
By induction using the definition of $\mapFonly$.
\end{proof}

\subsubsection{Indexing of Parameters}
\label{ssec:partly-substituted-ground-higher-order-level}

\begin{defi}[Indexing of Parameters]\label{def:PG:mapI}
The transformation $\mapIonly$ translates from $\termsPG$ to $\termsIPG$ by
encoding any occurrence of a constant with parameters $\cst{f}\typeargs{\tuple{\upsilon}}(\tuple{u})$ as
$\cst{f}^{\tuple{\upsilon}}_{\tuple{u}\slimfull{}{\theta}}$\slimfull{.}{,}
\begin{full}where $\theta$ denotes the substitution of the corresponding closure.\end{full}
Formally:
\begin{align*}
\mapIonly_{\slimfull{}{\theta}}(x) &= x \\
\mapIonly_{\slimfull{}{\theta}}(\lambda \> t) &= \lambda \> \mapIonly_{\slimfull{}{\theta}}(t) \\
\mapIonly_{\slimfull{}{\theta}}(\cst{f}\typeargs{\tuple{\upsilon}}  \> \tuple{s})
&= \cst{f}\typeargs{\tuple{\upsilon}} \> \mapIonly_{\slimfull{}{\theta}}(\tuple{s})
\\
\mapIonly_{\slimfull{}{\theta}}(\cst{f}\typeargs{\tuple{\upsilon}} \> (\tuple{u}_k) \> \tuple{s})
&= \cst{f}^{\tuple{\upsilon}}_{\tuple{u}_k\slimfull{}{\theta}} \> \mapIonly_{\slimfull{}{\theta}}(\tuple{s})
\text{ if $k > 0$}\\
\mapIonly_{\slimfull{}{\theta}}(m \> \tuple{s}) &= m \> \mapIonly_{\slimfull{}{\theta}}(\tuple{s})
\end{align*}%
We extend $\mapIonly_{\slimfull{}{\theta}}$ to clauses by mapping each side of each literal individually.
\begin{full}
If $t$ is a ground term, the given substitution is irrelevant, so we omit the subscript and simply write $\mapI{t}$.
We extend $\mapIonly$ to grounding substitutions by
defining $\mapI{\theta}$ as $x \mapsto \mapI{x\theta}$.
The transformation $\mapIonly$ w.r.t.\ a closure $C\closure\theta$ is defined as $\mapI{C\closure\theta} = \mapIonly_\theta(C)\closure\mapI{\theta}$.
\end{full}
\end{defi}

\begin{full}
\begin{lem}\label{lem:PG:mapI-subst}
For all $t \in \termsPG$,
all clauses $C$ over $\termsPG$,
and all grounding substitutions $\theta$, we have
$\mapI{t\theta} = \mapIonly_\theta(t)\mapI{\theta}$
and $\mapI{C\theta} = \mapIonly_\theta(C)\mapI{\theta}$.
\end{lem}
\begin{proof}
Since $\mapIonly$ maps each side of each literal individually,
it suffices to show that
$\mapI{t\theta} = \mapIonly_\theta(t)\mapI{\theta}$.
We prove this by induction on the structure of $t$.
If $t$ is a variable, the claim is trivial.
For all other cases, the claim follows from the definition of $\mapIonly_\theta$
and the induction hypothesis.
\end{proof}

\begin{lem}\label{lem:PG:mapI-preserves-identities}
If 
$\mapIonly_\theta(t)\mapI{\theta} = \mapIonly_\theta(t')\mapI{\theta}$,
then
$t\theta = t'\theta$.
\end{lem}
\begin{proof}
  This becomes clear when viewing $\mapIonly$
  as composed of two operations: First,
  we apply the substitution to all variables in
  parameters. Second, we move type arguments and
  parameters into indices.
  Both parts clearly fulfill this lemma's statement.
\end{proof}

\end{full}

\begin{lem}\label{lem:PG:mapI-mapF-fol}
  Let $t\in\TT_\mathrm{ground}(\SigmaH)$, and
  let $C$ be a clause over $\TT_\mathrm{ground}(\SigmaH)$.
  Then $\mapF{\mapI{t}} = \mapF{t}$
  and $\mapF{\mapI{C}} = \mapF{C}$.
\end{lem}
\begin{proof}
For $t$, the claim follows directly from
the definitions of $\mapIonly$ (Definition~\ref{def:PG:mapI}) and $\mapFonly$
(Definition\slimfull{}{s}~\ref{def:fol}\begin{full} and~\ref{def:IPG:mapF}\end{full}),
relying on the identification of $\cst{fun}_t$ and $\cst{fun}_{\mapI{t}}$ (Convention~\ref{convention:first-order-signatures}).
For $C$, the claim holds because
$\mapIonly$ and $\mapFonly$ map each side of each literal individually.
\end{proof}

\begin{full}
\subsubsection{Partial Substitution}

  Let $\theta$ be a grounding substitution from $\VVH$ to $\TT_\mathrm{ground}(\SigmaH)$. We define a substitution
  $\mapp{\theta}$, mapping from $\VVH$ to $\TT(\SigmaH, \VVPG)$, and 
  a substitution $\mapq{\theta}$,
  mapping from $\VVPG$ to $\TT_\mathrm{ground}(\SigmaH)$,
  as follows.
  For each type variable $\alpha$,
  let $\alpha\mapp{\theta} = \alpha\theta$.
  For each variable $y\in \VVH$, let $y\mapp{\theta}$ be the term resulting
  from replacing each nonfunctional yellow subterm at a yellow position $p$ in $y\theta$ by
  $y_p \in \VVPG$.
  We call these variables $y_p$ the variables introduced by $\mapp{\theta}$.
  If a nonfunctional yellow subterm is contained inside another nonfunctional yellow subterm,
  the outermost nonfunctional yellow subterm is replaced by a variable.
  The substitution
  $\mapq{\theta}$ is defined
  as $y_p\mapq{\theta} = y\theta|_p$
  for all variables $y_p$ introduced by $\mapp{\theta}$,
  and for all other variables $y\in \VVPG$, we set
  $y\mapq{\theta}$ to be some arbitrary ground term that is independent of $\theta$.

  Finally, we define
   \[\mapPonly : \clausesG \to \clausesPG,\quad C \closure \theta\mapsto C\mapp{\theta}\closure\mapq{\theta}\]

For example, if $y\theta = \lambda\> \cst{f}\>(\DB{0}\>(\cst{g}\>\cst{a}))$,
then $y\mapp{\theta} = \lambda\> \cst{f}\>(\DB{0}\>y_{1.1.1})$
and $y_{1.1.1} \mapq{\theta} = \cst{g}\>\cst{a}$.

Technically, this definition of $\mapponly$, $\mapqonly$, and $\mapPonly$
depends on the choice of a $\beta\eta$-normalizer $\benf{}$
because it relies on yellow positions. However,
this choice affects the resulting terms and clauses only up to
renaming of variables, and our proofs
work for any choice of $\benf{}$
as long as we use the
same fixed $\benf{}$
for $\mapponly$, $\mapqonly$, and $\mapPonly$.

\begin{lem}\label{lem:G:mapp-comp-mapq}
  Let $\theta$ be a grounding substitution. Then
  $\mapp{\theta}\mapq{\theta} = \theta$.
\end{lem}
\begin{proof}
  For type variables $\alpha$, we have $\alpha\mapp{\theta}\mapq{\theta} = \alpha\theta$ by definition of $\mapponly$.
  We must show that $y\mapp{\theta}\mapq{\theta} = y\theta$ for all variables $y\in \VVH$.
  By definition of $\mapponly$, $y\mapp{\theta}$ is obtained from $y\theta$ by replacing each nonfunctional yellow subterm at a yellow position $p$ with $y_p$.
  By definition of $\mapqonly$, we have $y_p\mapq{\theta} = y\theta|_p$ for all such positions $p$.
  Therefore, when we apply $\mapq{\theta}$ to $y\mapp{\theta}$, we replace each $y_p$ with the original subterm $y\theta|_p$, effectively reconstructing $y\theta$.
\end{proof}

\begin{lem}\label{lem:G:mapp-mapp}
  Let $\theta$ be a substitution from $\VVH$ to $\TT_\mathrm{ground}(\SigmaH)$ and $\rho$ be a substitution
  from $\VVPG$ to $\TT_\mathrm{ground}(\SigmaH)$.
  Then
  $\mapp{\mapp{\theta}\rho} = \mapp{\theta}$.
\end{lem}
\begin{proof}
  Let $y\in \VVH$.
  By definition of $\mapponly$, $y\mapp{\theta}$ is obtained from $y\theta$
  by replacing each nonfunctional yellow subterm at a yellow position $p$ with $y_p$.
  Now, consider $y\mapp{\theta}\rho$. Since $\rho$ is grounding,
  it will replace each $y_p$ with a ground term.
  To obtain $y\mapp{\mapp{\theta}\rho}$, we take $y\mapp{\theta}\rho$ and again replace each nonfunctional yellow subterm at a yellow position $p$ with $y_p$.
  These positions and the resulting structure will be identical to those in $y\mapp{\theta}$
  because the ground terms introduced by $\rho$ do not affect the overall structure of yellow positions.
  Therefore, $y\mapp{\mapp{\theta}\rho} = y\mapp{\theta}$ for each variable $y\in \VVH$.

  For type variables $\alpha$, we have $\alpha\mapp{\theta} = \alpha\theta = \alpha\mapp{\mapp{\theta}\rho}$.
  Thus, we can conclude that $\mapp{\mapp{\theta}\rho} = \mapp{\theta}$.
\end{proof}

\begin{lem}\label{lem:G:mapq-mapp}
  Let $\theta$ be a substitution from $\VVH$ to $\TT_\mathrm{ground}(\SigmaH)$.
  Let $\rho$ be a substitution
  from $\VVPG$ to $\TT_\mathrm{ground}(\SigmaH)$
  such that $y\rho = y\mapq{\theta}$ for all $y$ not introduced by $\mapp{\theta}$.
  Then
  $\mapq{\mapp{\theta}\rho} = \rho$.
\end{lem}
\begin{proof}
  Let $y_p$ be a variable introduced by $\mapp{\mapp{\theta}\rho}$.
  By definition of $\mapqonly$, we have $y_p\mapq{\mapp{\theta}\rho} = y\mapp{\theta}\rho|_p$.
  Moreover, since $y\mapp{\theta}|_p = y_p$, we have $y\mapp{\theta}\rho|_p = y_p\rho$.
  So $y_p\mapq{\mapp{\theta}\rho} = y_p\rho$.

  Let $y$ be a variable not introduced by $\mapp{\mapp{\theta}\rho}$.
  It remains to show that $y\mapq{\mapp{\theta}\rho} = y\rho$.
  By Lemma~\ref{lem:G:mapp-mapp},
  the variables introduced by $\mapp{\theta}$
  are the same as the variables introduced by $\mapp{\mapp{\theta}\rho}$.
  So $y$ is not introduced by $\mapp{\theta}$ either.
  So, by definition of $\mapqonly$, we have $y\mapq{\mapp{\theta}\rho} = y\mapq{\theta}$, and
  with the assumption of this lemma that $y\rho = y\mapq{\theta}$,
  we conclude that $y\mapq{\mapp{\theta}\rho} = y\rho$.  
\end{proof}

\begin{lem}\label{lem:G:mapp-most-general}
  Let $\theta$ be a grounding substitution.
  For each variable $y \in \VVH$,
  $y\mapp{\theta}$ is the most general term $t$ (unique up to renaming of variables) with the following properties:
  \begin{enumerate}[label=\arabic*.,ref=\arabic*]
    \item \label{mapp-most-gen-one} there exists a substitution $\rho$ such that $t\rho = y\theta$;
    \item \label{mapp-most-gen-two} $t$ contains no type variables and no functional variables.
  \end{enumerate}
\end{lem}
\begin{proof}
  Let $y\in \VVH$.
  By Lemma~\ref{lem:G:mapp-comp-mapq}, $y\mapp{\theta}$ satisfies property \ref{mapp-most-gen-one},
  and by definition of $\mapponly$, it satisfies property \ref{mapp-most-gen-two}.

  To show that $y\mapp{\theta}$ is the most general such term,
  let $s$ be any term satisfying properties \ref{mapp-most-gen-one} and \ref{mapp-most-gen-two},
  and let $\sigma$ be a substitution such that $s\sigma = y\theta$.
  We must show there exists a substitution $\pi$ such that $y\mapp{\theta}\pi = s$.

  Since $s$ contains no type variables and no functional variables,
  it is easy to see from the definition of orange subterms
  that for any orange position $p$ of $s\sigma$, either $p$
  is also an orange position of $s$ or there exists a proper prefix $q$ of $p$ such that $q$ is an orange position of $s$
  and $s|_q$ is a nonfunctional variable.
  If there exists such a prefix $q$, then $q$ must be a nonfunctional yellow position of $s\sigma$
  since $\sigma$ cannot introduce free De Bruijn indices at that position,
  and thus $p$ cannot be an outermost nonfunctional yellow subterm of $s\sigma$.
  From these observations, we conclude that any outermost nonfunctional yellow position $p$ of $s\sigma$
  must be an orange position of $s$.
  In fact, since substituting
  nonfunctional variables
  cannot eliminate De Bruijn indices,
  any outermost nonfunctional yellow position of $s\sigma$
  must be a yellow position of $s$.

  Let $\pi$ map each variable $y_p$ introduced by $\mapp{\theta}$ 
  to the corresponding term in $s$ at position $p$.
  This term exists because $p$ is, by definition of $\mapp{\theta}$,
  an outermost nonfunctional yellow position of $y\theta = s\sigma$
  and thus, by the above, a yellow position of $s$.
  Then $y\mapp{\theta}\pi = s$ by construction,
  showing that $y\mapp{\theta}$ is indeed most general.
\end{proof}

\begin{lem}\label{lem:G:mapp-comp-subst}
  Let $\sigma$ be a substitution, and let $\zeta$ be a grounding substitution.
  Then there exists a substitution $\pi$ such that
  $\sigma\mapp{\zeta} = \mapp{\sigma\zeta}\pi$
  and 
  $\mapq{\sigma\zeta} = \pi\mapq{\zeta}$.
\end{lem}
\begin{proof}
  Let $x\in \VVH$.
  By Lemma~\ref{lem:G:mapp-comp-mapq},
  $x\sigma\mapp{\zeta}\mapq{\zeta} = x\sigma\zeta$.
  Moreover, $x\sigma\mapp{\zeta}$ contains only nonfunctional variables.
  By Lemma~\ref{lem:G:mapp-most-general}, since $x\mapp{\sigma\zeta}$ is the most general term
  with these properties, there must exist a substitution $\pi$ such that
  $x\sigma\mapp{\zeta} = x\mapp{\sigma\zeta}\pi$.
  Since the variables in $x_1\mapp{\sigma\zeta}$ and $x_2\mapp{\sigma\zeta}$
  are disjoint for $x_1\ne x_2$,
  we can construct a single substitution $\pi$ that satisfies $\sigma\mapp{\zeta} = \mapp{\sigma\zeta}\pi$,
  proving the first part of the lemma.
  
  For this construction of $\pi$, only the values of $\pi$ for
  variables introduced by $\mapp{\sigma\zeta}$ are relevant.
  Thus, we can define $y\pi = y\mapq{\sigma\zeta}$ for all other variables $y$.
  Then $y\mapq{\sigma\zeta} = y\pi\mapq{\zeta}$ for all variables $y$ not introduced by $\mapp{\sigma\zeta}$.
  
  Finally, let $y_p$ be a variable introduced by $\mapp{\sigma\zeta}$.
  By Lemma~\ref{lem:G:mapp-comp-mapq} and the above,
  $\sigma\zeta = \sigma\mapp{\zeta}\mapq{\zeta} = \mapp{\sigma\zeta}\pi\mapq{\zeta}$.
  By Lemma~\ref{lem:G:mapq-mapp},
  $y_p\mapq{\sigma\zeta} = y_p\pi\mapq{\zeta}$
  for all variables $y_p$ introduced by $\mapp{\sigma\zeta}$,
  completing the proof of the second part of the lemma.
\end{proof}

The redundancy notions of the $\levelH$ level use the map $\mapFonly$,
defined in Section \ref{ssec:redundancy}.
It is closely related to the maps defined above.

\begin{lem}\label{lem:fo-eq-tfip}%
For all clauses $C$ and grounding substitutions $\theta$,
\[\mapF{C\theta} = \mapT{\mapF{\mapI{\mapP{C\closure\theta}}}}\]
\end{lem}
\begin{proof}
All of the maps $\mapTonly$, $\mapFonly$, $\mapIonly$, and $\mapPonly$
map each literal and each side of a literal individually.
So we can focus on one side $s$ of some literal in $C$,
and we must show that
\[\tag{$*$}\mapF{s\theta} = \mapF{\mapIonly_{\mapq{\theta}}(s\mapp{\theta})} \mapF{\mapI{\mapq{\theta}}}\]

By Lemma~\ref{lem:PG:mapI-mapF-fol},
\[\mapF{t} = \mapF{\mapI{t}}\]
for all ground terms $t \in \TT_\mathrm{ground}(\SigmaH)$.
By Lemma~\ref{lem:PG:mapI-subst},
\[\mapIonly_{\rho}(t)\mapI{\rho} = \mapI{t\rho}\]
for all $t\in\termsPG$ and grounding substitutions $\rho$.
Also, by Lemma~\ref{lem:IPG:mapF-subst},
\[\mapF{t}\mapF{\rho} = \mapF{t\rho}\]
for all $t\in\termsIPG$ and grounding substitutions $\rho$.
From the last three equations, we obtain that
\[\mapF{t\rho} = \mapF{\mapIonly_{\rho}(t)} \mapF{\mapI{\rho}}\]
for all $t\in\termsPG$ and grounding substitutions $\rho$.

Using the term $s$ and the substitution $\theta$ introduced at the beginning of this proof,
take $t$ to be $s\mapp{\theta}$ and $\rho$ to be $\mapq{\theta}$.
We then have
\[\mapF{s\mapp{\theta}\mapq{\theta}} = \mapF{\mapIonly_{\mapq{\theta}}(s\mapp{\theta})} \mapF{\mapI{\mapq{\theta}}}\]
By Lemma~\ref{lem:G:mapp-comp-mapq}, this implies ($*$).
\end{proof}

\subsubsection{Grounding}
The terms of level $\levelH$ are $\TT(\SigmaH)$.
Its clauses $\clausesH$ are constrained clauses over these terms.
We define the function $\mapGonly: \clausesH \to \clausesG$ by
\[ \mapG{C\slimfull{}{\constraint{S}}} =
\{C\closure\theta\mid \theta\text{ is grounding and }S\theta\text{ is true}\}
\]
for each \slimfull{}{constrained} clause $C\slimfull{}{\constraint{S}} \in \clausesH$.
\end{full}

\subsection{Calculi}

In this section, we define the calculi
\begin{slim}$\IGInf$ and $\PGInf$,\end{slim}
\begin{full}$\PFInf$, $\IPGInf$, and $\PGInf$,\end{full}
for the respective levels 
\begin{slim}$\levelIPG$ and $\levelG$. Both of these calculi are\end{slim}%
\begin{full}$\levelPF$, $\levelIPG$, and $\levelPG$. Each of these calculi is\end{full}
parameterized by
a relation $\succ$ on 
\slimfull{ground terms and ground clauses}{ground terms, ground clauses, and closures}
and by a selection function $\mathit{sel}$.
\begin{full}
Based on these parameters, we define
the notion of eligibility.
\end{full}
The specific requirements on $\succ$ depend on the calculus and are given
in the corresponding subsection below.
\begin{full}
For each of the levels, we define selection functions and
the notion of eligibility as follows:
\begin{defi}[Selection Function]\label{def:PF:selection-function}
  For each level $X \in \{\levelPF, \levelIPG, \levelPG\}$,
  we define
  a selection function $\mathit{sel}$ to be a function mapping each closure $C\closure\theta \in \CC_X$ to
  a subset of $C$'s literals. We call those literals \emph{selected}.
  Only negative literals and
  literals of the form $t \ceq \ifalse$ may be selected.
\end{defi}
\end{full}

\begin{full}

\begin{defi}[Eligibility in Closures]\label{def:closure-eligible}
  Let $X \in \{\levelPF, \levelIPG, \levelPG\}$.
  Let $C\closure\theta \in \CC_X$.
  Given a relation $\succ$
  and a selection function, 
  a literal $L\in C$ is (\emph{strictly}) \emph{eligible} in $C\closure\theta$ if it is
  selected in $C\closure\theta$
  or there are no selected literals in $C\closure\theta$
  and $L\theta$ is (strictly) maximal in $C\theta$.
  A position
  $L.s.p$
  of a closure $C\closure\theta$ is \emph{eligible}
  if the literal $L$ is
  of the form $s \doteq t$ with $s\theta \succ t\theta$
  and $L$ is either
  negative and
  eligible
  or positive and strictly eligible.
\end{defi}

For inferences, we follow the same conventions as in Definition~\ref{def:inference},
but our inference rules operate on closures instead of constrained clauses.

\subsubsection{First-Order Levels}
The calculus $\PFInf^{\succ, \mathit{sel}}$ is parameterized by a relation $\succ$
and a selection function $\mathit{sel}$.
\end{full}%
\begin{slim}
For the $\levelF$ level, we use the calculus $\FInf^{\succ, \mathit{sel}}$
introduced in Section~\ref{ssec:redundancy}.
\end{slim}
We require that $\succ$ is an admissible term order for $\PFInf$ in the following sense:
\begin{defi}\label{def:PF:admissible-term-order}
Let $\succ$ be a relation on 
\begin{slim}ground terms and ground clauses.\end{slim}
\begin{full}ground terms, ground clauses, and closures.\end{full}
Such a relation $\succ$ is an \emph{admissible term order for $\PFInf$}
if it fulfills the following properties:
\begin{enumerate}[label=$(\text{O\arabic*})_{\levelPF}$,leftmargin=4em]
    \item the relation $\succ$ on ground terms is a well-founded total order\label{cond:PF:order:total};
    \item ground compatibility with contexts: if $s' \succ s$, then $\subterm{s'}{t} \succ \subterm{s}{t}$;\label{cond:PF:order:comp-with-contexts}
    \item ground subterm property: $\subterm{t}{s} \succ s$
    for ground terms $s$ and $t$;\label{cond:PF:order:subterm}
    \item $u \succ \ifalse \succ \itrue$ for all ground terms $u \notin \{\itrue, \ifalse\}$;\label{cond:PF:order:t-f-minimal}
    \item $\mapF{u} \succ \mapF{u\>\diff^{\tau,\upsilon}_{s,t}}$ for all $s,t,u \oftype \tau \fun \upsilon \in \TT_\mathrm{ground}(\SigmaI)$;\label{cond:PF:order:ext}
    \item the relation $\succ$ on ground clauses is the standard extension of $\succ$ on ground terms via multisets \cite[\Section~2.4]{bachmair-ganzinger-1994};\label{cond:PF:order:clause-extension}
    \begin{full}
    \item for closures $C\closure\theta$ and $D\closure\rho$, we have $C\closure\theta \succ D\closure\rho$ if and only if $C\theta \succ D\rho$.\label{cond:PF:order:closures}
    \end{full}
\end{enumerate}
\end{defi}
\begin{slim}
\begin{rem}\label{rem:FInf-SigmaI-SigmaH-equivalence}
By Lemma~\ref{lem:PG:mapI-mapF-fol} and because $\mapIonly$ is a bijection,
the rules $\FArgCong$, $\FExt$, and $\FDiff$, can equivalently be described by using
$s,s',u,w$ from $\TT_\mathrm{ground}(\SigmaI)$ instead of $\TT_\mathrm{ground}(\SigmaH)$
and replacing $\diff\typeargs{\tau,\upsilon}(u,w)$ with $\diff^{\tau,\upsilon}_{u,w}$.
\end{rem}
\end{slim}

\begin{full}
We use the notion of green subterms in first-order terms introduced
in Section~\ref{ssec:redundancy}.
We define $x(\rho\cup\theta)$ as $x\rho$ if
$x$ occurs in the left premise
and as $x\theta$ otherwise.

\begin{gather*}
\begin{aligned}
  &
  \namedinference{\PFSup}
    {\overbrace{(D' \llor t \ceq t')}^{D}\closure\>\rho
    \quad
    \greensubterm{C}{u}\closure\>\theta}
    {(D' \llor \greensubterm{C}{t'})\closure(\rho \cup \theta)}\quad
  &&\quad
  \namedinference{\PFEqRes}
    {\overbrace{(C' \llor u \cneq u')}^{C}\closure\>\theta}
    {C'\closure\>\theta}
\end{aligned}\\[\jot]
\begin{aligned}
  &
  \namedinference{\PFEqFact}
    {\overbrace{(C' \llor u' \ceq v' \llor u \ceq v)}^{C}\closure\>\theta}
    {(C' \llor v \cneq v' \llor u \ceq v')\closure\>\theta}
  \;\,
  ~~\quad
  \namedinference{\PFClausify}
    {(C' \llor s \ceq t) \closure\>\theta}
    {(C' \llor D) \closure\>\theta}
\end{aligned}\\[\jot]
\begin{aligned}
  \namedinference{\PFBoolHoist}
    {\greensubterm{C}{u}\closure\>\theta}
    {(\greensubterm{C}{\ifalse} \llor u \ceq \itrue)\closure\>\theta}
    \;\,
    ~~\quad
    \namedinference{\PFLoobHoist}
      {\greensubterm{C}{u}\closure\>\theta}
      {(\greensubterm{C}{\itrue} \llor u \ceq \ifalse)\closure\>\theta}
\end{aligned}\\[\jot]
\begin{aligned}
  \namedinference{\PFFalseElim}
    {\overbrace{(C' \llor s \ceq t)}^{C}\closure\>\theta}
    {C'\closure\>\theta}
\end{aligned}\\[\jot]
\begin{aligned}
  \namedinference{\PFArgCong}
    {\overbrace{(C' \llor \mapF{s} \ceq \mapF{s'})}^{C}\closure\>\theta}
    {C'\llor \mapF{s\>\diff^{\tau,\upsilon}_{u,w}} \ceq \mapF{s'\>\diff^{\tau,\upsilon}_{u,w}}\closure\>\theta}
    \;\,
\end{aligned}\\[\jot]
\begin{aligned}
    \namedinference{\PFExt}
      {\greensubterm{C}{\mapF{u}}\closure\>\theta}
      {\greensubterm{C}{\mapF{w}}\llor \mapF{u\>\diff^{\tau,\upsilon}_{u\theta,w\rho}} \cneq \mapF{w\>\diff^{\tau,\upsilon}_{u\theta,w\rho}}\closure\>\rho}
\end{aligned}\\[\jot]
\begin{aligned}
    \namedinference{\PFDiff}
      {}
      {\mapF{u\>\diff^{\tau,\upsilon}_{u\theta,w\theta}} \cneq \mapF{w\>\diff^{\tau,\upsilon}_{u\theta,w\theta}} \llor \mapF{u\>s} \ceq \mapF{w\>s}\closure\>\theta}
\end{aligned}
\end{gather*}
Side conditions for \PFSup:
  \begin{enumerate}[label=\arabic*.,ref=\arabic*]
  \item $t\rho =u\theta$;
  \item $u$ is not a variable;
  \item $u$ is nonfunctional;
  \item $t\rho \succ t'\rho$;
  \item $D\rho \prec C[u]\theta$;
  \item the position of $u$ is eligible in $C\closure\theta$;
  \item $t \ceq t'$ is strictly eligible in $D\closure\rho$;
  \item if $t'\rho$ is Boolean, then $t'\rho = \itrue$.
  \end{enumerate}
Side conditions for \PFEqRes:
  \begin{enumerate}[label=\arabic*.,ref=\arabic*]
  \item $u\theta = u'\theta$;
  \item $u \cneq u'$ is eligible in $C\closure\theta$.
  \end{enumerate}
Side conditions for \PFEqFact:
  \begin{enumerate}[label=\arabic*.,ref=\arabic*]
  \item $u\theta = u'\theta$;
  \item $u\ceq v\closure\theta$ is maximal in $C\closure\theta$;
  \item there are no selected literals in $C\closure\theta$;
  \item $u\theta \succ v\theta$,
  \end{enumerate}
Side conditions for \PFClausify:
  \begin{enumerate}[label=\arabic*.,ref=\arabic*]
  \item $s \ceq t$ is strictly eligible in $(C' \llor s \ceq t) \closure\>\theta$;
  \item The triple ($s, t\theta, D)$ has one of the following forms, where $\tau$ is an arbitrary type and $u$, $v$ are arbitrary terms:%
  \begin{align*}
  &(u \iand v,\ \itrue,\ u \ceq \itrue)&
  &(u \iand v,\ \itrue,\ v \ceq \itrue)&
  &(u \iand v,\ \ifalse,\ u \ceq \ifalse \llor v \ceq \ifalse)\\
  &(u \ior v,\ \itrue,\ u \ceq \itrue \llor v \ceq \itrue)&
  &(u \ior v,\ \ifalse,\ u \ceq \ifalse)&
  &(u \ior v,\ \ifalse,\ v \ceq \ifalse)\\
  &(u \iimplies v,\ \itrue,\ u \ceq \ifalse \llor v \ceq \itrue)&
  &(u \iimplies v,\ \ifalse,\ u \ceq \itrue)&
  &(u \iimplies v,\ \ifalse,\ v \ceq \ifalse)\\
  &(u \ieq^\tau v,\ \itrue,\ u \ceq v)&
  &(u \ieq^\tau v,\ \ifalse,\ u \cneq v)\\
  &(u \ineq^\tau v,\ \itrue,\ u \cneq v)&
  &(u \ineq^\tau v,\ \ifalse,\ u \ceq v)\\
  &(\inot u,\ \itrue,\ u \ceq \ifalse)&
  &(\inot u,\ \ifalse,\ u \ceq \itrue)
  \end{align*}
  \end{enumerate}
Side conditions for \PFBoolHoist and \PFLoobHoist:
  \begin{enumerate}[label=\arabic*.,ref=\arabic*]
  \item $u$ is of Boolean type
  \item $u$ is not a variable and is neither $\itrue$ nor $\ifalse$;
  \item the position of $u$ is eligible in $C\closure\theta$;
  \item the occurrence of $u$ is not in a literal $L$ with $L\theta = (u\theta \ceq \ifalse)$ or $L\theta = (u\theta \ceq \itrue)$.
  \end{enumerate}
Side conditions for \PFFalseElim:
  \begin{enumerate}[label=\arabic*.,ref=\arabic*]
  \item $(s \ceq t)\theta = \ifalse \ceq \itrue$;
  \item $s \ceq t$ is strictly eligible in $C\closure\theta$.
  \end{enumerate}
Side conditions for \PFArgCong:
  \begin{enumerate}[label=\arabic*.,ref=\arabic*]
  \item $s$ is of type $\tau \fun \upsilon$;
  \item $u,w$ are ground terms of type $\tau \fun \upsilon$;
  \item $\mapF{s} \ceq \mapF{s'}$ is strictly eligible in $C\closure\theta$.
  \end{enumerate}
Side conditions for \infname{PFExt}:
  \begin{enumerate}[label=\arabic*.,ref=\arabic*]
  \item the position of $\mapF{u}$ is eligible in $C\closure\theta$;
  \item the type of $u$ is $\tau \fun \upsilon$;
  \item $w\in\termsIPG$ is a term of type $\tau \fun \upsilon$ whose nonfunctional yellow subterms are different variables
  and the variables in $\mapF{w}$ do not occur in $\subterm{C}{\mapF{u}}$.
  \item $u\theta \succ w\rho$;
  \item $\rho$ is a grounding substitution that coincides with $\theta$ on all variables in $\subterm{C}{\mapF{u}}$.
  \end{enumerate}
Side conditions for \PFDiff:
  \begin{enumerate}[label=\arabic*.,ref=\arabic*]
  \item $\tau$ and $\upsilon$ are ground types;
  \item $u,w,s\in \termsIPG$ are terms whose nonfunctional yellow subterms are different fresh variables;
  \item $\theta$ is a grounding substitution.
  \end{enumerate}
\end{full}

\begin{slim}\subsubsection{Indexed Ground Higher-Order Level}The \end{slim}%
\begin{full}\subsubsection{Indexed Partly Substituted Ground Higher-Order Level}The \end{full}
calculus $\IPGInf^{\succ, \mathit{sel}}$ is parameterized
by a relation $\succ$ and a selection function $\mathit{sel}$.
We require that $\succ$ is an admissible term order for $\IPGInf$ in the following sense:
\begin{defi}\label{def:IPG:admissible-term-order}
Let $\succ$ be a relation
on $\TT_\mathrm{ground}(\SigmaI)$,
\begin{slim}and on clauses over $\TT_\mathrm{ground}(\SigmaI)$.\end{slim}
\begin{full}on clauses over $\TT_\mathrm{ground}(\SigmaI)$, and on closures $\clausesIPG$.\end{full}
Such a relation $\succ$
is \emph{an admissible term order for $\IPGInf$}
if it fulfills the following properties:
\begin{enumerate}[label=$(\text{O\arabic*})_\levelIPG$,leftmargin=4em]
  \item the relation $\succ$ on ground terms is a well-founded total order\label{cond:IPG:order:total};
  \item ground compatibility with yellow contexts:\enskip $s' \succ s$ implies
  $\orangesubterm{t}{s'} \succ \orangesubterm{t}{s}$
  for ground terms $s$, $s'$, and $t$;\label{cond:IPG:order:comp-with-contexts}
  \item ground yellow subterm property:\enskip $\orangesubterm{t}{s} \succeq s$
  for ground terms $s$ and $t$;\label{cond:IPG:order:subterm}
  \item $u \succ \ifalse \succ \itrue$ for all ground terms $u \notin \{\itrue, \ifalse\}$;\label{cond:IPG:order:t-f-minimal}
  \item $u \succ u\>\diff^{\tau,\upsilon}_{s,t}$ for all ground terms $s,t,u \oftype \tau \fun \upsilon$.\label{cond:IPG:order:ext}
  \item the relation $\succ$ on ground clauses is the standard extension of $\succ$ on ground terms via multisets \cite[\Section~2.4]{bachmair-ganzinger-1994};\label{cond:IPG:order:clause-extension}
  \begin{full}
  \item for closures $C\closure\theta$ and $D\closure\rho$, we have $C\closure\theta \succ D\closure\rho$ if and only if $C\theta \succ D\rho$.\label{cond:IPG:order:closures}
  \end{full}
\end{enumerate}
\end{defi}

The rules of $\IPGInf^{\succ,\mathit{sel}}$ (abbreviated $\IPGInf$) are the following.
\begin{full}
We assume that for the binary inference $\IPGSup$, the
premises do not have any variables in common,
and we define $x(\rho\cup\theta)$ as $x\rho$ if
$x$ occurs in the left premise
and as $x\theta$ otherwise.

\[\namedinference{\IPGSup}
{\overbrace{D' \llor { t \ceq t'}}^{\vphantom{\cdot}\smash{D}} \closure\> \rho \hypsep
 \greensubterm{C}{u}\closure\theta}
{D' \llor \greensubterm{C}{t'}\>\closure(\rho\cup\theta)}\]
with the following side conditions:
\begin{enumerate}[label=\arabic*.,ref=\arabic*]
  \item $t\rho = u\theta$;
  \item $u$ is not a variable;
  \item $u$ is nonfunctional;
  \item $t\rho \succ t'\negvthinspace\rho$;
  \item $D\rho\prec \greensubterm{C}{u}\theta$;
  \item the position of $u$ is eligible in $C\closure\theta$;
  \item $t \ceq t'$ is strictly eligible in $D\closure\rho$;
  \item if $t'\rho$ is Boolean, then $t'\rho = \itrue$.
\end{enumerate}
\begin{align*}
 &\namedinference{\IPGEqRes}
{\overbrace{C' \llor {u \cneq u'}}^{\vphantom{\cdot}\smash{C}}\>\closure\>\theta}
{C'\closure\theta}
&&\namedinference{\IPGEqFact}
{\overbrace{C' \llor {u'} \ceq v' \llor {u} \ceq v}^{\vphantom{\cdot}\smash{C}}\>\closure\>\theta}
{C' \llor v \cneq v' \llor u \ceq v'\closure\theta}
\end{align*}
Side conditions for \IPGEqRes{}:
\begin{enumerate}[label=\arabic*.,ref=\arabic*]
\item $u\theta = u'\theta$;
\item $u \cneq u'$ is eligible in $C\closure\theta$.
\end{enumerate}
Side conditions for \IPGEqFact{}:
\begin{enumerate}[label=\arabic*.,ref=\arabic*]
\item $u\theta = u'\theta$;
\item $u \ceq v\closure\theta$ is maximal in $C\closure\theta$;
\item there are no selected literals in $C\closure\theta$;
\item $u\theta \succ v\theta$.
\end{enumerate}

\begin{align*}
  \namedinference{\IPGClausify}{C' \llor s \ceq t  \closure \theta}
  {C' \llor D  \closure \theta}
\end{align*}
with the following side conditions:
\begin{enumerate}[label=\arabic*.,ref=\arabic*]
\item $s \ceq t$ is strictly eligible in $C' \llor s \ceq t\closure\theta$;%
\item The triple ($s, t\theta, D)$ has one of the following forms, where $\tau$ is an arbitrary type and $u$, $v$ are arbitrary terms:%
\begin{align*}
&(u \iand v,\ \itrue,\ u \ceq \itrue)&
&(u \iand v,\ \itrue,\ v \ceq \itrue)&
&(u \iand v,\ \ifalse,\ u \ceq \ifalse \llor v \ceq \ifalse)\\
&(u \ior v,\ \itrue,\ u \ceq \itrue \llor v \ceq \itrue)&
&(u \ior v,\ \ifalse,\ u \ceq \ifalse)&
&(u \ior v,\ \ifalse,\ v \ceq \ifalse)\\
&(u \iimplies v,\ \itrue,\ u \ceq \ifalse \llor v \ceq \itrue)&
&(u \iimplies v,\ \ifalse,\ u \ceq \itrue)&
&(u \iimplies v,\ \ifalse,\ v \ceq \ifalse)\\
&(u \ieq^\tau v,\ \itrue,\ u \ceq v)&
&(u \ieq^\tau v,\ \ifalse,\ u \cneq v)\\
&(u \ineq^\tau v,\ \itrue,\ u \cneq v)&
&(u \ineq^\tau v,\ \ifalse,\ u \ceq v)\\
&(\inot u,\ \itrue,\ u \ceq \ifalse)&
&(\inot u,\ \ifalse,\ u \ceq \itrue)
\end{align*}
\end{enumerate}

\begin{align*}
  \namedinference{\IPGBoolHoist}{\greensubterm{C}{u} \closure\theta}
  {\greensubterm{C}{\ifalse} \llor u \ceq \itrue\closure\theta}
  \quad
  \namedinference{\IPGLoobHoist}{\greensubterm{C}{u} \closure\theta}
  {\greensubterm{C}{\itrue} \llor u \ceq \ifalse\closure\theta}
\end{align*}
each with the following side conditions:
\begin{enumerate}[label=\arabic*.,ref=\arabic*]
  \item $u$ is of Boolean type;%
  \item $u$ is not a variable and is neither $\itrue$ nor $\ifalse$;%
  \item the position of $u$ is eligible in $C\closure\theta$;%
  \item the occurrence of $u$ is not in a literal $L$ with $L\theta = (u\theta \ceq \ifalse)$ or $L\theta = (u\theta \ceq \itrue)$.
\end{enumerate}

\begin{align*}
  \namedinference{\IPGFalseElim}{\overbrace{C' \llor s \ceq t}^C \closure\theta}
  {C'\closure\theta}
\end{align*}
with the following side conditions:
\begin{enumerate}[label=\arabic*.,ref=\arabic*]
  \item $(s\ceq t)\theta = \ifalse \ceq \itrue$;
  \item $s\ceq t$ is strictly eligible in $C \closure\theta$.
\end{enumerate}

\[\namedinference{IPGArgCong}
{\overbrace{C' \llor s \eq s'}^C\closure\theta}
{C' \llor s\>\diff^{\tau,\upsilon}_{u,w} \eq s'\>\diff^{\tau,\upsilon}_{u,w}\closure\theta}\]
with the following side conditions:
\begin{enumerate}[label=\arabic*.,ref=\arabic*]
  \item $s$ is of type $\tau \fun \upsilon$;
  \item $u,w$ are ground terms of type $\tau \fun \upsilon$;
  \item $s \eq s'$ is strictly eligible in $C\closure\theta$.
\end{enumerate}

\begin{align*}
  \namedinference{\IPGExt}{\greensubterm{C}{u} \closure\theta}
  {\greensubterm{C}{w} \llor u\>\diff^{\tau,\upsilon}_{u\theta,w\rho}\noteq w\>\diff^{\tau,\upsilon}_{u\theta,w\rho} \closure\rho}
\end{align*}
with the following side conditions:
\begin{enumerate}[label=\arabic*.,ref=\arabic*]
  \item the position of $u$ is eligible in $\greensubterm{C}{u}\closure\theta$;
  \item the type of $u$ is $\tau\to\upsilon$;
  \item $w\in \termsIPG$ is a term whose nonfunctional yellow subterms are different fresh variables;
  \item $u\theta \succ w\rho$;
  \item $\rho$ is a grounding substitution that coincides with $\theta$ on all variables in $\greensubterm{C}{u}$.
\end{enumerate}

\begin{align*}
  \namedinference{\IPGDiff}{}
  {u\>\diff^{\tau,\upsilon}_{u\theta,w\theta} \noteq u\>\diff^{\tau,\upsilon}_{u\theta,w\theta} \llor u\>s \eq w\>s \closure\theta}
\end{align*}
with the following side conditions:
\begin{enumerate}[label=\arabic*.,ref=\arabic*]
  \item $\tau$ and $\upsilon$ are ground types;
  \item $u,w,s\in \termsIPG$ are terms whose nonfunctional yellow subterms are different variables;
  \item $\theta$ is a grounding substitution.
\end{enumerate}
\end{full}%
\begin{slim}
\[\namedinference{\IGSup}
{\overbrace{D' \llor { t \ceq t'}}^{\vphantom{\cdot}\smash{D}} \hypsep
 \greensubterm{C}{t}}
{D' \llor \greensubterm{C}{t'}}\]
with the following side conditions:
\begin{enumerate}[label=\arabic*.,ref=\arabic*]
  \item $t$ is nonfunctional;
  \item $t \succ t'$;
  \item $D\prec \greensubterm{C}{t}$;
  \item the position of $t$ is eligible in $C$;
  \item $t \ceq t'$ is strictly eligible in $D$;
  \item if $t'$ is Boolean, then $t' = \itrue$.
\end{enumerate}
\begin{align*}
 &\namedinference{\IGEqRes}
{\overbrace{C' \llor {u \cneq u}}^{\vphantom{\cdot}\smash{C}}}
{C'}
&&\namedinference{\IGEqFact}
{\overbrace{C' \llor u \ceq v' \llor u \ceq v}^{\vphantom{\cdot}\smash{C}}}
{C' \llor v \cneq v' \llor u \ceq v'}
\end{align*}
Side conditions for \IGEqRes{}:
\begin{enumerate}[label=\arabic*.,ref=\arabic*]
\item $u \cneq u$ is eligible in $C$.
\end{enumerate}
Side conditions for \IGEqFact{}:
\begin{enumerate}[label=\arabic*.,ref=\arabic*]
\item $u \ceq v$ is maximal in $C$;
\item there are no selected literals in $C$;
\item $u \succ v$.
\end{enumerate}

\begin{align*}
  \namedinference{\IGClausify}{C' \llor s \ceq t}
  {C' \llor D}
\end{align*}
with the following side conditions:
\begin{enumerate}[label=\arabic*.,ref=\arabic*]
\item $s \ceq t$ is strictly eligible in $C' \llor s \ceq t$;%
\item The triple ($s, t, D)$ has one of the following forms, where $\tau$ is an arbitrary type and $u$, $v$ are arbitrary terms:%
\begin{align*}
&(u \iand v,\ \itrue,\ u \ceq \itrue)&
&(u \iand v,\ \itrue,\ v \ceq \itrue)&
&(u \iand v,\ \ifalse,\ u \ceq \ifalse \llor v \ceq \ifalse)\\
&(u \ior v,\ \itrue,\ u \ceq \itrue \llor v \ceq \itrue)&
&(u \ior v,\ \ifalse,\ u \ceq \ifalse)&
&(u \ior v,\ \ifalse,\ v \ceq \ifalse)\\
&(u \iimplies v,\ \itrue,\ u \ceq \ifalse \llor v \ceq \itrue)&
&(u \iimplies v,\ \ifalse,\ u \ceq \itrue)&
&(u \iimplies v,\ \ifalse,\ v \ceq \ifalse)\\
&(u \ieq^\tau v,\ \itrue,\ u \ceq v)&
&(u \ieq^\tau v,\ \ifalse,\ u \cneq v)\\
&(u \ineq^\tau v,\ \itrue,\ u \cneq v)&
&(u \ineq^\tau v,\ \ifalse,\ u \ceq v)\\
&(\inot u,\ \itrue,\ u \ceq \ifalse)&
&(\inot u,\ \ifalse,\ u \ceq \itrue)
\end{align*}
\end{enumerate}

\begin{align*}
  \namedinference{\IGBoolHoist}{\greensubterm{C}{u}}
  {\greensubterm{C}{\ifalse} \llor u \ceq \itrue}
  \quad
  \namedinference{\IGLoobHoist}{\greensubterm{C}{u}}
  {\greensubterm{C}{\itrue} \llor u \ceq \ifalse}
\end{align*}
each with the following side conditions:
\begin{enumerate}[label=\arabic*.,ref=\arabic*]
  \item $u$ is of Boolean type;%
  \item $u$ is neither $\itrue$ nor $\ifalse$;%
  \item the position of $u$ is eligible in $C$;%
  \item the occurrence of $u$ is not in a literal of the form $u \ceq \ifalse$ or $u \ceq \itrue$.
\end{enumerate}

\begin{align*}
  \namedinference{\IGFalseElim}{\overbrace{C' \llor \ifalse \ceq \itrue}^C}
  {C'}
\end{align*}
with the following side conditions:
\begin{enumerate}[label=\arabic*.,ref=\arabic*]
  \item $\ifalse \ceq \itrue$ is strictly eligible in $C$.
\end{enumerate}

\[\namedinference{IGArgCong}
{\overbrace{C' \llor s \eq s'}^C}
{C' \llor s\>\diff^{\tau,\upsilon}_{u,w} \eq s'\>\diff^{\tau,\upsilon}_{u,w}}\]
with the following side conditions:
\begin{enumerate}[label=\arabic*.,ref=\arabic*]
  \item $s$ is of type $\tau \fun \upsilon$;
  \item $u,w$ are ground terms of type $\tau \fun \upsilon$;
  \item $s \eq s'$ is strictly eligible in $C$.
\end{enumerate}

\begin{align*}
  \namedinference{\IGExt}{\greensubterm{C}{u}}
  {\greensubterm{C}{w} \llor u\>\diff^{\tau,\upsilon}_{u,w}\noteq w\>\diff^{\tau,\upsilon}_{u,w}}
\end{align*}
with the following side conditions:
\begin{enumerate}[label=\arabic*.,ref=\arabic*]
  \item the position of $u$ is eligible in $\greensubterm{C}{u}$;
  \item the type of $u$ is $\tau\to\upsilon$;
  \item $w$ is a ground term of type $\tau\to\upsilon$;
  \item $u \succ w$;
\end{enumerate}

\begin{align*}
  \namedinference{\IGDiff}{}
  {u\>\diff^{\tau,\upsilon}_{u,w} \noteq u\>\diff^{\tau,\upsilon}_{u,w} \llor u\>s \eq w\>s}
\end{align*}
with the following side conditions:
\begin{enumerate}[label=\arabic*.,ref=\arabic*]
  \item $\tau$ and $\upsilon$ are ground types;
  \item $u,w$ are ground terms of type $\tau \fun \upsilon$;
  \item $s$ is a ground term of type $\tau$.
\end{enumerate}
\end{slim}

\begin{full}\subsubsection{Partly Substituted Ground Higher-Order Level}Like \end{full}%
\begin{slim}\subsubsection{Ground Higher-Order Level}Like \end{slim}%
on the other levels,
the calculus $\PGInf$ is parameterized
by a relation $\succ$ and a selection function $\mathit{sel}$.
\begin{defi}\label{def:PG:admissible-term-order}
Let $\succ$ be a relation on $\TT_\mathrm{ground}(\SigmaH)$%
\begin{full}, on clauses over $\TT_\mathrm{ground}(\SigmaH)$, and on $\clausesPG$.\end{full}%
\begin{slim}and on clauses over $\TT_\mathrm{ground}(\SigmaH)$.\end{slim}
Such a relation $\succ$ is \emph{an admissible term order for $\PGInf$}
if it fulfills the following properties:
\begin{enumerate}[label=$(\text{O\arabic*})_\levelPG$,leftmargin=4em]
  \item the relation $\succ$ on ground terms is a well-founded total order\label{cond:PG:order:total};
  \item ground compatibility with yellow contexts:\enskip $s' \succ s$ implies
  $\orangesubterm{t}{s'} \succ \orangesubterm{t}{s}$
  for ground terms $s$, $s'$, and $t$;\label{cond:PG:order:comp-with-contexts}
  \item ground yellow subterm property:\enskip $\orangesubterm{t}{s} \succeq s$
  for ground terms $s$ and $t$;\label{cond:PG:order:subterm}
  \item $u \succ \ifalse \succ \itrue$ for all ground terms $u \notin \{\itrue, \ifalse\}$;\label{cond:PG:order:t-f-minimal}
  \item $u \succ u\>\diff\typeargs{\tau,\upsilon}(s,t)$ for all ground terms $s,t,u \oftype \tau \fun \upsilon$.\label{cond:PG:order:ext}
  \item the relation $\succ$ on ground clauses is the standard extension of $\succ$ on ground terms via multisets \cite[\Section~2.4]{bachmair-ganzinger-1994};\label{cond:PG:order:clause-extension}
  \begin{full}
  \item for closures $C\closure\theta$ and $D\closure\rho$, we have $C\closure\theta \succ D\closure\rho$ if and only if $C\theta \succ D\rho$.\label{cond:PG:order:closures}
  \end{full}
\end{enumerate}
\end{defi}

The calculus rules of $\PGInf$ are a verbatim copy of those of $\IPGInf$, with the
following exceptions:
\begin{itemize}
  \item $\PGInf$ uses $\SigmaH$ instead of $\SigmaI$ and $\clausesPG$ instead of $\clausesIPG$.
  \item The rules are prefixed by $\infname{\slimfull{}{P}G}$ instead of $\infname{I\slimfull{}{P}G}$.
  \item $\slimfull{\GArgCong}{\PGArgCong}$ uses $\diff\typeargs{\tau,\upsilon}(u,w)$ instead of $\diff^{\tau,\upsilon}_{u,w}$.
  \item $\slimfull{\GExt}{\PGExt}$ uses $\diff\typeargs{\tau,\upsilon}(u,w)$ instead of \slimfull{$\diff^{\tau,\upsilon}_{u,w}$}{$\diff^{\tau,\upsilon}_{u\theta,w\rho}$}.
  \item $\slimfull{\GDiff}{\PGDiff}$ uses $\diff\typeargs{\tau,\upsilon}(u,w)$ instead of \slimfull{$\diff^{\tau,\upsilon}_{u,w}$}{$\diff^{\tau,\upsilon}_{u\theta,w\theta}$}.
\end{itemize}

\subsection{Redundancy Criteria and Saturation}
\label{ssec:redundancy-criteria-and-saturation}

In this subsection, we define redundancy criteria
for the levels 
\begin{full}$\levelPF$, $\levelIPG$, $\levelPG$, and $\levelH$\end{full}
\begin{slim}$\levelF$, $\levelIPG$, and $\levelG$\end{slim}
and show that saturation up to redundancy on one level implies saturation up to
redundancy on the previous level.
We will use these results in Section~\ref{ssec:model-construction} to lift
refutational completeness from level $\levelPF$ to level $\levelH$.

\begin{defi}\label{def:saturation}
  A set $N$ of clauses is called \emph{saturated up to redundancy}
  if every inference with premises in $N$ is redundant \wrt\ $N$.
\end{defi}
\subsubsection{First-Order Level}

In this subsection, let $\succ$ be an admissible term order for 
$\PFInf$ (Definition~\ref{def:PF:admissible-term-order}),
and let $\mathit{\slimfull{}{p}fsel}$ be a selection function on
$\clausesPF$\slimfull{}{ (Definition~\ref{def:PF:selection-function})}.

\begin{full}
We define a notion of variable-irreducibility,
roughly following Nieuwenhuis and Rubio~\cite{nieuwenhuis-rubio-1995}
and Bachmair et al.~\cite{bachmair-et-al-1992} (where it is called ``order-irreducibility''):
\begin{defi}\label{def:PF:irred}
  A closure literal $L\closure \theta \in\clausesPF$ is \emph{variable-irreducible}
\wrt\ a ground term rewrite system $R$ if, for all variables $x$ in $L$,
$x\theta$ is irreducible \wrt\ the rules $s \rewrite t \in R$ with $L\theta \succ s \eq t$
and all Boolean subterms of $x\theta$ are either $\itrue$ or $\ifalse$.
A closure $C \closure \theta \in\clausesPF$ is \emph{variable-irreducible} \wrt\ $R$ if all its literals are variable-irreducible \wrt\ $R$.
Given a set $N$ of closures, we write $\irred_R(N)$ for the set of variable-irreducible closures in $N$ \wrt\ $R$.
\end{defi}
\begin{rem}
The restriction that $L\theta \succ s \eq t$
cannot be replaced by $C\theta \succ s \eq t$
because we need it to ensure that $\irred_R(N)$ is saturated when $N$ is.

Here is an example demonstrating that variable-irreducibilty would not
be closed under inferences if we used the entire clause for comparison:
Let $\cst{e} \succ \cst{d} \succ \cst{c} \succ \cst{b} \succ \cst{a}$ and 
$R = \{\cst{e} \rewrite \cst{b}\}$. Then a notion replacing the restriction $L\theta \succ s \eq t$
by $C\theta \succ s \eq t$
would say that
$x \ceq \cst{a}  \llor  x \ceq \cst{b} \ \constraint{x \equiv \cst{e}}$
and
$\cst{e} \ceq \cst{c}  \llor  \cst{e} \ceq \cst{d}$
are order-irreducible \wrt\ $R$. By \Sup{} (second literal of the first clause into
the second literal of the second clause), we obtain
\[x \ceq \cst{a}  \llor  \cst{e} \ceq \cst{c}  \llor  \cst{b} \ceq \cst{d} \ \constraint{ x \equiv e }\]
Now $x\theta$ in the first literal has become order-reducible by $\cst{e} \rewrite \cst{b}$
since $e \ceq b$ is now smaller than the largest literal $e \ceq c$ of the clause.
\end{rem}

\begin{defi}[Inference Redundancy]\label{def:PF:RedI}
  Given $\iota \in \PFInf$ and $N\subseteq\clausesPF$, let $\iota\in\PFRedI(N)$
  if for all confluent term rewrite systems $R$ oriented by $\succ$
  whose only Boolean normal forms are $\itrue$ and $\ifalse$
  such that $\concl(\iota)$ is variable-irreducible, we have
  $R \cup O \modelsolam \concl(\iota)$,
  where $O = \irred_R(N)$ if $\iota$ is a $\Diff$ inference, and $O = \{E \in \irred_R(N) \mid E \prec \mprem(\iota)\}$ otherwise.
\end{defi}
\end{full}

\begin{slim}
\begin{defi}[Inference Redundancy]\label{def:PF:RedI}
  Given $\iota \in \FInf$ and $N\subseteq\clausesF$, let $\iota\in\FRedI(N)$
  if $\iota$ is a $\Diff$ inference and 
  $N \modelsolam \concl(\iota)$
  or if $\iota$ is not a $\Diff$ inference and
  $\{E \in N \mid E \prec \mprem(\iota)\} \modelsolam \concl(\iota)$.
\end{defi}
\end{slim}

\begin{full}
To connect to the redundancy criteria of the higher levels, we need to establish
a connection to the $\FInf$ inference system defined in Section~\ref{ssec:redundancy}:
\begin{lem}\label{lem:PF:inference-to-F-inference}
  Let $\iota_\levelPF \in \PFInf^{\succ,\mathit{pfsel}}$.
  Let $C_1\closure\theta_1,\dots,C_m\closure\theta_m$ be its premises
  and $C_{m+1}\closure\theta_{m+1}$ its conclusion.
  Then
      \[\inference
      {C_1\theta_1 \cdots C_m\theta_m}
      {C_{m+1}\theta_{m+1}}\]
      is a valid $\FInf^{\succ}$ inference $\iota_\levelF$,
   and the rule names of $\iota_\levelPF$ and  $\iota_\levelF$ correspond up to the prefixes $\infname{PF}$ and $\infname{F}$.
\end{lem}
\begin{proof}
  This is easy to see by comparing the rules of $\PFInf$ and $\FInf$.
  It is crucial that the concepts of eligibility match:
  If a literal or a position is (strictly) eligible in a closure $C\closure\theta \in \clausesPF$
    (according to the $\levelPF$ concept of eligibility),
    then the corresponding literal or position is (strictly) eligible in $C\theta$
    (according to the $\levelF$ concept of eligibility).
\end{proof}
\end{full}

\begin{full}\subsubsection{Indexed Partly Substituted Ground Higher-Order Level}\mbox{}\end{full}%
\begin{slim}\subsubsection{Indexed Ground Higher-Order Level}\mbox{}\end{slim}%
In this subsubsection, let $\succ$ be an admissible term order for $\IPGInf$ (Definition~\ref{def:IPG:admissible-term-order}),
and let $\mathit{i\slimfull{}{p}gsel}$ be a selection function on
\begin{slim}$\clausesIPG$\end{slim}%
\begin{full}$\clausesIPG$ (Definition~\ref{def:PF:selection-function})\end{full}.

To lift the notion of inference redundancy, we need to connect the inference systems
$\PFInf$ and $\IPGInf$ as follows.
\begin{full}
Since the mapping $\mapFonly$ is bijective (Lemma~\slimfull{\ref{lem:fol-bijection}}{\ref{lem:IPG:F-bijection}}),
we can transfer the order $\succ$ from the $\levelIPG$ level to the $\levelPF$ level:
\begin{defi}\label{def:IPG:succ-transfer}
Based on $\succ$,
we define
a relation $\succ_\mapFonly$
on ground terms $\termsF$\slimfull{ and}{,} ground clauses
$\clausesF$\slimfull{}{, and closures $\clausesPF$} as
 $d \mathrel{\succ_\mapFonly} e$ if and only if $\mapFonly^{-1}(d) \succ \mapFonly^{-1}(e)$
 for all terms\slimfull{ or}{,} clauses\slimfull{}{, or closures} $d$ and $e$.
\end{defi}
\end{full}
\begin{lem}\label{lem:IPG:succ-transfer}
  Since $\succ$ is an admissible term order for \slimfull{$\IGInf$}{$\IPGInf$} (Definition~\ref{def:IPG:admissible-term-order}),
  the relation $\succ_\mapFonly$ \begin{slim}defined in Definition~\ref{def:order-mapF}\end{slim} is an admissible term order for $\PFInf$ (Definition~\ref{def:PF:admissible-term-order}).
\end{lem}
\begin{proof}
  This is easy to see, considering
  that $\mapFonly$ is a bijection between $\TT_\mathrm{ground}(\SigmaI)$ and \slimfull{$\termsF$}{$\termsPF^\mathsf{gnd}$}
  (Lemma~\slimfull{\ref{lem:fol-bijection}}{\ref{lem:IPG:F-bijection}}),
  \begin{full}that higher-order yellow subterms and first-order subterms correspond\end{full}%
  \begin{slim}that every first-order subterm corresponds to a higher-order yellow subterm\end{slim}
  by Lemma~\slimfull{\ref{lem:IPG:mapF-green-subterms}}{\ref{lem:IPG:mapF-subterms}},
  \slimfull{and }{}that
  $\mapFonly$ maps each side of each literal individually\slimfull{.}{,}
  \begin{full}and that $\mapF{C}\mapF{\theta} = \mapF{C\theta}$
  for all $C\closure \theta \in\clausesIPG$ (Lemma~\ref{lem:IPG:mapF-subst}).\end{full}
\end{proof}

\begin{full}
Since $\mapFonly$ is bijective, we can transfer the selection function as follows:
\begin{defi}\label{def:IPG:sel-transfer}
  Based on $\mathit{i\slimfull{}{p}gsel}$,
  we define $\mapF{\mathit{i\slimfull{}{p}gsel}}$ as a selection function that selects
  the literals of $C \in \clausesPF$ corresponding to the $\mathit{i\slimfull{}{p}gsel}$-selected literals in
  $\mapFonly^{-1}(C)$.
\end{defi}
\end{full}

\begin{defi}\label{def:IPG:mapF-inference}
We extend $\mapFonly$ to inference rules
by mapping an inference $\iota \in \IPGInf$
to the inference
\[
\inference{\mapF{\prems(\iota)}}{\mapF{\concl(\iota)}}
\]
\end{defi}

\begin{lem}\label{IPG:lem:mapF-inference}
The mapping $\mapFonly$ is a bijection between $\IPGInf^{\succ,\mathit{i\slimfull{}{p}gsel}}$ and 
$\PFInf^{\succ_\mapFonly,\mapF{\mathit{i\slimfull{}{p}gsel}}}$\slimfull{, where $\mapF{\mathit{igsel}}$ is defined in Definition~\ref{def:sel-mapF}.}{.}
\end{lem}
\begin{proof}
This is easy to see by comparing the rules of $\IPGInf$ and $\PFInf$\slimfull{ and considering Remark~\ref{rem:FInf-SigmaI-SigmaH-equivalence}}{}.
It is crucial that the following concepts match:
  \begin{itemize}
    \item \slimfull{Subterms}{Green subterms} on the $\levelPF$ level correspond to green subterms on the $\levelIPG$ level
    by Lemma~\ref{lem:IPG:mapF-green-subterms}.
    \item The term orders correspond (Definition~\slimfull{\ref{def:order-mapF}}{\ref{def:IPG:succ-transfer}}).
    \item The selected literals correspond; i.e.,
    a literal $L$ is selected in a \slimfull{clause $C$}{closure $C\closure\theta$} if and only if
    the literal $\mapF{L}$ is selected in $\mapF{C\slimfull{}{\closure\theta}}$.
    This follows directly from the definition of $\mapF{\mathit{i\slimfull{}{p}gsel}}$ (Definition~\slimfull{\ref{def:sel-mapF}}{\ref{def:IPG:sel-transfer}}).
    \item The concepts of eligibility correspond; i.e.,
    a literal $L$ of a \slimfull{clause $C\in\clausesIPG$}{closure $C\closure\theta\in\clausesIPG$}
    is (strictly) eligible \wrt\ $\succeq$
    if and only if 
    the literal $\mapF{L}$ of the \slimfull{clause $C$}{closure $\mapF{C\closure\theta}$}
    is (strictly) eligible \wrt\ $\succeq_\mapFonly$; and
    a position $L.s.p$ of a \slimfull{clause $C\in\clausesIPG$}{closure $C\closure\theta\in\clausesIPG$} is eligible \wrt\ $\succeq$
    if and only if 
    the position $\mapF{L}.\mapF{s}.\slimfull{p}{q}$ of the \slimfull{clause $\mapF{C}$}{closure $\mapF{C\closure\theta}$}
    is eligible \wrt\ $\succeq_\mapFonly$\slimfull{.}{,}
    \begin{full}where $q$ is the position corresponding to $p$.\end{full}
    This is true because eligibility (Definition~\slimfull{\ref{def:ho-eligible}}{\ref{def:closure-eligible}}) depends only on the selected literals and the term order,
    which correspond as discussed above.\qedhere
  \end{itemize}
\end{proof}

\begin{full}
\begin{defi}\label{def:IPG:irred}
  Given a term rewrite system $R$ on $\termsF$,
  we say that a closure $C\closure\theta\in\clausesIPG$ is variable-irreducible \wrt\ $R$ and $\succ$
  if $\mapF{C\closure\theta}$ is variable-irreducible \wrt\ $R$ and $\succ_\mapFonly$.
  We write $\irred_R(N)$ for all variable-irreducible closures in a set $N\subseteq\clausesIPG$.
\end{defi}
\end{full}

\begin{defi}[Inference Redundancy]\label{def:IPG:RedI}
  Given $\iota \in \IPGInf^{\succ,\mathit{i\slimfull{}{p}gsel}}$ and $N\subseteq\clausesIPG$, let $\iota\in\IPGRedI(N)$ if
  \begin{full}
  for all confluent term rewrite systems $R$ on $\termsF$ oriented by $\succ_\mapFonly$
  whose only Boolean normal forms are $\itrue$ and $\ifalse$
  such that $\concl(\iota)$ is variable-irreducible \wrt\ $R$, we have
  \[
    R \cup O \modelsolam \mapF{\concl(\iota)}
  \]
  where $O = \irred_R(\mapF{N})$ if $\iota$ is a $\IPGDiff$ inference,
  and $O = \{ E \in \irred_R(\mapF{N}) \mid E \prec_\mapFonly \mapF{\mprem(\iota)} \}$ otherwise.
  \end{full}
  \begin{slim}
  $\mapF{\iota} \in \FRedI(\mapF{N})$ (Definition~\ref{def:PF:RedI}) \wrt\ $\succ_\mapFonly$.
  \end{slim}
\end{defi}

\begin{full}
\begin{lem}\label{lem:IPG:inference-to-F-inference}
  Let $\iota_\levelIPG \in \IPGInf^{\succ,\mathit{i\slimfull{}{p}gsel}}$.
  Let $C_1\closure\theta_1,\dots,C_m\closure\theta_m$ be its premises
  and $C_{m+1}\closure\theta_{m+1}$ its conclusion.
  Then
      \[\inference
      {\mapF{C_1\theta_1} \cdots \mapF{C_m\theta_m}}
      {\mapF{C_{m+1}\theta_{m+1}}}\]
      is a valid $\FInf^{\succ_\mapFonly}$ inference $\iota_\levelF$, and
  the rule names of $\iota_\levelIPG$ and  $\iota_\levelF$ correspond up to the prefixes $\infname{IPG}$ and $\infname{F}$.
\end{lem}
\begin{proof}
  By Lemma~\ref{IPG:lem:mapF-inference}, we know that $\mapF{\iota_\levelIPG}$ is a valid $\PFInf^{\succ_\mapFonly,\mapF{\mathit{ipgsel}}}$ inference, and the rule names coincide up to the prefixes $\infname{PF}$ and $\infname{IPG}$.

  Now, applying Lemma~\ref{lem:PF:inference-to-F-inference} to $\mapF{\iota_\levelIPG}$
  and using the fact that $\mapF{C_i}\mapF{\theta_i} = \mapF{C_i\theta_i}$ (Lemma~\ref{lem:IPG:mapF-subst}), we obtain that
  \[\inference
    {\mapF{C_1\theta_1} \quad \cdots \quad \mapF{C_m\theta_m}}
    {\mapF{C_{m+1}\theta_{m+1}}}\]
  is a valid $\FInf^{\succ_\mapFonly}$ inference $\iota_\levelF$, and the rule names of $\mapF{\iota_\levelIPG}$ and $\iota_\levelF$ correspond up to the prefixes $\infname{PF}$ and $\infname{F}$.

  Combining these two results, we conclude that the rule names of $\iota_\levelIPG$ and $\iota_\levelF$ correspond up to the prefixes $\infname{IPG}$ and $\infname{F}$, which completes the proof.
\end{proof}
\end{full}

Using the bijection between $\IPGInf$ and $\PFInf$,
we can show that saturation \wrt\ $\IPGInf$ implies saturation \wrt\ $\PFInf$:
\begin{lem}\label{lem:IPG:saturated}
  Let $N$ be saturated up to redundancy w.r.t.\ $\IPGInf^{\succ,\mathit{i\slimfull{}{p}gsel}}$.
  Then $\mapF{N}$ is saturated up to redundancy w.r.t.\ $\PFInf^{\succ_\mapFonly,\mapF{\mathit{i\slimfull{}{p}gsel}}}$.
\end{lem}
\begin{proof}
  By Lemma~\ref{IPG:lem:mapF-inference}
  \begin{slim} and Definition~\ref{def:IPG:RedI}.\end{slim}%
  \begin{full} because the notions of inference redundancy correspond.\end{full}
\end{proof}

\begin{full}\subsubsection{Partly Substituted Ground Higher-Order Level}\mbox{}\end{full}%
\begin{slim}\subsubsection{Ground Higher-Order Level}\mbox{}\end{slim}%
In this subsubsection,
let $\succ$ be an admissible term order for $\PGInf$ (Definition~\ref{def:PG:admissible-term-order}),
and let $\mathit{\slimfull{}{p}gsel}$ be a selection function on $\clausesPG$\slimfull{}{ (Definition~\ref{def:PF:selection-function})}.

Since mapping $\mapIonly$ is clearly bijective%
\slimfull{,}{ for ground terms and ground clauses,}
we can transfer $\succ$ from the $\levelPG$ level to the $\levelIPG$ level
as follows:
\begin{defi}\label{def:PG:succ-transfer}
Let $\succ$ be a relation on
\begin{slim}$\TT_\mathrm{ground}(\SigmaH)$ and on clauses over $\TT_\mathrm{ground}(\SigmaH)$.\end{slim}%
\begin{full}$\TT_\mathrm{ground}(\SigmaH)$, on clauses over $\TT_\mathrm{ground}(\SigmaH)$, and on closures $\clausesPG$.\end{full}
We define
a relation $\succ_\mapIonly$
on $\TT_\mathrm{ground}(\SigmaI)$ and
on clauses over $\TT_\mathrm{ground}(\SigmaI)$
as
$d \mathrel{\succ_\mapIonly} e$ if and only if $\mapIonly^{-1}(d) \succ \mapIonly^{-1}(e)$
for all terms or clauses $d$ and $e$.
\begin{full}For closures $C\closure\theta, D\closure\rho \in \clausesIPG$, we
define $C\closure\theta \mathrel{\succ_\mapIonly} D\closure\rho$ if
$C\theta \mathrel{\succ_\mapIonly} D\rho$.\end{full}
\end{defi}

\begin{lem}\label{lem:PG:succ-transfer}
  Since $\succ$ is an admissible term order for $\PGInf$ (Definition~\ref{def:PG:admissible-term-order}),
  the relation $\succ_\mapIonly$ is an admissible term order for $\IPGInf$ (Definition~\ref{def:IPG:admissible-term-order}).
\end{lem}
\begin{proof}
  This is easy to see, considering that
  $\mapIonly$ is a bijection \begin{full}on ground terms\end{full}
  and that $\mapIonly$ and $\mapIonly^{-1}$ preserve yellow subterms.
\end{proof}

\begin{full}%
\begin{lem}\label{lem:PG:closure-extension-mapI}
  Given $D \closure\rho, C \closure\theta \in \clausesPG$, we have
  $D\closure\rho \succ C\closure\theta$ if and only if
  $\mapI{D\closure\rho} \mathrel{\succ_\mapIonly} \mapI{C\closure\theta}$.
\end{lem}
\begin{proof}
By \ref{cond:IPG:order:closures},
$D\closure\rho \succ C\closure\theta$ if and only if
$C\theta \succ D\rho$.
By Definition~\ref{def:PG:succ-transfer},
this is equivalent to $\mapI{D\rho} \mathrel{\succ_\mapIonly} \mapI{C\theta}$.
Since $\mapI{D\rho} = \mapIonly_\rho(D)\mapI{\rho}$
and $\mapI{C\theta} = \mapIonly_\theta(C)\mapI{\theta}$ by Lemma~\ref{lem:PG:mapI-subst},
this is equivalent to
$\mapIonly_\rho(D)\mapI{\rho} \mathrel{\succ_\mapIonly} \mapIonly_\theta(C)\mapI{\theta}$.
By \ref{cond:IPG:order:closures}, this is equivalent to
$\mapI{D\closure\rho} \mathrel{\succ_\mapIonly} \mapI{C\closure\theta}$.
\end{proof}

\begin{defi}\label{def:PG:irred}
  Given a term rewrite system $R$ on $\termsF$,
  we say that a closure $C\closure\theta\in\clausesPG$ is variable-irreducible \wrt\ $R$
  if $\mapF{\mapI{C\closure\theta}}$ is.
  We write $\irred_R(N)$ for all variable-irreducible closures in a set $N\subseteq\clausesPG$.
\end{defi}
\end{full}

\begin{slim}

  Since $\mapIonly$ is bijective, we can transfer the selection function as follows:
  \begin{defi}\label{def:PG:sel-transfer}
    Based on $\mathit{gsel}$,
    we define $\mapI{\mathit{gsel}}$ as a selection function that selects
    the literals of $C \in \clausesIPG$ corresponding to the $\mathit{gsel}$-selected literals in
    $\mapIonly^{-1}(C)$.
  \end{defi}

  \begin{defi}\label{def:PG:mapI-inference}
    We extend $\mapIonly$ to inference rules
    by mapping an inference $\iota \in \PGInf$
    to the inference
    \[
    \inference{\mapI{\prems(\iota)}}{\mapI{\concl(\iota)}}
    \]
    \end{defi}

    \begin{lem}\label{PG:lem:mapI-inference}
    The mapping $\mapIonly$ is a bijection between $\PGInf^{\succ,gsel}$ and $\IGInf^{\succ_\mapIonly,\mapI{\mathit{gsel}}}$.
    \end{lem}
    \begin{proof}
    This is easy to see by comparing the rules of $\PGInf$ and $\IGInf$.
    It is crucial that the following concepts match:
      \begin{itemize}
        \item Green subterms on the $\levelIPG$ level correspond to green subterms on the $\levelG$ level.
        \item The term orders correspond (Definition~\ref{def:PG:succ-transfer}).
        \item The selected literals correspond; i.e.,
        a literal $L$ is selected in a clause $C$ if and only if
        the literal $\mapI{L}$ is selected in $\mapI{C}$.
        This follows directly from the definition of $\mapI{\mathit{gsel}}$ (Definition~\ref{def:PG:sel-transfer}).
        \item The concepts of eligibility correspond; i.e.,
        a literal $L$ of a clause $C\in\clausesG$ is (strictly) eligible \wrt\ $\succeq$
        if and only if 
        the literal $\mapI{L}$ of the clause $C$ is (strictly) eligible \wrt\ $\succeq_\mapIonly$; and
        a position $L.s.p$ of a clause $C\in\clausesG$ is eligible \wrt\ $\succeq$
        if and only if 
        the position $\mapI{L}.\mapI{s}.p$ of the clause $\mapI{C}$
        is eligible \wrt\ $\succeq_\mapIonly$.
        This is true because eligibility (Definition~\ref{def:ho-eligible}) depends only on the selected literals and the term order,
        which correspond as discussed above.\qedhere
      \end{itemize}
    \end{proof}
\end{slim}

\begin{defi}[Inference Redundancy]\label{def:PG:RedI}
\begin{full}
Let $N \subseteq \clausesPG$.
Let $\iota \in \PGInf$
an inference with premises $C_1\closure\theta_1$, \dots, $C_m\closure\theta_m$ and
conclusion $C_{m+1}\closure\theta_{m+1}$.
We define $\iota \in \PGRedI(N)$ if
\begin{enumerate}[label=\arabic*.,ref=\arabic*]
\item\label{cond:PG:RedI:FInf-not-valid}
the inference $\iota'$ given as\strut
\[
\inference{\mapF{\mapI{C_1\theta_1}}\quad \cdots \quad \mapF{\mapI{C_m\theta_m}}}
{\mapF{\mapI{C_{m+1}\theta_{m+1}}}}
\]
is not a valid $\FInf^{\succ_{\mapIonly\mapFonly}}$ inference
such that the names of $\iota$ and $\iota'$ correspond up to the prefixes $\infname{PG}$ and $\infname{F}$;
or
\item\label{cond:PG:RedI:entailed} for all confluent term rewrite systems $R$ oriented by $\succ_{\mapIonly\mapFonly}$
  whose only Boolean normal forms are $\itrue$ and $\ifalse$
  such that $C_{m+1}\closure\theta_{m+1}$ is variable-irreducible, we have
\[R \cup O \modelsolam \mapF{\mapI{C_{m+1}\closure\theta_{m+1}}}\]
where 
$O = \irred_R(\mapF{\mapI{N}})$
if $\iota$ is a $\PGDiff$ inference and
$O = \{E \in \irred_R(\mapF{\mapI{N}}) \mid E \prec_{\mapIonly\mapFonly} \mapF{\mapI{C_m\theta_m}}\}$
if $\iota$ is some other inference.
\end{enumerate}
\end{full}
\begin{slim}
Let $N \subseteq \clausesG$ and $\iota \in \PGInf$.
We define $\iota \in \GRedI$ if $\mapI{\iota} \in \IGRedI(\mapI{N})$.
\end{slim}
\end{defi}

\begin{full}%
We transfer the selection function $\mathit{pgsel}$ as follows:
\begin{defi}\label{def:PG:sel-transfer}
Let $N \subseteq \clausesPG$ be a set of closures.
Then we choose a function $\mapIonly^{-1}_N$, depending on this set $N$,
such that $\mapIonly_N^{-1}(C) \in N$ and $\mapI{\mapIonly^{-1}_N(C)} = C$ for all $C \in \mapI{N}$.
Then we define $\mapI{\mathit{pgsel},N}$ as a selection function that selects
the literals of $C \in \mapI{N}$ corresponding to the $\mathit{pgsel}$-selected literals in $\mapIonly^{-1}_N(C)$
and that selects arbitrary literals in all other closures.
\end{defi}
\end{full}

\begin{lem}\label{lem:PG:saturated}
Let $N \subseteq \clausesPG$ be saturated up to redundancy \wrt{} $\PGInf^{\succ,\mathit{\slimfull{}{p}gsel}}$.
Then $\mapI{N}$ is saturated up to redundancy \wrt{} $\IPGInf^{\succ_\mapIonly,\mapI{\slimfull{\mathit{gsel}}{\mathit{pgsel},N}}}$. 
\end{lem}
\begin{slim}
\begin{proof}
  By Lemma~\ref{PG:lem:mapI-inference} and Defintion~\ref{def:PG:RedI}.
\end{proof}
\end{slim}%
\begin{full}
\begin{proof}
Let $\iota'$ be a $\IPGInf^{\succ_\mapIonly,\mapI{\mathit{pgsel},N}}$ inference
from $\mapI{N}$.
We must show that $\iota' \in \IPGRedI(\mapI{N})$.
It suffices to construct a $\PGInf$ inference $\iota$ 
with premises $\mapIonly^{-1}_N(\prems(\iota'))$
such that $\mapI{\concl(\iota)} = \concl(\iota')$
and the rule names of $\iota$ and $\iota'$ coincide up to the prefixes
$\infname{PG}$ and $\infname{IPG}$.
Then, by saturation, $\iota \in \PGRedI(N)$; i.e., 
condition~\ref{cond:PG:RedI:FInf-not-valid}
or condition~\ref{cond:PG:RedI:entailed}
of Definition~\ref{def:PG:RedI},
is satisfied.
Condition~\ref{cond:PG:RedI:FInf-not-valid} cannot be satisfied
because it contradicts Lemma~\ref{lem:IPG:inference-to-F-inference}
applied to $\iota'$.
Thus, condition~\ref{cond:PG:RedI:entailed} must be satisfied.
Then, by Definition~\ref{def:IPG:RedI},
$\iota' \in \IPGRedI(\mapI{N})$.

Finding such an inference $\iota$ is straightforward for all inference rules.
We illustrate it with the rule $\IPGEqRes$:
Let $\iota'$ be\strut
\[\namedinference{\IPGEqRes}
{C' \llor {u \cneq u'}\>\closure\>\theta}
{C'\closure\theta}\]
Then there exists a corresponding $\PGInf$ inference $\iota$ from 
$\mapIonly^{-1}_N(C' \llor {u \cneq u'}\>\closure\>\theta)$.
(See Definition~\ref{def:PG:sel-transfer} for the definition of $\mapIonly^{-1}_N$.)
The eligibility condition is fulfilled because the term order, and 
the selections are transferred according to Definitions
\ref{def:PG:succ-transfer} and \ref{def:PG:sel-transfer}
and Lemma~\ref{lem:PG:closure-extension-mapI}.
The equality condition is fulfilled by Lemma~\ref{lem:PG:mapI-preserves-identities}.
The conclusion $\concl(\iota)$ of this inference has the property
$\mapI{\concl(\iota)} = \concl(\iota')$, as desired.
\end{proof}
\end{full}

\subsubsection{Full Higher-Order Level}
\label{ssec:H:redundancy}

In this subsubsection, let $\succ$ be an admissible term order
(Definition~\ref{def:admissible-term-order}) and let
$\mathit{hsel}$ be a selection function (Definition~\ref{def:lit-sel}).
\begin{full}
We extend $\succ$ to closures $\clausesPG$ by
$C\closure\theta \succ D\closure\rho$ if and only if $C\theta \succ D\rho$.
Then we can use it for $\PGInf$ as well:
\end{full}
\begin{lem}\label{lem:G:admissible-term-order}
  The relation $\succ$ is an admissible term order for $\PGInf$.
\end{lem}
\begin{proof}
  Conditions~\ref{cond:order:total} to \ref{cond:order:clause-extension}
  are identical to conditions~\ref{cond:PG:order:total} to \ref{cond:PG:order:clause-extension}.
  \begin{full}Condition~\ref{cond:PG:order:closures} is fulfilled by the given extension of $\succ$ to closures.\end{full}
\end{proof}

\begin{full}
\begin{defi}\label{def:G:irred}
  Given a term rewrite system $R$ on $\termsF$,
  we say that a closure $C\closure\theta\in\clausesG$ is variable-irreducible \wrt\ $R$
  if $\mapF{\mapI{\mapP{C\closure\theta}}}$ is.
  We write $\irred_R(N)$ for all variable-irreducible closures in a set $N\subseteq\clausesG$.
\end{defi}

For the $\levelH$ level, we define both clause and inference redundancy.
Below, we write $\fipg{C}$ for $\mapF{\mapI{\mapP{\mapG{C}}}}$
and $\fip{C}$ for $\mapF{\mapI{\mapP{C}}}$.

\begin{defi}[Clause Redundancy]\label{def:H:RedC}
  Given a constrained clause $C \in \clausesH$ and a set $N\subseteq\clausesH$, let $C\in\HRedC(N)$
  if for all confluent term rewrite systems $R$ on $\termsF$ oriented by $\succ_{\mapIonly\mapFonly}$
  whose only Boolean normal forms are $\itrue$ and $\ifalse$
  and all $C' \in \irred_R(\fipg{C})$, at least one of the following two conditions holds:
  \begin{enumerate}[label=\arabic*.,ref=\arabic*]
    \item \label{cond:H:RedC:entailment}
      $R \cup \{E \in \irred_R(\fipg{N}) \mid E \prec_{\mapIonly\mapFonly} C'\}\modelsolam C'$; or
    \item \label{cond:H:RedC:subsumed} there exists clauses $D \in N$ and $D' \in \irred_R(\fipg{D})$
    such that $C \sqsupset D$ and $\mapT{D'} = \mapT{C'}$.
  \end{enumerate}
\end{defi}

\begin{defi}[Inference Redundancy]\label{def:H:RedI}
  Let $N \subseteq \clausesH$.
  Let $\iota \in \HInf$
  an inference with premises $C_1\constraint{S_1}$, \dots, $C_m\constraint{S_m}$ and
   conclusion $C_{m+1}\constraint{S_{m+1}}$.
  We define $\HRedI$ so that $\iota \in \HRedI(N)$
  if for all substitutions $(\theta_1, \dots, \theta_{m+1})$
  for which $\iota$ is rooted in $\FInf{}$ (Definition~\ref{def:fol-inferences}),
  and for all confluent term rewrite systems $R$ oriented by $\succ_{\mapIonly\mapFonly}$
  whose only Boolean normal forms are $\itrue$ and $\ifalse$
  such that $C_{m+1}\closure\theta_{m+1}$ is variable-irreducible, we have
\[R \cup O \modelsolam \mapF{C_{m+1}\theta_{m+1}}\]
where 
$O = \irred_R(\fipg{N})$
if $\iota$ is a $\Diff$ inference and
$O = \{E \in \irred_R(\fipg{N}) \mid E \prec_{\mapIonly\mapFonly} \mapF{C_m\theta_m}\}$
if $\iota$ is some other inference.
\end{defi}
\end{full}

The selection function is transferred
\begin{slim}as follows:\end{slim}
\begin{full}in a similar way as with $\mapIonly$:\end{full}
\begin{defi}\label{def:H:sel-transfer}
\begin{full}
Let $N \subseteq \clausesG$.
We choose a function $\mapPonly^{-1}_N$, depending on this set $N$,
such that $\mapPonly_N^{-1}(C) \in N$ and $\mapP{\mapPonly^{-1}_N(C)} = C$ for all $C \in \mapP{N}$.
\end{full}
\slimfull{Let $N \subseteq \clausesH$. We}{Similarly, for $N \subseteq \clausesH$, we}
choose a function $\mapGonly^{-1}_N$, depending on this set $N$,
such that $\mapGonly_N^{-1}(C) \in N$ and $\mapG{\mapGonly^{-1}_N(C)} = C$ for all $C \in \mapG{N}$.
\begin{slim}
Then we define $\mapG{\mathit{hsel},N}$ as a selection function that selects
the literals of $C \in \mapG{N}$ corresponding to the $\mathit{hsel}$-selected literals in $\mapGonly^{-1}_N(C)$
and that selects arbitrary literals in all other clauses.
\end{slim}
\begin{full}

Then we define $\mapPG{\mathit{hsel},N}$ as a selection function for $\clausesPG$ as follows:
Given a clause $C_\levelPG \in \mapPG{N}$,
let $C_\levelH\closure\theta = \mapPonly^{-1}_{\mapG{N}}(C_\levelPG)$
and $C_\levelH\constraint{S} = \mapGonly^{-1}_N(C_\levelH\closure\theta)$.
We define $\mapPG{\mathit{hsel},N}$ to select
$L_\levelPG \in C_\levelPG$ if and only if there exists a
literal $L_\levelH$ selected in $C_\levelH\constraint{S}$ by $\mathit{hsel}$ such that
$L_\levelPG = L_\levelH\mapp{\theta}$.
Given a clause $C_\levelPG \not\in \mapPG{N}$,
$\mapPG{\mathit{hsel},N}$ can select arbitrary literals.
\end{full}
\end{defi}

\begin{full}
\begin{lem}\label{lem:instance-varirred}
  Let $R$ be a confluent term rewrite system on $\termsPF$ oriented by $\succ_{\mapIonly\mapFonly}$
  whose only Boolean normal forms are $\itrue$ and $\ifalse$.
  Let $C\closure\theta,D\closure\rho \in \clausesPG$.
  Let $C\closure\theta$ be variable-irreducible \wrt\ $R$.
  Let $\sigma$ be a substitution such that 
  $z\theta = z\sigma\rho$ for all variables $z$ in $C$ and $D = C\sigma$.
  Then $D\closure\rho$ is variable-irreducible \wrt\ $R$.
\end{lem}
\begin{proof}
Let $L$ be a literal in $\mapF{\mapIonly_\rho(D)}$.
We must show that
for all variables $x$ in $L$,
$\mapF{\mapI{x\rho}}$ is irreducible \wrt\ the rules $s \rewrite t \in R$ with $L\mapF{\mapI{\rho}} \succ_{\mapIonly\mapFonly} s \eq t$
and all Boolean subterms of $\mapF{\mapI{x\rho}}$ are either $\itrue$ or $\ifalse$.
Then there exists a literal $L_0 \in D$ such that $L = \mapF{\mapIonly_\rho(L_0)}$.
Let $x$ be a variable in $L$. Then it is also a variable in $L_0$,
occurring outside of parameters.
Since $D = C\sigma$, $x$ occurs in $C\sigma$ outside of parameters.
Since $C\closure\theta \in \clausesPG$, 
the clause $C$ contains only nonfunctional variables,
and thus a literal $L'_0 \in C$ with $L'_0\sigma = L_0$ must contain a variable $z$ outside of parameters
such that $x$ occurs outside of parameters in $z\sigma$.
So $\mapF{\mapI{x\rho}}$ is a subterm of $\mapF{\mapI{z\sigma\rho}} = \mapF{\mapI{z\theta}}$.
Let $L' = \mapF{\mapIonly_\theta(L'_0)} \in \mapF{\mapIonly_\theta(C)}$.
Then $L'$ also contains $z$.
By variable-irreducibility of $C\closure\theta$,
$\mapF{\mapI{z\theta}}$ is irreducible \wrt\ the rules $s \rewrite t \in R$ with $L'\mapF{\mapI{\theta}} \succ_{\mapIonly\mapFonly} s \eq t$
and all Boolean subterms of $\mapF{\mapI{z\theta}}$ are either $\itrue$ or $\ifalse$.
Then the subterm $\mapF{\mapI{x\rho}}$ of $\mapF{\mapI{z\theta}}$ 
is also irreducible \wrt\ the rules $s \rewrite t \in R$ with $L'\mapF{\mapI{\theta}} \succ_{\mapIonly\mapFonly} s \eq t$
and all Boolean subterms of $\mapF{\mapI{x\rho}}$ are either $\itrue$ or $\ifalse$.
It remains to show that $L'\mapF{\mapI{\theta}} = L\mapF{\mapI{\rho}}$.
We have 
$L'\mapF{\mapI{\theta}}
= \mapF{\mapIonly_\theta(L'_0)}\mapF{\mapI{\theta}}
= \mapF{\mapIonly_\theta(L'_0\theta)}
= \mapF{\mapIonly_\theta(L'_0\sigma\rho)}
= \mapF{\mapIonly_\rho(L_0\rho)}
= \mapF{\mapIonly_\rho(L_0)}\mapF{\mapI{\rho}}
= L\mapF{\mapI{\rho}}$, using Lemma~\ref{lem:IPG:mapF-subst} and Lemma~\ref{lem:PG:mapI-subst}.
\end{proof}
\end{full}

\begin{lem}[Lifting of Order Conditions]\label{lem:H:order-lifting}
  \begin{full}Let $t\constraint{T}$ and $s\constraint{S}$ be constrained terms over $\TT(\SigmaH)$,\end{full}%
  \begin{slim}Let $t,s \in \TT(\SigmaH)$,\end{slim}
  and let $\zeta$ be a grounding substitution\begin{full} such that $S\zeta$ and $T\zeta$ are true\end{full}.
  If $t\zeta \succ s\zeta$, then $t\slimfull{}{\constraint{T}} \not\preceq s\slimfull{}{\constraint{S}}$.
  The same holds for \slimfull{}{constrained} literals.
\end{lem}
\begin{proof}
We prove the contrapositive. If $t\slimfull{}{\constraint{T}} \preceq s\slimfull{}{\constraint{S}}$, then, by
\ref{cond:order:stability-terms},
$t\zeta \preceq s\zeta$.
Therefore, since $\succ$ is asymmetric by \ref{cond:order:total}, $t\zeta \not\succ s\zeta$.
The proof for \slimfull{}{constrained} literals is analogous, using \ref{cond:order:clause-extension}
and \ref{cond:order:stability-clauses}.
\end{proof}

\begin{lem}[Lifting of Maximality Conditions]\label{lem:H:maximality-lifting}
  Let $C\slimfull{}{\constraint{S}} \in \clausesH$.
  Let $\theta$ be a grounding substitution.
  Let $L_0$ be (strictly) maximal in $C\theta$.
  Then there exists a literal $L$ that is (strictly) maximal in $C\slimfull{}{\constraint{S}}$
  such that $L\theta = L_0$.
\end{lem}
\begin{proof}
  By Definition~\ref{def:maximality}, a literal $L$ of a \begin{full}constrained\end{full} clause $C\begin{full}\constraint{S}\end{full}$ is maximal
  if for all $K \in C$ such that $K\begin{full}\constraint{S}\end{full} \succeq L\begin{full}\constraint{S}\end{full}$,
  we have $K\begin{full}\constraint{S}\end{full} \preceq L\begin{full}\constraint{S}\end{full}$.

  Since $L_0 \in C\theta$, there exist literals $L$ in $C\slimfull{}{\constraint{S}}$ such that $L\theta = L_0$.
  Let $L$ be a maximal one among these literals.
  A maximal one must exist because $\succ$ is transitive on \begin{full}constrained\end{full} literals by
  \ref{cond:order:transitive}
  and transitivity implies existence of maximal elements in nonempty finite sets.
  Let $K$ be a literal in $C\slimfull{}{\constraint{S}}$ such that $K\slimfull{}{\constraint{S}} \succeq L\slimfull{}{\constraint{S}}$.
  We must show that $K\slimfull{}{\constraint{S}} \preceq L\slimfull{}{\constraint{S}}$.
  By Lemma~\ref{lem:H:order-lifting}, $K\theta \not\prec L\theta = L_0$.
  By \ref{cond:order:total}, $\succ$ is a total order on ground terms, and thus $K\theta \succeq L_0$.
  By maximality of $L_0$ in $C\theta$, we have $K\theta \preceq L_0$
  and thus $K\theta = L_0$ by \ref{cond:order:total}.
  Then $K\slimfull{}{\constraint{S}} \preceq L\slimfull{}{\constraint{S}}$
  because we chose $L$ to be maximal among all literals in $C\slimfull{}{\constraint{S}}$ such that $L\theta = L_0$.

  For \emph{strict} maximality, we simply observe that if $L$ occurs more than once in $C$,
  it also occurs more than once in $C\theta$.
\end{proof}

\begin{lem}[Lifting of Eligibility]\label{lem:H:eligibility-lifting}
Let $N \subseteq \clausesH$.
\begin{full}
Let $C_\levelPG\closure\theta_\levelPG \in \mapPG{N}$,
let $C_\levelH\closure\theta = \mapPonly^{-1}_{\mapG{N}}(C_\levelPG\closure\theta_\levelPG)$
and let $C_\levelH\constraint{S} = \mapGonly^{-1}_N(C_\levelH\closure\theta)$.
\end{full}%
\begin{slim}
Let $C_\levelG\in\mapG{N}$, let $C_\levelH = \mapGonly^{-1}_N(C_\levelG)$,
and let $\theta$ be a grounding substitution such that $C_\levelG = C_\levelH\theta$.
\end{slim}
  \begin{itemize}
    \item 
    Let $L_{\levelPG}$ be a literal in $C_{\levelPG}\slimfull{}{\closure\theta_\levelPG}$
    that is (strictly) eligible \wrt\ $\slimfull{\mapG{\mathit{hsel},N}}{\mapPG{\mathit{hsel},N}}$.
    Then there exists a literal $L_\levelH$ in $C_\levelH$
    such that $\slimfull{L_\levelG = L_\levelH\theta}{L_\levelPG = L_\levelH\mapp{\theta}}$ and,
    given substitutions $\sigma$ and $\zeta$ with $x\theta = x\sigma\zeta$ for all variables $x$ in $C_\levelH\slimfull{}{\constraint{S}}$,
    $L_\levelH$ is (strictly) eligible in $C_\levelH\slimfull{}{\constraint{S}}$  \wrt\ $\sigma$ and $\mathit{hsel}$.
    \item Let $\slimfull{L_\levelG.s_\levelG.p_\levelG}{L_\levelPG.s_\levelPG.p_\levelPG}$ be a green position of $\slimfull{C_\levelG}{C_\levelPG\closure\theta_\levelPG}$
    that is eligible \wrt\ $\slimfull{\mapG{\mathit{hsel},N}}{\mapPG{\mathit{hsel},N}}$.
    Then there exists 
    a green position $L_\levelH.s_\levelH.p_\levelH$ of $C_\levelH$
    such that
    \begin{itemize}
      \item $\slimfull{L_\levelG = L_\levelH\theta}{L_\levelPG = L_\levelH\mapp{\theta}}$;
      \item $\slimfull{s_\levelG = s_\levelH\theta}{s_\levelPG = s_\levelH\mapp{\theta}}$;
      \item 
      \begin{itemize}
        \item
        $\slimfull{p_\levelG = p_\levelH}{p_\levelPG = p_\levelH}$, or
        \item 
        $\slimfull{p_\levelG = p_\levelH.q}{p_\levelPG = p_\levelH.q}$ for some nonempty $q$,
        the subterm $u_\levelH$ at position $L_\levelH.s_\levelH.p_\levelH$ of $C_\levelH$ is
        \begin{full}not a variable but is\end{full}
        variable-headed, and
        $u_\levelH\theta$ is nonfunctional; and
      \end{itemize}
      \item given substitutions $\sigma$ and $\zeta$ with $x\theta = x\sigma\zeta$ for all variables $x$ in $C_\levelH\slimfull{}{\constraint{S}}$,
      $L_\levelH.s_\levelH.p_\levelH$ is eligible in $C_\levelH\slimfull{}{\constraint{S}}$  \wrt\ $\sigma$ and $\mathit{hsel}$.
    \end{itemize}
  \end{itemize}
\end{lem}
\begin{proof}
Let $L_{\levelPG}$ be a literal in $\slimfull{C_\levelG}{C_\levelPG\closure\theta_\levelPG}$ that is (strictly) eligible \wrt\ $\slimfull{\mapG{\mathit{hsel},N}}{\mapPG{\mathit{hsel},N}}$.
By the definition of eligibility \slimfull{(Definition~\ref{def:ho-eligible})}{(Definition~\ref{def:closure-eligible})},
there are two ways to be (strictly) eligible:
\begin{itemize}
  \item $L_{\levelPG}$ is selected by $\slimfull{\mapG{\mathit{hsel},N}}{\mapPG{\mathit{hsel},N}}$.
  By Definition~\ref{def:H:sel-transfer},
  there exists a literal $L_\levelH$ selected by $\mathit{hsel}$
  such that $\slimfull{L_\levelG = L_\levelH\theta}{L_\levelPG = L_\levelH\mapp{\theta}}$.
  By Definition~\ref{def:ho-eligible}, $L_\levelH$ is (strictly) eligible in $C_\levelH\slimfull{}{\constraint{S}}$ \wrt\ $\sigma$
  because it is selected.
  \item There are no selected literals in $\slimfull{C_\levelG}{C_\levelPG\closure\theta_\levelPG}$ and
  $L_{\levelPG}\slimfull{}{\theta_\levelPG}$ is (strictly) maximal in $\slimfull{C_\levelG}{C_\levelPG\theta_\levelPG}$.
  By Definition~\ref{def:H:sel-transfer},
  there are no selected literals in $C_\levelH\slimfull{}{\constraint{S}}$.
  Since $\slimfull{C_\levelG}{C_\levelPG\theta_\levelPG} = C_\levelH\theta = C_\levelH\sigma\zeta$,
  by Lemma~\ref{lem:H:maximality-lifting}, 
  there exists a literal
  $L_\levelH\in C_\levelH$ such that $L_\levelH\sigma$ is (strictly) maximal in $C_\levelH\sigma$.
  By Definition~\ref{def:ho-eligible}, $L_\levelH$ is (strictly) eligible in $C_\levelH\constraint{S}$ \wrt\ $\sigma$.
\end{itemize}

For the second part of the lemma,
let $\slimfull{L_\levelG.s_\levelG.p_\levelG}{L_\levelPG.s_\levelPG.p_\levelPG}$ be a green position of $\slimfull{C_\levelG}{C_\levelPG\closure\theta_\levelPG}$
that is eligible \wrt\ $\slimfull{\mapG{\mathit{hsel},N}}{\mapPG{\mathit{hsel},N}}$.
By Definition~\slimfull{\ref{def:ho-eligible}}{\ref{def:closure-eligible}},
the literal $L_{\levelPG}$ is
of the form $s_{\levelPG} \doteq t_{\levelPG}$
with $\slimfull{s_\levelG \succ t_\levelG}{s_\levelPG\theta_\levelPG \succ t_\levelPG\theta_\levelPG}$
and $L_{\levelPG}$ is either
negative and
eligible
or positive and strictly eligible.
By the first part of this lemma,
there exists a literal $L_\levelH$
in $C_\levelH$ that is
either
negative and
eligible
or positive and strictly eligible in $C_\levelH\slimfull{}{\constraint{S}}$
\wrt\ $\sigma$ and $\mathit{hsel}$ such that $\slimfull{L_\levelG = L_\levelH\theta}{L_\levelPG = L_\levelH\mapp{\theta}}$.
Then $L_\levelH$ must be of the form $s_\levelH \doteq t_\levelH$
with $\slimfull{s_\levelG = s_\levelH\theta}{s_\levelPG = s_\levelH\mapp{\theta}}$ and $\slimfull{t_\levelG = t_\levelH\theta}{t_\levelPG = t_\levelH\mapp{\theta}}$.
Since $\slimfull{s_\levelG \succ t_\levelG}{s_\levelPG\theta_\levelPG \succ t_\levelPG\theta_\levelPG}$,
we have $s_\levelH \centernot\preceq t_\levelH$.
By Definition~\ref{def:ho-eligible},
every green position in $L_\levelH.s_\levelH$ is eligible
in $C_\levelH\slimfull{}{\constraint{S}}$  \wrt\ $\sigma$ and $\mathit{hsel}$.

It remains to show that there exists a green position $L_\levelH.s_\levelH.p_\levelH$ in $C_\levelH$ such that
either $p_{\levelPG} = p_\levelH$ or
$p_{\levelPG} = p_\levelH.q$ for some nonempty $q$,
the subterm $u_\levelH$ at position $L_\levelH.s_\levelH.p_\levelH$ of $C_\levelH$
is
\begin{full}not a variable but variable-headed,\end{full}%
\begin{slim}variable-headed,\end{slim}
and $u_\levelH\theta$ is nonfunctional.

Since $p_{\levelPG}$ is a green position of $s_{\levelPG} = s_\levelH\slimfull{\theta}{\mapp{\theta}}$,
position $p_{\levelPG}$ must either be a green position of $s_\levelH$
or be below a variable-headed term in $s_\levelH$.
In the first case, we set $p_\levelH = p_{\levelPG}$.
In the second case, 
\begin{full}let $p_\levelH$ be the position of the variable-headed term.\end{full}%
\begin{slim}let $u_\levelH$ be that variable-headed term and let $p_\levelH$ be its position.\end{slim}
Then $p_\levelH.q = p_{\levelPG}$ for some nonempty $q$.
Moreover, since $p_{\levelPG}$ is a green position of $s_{\levelPG}$,
the subterm of $s_{\levelPG}$ at position $p_\levelH$, which is $u_\levelH\theta$, cannot be functional.
\begin{full}
Since $\mapp{\theta}$ maps nonfunctional variables to nonfunctional variables,
the subterm of $s_\levelH$ at position $p_\levelH$
cannot be a variable, but it is variable-headed.
\end{full}
\end{proof}

\begin{lem}[Lifting Lemma]\label{lem:H:saturation}
  Let $N \subseteq \clausesH$ be saturated up to redundancy \wrt{} $\HInf^{\succ,\mathit{hsel}}$.
  Then $\slimfull{\mapG{N}}{\mapPG{N}}$ is saturated up to redundancy \wrt{} $\PGInf^{\succ, \slimfull{\mapG{\mathit{hsel},N}}{\mapPG{\mathit{hsel},N}}}$.
\end{lem}
\begin{slim}
\begin{proof}
Let $\iota_\levelG$ be a $\PGInf$ inference from $\mapG{N}$.
We must show that $\iota_\levelPG \in \GRedI(\mapG{N})$.
By Definitions~\ref{def:PG:RedI}, \ref{def:IPG:RedI}, \ref{def:PF:RedI},
and Lemma~\ref{lem:PG:mapI-mapF-fol},
it suffices to show that
\begin{align*}
  \mapF{\mapG{N}} &\modelsolam \mapF{\concl(\iota_\levelG)}&&\text{ if $\iota_\levelG$ is a $\Diff$ inference, or}\\
  \{ E \in \mapF{\mapG{N}} \mid E \prec_\mapFonly \mapF{\mprem(\iota_\levelG)} \}&\modelsolam \mapF{\concl(\iota_\levelG)}
  &&\text{ if $\iota_\levelG$ is some other inference.}\tag{*}
\end{align*}
One strategy that we will apply below is to construct
an inference $\iota_\levelH$
such that the prefixes of $\iota_\levelH$ and $\iota_\levelG$ match up to the prefixes $\infname{G}$ and $\infname{Fluid}$
and to construct substitutions
$\theta_1, \dots, \theta_{m+1}$
such that
applying $\theta_1, \dots, \theta_{m}$ to $\prems(\iota_\levelH)$
yields $\prems(\iota_\levelG)$ and
applying $\theta_{m+1}$ to $\concl(\iota_\levelH)$
yields $\concl(\iota_\levelG)$. ($**$)

By
Lemmas~\ref{PG:lem:mapI-inference} and~\ref{IPG:lem:mapF-inference},
$\mapF{\mapI{\iota_\levelG}}$ is a valid $\PFInf^{\succ_{\mapIonly\mapFonly},\mapF{\mapI{\mapG{\mathit{hsel},N}}}}$ inference,
By Lemma~\ref{lem:PG:mapI-mapF-fol},
$\succ_{\mapIonly\mapFonly} = \succ_{\mapFonly}$ and
$\mapF{\mapI{\mapG{\mathit{hsel},N}}} = \mapF{{\mapG{\mathit{hsel},N}}}$.
Moreover, Lemma~\ref{lem:PG:mapI-mapF-fol} tells us that
$\mapF{\mapI{\iota_\levelG}}$ can also be obtained by applying $\mapFonly$ directly to
premisses and conclusion of $\iota_\levelG$.
Therefore, 
$\iota_\levelH$ is rooted in $\FInf$ for $(\theta_1, \dots, \theta_{m+1})$  (Definition~\ref{def:fol-inferences}).
By saturation of $N$ up to redundancy \wrt\ $\HInf{}$,
$\iota_\levelH$ is redundant and thus ($*$) by Definition~\ref{def:H:RedI}.

\medskip

\noindent{\GSup}:\enskip
Assume that $\iota_\levelG$ is a $\GSup$ inference
\[\namedinference{\IGSup}
{\overbrace{D'_\levelG \llor { t_\levelG \ceq t'_\levelG}}^{\vphantom{\cdot}\smash{D_\levelG}} \hypsep
 \greensubterm{C_\levelG}{t_\levelG}}
{D'_\levelG \llor \greensubterm{C_\levelG}{t'_\levelG}}\]
Let $C_\levelH = \mapGonly_{N}^{-1}(C_\levelG)$,
and let $\theta$ be a grounding substitution such that $C_\levelG = C_\levelH\theta$.
By condition~\REF{4} of $\GSup$, the position $L_\levelG.s_\levelG.p_\levelG$ of $t_\levelG$ is eligible in $C_\levelG$.
By Lemma~\ref{lem:H:eligibility-lifting},
there exists 
a green position $L_\levelH.s_\levelH.p_\levelH$ of $C_\levelH$
such that
\begin{itemize}
  \item $L_\levelG = L_\levelH\theta$;
  \item $s_\levelG = s_\levelH\theta$;
  \item given substitutions $\sigma$ and $\zeta$ with $x\theta = x\sigma\zeta$ for all variables $x$ in $C_\levelH\slimfull{}{\constraint{S}}$,
  $L_\levelH.s_\levelH.p_\levelH$ is eligible in $C_\levelH$  \wrt\ $\sigma$ and $\mathit{hsel}$
  (condition~\ref{sup:five} of $\Sup$ or $\FluidSup$);
\end{itemize}
and one of the following cases applies:
\medskip

\noindent\textsc{Case 1:}\enskip
$p_\levelG = p_\levelH$.

\medskip

\noindent\textsc{Case 1a:}\enskip
The subterm at position $L_\levelH.s_\levelH.p_\levelH$ of $C_\levelH$ is a variable $x$
and no occurrence of $x$ is inside of a parameter in $C_\levelH$.

Since free De Bruijn indices cannot occur in parameters,
$x$ does not occur in parameters in
$C_\levelH\theta[x \mapsto x]$ either.
By condition~\REF{1} of $\GSup$, $t_\levelG$ is nonfunctional and so $x$ is nonfunctional.
Thus, all occurrences of $x$ in $C_\levelH\theta[x \mapsto x]$
are in green positions.
Thus, $C_\levelH\theta[x \mapsto t'_\levelG]$
results from $C_\levelH\theta = C_\levelG$ by replacing $t_\levelG$ with $t'_\levelG$
at green positions.
Thus, by the $\termsPG$-analogue of Lemma~\ref{lem:IPG:mapF-green-subterms},
$\mapF{C_\levelG}$ results from $\mapF{C_\levelH\theta[x \mapsto t'_\levelG]}$ by replacing
$\mapF{t_\levelG}$ with $\mapF{t'_\levelG}$.
Therefore,
$\{\mapF{t_\levelG} \ceq \mapF{t'_\levelG}, \mapF{C_\levelH\theta[x \mapsto t'_\levelG]}\} \modelsolam \mapF{C_\levelG}$
and thus
\[\{\mapF{D_\levelG}, \mapF{C_\levelH\theta[x \mapsto t'_\levelG]}\} \modelsolam \mapF{D'_\levelG \llor \greensubterm{C_\levelG}{t'_\levelG}}\]
By condition~\REF{3} of $\GSup$, $D_\levelG \prec \mprem(\iota_\levelG)$,
and thus $\mapF{D_\levelG} \prec_\mapFonly \mapF{\mprem(\iota_\levelG)}$.
By Condition~\REF{2} of $\GSup$, $t_\levelG \succ t'_\levelG$,
and thus by \ref{cond:PF:order:comp-with-contexts},
$\mapF{\mprem(\iota_\levelG)} = \mapF{C_\levelH\theta[x \mapsto t_\levelG]} \succ_\mapFonly \mapF{C_\levelH\theta[x \mapsto t'_\levelG]}$.
Therefore,
\[\{ E \in \mapF{\mapG{N}} \mid E \prec_\mapFonly \mapF{\mprem(\iota_\levelG)} \}\modelsolam \mapF{\concl(\iota_\levelG)}\]
and thus ($*$).

\medskip

\noindent\textsc{Case 1b:}\enskip
The subterm $u_\levelH$ at position $L_\levelH.s_\levelH.p_\levelH$ of $C_\levelH$ is not a variable
or there exists an occurrence of that variable inside of a parameter in $C_\levelH$.

Then we construct a $\Sup$ inference $\iota_\levelH$ from 
$D_\levelH = \mapGonly_{N}^{-1}(D_\levelG)$,
and
$C_\levelH = \mapGonly_{N}^{-1}(C_\levelG)$.
Let $\rho$ be the grounding substitution such that $D_\levelG = D_\levelH\rho$ and
let $\theta$ be the grounding substitution such that $C_\levelG = C_\levelH\theta$.

The assumption of this cases matches exactly condition~\REF{2} of $\Sup$.

Condition~\REF{5} of $\GSup$ states that $t_\levelG \ceq t'_\levelG$ is strictly eligible in $D_\levelG$.
Let $D_\levelH = D'_\levelH \llor t_\levelH \ceq t'_\levelH$,
where $t_\levelH \ceq t'_\levelH$
is the literal that Lemma~\ref{lem:H:eligibility-lifting} guarantees to be strictly eligible
\wrt\ any suitable $\sigma$ (condition~\ref{sup:six} of $\Sup$),
with $t_\levelH\rho = t_\levelG$ and $t'_\levelH\rho = t'_\levelG$.
Condition~\REF{6} of $\GSup$ states that 
if $t'_\levelG$ is Boolean, then $t'_\levelG = \itrue$.
Thus, there are no selected literals in $D_\levelG$.
By Definition~\ref{def:H:sel-transfer},
it follows that there are no selected literals in $D_\levelH$
(condition~\ref{sup:seven} of $\Sup$).

Since $L_\levelG = L_\levelH\theta$, $s_\levelG = s_\levelH\theta$,
and $p_\levelG = p_\levelH$, we have
$u_\levelH\theta = t_\levelG = t_\levelH\rho$.
Since the variables of $D_\levelH$ and $C_\levelH$ are disjoint,
there exists a subsitution $\theta \cup \rho$
that matches $\theta$ on all variables of $C_\levelH$ and $\rho$ on all variables of $D_\levelH$.
This substitution is a unifier of $t_\levelH \equiv u_\levelH$.
By Definition~\ref{def:csu},
there exists a unifier $\sigma \in \csu(t_\levelH \equiv u_\levelH)$
and a substitution $\zeta$
such that $x\sigma\zeta = x(\theta \cup \rho)$ for all variables $x$ in $C_\levelH$ or $D_\levelH$ (condition~\REF{1} of $\Sup$).

Since $t_\levelG$ is nonfunctional by condition~\REF{1} of $\GSup$,
and since $u_\levelH\sigma\zeta = u_\levelH\theta = t_\levelG$,
$u_\levelH\sigma$ is nonfunctional (condition~\REF{3} of $\Sup$).

By condition~\REF{2} of $\GSup$, $t_\levelH\sigma\zeta = t_\levelG \succ t'_\levelG = t'_\levelH\sigma\zeta$,
and thus by Lemma~\ref{lem:H:order-lifting},
$t_\levelH\sigma \not\preceq t'_\levelH\sigma$ (condition~\REF{4} of $\Sup$).

Moreover, $\concl(\iota_\levelH)\zeta = \concl(\iota_\levelG)$.
Thus, we can apply ($**$).

\medskip

\noindent\textsc{Case 2:}\enskip
$p_\levelG = p_\levelH.q$ for some nonempty $q$, 
the subterm $u_\levelH$ at position $L_\levelH.s_\levelH.p_\levelH$ of $C_\levelH$ is
variable-headed, and $u_\levelH\theta$ is nonfunctional.

\medskip

\noindent\textsc{Case 2a:}\enskip
The subterm at position $L_\levelH.s_\levelH.p_\levelH$ of $C_\levelH$ is a variable
and no occurrence of that variable is inside of a parameter in $C_\levelH$.

Then we can proceed as in Case 1a.

\medskip

\noindent\textsc{Case 2b:}\enskip
The subterm $u_\levelH$ at position $L_\levelH.s_\levelH.p_\levelH$ of $C_\levelH$ is not a variable
or there exists an occurrence of that variable inside of a parameter in $C_\levelH$.

Then we construct a $\FluidSup$ inference $\iota_\levelH$ from 
$D_\levelH = \mapGonly_{N}^{-1}(D_\levelG)$,
and
$C_\levelH = \mapGonly_{N}^{-1}(C_\levelG)$.
Let $\rho$ be the grounding substitution such that $D_\levelG = D_\levelH\rho$ and
let $\theta$ be the grounding substitution such that $C_\levelG = C_\levelH\theta$.

The assumptions of this case imply condition~\REF{2} of $\FluidSup$.
We define $D'_\levelH$, $t_\levelH$, and $t'_\levelH$
in the same way as in Case 1b.
This ensures condition~\REF{6} and condition~\REF{7} of $\FluidSup$.

Let $z$ be a fresh variable (condition~\REF{8} of $\FluidSup$).
Let $v = \lambda\>\greensubterm{(u_\levelH\theta)}{\DB{n}}_q$,
where $n$ is the appropriate De Bruijn index to refer to the initial $\lambda$.
We define $\theta'$
by $z\theta' = v$,
$x\theta' = x\rho$ for all variables $x$ in $D_\levelH$
and $x\theta' = x\theta$ for all other variables $x$.
Then,
$(z\>t_\levelH)\theta' =
v\>(t_\levelH\rho) =
v\>t_\levelG =
\greensubterm{(u_\levelH\theta)}{t_\levelG}_q 
= u_\levelH\theta
= u_\levelH\theta'$.
So $\theta'$ is a unifier of $z\>t_\levelH$ and $u_\levelH$.
Thus, by definition of $\csu$ (Definition~\ref{def:csu}),
there exists a unifier $\sigma \in \csu(z\>t_\levelH \equiv u_\levelH)$ 
and a substitution $\zeta$ such that $x\sigma\zeta = x\theta'_\levelH$ 
for all relevant variables $x$
(condition~\REF{1} of $\FluidSup$).

The assumption of Case 2 tells us that $u_\levelH\theta = u_\levelH\sigma\zeta$ is nonfunctional.
It follows that $u_\levelH\sigma$ is nonfunctional (condition~\REF{3} of $\FluidSup$).

By condition~\REF{2} of $\GSup$, $t_\levelH\sigma\zeta = t_\levelG \succ t'_\levelG = t'_\levelH\sigma\zeta$,
and thus by Lemma~\ref{lem:H:order-lifting},
$t_\levelH\sigma \not\preceq t'_\levelH\sigma$ (condition~\REF{4} of $\FluidSup$).

By condition~\REF{2} of $\GSup$, $t_\levelH\rho = t_\levelG \ne t'_\levelG = t'_\levelH\rho$.
Thus, $(z\>t_\levelH)\sigma\zeta =
v\>(t_\levelH\rho) =
\greensubterm{u_\levelH\theta}{t_\levelH\rho}_q \ne
\greensubterm{u_\levelH\theta}{t'_\levelH\rho}_q =
v\>(t'_\levelH\rho) = (z\>t'_\levelH)\sigma\zeta$.
So,
$(z\>t'_\levelH)\sigma \ne (z\>t_\levelH)\sigma$ (condition~\REF{9} of $\FluidSup$).

Since $z\sigma\zeta = v$ and $v \ne \lambda\>\DB{0}$
because $q$ is nonempty,
we have $z\sigma \ne \lambda\>\DB{0}$
(condition~\REF{10} of $\FluidSup$).

Moreover, $\concl(\iota_\levelH)\zeta 
= (D'_\levelH \llor \greensubterm{C_\levelH}{z\>t'_\levelH}_{p_\levelH})\sigma\zeta
= D'_\levelG \llor \greensubterm{C_\levelG}{\greensubterm{(u_\levelH\theta)}{t'_\levelG}_q}_{p_\levelH}
= D'_\levelG \llor  \greensubterm{C_\levelG}{t'_\levelG}_{p_\levelG}
= \concl(\iota_\levelG)$.
Thus, we can apply ($**$).

\medskip

\noindent{\GEqRes}:\enskip
Assume that $\iota_\levelG$ is a $\GEqRes$ inference
\[\namedinference{\GEqRes}
{\overbrace{C'_\levelG \llor {u_\levelG \cneq u_\levelG}}^{\vphantom{\cdot}\smash{C_\levelG}}}
{C'_\levelG}\]
Let $C_\levelH = \mapGonly_{N}^{-1}(C_\levelG)$,
and let $\theta$ be a grounding substitution such that $C_\levelG = C_\levelH\theta$.
Condition~\REF{1} of $\GEqRes$ states that $u_\levelG \cneq u_\levelG$ is strictly eligible in $C_\levelG$.
Let $C_\levelH = C'_\levelH \llor u_\levelH \cneq u'_\levelH$,
where $u_\levelH \cneq u'_\levelH$
is the literal that Lemma~\ref{lem:H:eligibility-lifting} guarantees to be strictly eligible \wrt\ any suitable $\sigma$
(condition~\REF{2} of $\GEqRes$),
with $u_\levelH\theta = u_\levelG$ and $u'_\levelH\theta = u_\levelG$.
Then $\theta$ is a unifier of $u_\levelH$ and $u'_\levelH$,
and thus there exists a unifier $\sigma \in \csu(u_\levelH \equiv u'_\levelH)$
and a substitution $\zeta$ such that $x\sigma\zeta = x\theta$ for all variables $x$ in $C_\levelH$ (condition~\REF{1} of $\GEqRes$).

Moreover, $\concl(\iota_\levelH)\zeta = \concl(\iota_\levelG)$.
Thus, we can apply ($**$).

\medskip

\noindent{\GEqFact}:\enskip
Assume that $\iota_\levelG$ is a $\GEqFact$ inference
\[\namedinference{\GEqFact}
{\overbrace{C'_\levelG \llor u_\levelG \ceq v'_\levelG \llor u_\levelG \ceq v_\levelG}^{\vphantom{\cdot}\smash{C_\levelG}}}
{C'_\levelG \llor v_\levelG \cneq v'_\levelG \llor u_\levelG \ceq v_\levelG}\]
Let $C_\levelH = \mapGonly_{N}^{-1}(C_\levelG)$,
and let $\theta$ be a grounding substitution such that $C_\levelG = C_\levelH\theta$.
Condition~\REF{1} of $\GEqFact$ states that $u_\levelG \ceq v_\levelG$ is maximal in $C_\levelG$.
Let $u_\levelH \ceq v_\levelH$ be the literal in $C_\levelH$
that Lemma~\ref{lem:H:maximality-lifting} guarantees to be strictly maximal \wrt\ any suitable $\sigma$
(condition~\REF{2} of $\EqFact$),
with $u_\levelH\theta = u_\levelG$ and $v_\levelH\theta = v_\levelG$.
Choose $C'_\levelH$, $u'_\levelH$, and $v'_\levelH$ such that
$C_\levelH = C'_\levelH \llor u'_\levelH \ceq v'_\levelH \llor u_\levelH \ceq v_\levelH$,
$C'_\levelH\theta = C'_\levelPG$,
$u'_\levelH\theta = u_\levelG$, and $v'_\levelH\theta = v'_\levelPG$.

Then $\theta$ is a unifier of $u_\levelH$ and $u'_\levelH$,
and thus there exists a unifier $\sigma \in \csu(u_\levelH \equiv u'_\levelH)$
and a substitution $\zeta$ such that $x\sigma\zeta = x\theta$ for all variables $x$ in $C_\levelH$ (condition~\REF{1} of $\EqFact$).

By condition~\REF{2} of $\GEqFact$, there are no selected literals 
in $C_\levelG$ and thus there are no selected literals in $C_\levelH$ (condition~\REF{3} of $\EqFact$).

By condition~\REF{3} of $\GEqFact$, $u_\levelH\sigma\zeta = u_\levelG \succ v_\levelG = v_\levelH\sigma\zeta$.
By Lemma~\ref{lem:H:order-lifting}, $u_\levelH\sigma \not\preceq v_\levelH\sigma$ (condition~\REF{4} of $\EqFact$).

Moreover, $\concl(\iota_\levelH)\zeta = \concl(\iota_\levelG)$.
Thus, we can apply ($**$).

\medskip

\noindent{\GClausify}:\enskip
Assume $\iota_\levelG$ is a $\GClausify$ inference
\[\namedinference{\GClausify}{\overbrace{C'_\levelG \llor s_\levelG \ceq t_\levelG}^{C_\levelG}}
{C'_\levelG \llor D_\levelG}\]
with $\tau_\levelG$ being the type and $u_\levelG$ and $v_\levelG$ being the terms used for condition~\REF{2}.
Let $C_\levelH = \mapGonly_{N}^{-1}(C_\levelG)$,
and let $\theta$ be a grounding substitution such that $C_\levelG = C_\levelH\theta$.

Condition~\REF{1} of $\GClausify$ is that $s_\levelG \ceq t_\levelG$ is strictly eligible in $C_\levelG$.
Let $C_\levelH = C'_\levelH \llor s_\levelH \ceq t_\levelH$,
where $s_\levelH \ceq t_\levelH$ is the literal that Lemma~\ref{lem:H:eligibility-lifting} guarantees to be strictly eligible
\wrt\ any suitable $\sigma$ (condition~\REF{2} of $\Clausify$),
with $s_\levelH\theta = s_\levelG$ and $t_\levelH\theta = t_\levelG$.

We distinguish two cases:

\noindent\textsc{Case 1:}\enskip
$s_\levelH$ is a variable
and no occurrence of that variable is inside of a parameter in $C_\levelH$.

Then we define substitutions $\theta_{\itruelight}$ and $\theta_{\ifalselight}$
that coincide with $\theta$, except for mapping $s_\levelH$ to $\itrue$ and $\ifalse$, respectively.
Since the semantics of $\modelsolam$ interpret Booleans,
we have $\{\mapF{C_\levelH\theta_{\itruelight}}, \mapF{C_\levelH\theta_{\ifalselight}}\} \modelsolam \mapF{C_\levelH\theta}$
because $s_\levelH$ does not appear in parameters in $C_\levelH$.
By \ref{cond:PF:order:t-f-minimal} and \ref{cond:PF:order:comp-with-contexts}, $\mapF{C_\levelH\theta_{\itruelight}} \prec \mapF{C_\levelH\theta_{\ifalselight}} \prec \mapF{C_\levelH\theta}$.
Therefore, we can apply ($*$).

\medskip

\noindent\textsc{Case 2:}\enskip
$s_\levelH$ is not a variable
or there exists an occurrence of that variable inside of a parameter in $C_\levelH$
(condition~\REF{3} of $\Clausify$).
Then we construct a corresponding $\Clausify$ inference $\iota_\levelH$.

Comparing the listed triples in $\GClausify$ and $\Clausify$,
we see that there must be a triple
$(s'_\levelH, t'_\levelH, D_\levelH)$ listed for $\Clausify$
such that
$(s'_\levelH\rho, t'_\levelH\rho, D_\levelH\rho) = (s_\levelG, t_\levelG, D_\levelG)$
with $\rho = \{\alpha \mapsto \tau_\levelG, x\mapsto u_\levelG,y\mapsto v_\levelG\}$
is the triple used for $\iota_\levelG$ (condition~\REF{4} of $\Clausify$).

Moreover, we observe that $s_\levelH\theta = s_\levelG = s'_\levelH\rho$
and $t_\levelH\theta = t_\levelG = t'_\levelH\rho$.
Thus the substitution $\theta'$ mapping all variables $x$ in
$s'_\levelH$ and $t'_\levelH$ to $x\rho$ 
and all other variables $x$ to $x\theta$ is a unifier of $s_\levelH \equiv s'_\levelH$ and $t_\levelH \equiv t'_\levelH$.
So 
there exists a unifier $\sigma\in\csu(s_\levelH \equiv s'_\levelH, t_\levelH \equiv t'_\levelH)$ (condition~\REF{1} of $\Clausify$)
and a substitution $\zeta$
such that $x\sigma\zeta = x\theta'$ for all variables $x$ in $C_\levelH$.

Moreover, it follows that $\concl(\iota_\levelH)\zeta = \concl(\iota_\levelG)$.
Thus, we can apply ($**$).

\medskip

\noindent{\GBoolHoist}:\enskip
Analogous to $\GSup$, using the substitutions $\theta_{\itruelight}$ and $\theta_{\ifalselight}$
as in Case 1 of $\GClausify$.

\medskip

\noindent{\GLoobHoist}:\enskip
Analogous to $\GBoolHoist$.

\medskip

\noindent{\GFalseElim}:\enskip
Analogous to $\GEqRes$.

\medskip

\noindent{\GArgCong}:\enskip
Analogous to $\GEqRes$.

\medskip

\noindent{\GExt}:\enskip
Assume that $\iota_\levelG$ is a $\GExt$ inference
\begin{align*}
  \namedinference{\IGExt}{\greensubterm{C_\levelG}{u_\levelG}}
  {\greensubterm{C_\levelG}{w_\levelG} \llor u_\levelG\>\diff\typeargs{\tau_\levelG,\upsilon_\levelG}(u_\levelG,w_\levelG)\noteq w_\levelG\>\diff\typeargs{\tau_\levelG,\upsilon_\levelG}(u_\levelG,w_\levelG)}
\end{align*}
Let $C_\levelH = \mapGonly_{N}^{-1}(C_\levelG)$,
and let $\theta$ be a grounding substitution such that $C_\levelG = C_\levelH\theta$.
By condition~\REF{1} of $\GExt$, the position $L_\levelG.s_\levelG.p_\levelG$ of $u_\levelG$ is eligible in $C_\levelG$.
By Lemma~\ref{lem:H:eligibility-lifting},
there exists 
a green position $L_\levelH.s_\levelH.p_\levelH$ of $C_\levelH$
such that
\begin{itemize}
  \item $L_\levelG = L_\levelH\theta$;
  \item $s_\levelG = s_\levelH\theta$;
  \item given substitutions $\sigma$ and $\zeta$ with $x\theta = x\sigma\zeta$ for all variables $x$ in $C_\levelH\slimfull{}{\constraint{S}}$,
  $L_\levelH.s_\levelH.p_\levelH$ is eligible in $C_\levelH$  \wrt\ $\sigma$ and $\mathit{hsel}$
  (condition~\REF{3} of $\Ext$/condition~\ref{fluidext:eligible} of $\FluidExt$);
\end{itemize}
and one of the following cases applies:
\medskip

\noindent\textsc{Case 1:}\enskip
$p_\levelG = p_\levelH$.
Then we construct an $\Ext$ inference $\iota_\levelH$ from $C_\levelH$.

Let $u_\levelH$ be the subterm of $C_\levelH$ at position $L_\levelH.s_\levelH.p_\levelH$.
Since $p_\levelG = p_\levelH$, we have $u_\levelH\theta = u_\levelG$.
By condition~\REF{2} of $\GExt$, $u_\levelG$ is functional and thus
there exists a most general type substitution $\sigma$ such that $u_\levelH\sigma$ is functional (condition~\REF{1} of $\Ext$).
By definition of `most general',
there exists a substitution $\zeta$
such that $x\sigma\zeta = x\theta$ for all variables $x$ in $C_\levelH$.

Let $y$ be a fresh variable of the same type as $u_\levelH\sigma$ (condition~\REF{2} of $\Ext$).
Let $\zeta' = \zeta[y\mapsto w_\levelG]$.
Then, $\concl(\iota_\levelH)\zeta' = \concl(\iota_\levelG)$.
Thus, we can apply ($**$).

\medskip

\noindent\textsc{Case 2:}\enskip
$p_\levelG = p_\levelH.q$ for some nonempty $q$, 
the subterm $u_\levelH$ at position $L_\levelH.s_\levelH.p_\levelH$ of $C_\levelH$ is
variable-headed, and $u_\levelH\theta$ is nonfunctional.

Then we construct a $\FluidSup$ inference $\iota_\levelH$ from
$C_\levelH = \mapGonly_{N}^{-1}(C_\levelG)$.

The assumption of this case implies condition~\ref{fluidext:var} of $\FluidExt$.

Let $\alpha$ and $\beta$ be fresh type variables.
Let $x$ and $y$ be fresh variables of type $\alpha \fun \beta$,
and let $z$ be a fresh variable of function type from $\alpha \fun \beta$ to the type of $u_\levelH$
(condition~\ref{fluidext:fresh} of $\FluidExt$).

Let $v = \lambda\>\greensubterm{(u_\levelH\theta)}{\DB{n}}_q$,
where $n$ is the appropriate De Bruijn index to refer to the initial $\lambda$.
We define a substitution$\theta'$
that coincides with $\theta$ on all variables except
$\alpha\theta' = \tau_\levelG$,
$\beta\theta' = \upsilon_\levelG$,
$z\theta' = v$,
$x\theta' = u_\levelG$, and $y\theta' = w_\levelG$.
Then,
$(z\>x)\theta' =
v\>u_\levelG =
\greensubterm{(u_\levelH\theta)}{u_\levelG}_q
= u_\levelH\theta
= u_\levelH\theta'$.
So $\theta'$ is a unifier of $z\>x$ and $u_\levelH$.
Thus, by definition of $\csu$ (Definition~\ref{def:csu}),
there exists a unifier $\sigma \in \csu(z\>x \equiv u_\levelH)$ 
and a substitution $\zeta$ such that $x\sigma\zeta = x\theta'$ 
for all relevant variables $x$
(condition~\ref{fluidext:csu} of $\FluidExt$).

The assumption of Case 2 tells us that $u_\levelH\theta = u_\levelH\sigma\zeta$ is nonfunctional.
It follows that $u_\levelH\sigma$ is nonfunctional (condition~\ref{fluidext:nonfunctional} of $\FluidExt$).

By condition~\REF{4} of $\GExt$, $u_\levelG \ne w_\levelG$.
Thus, $(z\>x)\sigma\zeta =
v\>u_\levelG =
\greensubterm{(u_\levelH\theta)}{u_\levelG}_q \ne
\greensubterm{(u_\levelH\theta)}{w_\levelG}_q =
v\>w_\levelG = (z\>y)\sigma\zeta$.
So,
$(z\>x)\sigma \ne (z\>y)\sigma$ (condition~\ref{fluidext:csu-not-id} of $\FluidExt$).

Since $z\sigma\zeta = v$ and $v \ne \lambda\>\DB{0}$
because $q$ is nonempty,
we have $z\sigma \ne \lambda\>\DB{0}$
(condition~\ref{fluidext:csu-not-proj} of $\FluidExt$).

Moreover,
\begin{align*}
\concl(\iota_\levelH)\zeta 
&= (\greensubterm{C_\levelH}{z\>y}\llor x\>\diff\typeargs{\alpha,\beta}(x,y)\noteq y\>\diff\typeargs{\alpha,\beta}(x,y))\sigma\zeta\\
&= \greensubterm{C_\levelG}{\greensubterm{(u_\levelH\theta)}{w_\levelG}_q}_{p_\levelH} \llor u_\levelG\>\diff\typeargs{\tau_\levelG,\upsilon_\levelG}(u_\levelG,w_\levelG)\noteq w_\levelG\>\diff\typeargs{\tau_\levelG,\upsilon_\levelG}(u_\levelG,w_\levelG)\\
&= \greensubterm{C_\levelG}{w_\levelG}_{p_\levelG} \llor u_\levelG\>\diff\typeargs{\tau_\levelG,\upsilon_\levelG}(u_\levelG,w_\levelG)\noteq w_\levelG\>\diff\typeargs{\tau_\levelG,\upsilon_\levelG}(u_\levelG,w_\levelG)\\
&= \concl(\iota_\levelG)
\end{align*}
Thus, we can apply ($**$).

\medskip

\noindent{\GDiff}:\enskip
Assume that $\iota_\levelPG$ is a $\GDiff$ inference
\begin{align*}
  \namedinference{\GDiff}{}
  {u\>\diff\typeargs{\tau,\upsilon}(u,w) \noteq u\>\diff\typeargs{\tau,\upsilon}(u,w) \llor u\>s \eq w\>s}
\end{align*}
Then we use the following $\Diff$ inference $\iota_\levelH$:
\begin{gather*}
  \namedinference{\Diff}{}
  {{y\>(\diff\typeargs{\alpha,\beta}(y,z))\cneq z\>(\diff\typeargs{\alpha,\beta}(y,z)) \llor y\>x\ceq z\>x}}
\end{gather*}
Clearly, $\concl(\iota_\levelH)\theta = \concl(\iota_\levelPG)$
for $\theta = \{\alpha \mapsto \tau, \beta \mapsto \upsilon, y \mapsto u, z \mapsto w, x \mapsto s\}$.
Thus, we can apply ($**$).
\end{proof}
\end{slim}
\begin{full}
\begin{proof}
Let $\iota_\levelPG$ be a $\PGInf$ inference from $\mapPG{N}$.
We must show that $\iota_\levelPG \in \PGRedI(\mapPG{N})$.
It suffices to construct a $\HInf$ inference $\iota_\levelH$ from $N$ such that $\iota_\levelH$ and $\iota_\levelPG$ are of the form\strut
\[\namedinference{$\iota_\levelH$}
{C_1\constraint{S_1}\quad\cdots\quad C_m\constraint{S_m}}
{C_{m+1}\constraint{S_{m+1}}}
\quad\quad
\namedinference{$\iota_\levelPG$}
{\mapP{C_1\closure\theta_1}\quad\dots\quad \mapP{C_m\closure\theta_m}}
{E\closure\xi}\tag{$*$}\]
for some $C_1\constraint{S_1},\cdots, C_{m+1}\constraint{S_{m+1}} \in \clausesH$, $E\closure\xi\in\clausesPG$, and 
grounding substitutions
$\theta_1, \dots,\allowbreak \theta_{m+1}$ such that $S_1\theta_1$,\dots,$S_{m+1}\theta_{m+1}$ are true
and $C_{m+1}\mapp{\theta_{m+1}} = E\pi$ and
$x\pi\mapq{\theta_{m+1}} = x\xi$ for some substitution $\pi$ and all variables $x$ in $E$.

Here is why this suffices:
By Definition~\ref{def:PG:RedI},
it suffices to show that
for all confluent term rewrite systems $R$ on $\termsF$ oriented by $\succ_{\mapIonly\mapFonly}$
whose only Boolean normal forms are $\itrue$ and $\ifalse$,
such that $E\closure\xi$ is variable-irreducible \wrt\ $R$, we have
\[
  R \cup O \modelsolam \mapF{\mapI{E\closure\xi}}
\]
where $O = \irred_R(\mapF{\mapI{\mapG{N}}})$ if $\iota_\levelPG$ is a $\Diff$ inference,
and otherwise $O = \{ E \in \irred_R(\mapF{\mapI{\mapG{N}}}) \mid E \prec \mapF{\mapI{\mapP{C_m\closure\theta_m}}} \}$.
Let $R$ be such a rewrite system.
By Lemma~\ref{lem:instance-varirred},
since $E\closure\xi$ is variable-irreducible,
using $E\pi = C_{m+1}\mapp{\theta_{m+1}}$, and
$x\xi = x\pi\mapq{\theta_{m+1}}$,
also
$C_{m+1}\mapp{\theta_{m+1}}\closure\mapq{\theta_{m+1}}$ 
and thus $C_{m+1}\closure\theta_{m+1}$ are variable-irreducible \wrt\ $R$.
By saturation, $\iota_\levelH \in \HRedI(N)$,
and thus, by definition of $\HRedI$, (Definition~\ref{def:H:RedI}),
it suffices to show that $\iota_\levelH$ is rooted in $\FInf{}$ for $(\theta_1, \dots, \theta_{m+1})$
(Definition~\ref{def:fol-inferences}),
which holds by Lemma~\ref{lem:PG:mapI-mapF-fol}, Definition~\ref{def:PG:RedI}
and the fact that 
$C_{m+1}\theta_{m+1} = C_{m+1}\mapp{\theta_{m+1}}\mapq{\theta_{m+1}} =  E\pi\mapq{\theta_{m+1}} = E\xi$.

\smallskip

For most rules, the following special case of ($*$) suffices:
We construct a $\HInf$ inference $\iota_\levelH$ from $N$ such that
$\iota_\levelH$ and $\iota_\levelPG$ are of the form\strut
\[\namedinference{$\iota_\levelH$}
{C_1\constraint{S_1}\quad\cdots\quad C_m\constraint{S_m}}
{C_{m+1}'\sigma\>\constraint{S_{m+1}'}}
\quad\quad
\namedinference{$\iota_\levelPG$}
{\mapP{C_1\closure\theta_1}\quad\cdots\quad \mapP{C_m\closure\theta_m}}
{C_{m+1}'\mapp{\sigma\zeta}\closure\xi}\tag{$**$}\]
for some $C_1\constraint{S_1},\cdots, C_m\constraint{S_m}, C_{m+1}'\constraint{S_{m+1}'} \in \clausesH$ and substitutions
$\theta_1, \dots, \theta_m, \zeta, \sigma, \xi$ such that $S_1\theta_1$,\dots,$S_{m}\theta_{m}$, $S_{m+1}'\zeta$ are true
and $x\xi =x\mapq{\sigma\zeta}$ for all variables $x$ in $C_{m+1}'\mapp{\sigma\zeta}$.

Here is why this is a special case of ($*$) with
$C_{m+1} = C_{m+1}'\sigma$,
$S_{m+1} = S_{m+1}'$,
$E = C_{m+1}'\mapp{\sigma\zeta}$:
By Lemma~\ref{lem:G:mapp-comp-subst}, 
there exists a substitution $\pi$ such that
$\sigma\mapp{\zeta} = \mapp{\sigma\zeta}\pi$
and
$\mapq{\sigma\zeta} = \pi\mapq{\zeta}$.
Thus,
$C_{m+1}\mapp{\zeta} = C_{m+1}'\sigma\mapp{\zeta} = C_{m+1}'\mapp{\sigma\zeta}\pi = E\pi$
and
$x\pi\mapq{\zeta} = x\mapq{\sigma\zeta} = x\xi$ for all variables $x$ in $E$,
as required for ($*$).

\medskip

\noindent{\PGSup}:\enskip
If $\iota_\levelPG$ is a $\PGSup$ inference
\[\namedinference{\PGSup}
{\overbrace{D'_\levelPG \llor { t_\levelPG \ceq t'_\levelPG}}^{\vphantom{\cdot}\smash{D_\levelPG}} \closure\> \rho_\levelPG \hypsep
 \greensubterm{C_\levelPG}{u_\levelPG}\closure\theta_\levelPG}
{D'_\levelPG \llor \greensubterm{C_\levelPG}{t'_\levelPG}\>\closure(\rho_\levelPG\cup\theta_\levelPG)}\]
then we construct a corresponding $\Sup$ or $\FluidSup$ inference $\iota_\levelH$.
Let $D_\levelH\closure\rho_\levelH = \mapPonly_{\mapG{N}}^{-1}(D_\levelPG\closure\rho_\levelPG)$
and $C_\levelH\closure\theta_\levelH = \mapPonly_{\mapG{N}}^{-1}(C_\levelPG\closure\> \theta_\levelPG)$.
(See Definition~\ref{def:H:sel-transfer} for the definition of $\mapPonly_{\mapG{N}}^{-1}$.)
Let $D_\levelH\constraint{T} = \mapGonly_N^{-1}(D_\levelH\closure\> \rho_\levelH)$
and $C_\levelH\constraint{S} = \mapGonly_N^{-1}(C_\levelH\closure\> \theta_\levelH)$.
(See Definition~\ref{def:H:sel-transfer} for the definition of $\mapGonly_N^{-1}$.)
These clauses $D_\levelH\constraint{T}$ and $C_\levelH\constraint{S}$ will be the premises of $\iota_\levelH$.
Condition~\REF{7} of $\PGSup$ states that $t_\levelH \ceq t'_\levelH$ is strictly eligible in $D_\levelPG\closure\rho_\levelPG$.
Let $D_\levelH = D'_\levelH \llor t_\levelH \ceq t'_\levelH$,
where $t_\levelH \ceq t'_\levelH$
is the literal that Lemma~\ref{lem:H:maximality-lifting} guarantees to be strictly eligible
with $t_\levelH\mapp{\theta_\levelH} = t_\levelPG$ and $t'_\levelH\mapp{\theta_\levelH} = t'_\levelPG$.
Condition~\REF{8} of $\PGSup$ states that 
if $t'_\levelPG\rho_\levelPG$ is Boolean, then $t'_\levelPG\rho_\levelPG = \itrue$.
Thus, there are no selected literals in $D_\levelPG\closure\rho_\levelPG$.
By Definition~\ref{def:H:sel-transfer},
it follows that there are no selected literals in $D_\levelH\constraint{T}\closure\rho_\levelH$
(condition~\ref{sup:seven} of $\Sup$ or $\FluidSup$)
and that $t_\levelH \ceq t'_\levelH$ is strictly maximal
(condition~\ref{sup:six} of $\Sup$ or $\FluidSup$).
Let $L_\levelPG.s_\levelPG.p_\levelPG$ be the position of the green subterm $u_\levelPG$ in $C_\levelPG$.
Condition~\REF{6} of $\PGSup$ states that 
$L_\levelPG.s_\levelPG.p_\levelPG$ is eligible in $C_\levelPG\closure\theta_\levelPG$.
Let $L_\levelH.s_\levelH.p_\levelH$ be the position in $C_\levelH$
that Lemma~\ref{lem:H:eligibility-lifting} guarantees to be eligible in $C_\levelH\constraint{S}$ \wrt\ any suitable $\sigma$
(condition~\ref{sup:five} of $\Sup$ or $\FluidSup$).
Let $u_\levelH$ be the subterm of $C_\levelH$ at position $L_\levelH.s_\levelH.p_\levelH$.
By Lemma~\ref{lem:H:eligibility-lifting}, one of the following cases applies:

\medskip

\noindent\textsc{Case 1:}\enskip $p_\levelPG = p_\levelH$.
Then we construct a $\Sup$ inference.

Lemma~\ref{lem:H:eligibility-lifting} tells us
that $s_\levelPG = s_\levelH\mapp{\theta_\levelH}$ and thus
$u_\levelPG = u_\levelH\mapp{\theta_\levelH}$ because $p_\levelPG = p_\levelH$.
By condition~\REF{1} of $\PGSup$, $t_\levelPG\rho_\levelPG = u_\levelPG\theta_\levelPG$.
It follows that
$t_\levelH\rho_\levelH = t_\levelH \mapp{\rho_\levelH}\mapq{\rho_\levelH} =
t_\levelPG\rho_\levelPG = u_\levelPG\theta_\levelPG =
u_\levelH \mapp{\theta_\levelH}\mapq{\theta_\levelH} = u_\levelH\theta_\levelH$.
Let $\rho_\levelH \cup \theta_\levelH$ be the substitution that coincides with 
$\rho_\levelH$ on all variables in $D_\levelH\constraint{T}$ and 
with $\theta_\levelH$ on all other variables.
Then $\rho_\levelH \cup \theta_\levelH$ is a unifier of $t_\levelH$ and $u_\levelH$.
Moreover, by construction, $T\rho_\levelH$ and $S\theta_\levelH$ are true.
Thus, by definition of $\csuupto$ (Definition~\ref{def:csu-upto}),
there exists $(\sigma, U) \in \csuupto(T,S,t_\levelH \equiv u_\levelH)$ 
(condition~\ref{sup:one} of $\Sup$)
and a substitution $\zeta$ such that
$U\zeta$ is true and
$x\sigma\zeta = x(\rho_\levelH \cup \theta_\levelH)$ 
for all relevant variables $x$.

Conditions~\REF{2} and \REF{3} of $\PGSup$ state
that $u_\levelPG = u_\levelH\mapp{\theta_\levelH}$ is not a variable
and nonfunctional.
Thus, by definition of $\mapponly$,
$u_\levelH$ is not a variable and $u_\levelH\sigma$ is nonfunctional
(conditions \ref{sup:two}~and~\ref{sup:three} of $\Sup$).

By Lemma~\ref{lem:H:order-lifting},
the condition $(t_\levelH\constraint{T})\not\preceq(t'_\levelH\constraint{T})$
(condition~\ref{sup:four} of $\Sup$)
follows from condition~\REF{4} of $\PGSup$.

Finally, the given inference $\iota_\levelPG$ and 
the constructed inference $\iota_\levelH$ are of the form ($**$) with
$C_1\constraint{S_1} = D_\levelH\constraint{T}$,
$C_2\constraint{S_2} = C_\levelH\constraint{S}$,
$C_3' = D'_\levelH \llor \greensubterm{C_\levelH}{t'_\levelH}_{p_\levelH}$
$S_3' = U$,
$\theta_1 = \rho_\levelH$,
$\theta_2 = \theta_\levelH$, and
$\xi = \rho_\levelPG \cup \theta_\levelPG$.

\medskip

\noindent\textsc{Case 2:}\enskip $p_\levelPG = p_\levelH.q$ for some nonempty $q$,
$u_\levelH$ is not a variable but is variable-headed,
and $u_\levelH\theta$ is nonfunctional.
Then we construct a $\FluidSup$ inference.

Lemma~\ref{lem:H:eligibility-lifting} tells us
that $s_\levelPG = s_\levelH\mapp{\theta_\levelH}$.
Thus, the subterm of $s_\levelPG$ at position $p_\levelH$ is
$u_\levelH\mapp{\theta_\levelH}$.
So $q$ is a green position of
$u_\levelH\mapp{\theta_\levelH}$,
and the subterm at that position is $u_\levelPG$---i.e.,
$u_\levelH\mapp{\theta_\levelH} = \greensubterm{(u_\levelH\mapp{\theta_\levelH})}{u_\levelPG}_q$.
Let $v = \lambda\>\greensubterm{(u_\levelH\mapp{\theta_\levelH})}{\DB{n}}_q$,
where $n$ is the appropriate De Bruijn index to refer to the initial $\lambda$.

Let $z$ be a fresh variable (condition~\REF{8} of $\FluidSup$).
We define $\theta'_\levelH$
by $z\theta'_\levelH = v\theta_\levelPG$,
$x\theta'_\levelH = x\rho_\levelH$ for all variables $x$ in $D_\levelH\constraint{T}$
and $x\theta'_\levelH = x\theta_\levelH$ for all other variables $x$.
Then, using condition~\REF{1} of $\PGSup$, $z\>t_\levelH\theta'_\levelH =
v\theta_\levelPG\>(t_\levelH\rho_\levelH) =
v\theta_\levelPG\>(t_\levelPG\rho_\levelPG) =
v\theta_\levelPG\>(u_\levelPG\theta_\levelPG) =
(v\>u_\levelPG)\theta_\levelPG = 
\greensubterm{(u_\levelH\mapp{\theta_\levelH})}{u_\levelPG}_q \theta_\levelPG
= u_\levelH\mapp{\theta_\levelH}\theta_\levelPG
= u_\levelH\theta_\levelH
= u_\levelH\theta'_\levelH$.
So $\theta'_\levelH$ is a unifier of $z\>t_\levelH$ and $u_\levelH$.
Thus, by definition of $\csu$ (Definition~\ref{def:csu}),
there exists a unifier $\sigma \in \csu(z\>t_\levelH \equiv u_\levelH)$ 
and a substitution $\zeta$ such that $x\sigma\zeta = x\theta'_\levelH$ 
for all relevant variables $x$
(condition~\REF{1} of $\FluidSup$).

By the assumption of this case, 
$u_\levelH$ is not a variable but is variable-headed (condition~\REF{2} of $\FluidSup$).

Since $q$ is a green position of
$u_\levelH\mapp{\theta_\levelH}$, the type of $u_\levelH\mapp{\theta_\levelH}$ and
the type of $u_\levelH\sigma$
is nonfunctional (condition~\REF{3} of $\FluidSup$).

By Lemma~\ref{lem:H:order-lifting},
the condition $t_\levelH\constraint{T}\not\preceq t'_\levelH\constraint{T}$
(condition~\REF{4} of $\FluidSup$)
follows from condition~\REF{4} of $\PGSup$.

By condition~\REF{4} of $\PGSup$, $t_\levelH\rho_\levelH = t_\levelPG\rho_\levelPG \ne t'_\levelPG\rho_\levelPG = t'_\levelH\rho_\levelH$.
Thus, $(z\>t_\levelH)\sigma\zeta =
v\theta_\levelPG\>(t_\levelH\rho_\levelH) =
\greensubterm{(u_\levelH\mapp{\theta_\levelH})\theta_\levelPG}{t_\levelH\rho_\levelH}_q \ne
\greensubterm{(u_\levelH\mapp{\theta_\levelH})\theta_\levelPG}{t'_\levelH\rho_\levelH}_q =
v\theta_\levelPG\>(t'_\levelH\rho_\levelH) = (z\>t'_\levelH)\sigma\zeta$.
So,
$(z\>t'_\levelH)\sigma \ne (z\>t_\levelH)\sigma$ (condition~\REF{9} of $\FluidSup$).

Since $z\sigma\zeta = v\theta_\levelPG$ and $v\theta_\levelPG \ne \lambda\>\DB{0}$
because $q$ is nonempty,
we have $z\sigma \ne \lambda\>\DB{0}$
(condition~\REF{10} of $\FluidSup$).

The inferences $\iota_\levelPG$ and $\iota_\levelH$ are of the form ($*$) with
$C_1\constraint{S_1} = D_\levelH\constraint{T}$,
$C_2\constraint{S_2} = C_\levelH\constraint{S}$,
$C_3\constraint{S_3} = (D'_\levelH \llor \greensubterm{C_\levelH}{z\>t'_\levelH}_{p_\levelH}\constraint{T,S})\sigma$,
$\theta_1 = \rho_\levelH$,
$\theta_2 = \theta_\levelH$,
$\theta_3 = \zeta$,
$E = D'_\levelPG\llor\greensubterm{C_\levelPG}{t'_\levelPG}_{p_\levelPG}$, and
$\xi = \rho_\levelPG\cup\theta_\levelPG$.
In the following, we elaborate why
$C_3\mapp{\theta_3} = E\pi$ and
$x\pi\mapq{\theta_3} = x\xi$ for 
some substitution $\pi$ and all variables $x$ in $E$,
as required for ($*$).
We invoke Lemma~\ref{lem:G:mapp-comp-subst} to
obtain a substitution $\pi$ such that
$\sigma\mapp{\zeta} = \mapp{\sigma\zeta}\pi$
and
$\mapq{\sigma\zeta} = \pi\mapq{\zeta}$.
Then
\begin{align*}
  C_3\mapp{\theta_3}
  &= (D'_\levelH \llor \greensubterm{C_\levelH}{z\>t'_\levelH}_{p_\levelH})\sigma\mapp{\zeta}\\
  &=(D'_\levelH\sigma\mapp{\zeta} \llor \greensubterm{C_\levelH\sigma\mapp{\zeta}}{(z\>t'_\levelH)\sigma\mapp{\zeta}}_{p_\levelH})\\
  &\stackrel{(1)}{=}(D'_\levelH\sigma\mapp{\zeta} \llor \greensubterm{C_\levelH\sigma\mapp{\zeta}}{\greensubterm{(z\>t_\levelH)\sigma\mapp{\zeta}}{t'_\levelH\sigma\mapp{\zeta}}_q}_{p_\levelH})\\
  &\stackrel{(2)}{=}(D'_\levelH\sigma\mapp{\zeta} \llor \greensubterm{C_\levelH\sigma\mapp{\zeta}}{\greensubterm{u_\levelH\sigma\mapp{\zeta}}{t'_\levelH\sigma\mapp{\zeta}}_q}_{p_\levelH})\\
  &\stackrel{(3)}{=}(D'_\levelH\sigma\mapp{\zeta} \llor \greensubterm{C_\levelH\sigma\mapp{\zeta}}{t'_\levelH\sigma\mapp{\zeta}}_{p_\levelPG})\\
  &=(D'_\levelH\mapp{\sigma\zeta}\pi \llor \greensubterm{C_\levelH\mapp{\sigma\zeta}\pi}{t'_\levelH\mapp{\sigma\zeta}\pi}_{p_\levelPG})\\
  &=(D'_\levelH\mapp{\sigma\zeta} \llor \greensubterm{C_\levelH\mapp{\sigma\zeta}}{t'_\levelH\mapp{\sigma\zeta}}_{p_\levelPG})\pi\\
  &= (D'_\levelPG\llor\greensubterm{C_\levelPG}{t'_\levelPG}_{p_\levelPG})\pi\\
  &= E\pi
\end{align*}
Step (1) can be justified as follows: By Lemma~\ref{lem:G:mapp-comp-mapq},
$z\sigma\mapp{\zeta}\mapq{\zeta} = v\theta_\levelPG$.
Since $z\sigma\mapp{\zeta}$ and $v\theta_\levelPG$ contain only nonfunctional variables,
$z\sigma\mapp{\zeta}$ must be a $\lambda$-abstraction
whose $\lambda$ binds exactly one De Bruijn index,
which is located at orange position $1.q$ \wrt\ $\downarrow_{\beta\eta\mathrm{long}}$.
Thus, $(z\>t_\levelH)\sigma\mapp{\zeta}$
and $(z\>t'_\levelH)\sigma\mapp{\zeta}$
are identical up to the subterm at green position $q$,
which is $t_\levelH\sigma\mapp{\zeta}$ and $t'_\levelH\sigma\mapp{\zeta}$ respectively.
For step (2), we use that $(z\>t_\levelH)\sigma = u_\levelH\sigma$.
For step (3), we use that $u_\levelH$ is the green subterm of $C_\levelH$ at position $p_\levelH$.

\medskip

\noindent{\PGEqRes}:\enskip
If $\iota_\levelPG$ is an $\PGEqRes$ inference
\[\namedinference{\PGEqRes}
{\overbrace{C'_\levelPG \llor {u_\levelPG  \cneq u'_\levelPG}}^{\vphantom{\cdot}\smash{C_\levelPG}}\>\closure\>\theta_\levelPG}
{C'_\levelPG \closure\theta_\levelPG}\],
then we construct a corresponding $\EqRes$ inference $\iota_\levelH$.
Let $C_\levelH\closure\theta_\levelH = \mapPonly_{\mapG{N}}^{-1}(C_\levelPG\closure\theta_\levelPG)$
and $C_\levelH\constraint{S} = \mapGonly_N^{-1}(C_\levelH\closure\theta_\levelH)$.
This clause $C_\levelH\constraint{S}$ will be the premise of $\iota_\levelH$.

A condition of $\PGEqRes$ is that $u_\levelPG \cneq u'_\levelPG$ is eligible in $C_\levelPG\closure\theta_\levelPG$.
Let $C_\levelH = C'_\levelH \llor u_\levelH \cneq u'_\levelH$,
where $u_\levelH \cneq u'_\levelH$ is the literal that Lemma~\ref{lem:H:eligibility-lifting} guarantees to be eligible
\wrt\ any suitable $\sigma$,
with $u_\levelH\mapp{\theta_\levelH} = u_\levelPG$ and $u'_\levelH\mapp{\theta_\levelH} = u'_\levelPG$.

Another condition of $\PGEqRes$ states that $u_\levelPG\theta_\levelPG = u'_\levelPG\theta_\levelPG$.
Thus, $u_\levelH\theta_\levelH = u'_\levelH\theta_\levelH$, and therefore
there exists $(\sigma, U)\in\csuupto(S, u_\levelH \equiv u'_\levelH)$
and a substitution $\zeta$
such that $U\zeta$ is true
and $x\sigma\zeta = x\theta_\levelH$ for all variables $x$ in $C_\levelH\constraint{S}$.

Finally, the given inference $\iota_\levelPG$ and 
the constructed inference $\iota_\levelH$ are of the form ($**$) with
$C_1\constraint{S_1} = C_\levelH\constraint{S}$,
$C_2' = C'_\levelH$,
$S_2' = U$,
$\theta_1 = \theta_\levelH$,
and $\xi = \theta_\levelPG$.

\medskip

\noindent{\PGEqFact}:\enskip
If $\iota_\levelPG$ is an $\PGEqFact$ inference
\[\namedinference{\PGEqFact}
{\overbrace{C'_\levelPG \llor u'_\levelPG \ceq v'_\levelPG \llor {u_\levelPG} \ceq v_\levelPG}^{\vphantom{\cdot}\smash{C_\levelPG}}\>\closure\>\theta_\levelPG}
{C'_\levelPG \llor v_\levelPG \cneq v'_\levelPG \llor u_\levelPG \ceq v'_\levelPG\closure\theta_\levelPG}
\]
then we construct a corresponding $\EqFact$ inference $\iota_\levelH$.
Let $C_\levelH\closure\theta_\levelH = \mapPonly_{\mapG{N}}^{-1}(C_\levelPG\closure\theta_\levelPG)$
and $C_\levelH\constraint{S} = \mapGonly_N^{-1}(C_\levelH\closure\theta_\levelH)$.
This clause $C_\levelH\constraint{S}$ will be the premise of $\iota_\levelH$.

A condition of $\PGEqFact$ is that $u_\levelPG \ceq v_\levelPG\closure\theta_\levelPG$ is maximal in $C_\levelPG\closure\theta_\levelPG$.
Let $u_\levelH \ceq v_\levelH$ be the literal that Lemma~\ref{lem:H:maximality-lifting} guarantees to be maximal in $C_\levelH\constraint{S}$
\wrt\ any suitable $\sigma$,
with $u_\levelH\mapp{\theta_\levelH} = u_\levelPG$ and $v_\levelH\mapp{\theta_\levelH} = v_\levelPG$.
Choose $C'_\levelH$, $u'_\levelH$, and $v'_\levelH$ such that
$C_\levelH = C'_\levelH \llor u'_\levelH \ceq v'_\levelH \llor u_\levelH \ceq v_\levelH$,
$C'_\levelH\mapp{\theta_\levelH} = C'_\levelPG$,
$u'_\levelH\mapp{\theta_\levelH} = u'_\levelPG$, and $v'_\levelH\mapp{\theta_\levelH} = v'_\levelPG$.

Another condition of $\PGEqFact$ states that there are no selected literals in $C_\levelPG\closure\theta_\levelPG$.
By Definition~\ref{def:H:sel-transfer},
it follows that there are no selected literals in $C_\levelH\constraint{S}$.

Another condition of $\PGEqFact$ states that $u_\levelPG\theta_\levelPG = u'_\levelPG\theta_\levelPG$.
Thus, $u_\levelH\theta_\levelH = u'_\levelH\theta_\levelH$, and therefore
there exists $(\sigma, U)\in\csuupto(S, u_\levelH \equiv u'_\levelH)$
and a substitution $\zeta$
such that $U\zeta$ is true
and $x\sigma\zeta = x\theta_\levelH$ for all variables $x$ in $C_\levelH\constraint{S}$.

The last condition of $\PGEqFact$ is that $u_\levelPG\theta_\levelPG \succ v_\levelPG\theta_\levelPG$---i.e.,
$u_\levelH\sigma\zeta \succ v_\levelH\sigma\zeta$.
By Lemma~\ref{lem:H:order-lifting},
$(u_\levelH\constraint{S})\sigma \not\preceq (v_\levelH\constraint{S})\sigma$.

Finally, the given inference $\iota_\levelPG$ and 
the constructed inference $\iota_\levelH$ are of the form ($**$) with
$C_1\constraint{S_1} = C_\levelH\constraint{S}$,
$C_2' = C'_\levelH \llor v_\levelH \cneq v'_\levelH \llor u_\levelH \ceq v'_\levelH$,
$S_2' = U$,
$\theta_1 = \theta_\levelH$,
and $\xi = \theta_\levelPG$.

\medskip

\noindent{\PGClausify}:\enskip
If $\iota_\levelPG$ is a $\PGClausify$ inference
\[\namedinference{\PGClausify}{\overbrace{C'_\levelPG \llor s_\levelPG \ceq t_\levelPG}^{C_\levelPG}  \closure \theta_\levelPG}
{C'_\levelPG \llor D_\levelPG  \closure \theta_\levelPG}\]
with $\tau_\levelPG$ being the type and $u_\levelPG$ and $v_\levelPG$ being the terms used for condition~\REF{2}.
Then we construct a corresponding $\Clausify$ inference $\iota_\levelH$.
Let $C_\levelH \closure\theta_\levelH = \mapPonly_{\mapG{N}}^{-1}(C_\levelPG  \closure \theta_\levelPG)$
and $C_\levelH \constraint{S} = \mapGonly_N^{-1}(C_\levelH \closure\theta_\levelH)$.
This clause $C_\levelH \constraint{S}$ will be the premise of $\iota_\levelH$.

Condition~\REF{1} of $\PGClausify$ is that $s_\levelPG \ceq t_\levelPG$ is strictly eligible in $C_\levelPG\closure\theta_\levelPG$.
Let $C_\levelH = C'_\levelH \llor s_\levelH \ceq t_\levelH$,
where $s_\levelH \ceq t_\levelH$ is the literal that Lemma~\ref{lem:H:eligibility-lifting} guarantees to be strictly eligible
\wrt\ any suitable $\sigma$ (condition~\REF{2} of $\Clausify$),
with $s_\levelH\mapp{\theta_\levelH} = s_\levelPG$ and $t_\levelH\mapp{\theta_\levelH} = t_\levelPG$.

Comparing the listed triples in $\PGClausify$ and $\Clausify$,
we see that there must be a triple
$(s'_\levelH, t'_\levelH, D_\levelH)$ listed for $\Clausify$
such that
$(s'_\levelH\rho, t'_\levelH\rho, D_\levelH\rho) = (s_\levelPG, t_\levelPG\theta_\levelPG, D_\levelPG)$
with $\rho = \{\alpha \mapsto \tau_\levelPG, x\mapsto u_\levelPG,y\mapsto v_\levelPG\}$
is the triple used for $\iota_\levelPG$ (condition~\REF{4} of $\Clausify$).
Inspecting the listed triples, we see that $s_\levelPG$
cannot be a variable and that $s_\levelPG\theta_\levelPG = s_\levelH\theta_\levelH$
is of Boolean type.
It follows that 
$s_\levelH$ is not a variable (condition~\REF{3} of $\Clausify$)
because  if it were, then, by definition of $\mapponly$,
$s_\levelH\mapp{\theta_\levelH} = s_\levelPG$ would be a variable.

Moreover, we observe that $s_\levelH\theta_\levelH = s_\levelPG\theta_\levelPG = s'_\levelH\rho\theta_\levelPG$
and $t_\levelH\theta_\levelH = t_\levelPG\theta_\levelPG = t'_\levelH\rho = t'_\levelH\rho\theta_\levelPG$.
Thus the substitution mapping all variables $x$ in
$s'_\levelH$ and $t'_\levelH$ to $x\rho\theta_\levelPG$ 
and all other variables $x$ to $x\theta_\levelH$ is a unifier of $s_\levelH \equiv s'_\levelH$ and $t_\levelH \equiv t'_\levelH$.
So 
there exists a unifier $\sigma\in\csu(s_\levelH \equiv s'_\levelH, t_\levelH \equiv t'_\levelH)$ (condition~\REF{1} of $\Clausify$)
and a substitution $\zeta$
such that $x\sigma\zeta = x\theta_\levelH$ for all variables $x$ in $C_\levelH\constraint{S}$.

Finally, the given inference $\iota_\levelPG$ and 
the constructed inference $\iota_\levelH$ are of the form ($**$) with
$C_1\constraint{S_1} = C_\levelH\constraint{S}$,
$C_2' = C'_\levelH\llor D_\levelH$,
$S_2' = S\sigma$,
$\theta_1 = \theta_\levelH$,
and $\xi = \theta_\levelPG$.

\medskip

\noindent{\PGBoolHoist}:\enskip
If $\iota_\levelPG$ is a $\PGBoolHoist$ inference
\begin{align*}
  \namedinference{\PGBoolHoist}{\greensubterm{C_\levelPG}{u_\levelPG} \closure\theta_\levelPG}
  {\greensubterm{C_\levelPG}{\ifalse} \llor u_\levelPG \ceq \itrue\closure\theta_\levelPG}
\end{align*}
then we construct a corresponding $\BoolHoist$ or $\FluidBoolHoist$ inference $\iota_\levelH$.
Let $C_\levelH\closure\theta_\levelH = \mapPonly_{\mapG{N}}^{-1}(C_\levelPG\closure\theta_\levelPG)$
and $C_\levelH\constraint{S} = \mapGonly_N^{-1}(C_\levelH\closure\theta_\levelH)$.
This clause $C_\levelH\constraint{S}$ will be the premise of $\iota_\levelH$.
Let $L_\levelPG.s_\levelPG.p_\levelPG$ be the position of $u_\levelPG$ in $C_\levelPG$.
Condition~\REF{3} of $\PGBoolHoist$ states that $L_\levelPG.s_\levelPG.p_\levelPG$ is eligible in $C_\levelPG\closure\theta_\levelPG$.
Let $L_\levelH.s_\levelH.p_\levelH$ be the position that Lemma~\ref{lem:H:eligibility-lifting} guarantees to be eligible
in $C_\levelH\constraint{S}$ \wrt\ any suitable $\sigma$ (condition~\REF{3} of $\BoolHoist$ or condition~\REF{8} of $\FluidBoolHoist$).
Let $u_\levelH$ be the subterm at position $L_\levelH.s_\levelH.p_\levelH$ in $C_\levelH\constraint{S}$.
By Lemma~\ref{lem:H:eligibility-lifting}, one of the following cases applies:

\medskip

\noindent\textsc{Case 1:}\enskip
$p_\levelPG = p_\levelH$.
Then we construct a $\BoolHoist$ inference.

Lemma~\ref{lem:H:eligibility-lifting} tells us that $s_\levelPG = s_\levelH\mapp{\theta_\levelH}$
and thus $u_\levelPG = u_\levelH\mapp{\theta_\levelH}$ because $p_\levelPG = p_\levelH$.
By Condition~\REF{1} of $\PGBoolHoist$,
$u_\levelPG = u_\levelH\mapp{\theta_\levelH}$ is of Boolean type
and thus $u_\levelH\theta_\levelH$ is of Boolean type.
So there exists a 
most general type substitution
$\sigma$ such that $u_\levelH$ is Boolean (condition~\REF{1} of $\BoolHoist$)
and a substitution $\zeta$ such that $x\sigma\zeta = x\theta_\levelH$ for all variables $x$ in $C_\levelH\constraint{S}$.

By condition~\REF{2} of $\PGBoolHoist$,
$u_\levelPG = u_\levelH\mapp{\theta_\levelH}$ is
not a variable and is neither $\itrue$ nor $\ifalse$.
By condition~\REF{1} of $\PGBoolHoist$, $u_\levelPG= u_\levelH\mapp{\theta_\levelH}$
is nonfunctional.
So, using the definition of $\mapponly$, $u_\levelH$ is not a variable
and is neither $\itrue$ nor $\ifalse$ (condition~\REF{2} of $\BoolHoist$).

By condition~\REF{4} of $\PGBoolHoist$,
$L_\levelPG\theta_\levelPG$ is not of the form
$u_\levelPG\theta_\levelPG \ceq \itrue$ or $u_\levelPG\theta_\levelPG \ceq \ifalse$.
Since $L_\levelPG\theta_\levelPG = L_\levelH\theta_\levelH$
and $u_\levelPG\theta_\levelPG = u_\levelH\theta_\levelH$, it follows that
$L_\levelH$ is not of the form $u_\levelH \ceq \itrue$ or $u_\levelH \ceq \ifalse$ (condition~\REF{4} of $\BoolHoist$).

Finally, the given inference $\iota_\levelPG$ and 
the constructed inference $\iota_\levelH$ are of the form ($**$) with
$C_1\constraint{S_1} = C_\levelH\constraint{S}$,
$C_2' = \greensubterm{C_\levelH}{\ifalse}\llor u_\levelH \ceq \itrue$,
$S_2' = S\sigma$,
$\theta_1 = \theta_\levelH$,
and $\xi = \theta_\levelPG$.

\medskip

\noindent\textsc{Case 2:}\enskip
$p_\levelPG = p_\levelH.q$ for some nonempty $q$,
$u_\levelH$ is not a variable but is variable-headed,
and $u_\levelH\theta$ is nonfunctional.
Then we construct a $\FluidBoolHoist$ inference.
By the assumption of this case, $u_\levelH$ is not a variable but is variable-headed
(condition~\REF{1} of $\FluidBoolHoist$).

Lemma~\ref{lem:H:eligibility-lifting} tells us
that $s_\levelPG = s_\levelH\mapp{\theta_\levelH}$.
Thus, the subterm of $s_\levelPG$ at position $p_\levelH$ is
$u_\levelH\mapp{\theta_\levelH}$.
So $q$ is a green position of
$u_\levelH\mapp{\theta_\levelH}$,
and the subterm at that position is $u_\levelPG$.

Since $u_\levelPG$ is the green subterm at position $q$ of $u_\levelH\mapp{\theta_\levelH}$,
$u_\levelH\mapp{\theta_\levelH} = \greensubterm{(u_\levelH\mapp{\theta_\levelH})}{u_\levelPG}_q$.
Let $v = \lambda\>\greensubterm{(u_\levelH\mapp{\theta_\levelH})}{\DB{n}}_q$,
where $n$ is the appropriate De Bruijn index to refer to the initial $\lambda$.

Let $z_\levelH$ and $x_\levelH$ be fresh variables (condition~\REF{3} of $\FluidBoolHoist$).
We define
$\theta'_\levelH = (\theta_\levelH[z_\levelH \mapsto v\theta_\levelPG, x_\levelH \mapsto u_\levelPG\theta_\levelPG])$.
Then $(z_\levelH\>x_\levelH)\theta'_\levelH = 
v\theta_\levelPG\>(u_\levelPG\theta_\levelPG) =
(v\>u_\levelPG)\theta_\levelPG =
(v\>u_\levelH\mapp{\theta_\levelH})\theta_\levelPG =
 \greensubterm{(u_\levelH\mapp{\theta_\levelH})}{u_\levelH\mapp{\theta_\levelH}}_q\theta_\levelPG
 u_\levelH\mapp{\theta_\levelH}\theta_\levelPG
 = u_\levelH\mapp{\theta_\levelH}\mapq{\theta_\levelH} = u_\levelH\theta_\levelH
 u_\levelH\theta'_\levelH$.
 So $\theta'_\levelH$ is a unifier of $z_\levelH\>x_\levelH$ and $u_\levelH$.
Thus, by definition of $\csu$ (Definition~\ref{def:csu}),
there exists a unifier $\sigma \in \csu(z_\levelH\>x_\levelH \equiv u_\levelH)$ 
and a substitution $\zeta$ such that $x\sigma\zeta = x\theta'_\levelH$ 
for all relevant variables $x$
(condition~\REF{4} of $\FluidBoolHoist$).

Since $q$ is a green position of
$u_\levelH\mapp{\theta_\levelH}$, the type of $u_\levelH\mapp{\theta_\levelH}$ and
the type of $u_\levelH\sigma$
is nonfunctional (condition~\REF{2} of $\FluidBoolHoist$).

Since $z\sigma\zeta = v\theta_\levelPG$ and $v\theta_\levelPG \ne \lambda\>\DB{0}$
because $q$ is nonempty,
we have $z\sigma \ne \lambda\>\DB{0}$
(condition~\REF{6} of $\FluidBoolHoist$).

Condition~\REF{3} of $\PGBoolHoist$ states that
$u_\levelPG$ is not a variable and is neither $\itrue$ nor $\ifalse$.
So 
$u_\levelPG\theta_\levelPG = x_\levelH\sigma\zeta$
is neither $\itrue$ nor $\ifalse$ and
thus $x_\levelH\sigma$ is neither $\itrue$ nor $\ifalse$ (condition~\REF{7} of $\FluidBoolHoist$).
Moreover,
$(z_\levelH\>x_\levelH)\sigma\zeta =
v\theta_\levelPG\>(u_\levelPG\theta_\levelPG) =
\greensubterm{(u_\levelH\mapp{\theta_\levelH}\theta_\levelPG )}{u_\levelPG\theta_\levelPG}_q \ne
\greensubterm{(u_\levelH\mapp{\theta_\levelH}\theta_\levelPG )}{\ifalse}_q \theta_\levelPG =
v\theta_\levelPG\>\ifalse = (z_\levelH\>\ifalse)\sigma\zeta$.
Thus,
$(z_\levelH\>x_\levelH)\sigma \ne (z_\levelH\>\ifalse)\sigma$  (condition~\REF{5} of $\FluidBoolHoist$).

The inferences $\iota_\levelPG$ and $\iota_\levelH$ are of the form ($*$) with
$C_1\constraint{S_1} = C_\levelH\constraint{S}$,
$C_2\constraint{S_2} = (\greensubterm{C_\levelH}{z_\levelH\>\ifalse}_{p_\levelH}\llor x_\levelH \ceq \itrue\constraint{S})\sigma$,
$\theta_1 = \theta_\levelH$,
$\theta_2 = \zeta$,
$E = \greensubterm{C_\levelPG}{\ifalse}_{p_\levelPG}\llor u_\levelPG \ceq \itrue$,
and $\xi = \theta_\levelPG$.
In the following, we elaborate why
$C_2\mapp{\theta_2} = E\pi$ and
$x\pi\mapq{\theta_2} = x\xi$ for 
some substitution $\pi$ and all variables $x$ in $E$,
as required for ($*$).
The reasoning is similar to that in the $\FluidSup$ case.
We invoke Lemma~\ref{lem:G:mapp-comp-subst} to
obtain a substitution $\pi$ such that
$\sigma\mapp{\zeta} = \mapp{\sigma\zeta}\pi$
and
$\mapq{\sigma\zeta} = \pi\mapq{\zeta}$.
Then
\begin{align*}
  C_2\mapp{\theta_2}
  &= (\greensubterm{C_\levelH}{z_\levelH\>\ifalse}_{p_\levelH}\llor x_\levelH \ceq \itrue)\sigma\mapp{\zeta}\\
  &=\greensubterm{C_\levelH\sigma\mapp{\zeta}}{(z_\levelH\>\ifalse)\sigma\mapp{\zeta}}_{p_\levelH} \llor x_\levelH\sigma\mapp{\zeta} \ceq \itrue\\
  &\stackrel{(1)}{=}\greensubterm{C_\levelH\sigma\mapp{\zeta}}{\greensubterm{(z_\levelH\>x_\levelH)\sigma\mapp{\zeta}}{\ifalse}_q}_{p_\levelH}\llor u_\levelPG\pi \ceq \itrue\\
  &\stackrel{(2)}{=}\greensubterm{C_\levelH\sigma\mapp{\zeta}}{\greensubterm{u_\levelH\sigma\mapp{\zeta}}{\ifalse}_q}_{p_\levelH}\llor u_\levelPG\pi \ceq \itrue\\
  &\stackrel{(3)}{=}\greensubterm{C_\levelH\sigma\mapp{\zeta}}{\ifalse}_{p_\levelPG}\llor u_\levelPG\pi \ceq \itrue\\
  &=\greensubterm{C_\levelH\mapp{\sigma\zeta}\pi}{\ifalse}_{p_\levelPG}\llor u_\levelPG\pi \ceq \itrue\\
  &=(\greensubterm{C_\levelH\mapp{\sigma\zeta}}{\ifalse}_{p_\levelPG}\llor u_\levelPG \ceq \itrue)\pi\\
  &= (\greensubterm{C_\levelPG}{\ifalse}_{p_\levelPG}\llor u_\levelPG \ceq \itrue)\pi\\
  &= E\pi
\end{align*}
Step (1) can be justified as follows: By Lemma~\ref{lem:G:mapp-comp-mapq},
$z_\levelH\sigma\mapp{\zeta}\mapq{\zeta} = v\theta_\levelPG$.
Since $z_\levelH\sigma\mapp{\zeta}$ and $v$ contain only nonfunctional variables,
$z_\levelH\sigma\mapp{\zeta}$ must be a $\lambda$-abstraction
whose $\lambda$ binds exactly one De Bruijn index,
which is located at orange position $1.q$ \wrt\ $\downarrow_{\beta\eta\mathrm{long}}$.
Thus, $(z_\levelH\>x_\levelH)\sigma\mapp{\zeta}$
and $(z_\levelH\>\ifalse)\sigma\mapp{\zeta}$
are identical up to the subterm at green position $q$,
which is $x_\levelH\sigma\mapp{\zeta}$ and $\ifalse$ respectively.
So, $(z_\levelH\>\ifalse)\sigma\mapp{\zeta} = \greensubterm{(z_\levelH\>x_\levelH)\sigma\mapp{\zeta}}{\ifalse}_q$.
Moreover, since 
$\greensubterm{(z_\levelH\>x_\levelH)\sigma\mapp{\zeta}}{x_\levelH\sigma\mapp{\zeta}}_q = (z_\levelH\>x_\levelH)\sigma\mapp{\zeta} = u_\levelH\sigma\mapp{\zeta} = u_\levelH\mapp{\sigma\theta}\pi = u_\levelH\mapp{\theta_\levelH}\pi = \greensubterm{u_\levelH\mapp{\theta_\levelH}}{u_\levelPG}_q\pi$,
we have $x_\levelH\sigma\mapp{\zeta} = u_\levelPG\pi$.
For step (2), we use that $(z_\levelH\>x_\levelH)\sigma = u_\levelH\sigma$.
For step (3), we use that $u_\levelH$ is the green subterm of $C_\levelH$ at position $p_\levelH$.

\medskip

\noindent{\PGLoobHoist}:\enskip
Analogous to $\PGBoolHoist$.

\medskip

\noindent{\PGFalseElim}:\enskip
If $\iota_\levelPG$ is a $\PGFalseElim$ inference
\[\namedinference{PGFalseElim}
{\overbrace{C'_\levelPG \llor s_\levelPG \ceq t_\levelPG}^{C_\levelPG} \closure\theta_\levelPG}
{C'_\levelPG \closure\theta_\levelPG}\]
then we construct a corresponding $\FalseElim$ inference $\iota_\levelH$.
Let $C_\levelH \closure\theta_\levelH = \mapPonly_{\mapG{N}}^{-1}(C_\levelPG  \closure \theta_\levelPG)$
and $C_\levelH \constraint{S} = \mapGonly_N^{-1}(C_\levelH \closure\theta_\levelH)$.
This clause $C_\levelH \constraint{S}$ will be the premise of $\iota_\levelH$.

Condition~\REF{2} of $\PGFalseElim$ states that 
$s_\levelPG \eq t_\levelPG$ is strictly eligible in in $C_\levelPG\closure\theta_\levelPG$.
Let $C_\levelH = C'_\levelH \llor s_\levelH \ceq t_\levelH$,
where $s_\levelH \ceq t_\levelH$ is the literal that Lemma~\ref{lem:H:eligibility-lifting} guarantees to be strictly eligible
\wrt\ any suitable $\sigma$ (condition~\REF{2} of $\FalseElim$),
with $s_\levelH\mapp{\theta_\levelH} = s_\levelPG$ and $t_\levelH\mapp{\theta_\levelH} = t_\levelPG$.

Condition~\REF{1} of $\PGFalseElim$ states that
$(s_\levelPG \ceq t_\levelPG)\theta_\levelPG = \ifalse \ceq \itrue$.
So $\theta_\levelH$ is a unifier of $s_\levelH \equiv \ifalse$ and 
$t_\levelH \equiv \itrue$. By construction, $S\theta_\levelH$ is true.
So there exists $(\sigma, U) \in \csuupto(S,s_\levelH \equiv \ifalse, t_\levelH \equiv \itrue)$
(condition~\REF{1} of $\FalseElim$) and a substitution $\zeta$
such that $U\zeta$ is true
and $x\sigma\zeta = x\theta_\levelH$ for all variables $x$ in $C_\levelH\constraint{S}$.

The inferences $\iota_\levelPG$ and $\iota_\levelH$ are of the form ($**$) with
$C_1\constraint{S_1} = C_\levelH\constraint{S}$,
$C_2' = C'_\levelH\llor s_\levelH \ceq t_\levelH$,
$S_2' = U$,
$\theta_1 = \theta_\levelH$,
and $\xi = \theta_\levelPG$.

\medskip

\noindent{\PGArgCong}:\enskip
If $\iota_\levelPG$ is a $\PGArgCong$ inference
\[\namedinference{PGArgCong}
{\overbrace{C'_\levelPG \llor s_\levelPG \eq s'_\levelPG}^{C_\levelPG} \closure\theta_\levelPG}
{C'_\levelPG \llor s_\levelPG\>\diff\typeargs{\tau,\upsilon}(u_\levelPG,w_\levelPG) \eq s'_\levelPG\>\diff\typeargs{\tau,\upsilon}(u_\levelPG,w_\levelPG)\closure\theta_\levelPG}\]
then we construct a corresponding $\ArgCong$ inference $\iota_\levelH$.
Let $C_\levelH \closure\theta_\levelH = \mapPonly_{\mapG{N}}^{-1}(C_\levelPG  \closure \theta_\levelPG)$
and $C_\levelH \constraint{S} = \mapGonly_N^{-1}(C_\levelH \closure\theta_\levelH)$.
This clause $C_\levelH \constraint{S}$ will be the premise of $\iota_\levelH$.

Condition~\REF{3} of $\PGArgCong$ states that 
$s_\levelPG \eq s'_\levelPG$ is strictly eligible in in $C_\levelPG\closure\theta_\levelPG$.
Let $C_\levelH = C'_\levelH \llor s_\levelH \ceq s'_\levelH$,
where $s_\levelH \ceq s'_\levelH$ is the literal that Lemma~\ref{lem:H:eligibility-lifting} guarantees to be strictly eligible
\wrt\ any suitable $\sigma$ (condition~\REF{2} of $\ArgCong$),
with $s_\levelH\mapp{\theta_\levelH} = s_\levelPG$ and $s'_\levelH\mapp{\theta_\levelH} = s'_\levelPG$.

Let $x$ be a fresh variable (condition~\REF{3} of $\ArgCong$).
Let $\theta'_\levelH = \theta_\levelH[x \mapsto \diff\typeargs{\tau,\upsilon}(u_\levelPG,\allowbreak w_\levelPG)]$.

Condition~\REF{1} of $\PGArgCong$ states that $s_\levelPG$ is of type $\tau \fun \upsilon$.
Since $\levelPG$ does not use type variables, $\tau$ and $\upsilon$ are
ground types, and 
thus $s_\levelPG\theta_\levelPG = s_\levelH\theta_\levelH = s_\levelH\theta'_\levelH$ is of type $\tau \fun \upsilon$.
Let $\sigma$ be the most general type substitution such that $s_\levelH\sigma$
is of functional type (condition~\REF{1} of $\ArgCong$), and let $\zeta$ be a
substitution such that $y\sigma\zeta = y\theta'_\levelH$ for all type and term variables $y$ in $C_\levelH\constraint{S}$.

Then, the given inference $\iota_\levelPG$ and
the constructed inference $\iota_\levelH$ are of the form ($**$) with
$C_1\constraint{S_1} = C_\levelH\constraint{S}$,
$C_2' =  C'_\levelH\sigma\llor s_\levelH\sigma\>x \ceq s'_\levelH\sigma\>x$,
$S_2' = S\sigma$,
$\theta_1 = \theta'_\levelH$,
$\xi = \mapq{\theta'_\levelH}$,
$\zeta = \theta'_\levelH$,
and $\sigma = \{\}$.

\medskip

\noindent{\PGExt}:\enskip
If $\iota_\levelPG$ is a $\PGExt$ inference
\begin{align*}
  \namedinference{\PGExt}{\greensubterm{C_\levelPG}{u_\levelPG} \closure\theta_\levelPG}
  {\greensubterm{C_\levelPG}{w_\levelPG} \llor u_\levelPG\>\diff\typeargs{\tau,\upsilon}(u_\levelPG,w_\levelPG) \noteq w_\levelPG\>\diff\typeargs{\tau,\upsilon}(u_\levelPG,w_\levelPG) \closure\rho}
\end{align*}
then we construct a corresponding $\Ext$ or $\FluidExt$ inference $\iota_\levelH$.
Let $C_\levelH \closure\theta_\levelH = \mapPonly_{\mapG{N}}^{-1}(C_\levelPG  \closure \theta_\levelPG)$
and $C_\levelH \constraint{S} = \mapGonly_N^{-1}(C_\levelH \closure\theta_\levelH)$.
This clause $C_\levelH \constraint{S}$ will be the premise of $\iota_\levelH$.
Let $L_\levelPG.s_\levelPG.p_\levelPG$ be the position of the green subterm $u_\levelPG$ in $C_\levelPG$.
Condition~\REF{1} of $\PGExt$ states that 
$L_\levelPG.s_\levelPG.p_\levelPG$ is eligible in $C_\levelPG\closure\theta_\levelPG$.
Let $L_\levelH.s_\levelH.p_\levelH$ be the position in $C_\levelH$
that Lemma~\ref{lem:H:eligibility-lifting} guarantees to be eligible in $C_\levelH\constraint{S}$ \wrt\ any suitable $\sigma$
(condition~\REF{3} of $\Ext$ or condition~\ref{fluidext:eligible} of $\FluidExt$).
Let $u_\levelH$ be the subterm of $C_\levelH$ at position $L_\levelH.s_\levelH.p_\levelH$.
By Lemma~\ref{lem:H:eligibility-lifting}, one of the following cases applies.

\medskip

\noindent\textsc{Case 1:}\enskip
$p_\levelPG = p_\levelH$.
Then we construct an $\Ext$ inference.

Condition~\REF{2} of $\PGExt$ states that $u_\levelPG = u_\levelH\mapp{\theta_\levelH}$ is
of functional type.
Let $\sigma$ be the most general type substitution such that $u_\levelH\sigma$
is of type $\tau\fun\upsilon$ for some types $\tau$ and $\upsilon$ (condition~\REF{1} of $\Ext$).
Let $y$ be a fresh variable of the same type as $u_\levelH\sigma$ (condition~\REF{2} of $\Ext$).
Let $\theta'_\levelH = \theta_\levelH[y\mapsto w_\levelPG\rho]$.
Let $\zeta$ be a substitution such that $x\sigma\zeta = x\theta'_\levelH$ for all type and term variables $x$ in $C_\levelH\constraint{S}$
and for $x = y$.

By condition~\REF{3} of $\PGExt$,
$w_\levelPG\in \termsPG$ is a term whose nonfunctional yellow subterms are different fresh variables.
Then $w_\levelPG$ and $y\mapp{\theta_\levelH}$
are equal up to renaming of variables because $\mapponly$
replaces the nonfunctional yellow subterms of $w_\levelPG\theta_\levelH$
by distinct fresh variables.
Since the purpose of this proof is to show that $\iota_\levelPG$ is redundant,
a property that is independent of the variable names in $\iota_\levelPG$'s conclusion,
we can assume without loss of generality that
$w_\levelPG = y\mapp{\theta_\levelH}$ and
$\theta_\levelPG = \mapq{\theta'_\levelH}$.
Then the given inference $\iota_\levelPG$ and 
the constructed inference $\iota_\levelH$ are of the form ($**$) with
$C_1\constraint{S_1} = C_\levelH\constraint{S}$,
$C_2' = \greensubterm{C_\levelH\sigma}{y}\llor u_\levelH\sigma\>\diff\typeargs{\tau,\upsilon}(u_\levelH\sigma,y)\noteq y\>\diff\typeargs{\tau,\upsilon}(u_\levelH\sigma,y)$,
$S_2' = S\sigma$,
$\theta_1 = \theta'_\levelH$,
$\xi = \rho$, $\zeta=\theta'_\levelH$, and $\sigma = \{\}$.

\medskip

\noindent\textsc{Case 2:}\enskip $p_\levelPG = p_\levelH.q$ for some nonempty $q$,
$u_\levelH$ is not a variable but is variable-headed,
and $u_\levelH\theta$ is nonfunctional.
Then we construct a $\FluidExt$ inference.

Lemma~\ref{lem:H:eligibility-lifting} tells us
that $s_\levelPG = s_\levelH\mapp{\theta_\levelH}$.
Thus, the subterm of $s_\levelPG$ at position $p_\levelH$ is
$u_\levelH\mapp{\theta_\levelH}$.
So $q$ is a green position of
$u_\levelH\mapp{\theta_\levelH}$,
and the subterm at that position is $u_\levelPG$---i.e.,
$u_\levelH\mapp{\theta_\levelH} = \greensubterm{(u_\levelH\mapp{\theta_\levelH})}{u_\levelPG}_q$.
Let $t = \lambda\>\greensubterm{(u_\levelH\mapp{\theta_\levelH})}{\DB{n}}_q$,
where $n$ is the appropriate De Bruijn index to refer to the initial $\lambda$.

We define $\theta'_\levelH = \theta_\levelH [x \mapsto u_\levelPG\theta_\levelPG, y\mapsto w_\levelPG\rho, z\mapsto t\theta_\levelPG]$.
Then $(z\>x)\theta'_\levelH = 
(t\>u_\levelPG)\theta_\levelPG
 = u_\levelH\mapp{\theta_\levelH}\theta_\levelPG
 = u_\levelH\mapp{\theta_\levelH}\mapq{\theta_\levelH} = u_\levelH\theta_\levelH
 = u_\levelH\theta'_\levelH$.
So
$\theta'_\levelH$ is a unifier of $z\>x$ and $u_\levelH$.
Thus, by definition of $\csu$ (Definition~\ref{def:csu}),
there exists a unifier $\sigma \in \csu(z\>x \equiv u_\levelH)$ 
and a substitution $\zeta$ such that $x\sigma\zeta = x\theta'_\levelH$ 
for all relevant variables $x$
(condition~\ref{fluidext:csu} of $\FluidExt$).

By the assumption of this case, $u_\levelH$ is not a variable but is variable-headed (condition~\ref{fluidext:var} of $\FluidExt$)
and $u_\levelH\theta$ is nonfunctional (condition~\ref{fluidext:nonfunctional} of $\FluidExt$).

By condition~\REF{4} of $\PGExt$, $u_\levelPG\theta_\levelPG \ne w_\levelPG\rho$.
Thus, $(z\>x)\sigma\zeta =
(z\>x)\theta'_\levelH = 
(t\>u_\levelPG)\theta_\levelPG =
\greensubterm{(u_\levelH\mapp{\theta_\levelH})}{u_\levelPG\theta_\levelPG}_q \ne
\greensubterm{(u_\levelH\mapp{\theta_\levelH})}{w_\levelPG\rho}_q =
(z\>y)\theta'_\levelH  =
(z\>y)\sigma\zeta$.
So,
$(z\>x)\sigma \ne (z\>y)\sigma$  (condition~\ref{fluidext:csu-not-id} of $\FluidExt$).

Since $z\sigma\zeta = t\theta_\levelPG$ and $t\theta_\levelPG \ne \lambda\>\DB{0}$
because $q$ is nonempty,
we have $z\sigma \ne \lambda\>\DB{0}$
(condition~\ref{fluidext:csu-not-proj} of $\FluidExt$).

By condition~\REF{3} of $\PGExt$,
$w_\levelPG\in \termsPG$ is a term whose nonfunctional yellow subterms are different fresh variables.
Then $w_\levelPG$ and $y\mapp{\theta'_\levelH}$
are equal up to renaming of variables because $\mapponly$
replaces the nonfunctional yellow subterms of $w_\levelPG\rho$
by distinct fresh variables.
As above, we can assume without loss of generality that $y\mapp{\theta'_\levelH} = w_\levelPG$
and $w\rho = w\mapq{\theta'_\levelH}$ for all variables $w$ in $C_\levelPG$ and in $w_\levelPG$,
using the fact that $\rho$ coincides with $\theta_\levelPG$ on all variables in $C_\levelPG$.

The inferences $\iota_\levelPG$ and $\iota_\levelH$ are of the form ($*$) with
$C_1\constraint{S_1} = C_\levelH\constraint{S}$,
$C_2\constraint{S_2} = (\greensubterm{C_\levelH}{z\>y}_{p_\levelH}\llor x\>\diff\typeargs{\alpha,\beta}(x,y)\noteq y\>\diff\typeargs{\alpha,\beta}(x,y)\constraint{S})\sigma$,
$\theta_1 = \theta_\levelH$,
$\theta_2 = \zeta$,
$E = \greensubterm{C_\levelPG}{w_\levelPG}_{p_\levelPG}\llor u_\levelPG\>\diff\typeargs{\tau,\upsilon}(u_\levelPG,w_\levelPG)\noteq w_\levelPG\>\diff\typeargs{\tau,\upsilon}(u_\levelPG,w_\levelPG)$,
$\xi = \rho$.
In the following, we elaborate why
$C_2\mapp{\theta_2} = E\pi$ and
$x\pi\mapq{\theta_2} = x\xi$ for 
some substitution $\pi$ and all variables $x$ in $E$,
as required for ($*$).
The reasoning is similar to that in the $\FluidSup$ case.
We invoke Lemma~\ref{lem:G:mapp-comp-subst} to
obtain a substitution $\pi$ such that
$\sigma\mapp{\zeta} = \mapp{\sigma\zeta}\pi$
and
$\mapq{\sigma\zeta} = \pi\mapq{\zeta}$.
Then
\begin{align*}
  C_2\mapp{\theta_2}
  &= (\greensubterm{C_\levelH}{z\>y}_{p_\levelH}\llor x\>\diff\typeargs{\alpha,\beta}(x,y)\noteq y\>\diff\typeargs{\alpha,\beta}(x,y))\sigma\mapp{\zeta}\\
  &=\greensubterm{C_\levelH\sigma\mapp{\zeta}}{(z\>y)\sigma\mapp{\zeta}}_{p_\levelH} \llor (x\>\diff\typeargs{\alpha,\beta}(x,y)\noteq y\>\diff\typeargs{\alpha,\beta}(x,y))\sigma\mapp{\zeta}\\
  &\stackrel{(1)}{=}(\greensubterm{C_\levelH\sigma\mapp{\zeta}}{\greensubterm{(z\>x)\sigma\mapp{\zeta}}{y\sigma\mapp{\zeta}}_q}_{p_\levelH}
  \llor (x\>\diff\typeargs{\alpha,\beta}(x,y)\noteq y\>\diff\typeargs{\alpha,\beta}(x,y))\sigma\mapp{\zeta}\\
  &\stackrel{(2)}{=}(\greensubterm{C_\levelH\sigma\mapp{\zeta}}{\greensubterm{u_\levelH\sigma\mapp{\zeta}}{y\sigma\mapp{\zeta}}_q}_{p_\levelH}
  \llor (x\>\diff\typeargs{\alpha,\beta}(x,y)\noteq y\>\diff\typeargs{\alpha,\beta}(x,y))\sigma\mapp{\zeta}\\
  &\stackrel{(3)}{=}\greensubterm{C_\levelH\sigma\mapp{\zeta}}{y\sigma\mapp{\zeta}}_{p_\levelPG}
  \llor (x\>\diff\typeargs{\alpha,\beta}(x,y)\noteq y\>\diff\typeargs{\alpha,\beta}(x,y))\sigma\mapp{\zeta}\\
  &\stackrel{(4)}{=}\greensubterm{C_\levelH\sigma\mapp{\zeta}}{w_\levelPG}_{p_\levelPG}\llor (u_\levelPG\>\diff\typeargs{\tau,\upsilon}(u_\levelPG,w_\levelPG)\noteq w_\levelPG\>\diff\typeargs{\tau,\upsilon}(u_\levelPG,w_\levelPG))\pi\\
  &=(\greensubterm{C_\levelH\mapp{\sigma\zeta}}{w_\levelPG}_{p_\levelPG}\llor u_\levelPG\>\diff\typeargs{\tau,\upsilon}(u_\levelPG,w_\levelPG)\noteq w_\levelPG\>\diff\typeargs{\tau,\upsilon}(u_\levelPG,w_\levelPG))\pi\\
  &= (\greensubterm{C_\levelPG}{w_\levelPG}_{p_\levelPG}\llor u_\levelPG\>\diff\typeargs{\tau,\upsilon}(u_\levelPG,w_\levelPG)\noteq w_\levelPG\>\diff\typeargs{\tau,\upsilon}(u_\levelPG,w_\levelPG))\pi\\
  &= E\pi
\end{align*}
Step (1) can be justified as follows: By Lemma~\ref{lem:G:mapp-comp-mapq},
$z\sigma\mapp{\zeta}\mapq{\zeta} = t\theta_\levelPG$.
Since $z\sigma\mapp{\zeta}$ and $t$ contain only nonfunctional variables,
$z\sigma\mapp{\zeta}$ must be a $\lambda$-abstraction
whose $\lambda$ binds exactly one De Bruijn index,
which is located at orange position $1.q$ \wrt\ $\downarrow_{\beta\eta\mathrm{long}}$.
Thus, $(z\>x)\sigma\mapp{\zeta}$
and $(z\>y)\sigma\mapp{\zeta}$
are identical up to the subterm at green position $q$,
which is $x\sigma\mapp{\zeta}$ and $y\sigma\mapp{\zeta}$ respectively.
So, $(z\>y)\sigma\mapp{\zeta} = \greensubterm{(z\>x)\sigma\mapp{\zeta}}{y\sigma\mapp{\zeta}}_q$.

For step (2), we use that $(z\>x)\sigma = u_\levelH\sigma$.
For step (3), we use that $u_\levelH$ is the green subterm of $C_\levelH$ at position $p_\levelH$.
For step (4),
we use that
$x\sigma\mapp{\zeta} =  x\mapp{\sigma\zeta}\pi = x\mapp{\theta'_\levelH}\pi = u_\levelPG\pi$
and similarly
$y\sigma\mapp{\zeta} = w_\levelPG\pi$.

\medskip

\noindent{\PGDiff}:\enskip
If $\iota_\levelPG$ is a $\PGDiff$ inference
\begin{align*}
  \namedinference{\PGDiff}{}
  {u\>\diff\typeargs{\tau,\upsilon}(u,w) \noteq u\>\diff\typeargs{\tau,\upsilon}(u,w) \llor u\>s \eq w\>s \closure\theta_\levelPG}
\end{align*}
then we use the following $\Diff$ inference $\iota_\levelH$:
\begin{gather*}
  \namedinference{\Diff}{}
  {\underbrace{y\>(\diff\typeargs{\alpha,\beta}(y,z))\cneq z\>(\diff\typeargs{\alpha,\beta}(y,z)) \llor y\>x\ceq z\>x}_{C_\levelH}}
\end{gather*}
Let $\theta_\levelH$ be a grounding substitution with
$\alpha\theta_\levelH = \tau$, $\beta\theta_\levelH = \upsilon$, $y\theta_\levelH = u\theta_\levelPG$, and $z\theta_\levelH = w\theta_\levelPG$.
By condition~\REF{2} of $\PGDiff$,
$u,w,s\in \termsPG$ are terms whose nonfunctional yellow subterms are different fresh variables.
Then $u$ and $y\mapp{\theta_\levelH}$
are equal up to renaming of variables because $\mapponly$
replaces the nonfunctional yellow subterms of $u\theta_\levelPG$
by distinct fresh variables.
The same holds for $w$ and $z\mapp{\theta_\levelH}$ and for $s$ and $x\mapp{\theta_\levelH}$.
Since the purpose of this proof is to show that $\iota_\levelPG$ is redundant,
a property that is independent of the variable names in $\iota_\levelPG$'s conclusion,
we can assume without loss of generality that
$u = y\mapp{\theta_\levelH}$,
$w = z\mapp{\theta_\levelH}$,
$s = x\mapp{\theta_\levelH}$, and
$\theta_\levelPG = \mapq{\theta_\levelH}$.
Then 
$C_\levelH\mapp{\theta_\levelH} = C_\levelPG$
and $x_0\mapq{\theta_\levelH} = x_0\theta_\levelPG$
for all variables $x_0$ in $C_\levelPG$.
So $\iota_\levelPG$ and $\iota_\levelH$ have the form ($*$)
with $C_1 = C_\levelH$, $E = C_\levelPG$, $\xi = \theta_\levelPG$, $\pi = \{\}$, and $\theta_1 = \theta_\levelH$.
\end{proof}

\subsection{Trust and Simple Redundancy}

In this subsection, we define a notion of trust for each level and connect them.
Ultimately, we prove that simple redundancy $(\HRedC^\star, \HRedI^\star)$ as defined in Section~\ref{ssec:redundancy}
implies redundancy $(\HRedC, \HRedI)$ as defined in Section~\ref{ssec:H:redundancy}.

\subsubsection{First-Order Level}

In this subsubsection, let $\succ$ be an admissible term order for $\PFInf$ (Definition~\ref{def:PF:admissible-term-order}).

\begin{defi}[Trust]\label{def:PF:trust}
  A closure $C\closure \theta\in \clausesPF$ \emph{trusts} a closure $D\closure\rho \in N \subseteq \clausesPF$ if
  the variables in $D$ can be split into two sets $X$ and $Y$ such that
    \begin{enumerate}[label=\arabic*.,ref=\arabic*]
      \item
      \label{cond:pf:trust:corresponding-var}
      for every literal $L \in D$ containing a variable $x\in X$,
      there exists a variable $z$ in a literal $K \in C$ such that $x\rho$ is a subterm of $z\theta$
      and $L\rho \preceq K\theta$; and
      \item
      \label{cond:pf:trust:unconstrained}
      for all grounding substitutions $\rho'$ with $x\rho' = x\rho$ for all $x \not\in Y$,
      we have $D\closure \rho' \in N$.
    \end{enumerate}
  \end{defi}
  
  \begin{lem} \label{lem:pf:trust-irred}
    Let $R$ be a confluent term rewrite system oriented by $\succ$ whose only Boolean normal forms are $\itrue$ and $\ifalse$.
    If a closure $C\closure \theta\in \clausesPF$ trusts a closure $D\closure\rho \in N \subseteq \clausesPF$
    and $C\closure\theta$ is variable-irreducible, then
    there exists a closure $D\closure\rho' \in \irred_R(N)$
    with $D\closure\rho' \preceq D\closure\rho$
    such that $R \cup \{ D\closure\rho' \} \modelsolam D\closure\rho$.
  \end{lem}
  \begin{proof}
    Let $X$ and $Y$ the sets from Definition~\ref{def:PF:trust}.
    We define a substitution $\rho'$
    by $y\rho' = y\rho{\downarrow_R}$ for all variables $y\in Y$ and
    $x\rho' = x\rho$ for all variables $x\not\in Y$.
    By condition \ref{cond:pf:trust:unconstrained} of the definition of trust, we have $D\closure\rho'\in N$.
    Moreover, $D\closure\rho' \preceq D\closure\rho$ by \ref{cond:PF:order:comp-with-contexts}
    because $R$ is oriented by $\succ$.
  
    We show that $D\closure\rho'$ is also variable-irreducible.
    If a variable $y$ is in $Y$, then $y\rho'$ is reduced \wrt\ $R$ by definition of $\rho'$.
    If a variable $x$ is in $X$, then
    consider an occurrence of $x$ in a literal $L\closure\rho' \in D\closure\rho'$.
    We must show that $x\rho' = x\rho$ is irreducible \wrt\ the rules in $R$ smaller than $L \closure \rho'$.
    By condition \ref{cond:pf:trust:corresponding-var},
    there exists a variable $z$ in a literal $K \in C$ such that $x\rho$  is a subterm of $z\theta$ and $L\rho \preceq K\theta$.
    Since $C\closure\theta$ and thus $K\closure\theta$ is variable-irreducible \wrt\ $R$,
    $z\theta$ is irreducible \wrt\ the rules in $R$ smaller than $K\closure\theta$.
    Since $L\rho' \preceq L\rho \preceq K\theta$ and $x\rho = x\rho'$ is a subterm of $z\theta$, this implies that
    $x\rho'$ is irreducible \wrt\ the rules in $R$ smaller than $L \closure \rho'$.
  
    For the Boolean condition of variable-irreducibility,
    we use that $R$'s only Boolean normal forms are $\itrue$ and $\ifalse$.
  
    To show that $R \cup \{ D\closure\rho' \} \models D\closure\rho$,
    it suffices to prove that $x\rho \rewrite_R^* x\rho'$ for all $x$ in $D$.
    If $x$ is in $X$ then $x\rho = x\rho'$ and hence $x\rho \rewrite_R^* x\rho'$. Otherwise, $x$ is in $Y$ and $x\rho' = x\rho{\downarrow_R}$ and thus $x\rho \rewrite_R^* x\rho'$.
  \end{proof}

\subsubsection{Indexed Partly Substituted Ground Higher-Order Level}

In this subsubsection, let $\succ$ be an admissible term order for $\IPGInf$ (Definition~\ref{def:IPG:admissible-term-order}).

The definition of trust is almost identical to the
corresponding definition on the $\levelPF{}$ level.
However, we need to take into account that
$\levelPF{}$ level subterms correspond
to yellow subterms on the $\levelIPG{}$ level.
\begin{defi}[Trust]\label{def:IPG:trust}
A closure $C\closure\theta \in \clausesIPG$ \emph{trusts} a closure $D\closure\rho\in N\subseteq\clausesIPG$ if the variables in $D$ can be split into two sets $X$ and $Y$ such that
\begin{enumerate}[label=\arabic*.,ref=\arabic*]
\item \label{cond:IPG:trust:corresponding-var} for every literal $L\in D$ containing a variable $x\in X$,
there exists a variable $z$ in a literal $K\in C$ such that $x\rho$ is a yellow subterm of $z\theta$ and $L\rho \preceq K\theta$; and
\item \label{cond:IPG:trust:unconstrained}
for all grounding substitutions $\rho'$ with $x\rho' = x\rho$ for all $x \not\in Y$,
we have $D\closure \rho' \in N$.
\end{enumerate}
\end{defi}

\begin{lem}\label{lem:IPG:mapF-trust}
If a closure $C\closure\theta \in \clausesIPG$ trusts a closure $D\closure\rho \in N \subseteq \clausesIPG$
\wrt\ $\succ$, then $\mapF{C\closure\theta}$ trusts $\mapF{D\closure\rho} \in \mapF{N}$
\wrt\ $\succ_\mapFonly$.
\end{lem}
\begin{proof}
Assume $C\closure\theta$ trusts $D\closure\rho$. Thus, the variables of $D$ can be split into two sets $X$ and $Y$ such that conditions \ref{cond:IPG:trust:corresponding-var} and \ref{cond:IPG:trust:unconstrained} of the definition of trust are satisfied.
Note that the variables of $D$ and $\mapF{D}$ coincide.
We claim that $\mapF{C\closure\theta}$ trusts $\mapF{D\closure\rho}$ via the variable sets $X$ and $Y$.

Let $\mapF{L} \in \mapF{D}$ be a literal containing a variable $x\in X$.
Then $L\in D$ contains $x$ as well.
By the definition of trust (Definition~\ref{def:IPG:trust}), there exists a variable $z$ in some literal $K\in C$ such that $x\rho$ is a yellow subterm of $z\theta$ and $L\rho \preceq K\theta$.
It holds that $\mapF{L} \in \mapF{D}$ is a literal containing $x\in X$.
Since $x\rho$ is a yellow subterm of $z\theta$ we have that $\mapF{x\rho}$ is a subterm of $\mapF{z\theta}$.
Moreover, we have $\mapF{L\rho} \preceq_\mapFonly \mapF{K\theta}$.
Hence, condition \ref{cond:pf:trust:corresponding-var} of Definition~\ref{def:PF:trust} is satisfied.

It remains to show that also condition \ref{cond:pf:trust:unconstrained}
of Definition~\ref{def:PF:trust} holds---i.e.,
that for any grounding substitution $\rho'$ with $x\rho' = x\mapF{\rho}$ for all $x \not\in Y$, we have
$\mapF{D}\closure \rho' \in \mapF{N}$.
We define a substitution 
$\rho'' : x \mapsto \mapFonly^{-1}(\rho'(x))$.
Since $\mapFonly^{-1}$ maps ground terms to ground terms,
$\rho''$ is grounding.
Moreover, 
since $\mapFonly$ is bijective on ground terms,
$x\rho'' = \mapFonly^{-1}(\rho'(x)) = \mapFonly^{-1}(x\mapF{\rho}) = x\rho$ for all $x \not\in Y$.
By the definition of trust (Definition~\ref{def:IPG:trust}),
$D\closure \rho'' \in N$
and therefore $\mapF{D}\closure\rho' = \mapF{D\closure \rho''} \in \mapF{N}$.
\end{proof}

\begin{lem} \label{lem:ipg:trust-irred}
  Let $R$ be a confluent term rewrite system on $\termsPF$ oriented by $\succ_\mapFonly$ whose only Boolean normal forms are $\itrue$ and $\ifalse$.
  If a closure $C\closure \theta\in \clausesIPG$ trusts a closure $D\closure\rho \in N \subseteq \clausesIPG$
  and $C\closure\theta$ is variable-irreducible, then
  there exists a closure $D\closure\rho' \in \irred_R(N)$
  with $\mapF{D\closure\rho'} \preceq_\mapFonly \mapF{D\closure\rho}$
  such that $R \cup \{ \mapF{D\closure\rho'} \} \modelsolam \mapF{D\closure\rho}$.
\end{lem}
\begin{proof}
By Lemma~\ref{lem:IPG:mapF-trust},
$\mapF{C\closure\theta}$ trusts $\mapF{D\closure\rho} \in \mapF{N}$.
By Lemma~\ref{lem:pf:trust-irred},
there exists a closure $D_0\closure\rho'_0 \in \irred_R(\mapF{N}) = \mapF{\irred_R(N)}$
with $D_0\closure\rho'_0 \preceq_\mapFonly \mapF{D\closure\rho}$
such that $R \cup \{ D_0\closure\rho'_0 \} \modelsolam \mapF{D\closure\rho}$.
Thus, there must exist a closure $D\closure\rho' \in \irred_R(N)$
with $\mapF{D\closure\rho'} \preceq_\mapFonly \mapF{D\closure\rho}$
such that $R \cup \{ \mapF{D\closure\rho'} \} \modelsolam \mapF{D\closure\rho}$.
\end{proof}

\subsubsection{Partly Substituted Ground Higher-Order Level}

In this subsubsection, let $\succ$ be an admissible term order for $\PGInf$ (Definition~\ref{def:PG:admissible-term-order}).

As terms can also contain parameters we also need to account for this case in the definition of trust for the $\levelPG$ level.
Specifically, we must add a condition that the variables
in $Y$ may not appear in parameters in the trusted closure.

\begin{defi}[Trust]\label{def:PG:trust}
A closure $C\closure\theta \in \clausesPG$ \emph{trusts} a closure $D\closure\rho\in N \subseteq\clausesPG$ if the variables in $D$ can be partitioned into two sets $X$ and $Y$ such that
\begin{enumerate}[label=\arabic*.,ref=\arabic*]
	\item\label{cond:PG:trust:corresponding-var} for every literal $L\in D$ containing a variable $x\in X$ outside of parameters,
  there exists a literal $K\in C$ containing a variable $z$ outside of parameters such that $x\rho$ is a yellow subterm of $z\theta$ and $L\rho \preceq K\theta$; and
	\item\label{cond:PG:trust:unconstrained}
  all variables in $Y$ do not appear in parameters in $D$ and
  for all grounding substitutions $\rho'$ with $x\rho' = x\rho$ for all $x \not\in Y$, $D\closure \rho'\in N$.
	\end{enumerate}
\end{defi}

\begin{lem}\label{lem:PG:mapI-trust}
If a closure $C\closure\theta \in \clausesPG$ trusts a closure $D\closure\rho \in  N \subseteq\clausesPG$ \wrt\ $\succ$,
then $\mapI{C\closure\theta}$ trusts $\mapI{D\closure\rho} \in \mapI{N}$ \wrt\ $\succ_\mapIonly$.
\end{lem}
\begin{proof}
	Let $C\closure\theta \in \clausesPG$ and $D\closure\rho \in \clausesPG$ such that $C\closure\theta$ trusts $D\closure\rho$.
	Thus, the variables of $D$ can be partitioned into two sets $X$ and $Y$
  such that conditions \ref{cond:PG:trust:corresponding-var} and \ref{cond:PG:trust:unconstrained} of the definition of trust are satisfied.
	Let $X'$ be the set $X$ without the variables that only occur in parameters in $D$.
  We claim that $\mapI{C\closure\theta}$ trusts $\mapI{D\closure\rho}$ using the sets $X'$ and $Y$ that split the variables in $\mapI{D\closure\rho}$.

	Since the parameters disappear after the application of the transformation $\mapIonly$,
  the set $X' \cup Y$ contains all variables occurring in $\mapI{D\closure\rho}$.
	As all the other variables are preserved by $\mapIonly$ we have that $X'$ and $Y$ are partition of the variables in $\mapI{D\closure\rho}$.

  Let $\mapIonly_{\rho}(L) \in \mapIonly_{\rho}(D)$ be
  a literal containing a variable $x\in X$.
	Then $L\in D$ contains $x\in X' \subseteq X$ outside of parameters.
	By Defintion~\ref{def:PG:trust}, there exists a literal $K\in C$ containing a variable $z$ outside of parameters such that $x\rho$ is a yellow subterm of $z\theta$ and $L\rho \preceq K\theta$.
	It holds that $\mapIonly_{\rho}(K) \in \mapIonly_{\rho}(D)$ is a literal containing $x\in X$.
	Since $x\rho$ is a yellow subterm of $z\theta$, we have that $\mapI{x\rho}$ is a yellow subterm of $\mapI{z\theta}$, since $\mapIonly$ preserves yellow subterms.
	Then, by Lemma~\ref{lem:PG:mapI-subst}, $\mapIonly_\rho(x)\mapI{\rho}$ is a yellow subterm of $\mapIonly_\theta(z)\mapI{\theta}$.
  Moreover, we have $\mapIonly_\rho(K)\mapI{\rho}=\mapI{K\rho} \preceq_\mapIonly \mapI{L\theta} = \mapIonly_\rho(L)\mapI{\theta}$ again using Lemma~\ref{lem:PG:mapI-subst}.
	Hence, condition \ref{cond:IPG:trust:corresponding-var} of Definition~\ref{def:IPG:trust} is satisfied.

  Let $\rho'$ be a grounding substitution on the $\levelIPG$ level with $x\rho' = x\mapI{\rho}$ for all $x \not\in Y$.
  Since $\mapI{D\closure\rho} = \mapIonly_\rho (D)\closure \mapI{\rho}$,
  we must show that
   $\mapIonly_\rho (D) \closure \rho' \in \mapI{N}.$
   Since $\rho'$ is a grounding substitution
   and $\mapIonly$ is a bijection on ground
   terms, there exists a substitution $\rho'' = \mapIonly^{-1}(\rho')$.
   This substitution $\rho''$ is grounding and $x\rho'' = \mapIonly^{-1}(x\rho') = \mapIonly^{-1}(\mapI{x\rho}) =  x\rho$ for all $x \not\in Y$.
  Since $C\closure\theta$ trusts $D\closure\rho$, we have
  $D\closure\rho''\in N$ and thus $\mapIonly_{\rho''}(D)\closure I(\rho'') = \mapI{D\closure\rho''} \in \mapI{N}$.
  Since the variables in $Y$ do not occur in parameters in $D$,
  we have $\mapIonly_{\rho''}(D) = \mapIonly_{\rho}(D)$.
  Moreover, $\mapI{\rho''} = \rho'$.
  Hence, condition \ref{cond:IPG:trust:unconstrained} of Definition~\ref{def:IPG:trust} is satisfied.
\end{proof}

\begin{lem} \label{lem:pg:trust-irred}
  Let $R$ be a confluent term rewrite system on $\termsPF$ oriented by $\succ_{\mapIonly\mapFonly}$
  whose only Boolean normal forms are $\itrue$ and $\ifalse$.
  If a closure $C\closure \theta\in \clausesPG$ trusts a closure $D\closure\rho \in N \subseteq \clausesPG$
  and $C\closure\theta$ is variable-irreducible, then
  there exists a closure $D\closure\rho' \in \irred_R(N)$
  with $\mapF{\mapI{D\closure\rho'}} \preceq_{\mapIonly\mapFonly} \mapF{\mapI{D\closure\rho}}$
  such that $R \cup \{ \mapF{\mapI{D\closure\rho'}} \} \modelsolam \mapF{\mapI{D\closure\rho}}$.
\end{lem}
\begin{proof}
By Lemma~\ref{lem:PG:mapI-trust},
$\mapI{C\closure\theta}$ trusts $\mapI{D\closure\rho} \in \mapI{N}$.
By Lemma~\ref{lem:ipg:trust-irred},
there exists a closure $D_0\closure\rho'_0 \in \irred_R(\mapI{N}) = \mapI{\irred_R(N)}$
with $\mapF{D_0\closure\rho'_0} \preceq \mapF{\mapI{D\closure\rho}}$
such that $R \cup \{ \mapF{D_0\closure\rho'_0} \} \modelsolam \mapF{\mapI{D\closure\rho}}$.
Thus, there must exist a closure $D\closure\rho' \in \irred_R(N)$
with $\mapF{\mapI{D\closure\rho'}} \preceq \mapF{\mapI{D\closure\rho}}$
such that $R \cup \{ \mapF{\mapI{D\closure\rho'}} \} \modelsolam \mapF{\mapI{D\closure\rho}}$.
\end{proof}

\subsubsection{Full Higher-Order Level}

In this subsubsection, let $\succ$ be an admissible term order (Definition~\ref{def:admissible-term-order}),
extended to be an admissible term order for $\PGInf$ as in Section~\ref{ssec:H:redundancy}.
We have defined trust for level $\levelH$ in Defintion~\ref{def:H:trust}.

\begin{lem}\label{lem:G:mapP-trust}
Let $C\constraint{S} \in \clausesH$ and $D\constraint{T} \in N \subseteq \clausesH$.
Let $C\theta \in \gnd(C\constraint{S})$ and  $D\rho \in \gnd(D\constraint{T})$.
If the $\theta$-instance of $C\constraint{S}$ trusts the $\rho$-instance of $D\constraint{T}$,
then $\mapP{C\closure\theta}$ trusts $\mapP{D\closure\rho} \in \mapPG{N}$.
\end{lem}
\begin{proof}
Let $X$ and $Y$ be a partition of the variables in $D$
such that the variables in $X$ fulfill condition~\ref{cond:H:trust:corresponding-var}
and the variables in $Y$ fulfill condition~\ref{cond:H:trust:unconstrained}.
Let $X'$ be the set of variables occurring in $\mapP{D\closure\rho}$
originating from $x\mapp{\rho}$ for some $x \in X$.
Define the set $Y'$ analogously.
We claim that $\mapP{C\closure\theta}$ trusts $\mapP{D\closure\rho}$ using the sets $X'$ and $Y'$.
It holds that $X'$ and $Y'$ are a partition of the variables in $\mapP{D\closure\rho}$.

\medskip\noindent
\textsc{Regarding $X'$:}\enskip We need to show that for every literal $L' \in D\mapp{\rho}$ containing a variable $x'\in X'$ outside of parameters, there
exists a literal $K' \in C\mapp{\theta}$ containing a variable $z'$ outside of parameters such that $x'\mapq{\rho}$
is a yellow subterm of $z'\mapq{\theta}$ and $L'\mapq{\rho} \preceq K'\mapq{\theta}$.

Any literal $L' \in D\mapp{\rho}$ containing a variable $x'\in X'$ outside of parameters
must originate from a literal $L \in D$ containing a variable 
$x \in X$ outside of parameters, where $L' = L\mapp{\rho}$ and $x'$ occurs in $x\mapp{\rho}$.
By condition \ref{cond:H:trust:corresponding-var}
of Definition~\ref{def:H:trust},
this implies that there exists a literal $K \in C$ 
and a substitution $\sigma$ such that $z\theta = z\sigma\rho$ for all variables $z$ in
$C$ and $L \preceq K\sigma$.

By \ref{cond:order:variable} with the substitution $\mapp{\rho}$,
since $x'$ occurs outside of parameters of $L' = L\mapp{\rho} = L\mapp{\rho}$,
it also occurs outside of parameters of $K\sigma\mapp{\rho}$.

By 
Lemma~\ref{lem:G:mapp-comp-subst},
there exists a substitution $\pi$ such that
$\sigma\mapp{\rho} = \mapp{\sigma\rho}\pi$ 
and $\mapq{\sigma\rho} = \pi\mapq{\rho}$.
So $K\sigma\mapp{\rho} = K\mapp{\sigma\rho}\pi = K\mapp{\theta}\pi$
and $x'$ occurs outside of parameters in this literal.
Since $K\mapp{\theta}$ contains only nonfunctional variables,
there exists a variable $z'$ occurring outside of parameters in $K\mapp{\theta}$ such that
$x'$ is a yellow subterm of $z'\pi$.
Thus,
$x'\mapq{\rho}$ 
is a yellow subterm of $z'\pi\mapq{\rho} = z'\mapq{\sigma\rho} = z'\mapq{\theta}$, which is what we needed to show.

\medskip\noindent
\textsc{Regarding $Y'$:}\enskip We must show that for all grounding substitutions $\rho'$ 
with $x\rho' = x\mapq{\rho}$ for all $x \not\in Y'$, we have $D\mapp{\rho}\closure\rho' \in \mapPG{N}$.

Let $\rho'$ be a substitution with $x\rho' = x\mapq{\rho}$ for all $x \not\in Y'$.
Then, for all variables $x \not\in Y$, we have $x\mapp{\rho}\rho' = x\mapq{\rho}\mapq{\rho} = x\rho$
by Lemma~\ref{lem:G:mapp-comp-mapq}.
By condition \ref{cond:H:trust:unconstrained}
of Definition~\ref{def:H:trust},
the variables in $Y$ do not appear in the constraints $T$ of $D\constraint{T}$.
So, $T\mapp{\rho}\rho' = T\rho$ and thus
$\mapp{\rho}\rho'$ is true. Therefore,
$D\closure\mapp{\rho}\rho' \in \mapG{N}$
and 
$\mapP{D\closure\mapp{\rho}\rho'} \in \mapPG{N}$.
By Lemma~\ref{lem:G:mapp-mapp} and \ref{lem:G:mapq-mapp},
$\mapP{D\closure\mapp{\rho}\rho'} = 
D\mapp{\mapp{\rho}\rho'}\closure\mapq{\mapp{\rho}\rho'} =
D\mapp{\rho}\closure\rho'$
because $x\rho' = x\mapq{\rho}$ for all $x \not\in Y'$
and in particular for all $x$ not introduced by $\mapp{\rho}$.
Therefore, $D\mapp{\rho}\closure\rho' \in \mapPG{N}$.

A variable in $y\mapp{\rho}$ can occur in a parameter in $\mapP{D\closure\rho}$
only if the variable $y$ occurs in a parameter in $D$.
Since all variables in $Y$ do not appear in parameters, the variables in $Y'$ do not appear in parameters either.
\end{proof}

\begin{lem} \label{lem:g:trust-irred}
  Let $R$ be a confluent term rewrite system on $\termsPF$ oriented by $\succ_{\mapIonly\mapFonly}$
  whose only Boolean normal forms are $\itrue$ and $\ifalse$.
  Let $C\constraint{S} \in \clausesH$ and $D\constraint{T} \in N \subseteq \clausesH$.
  Let $C\theta \in \gnd(C\constraint{S})$ and  $D\rho \in \gnd(D\constraint{T})$.
  If the $\theta$-instance of $C\constraint{S}$ trusts the $\rho$-instance of $D\constraint{T}$
  and $C\closure\theta$ is variable-irreducible, then
  there exists a closure $D\closure\rho' \in \irred_R(\mapG{N})$
  with $\mapF{\mapI{\mapP{D\closure\rho'}}} \preceq_{\mapIonly\mapFonly} \mapF{\mapI{\mapP{D\closure\rho}}}$
  such that $R \cup \{ \mapF{\mapI{\mapP{D\closure\rho'}}} \} \modelsolam \mapF{\mapI{\mapP{D\closure\rho}}}$.
\end{lem}
\begin{proof}
By Lemma~\ref{lem:G:mapP-trust},
$\mapI{C\closure\theta}$ trusts $\mapP{D\closure\rho} \in \mapPG{N}$.
By Lemma~\ref{lem:pg:trust-irred},
there exists a closure $D_0\closure\rho'_0 \in \irred_R(\mapP{N}) = \mapP{\irred_R(N)}$
with $\mapF{\mapI{D_0\closure\rho'_0}} \preceq_{\mapIonly\mapFonly} \mapF{\mapI{\mapP{D\closure\rho}}}$
such that $R \cup \{ \mapF{\mapI{D_0\closure\rho'_0}} \} \modelsolam \mapF{\mapI{\mapP{D\closure\rho}}}$.
Thus, there must exist a closure $D\closure\rho' \in \irred_R(N)$
with $\mapF{\mapI{\mapP{D\closure\rho'}}} \preceq_{\mapIonly\mapFonly} \mapF{\mapI{\mapP{D\closure\rho}}}$
such that $R \cup \{ \mapF{\mapI{\mapP{D\closure\rho'}}} \} \modelsolam \mapF{\mapI{\mapP{D\closure\rho}}}$.
\end{proof}

\begin{lem}\label{lem:H:simple-RedC}
  Let $N \subseteq \clausesH$. Then $\HRedC^\star(N) \subseteq \HRedC(N)$.
\end{lem}
\begin{proof}
Let $C\constraint{S} \in \HRedC^\star(N)$.
Let $R$ be a confluent term rewrite system on $\termsF$ oriented by $\succ_{\mapIonly\mapFonly}$
whose only Boolean normal forms are $\itrue$ and $\ifalse$.
Let $C'\closure\theta' \in \irred_R(\fipg{C\constraint{S}})$, i.e,
there exists some $C\theta \in \gnd(C\constraint{S})$
such that $\fip{C\closure\theta} = C'\closure\theta'$.
By Lemma~\ref{lem:fo-eq-tfip},
$\mapF{C\theta} = C'\theta'$.

We make a case distinction on which condition of Definition~\ref{def:H:simple-redundancy}
applies to $C\theta \in \gnd(C\constraint{S})$.

\medskip

\noindent\textsc{Condition~\ref{cond:red:entailed-by-smaller}:}\enskip
There exist
an indexing set $I$ and for each $i\in I$
a ground instance $D_i \rho_i$ of a clause
$D_i \constraint{T_i} \in N$,
such that
\begin{enumerate}[label=(\alph*)]
  \item\label{cond:redproof:entailment} $\mapF{\{D_i\rho_i\mid i\in I\}} \modelsolam \mapF{C\theta}$; %
  \item\label{cond:redproof:order} for all $i \in I$, $D_i \rho_i \prec C \theta$; and
  \item\label{cond:redproof:trust} for all $i \in I$, the $\theta$-instance of $C\constraint S$ trusts the $\rho_i$-instance
  of $D_i\constraint {T_i}$.
\end{enumerate}
We show that $C\constraint{S} \in \HRedC(N)$
by condition~\ref{cond:H:RedC:entailment} of Definition~\ref{def:H:RedC}, i.e.,
we show that
\[R \cup \{E \in \irred_R(\fipg{N}) \mid E \prec_{\mapIonly\mapFonly} C'\theta'\}\modelsolam C'\theta'\]
By Lemma~\ref{lem:g:trust-irred} and point \ref{cond:redproof:trust} above,
there exists a closure $D'_i\closure\rho'_i \in \irred_R(\fipg{N})$
with $D'_i\closure\rho'_i \preceq_{\mapIonly\mapFonly} \mapF{\mapI{\mapP{D_i\closure\rho_i}}}$
such that $R \cup \{ D'_i\closure\rho'_i \} \modelsolam \mapF{\mapI{\mapP{D_i\closure\rho_i}}}$.
By Lemma~\ref{lem:fo-eq-tfip},
$D'_i\rho'_i \preceq_{\mapIonly\mapFonly} \mapF{D_i\rho_i}$
and
$R \cup \{ D'_i\rho'_i \} \modelsolam \mapF{D_i\rho_i}$.
With point \ref{cond:redproof:entailment} above, it follows that
\[R \cup \{ D'_i\rho'_i \mid i \in I \} \modelsolam \mapF{C\theta}\]
It remains to show that
$D'_i\rho'_i \prec_{\mapIonly\mapFonly} \mapF{C\theta}$ for all $i \in I$.
Since $D'_i\rho'_i \preceq_{\mapIonly\mapFonly} \mapF{D_i\rho_i}$,
it suffices to show that
$\mapF{D_i\rho_i} \prec_{\mapIonly\mapFonly} \mapF{C\theta}$.
This follows from point \ref{cond:redproof:order} above,
Lemma~\ref{lem:fo-eq-tfip},
and Definitions~\ref{def:IPG:succ-transfer} and \ref{def:PG:succ-transfer}.

\medskip

\noindent\textsc{Condition~\ref{cond:red:subsumed}:}\enskip
There exists a
ground instance $D\rho$ of some
$D\constraint {T}\in N$
such that
\begin{enumerate}[label=(\alph*)]
  \item\label{cond:redproof:subsumed:eq} $D\rho = C\theta$; %
  \item\label{cond:redproof:subsumed:sqsupset} $C \constraint S \sqsupset D\constraint {T}$; and
  \item\label{cond:redproof:subsumed:trust} the $\theta$-instance of $C\constraint S$ trusts the $\rho$-instance
  of $D\constraint {T}$.
\end{enumerate}
By Lemma~\ref{lem:g:trust-irred} and point \ref{cond:redproof:subsumed:trust} above,
there exists a closure $D'\closure\rho' \in \irred_R(\fipg{N})$
with $D'\closure\rho' \preceq_{\mapIonly\mapFonly} \mapF{\mapI{\mapP{D\closure\rho}}}$
such that $R \cup \{ D'\closure\rho'\} \modelsolam \mapF{\mapI{\mapP{D\closure\rho}}}$.

We distinguish two subcases.

\medskip

\noindent\textsc{Case 1:}\enskip
$D'\closure\rho' = \mapF{\mapI{\mapP{D\closure\rho}}}$.
Then we can show that $C\constraint{S} \in \HRedC(N)$
by condition~\ref{cond:H:RedC:subsumed} of Definition~\ref{def:H:RedC},
using points \ref{cond:redproof:subsumed:eq} and \ref{cond:redproof:subsumed:sqsupset} above
and Lemma~\ref{lem:fo-eq-tfip}.

\medskip

\noindent\textsc{Case 2:}\enskip
$D'\closure\rho' \prec_{\mapIonly\mapFonly} \mapF{\mapI{\mapP{D\closure\rho}}}$.
Then we show that $C\constraint{S} \in \HRedC(N)$
by condition~\ref{cond:H:RedC:entailment} of Definition~\ref{def:H:RedC}---i.e.,
\[R \cup \{E \in \irred_R(\fipg{N}) \mid E \prec_{\mapIonly\mapFonly} C'\theta'\}\modelsolam C'\theta'\]
Since $C'\theta' = \mapF{C\theta} = \mapF{D\rho} = \tfip{D\closure\rho}$ by point \ref{cond:redproof:subsumed:eq} above and Lemma~\ref{lem:fo-eq-tfip},
it suffices to show
\[R \cup \{E \in \irred_R(\fipg{N}) \mid E \prec_{\mapIonly\mapFonly} \fip{D\closure\rho}\}\modelsolam \fip{D\closure\rho}\]
This follows directly from the three defining properties of $D'\closure\rho'$ above.
\end{proof}

\begin{lem}\label{lem:H:simple-RedI}
  Let $N \subseteq \clausesH$. Then $\HRedI^\star(N) \subseteq \HRedI(N)$.
 \end{lem}
\begin{proof}
Let $\iota \in \HRedI^\star(N)$.
Let $C_1\constraint{S_1}$, \dots, $C_m\constraint{S_m}$ be its premises and
  $C_{m+1}\constraint{S_{m+1}}$ its conclusion.
Let $\theta_1, \dots, \theta_{m+1}$ be a tuple of substitutions
for which $\iota$ is rooted in $\FInf{}$ (Definition~\ref{def:fol-inferences}).
Since $\iota \in \HRedI^\star(N)$, by Definition~\ref{def:H:simple-inference-redundancy},
there exists an indexing set $I$ and for each $i \in I$
  a ground instance $D_i \rho_i$ of a clause
  $D_i \constraint{T_i} \in N$,
  such that
  \begin{enumerate}[label=\arabic*.,ref=\arabic*]
    \item[\ref{cond:HRedI:entailment}.] $\mapF{\{D_i\rho_i\mid i\in I\}} \modelsolam \mapF{C_{m+1}\theta_{m+1}}$; %
    \item[\ref{cond:HRedI:order}.] $\iota$ is a $\Diff$ inference or for all $i \in I$, $D_i  \rho_i \prec C_m\theta_m$; and
    \item[\ref{cond:HRedI:trust}.] for all $i \in I$, the $\theta_{m+1}$-instance of $C_{m+1}\constraint{S_{m+1}}$ trusts the $\rho_i$-instance
    of $D_i\constraint {T_i}$.
  \end{enumerate}
By Definition~\ref{def:H:RedI}, we must show that for all confluent term rewrite systems $R$ oriented by $\succ_{\mapIonly\mapFonly}$
whose only Boolean normal forms are $\itrue$ and $\ifalse$
such that $C_{m+1}\closure\theta_{m+1}$ is variable-irreducible, we have
\[R \cup O \modelsolam \mapF{C_{m+1}\theta_{m+1}}\]
where 
$O = \irred_R(\fipg{N})$
if $\iota$ is a $\Diff$ inference and
$O = \{E \in \irred_R(\fipg{N}) \mid E \prec_{\mapIonly\mapFonly} \mapF{C_m\theta_m}\}$
otherwise.

By Lemma~\ref{lem:g:trust-irred} and point \ref{cond:HRedI:trust} above,
there exists a closure $D'_i\closure\rho'_i \in \irred_R(\fipg{N})$
with $D'_i\closure\rho'_i \preceq_{\mapIonly\mapFonly} \mapF{\mapI{\mapP{D_i\closure\rho_i}}}$
such that $R \cup \{ D'_i\closure\rho'_i \} \modelsolam \mapF{\mapI{\mapP{D_i\closure\rho_i}}}$.
By Lemma~\ref{lem:fo-eq-tfip},
$D'_i\rho'_i \preceq_{\mapIonly\mapFonly} \mapF{D_i\rho_i}$
and
$R \cup \{ D'_i\rho'_i \} \modelsolam \mapF{D_i\rho_i}$.
With point \ref{cond:HRedI:entailment} above, it follows that
\[R \cup \{ D'_i\rho'_i \mid i \in I \} \modelsolam \mapF{C_{m+1}\theta_{m+1}}\]
If $\iota$ is a $\Diff$ inference, we are done.
For the other inferences, it remains to show that
$D'_i\rho'_i \prec_{\mapIonly\mapFonly} \mapF{C_m\theta_m}$ for all $i \in I$.
Since $D'_i\rho'_i \preceq_{\mapIonly\mapFonly} \mapF{D_i\rho_i}$,
it suffices to show that
$\mapF{D_i\rho_i} \prec_{\mapIonly\mapFonly} \mapF{C_m\theta_m}$.
This follows from point \ref{cond:HRedI:order} above,
Lemma~\ref{lem:fo-eq-tfip},
and Definitions~\ref{def:IPG:succ-transfer} and \ref{def:PG:succ-transfer}.
\end{proof}

\end{full}

\subsection{Model Construction}
\label{ssec:model-construction}

In this subsection,
we construct models of saturated clause sets,
starting with a first-order model and lifting it through the levels.
Using the results of Section~\ref{ssec:redundancy-criteria-and-saturation},
we prove a completeness property for each of the calculi
that roughly states
the following.
\begin{full}%
For any saturated set $N_\infty$
that does not contain an empty closure,
there exists a term rewrite system $R$ and a corresponding
interpretation $\III$ such that $\III$ is a model of the
closures in $N_\infty$ that are variable-irreducible \wrt\ $R$.
\end{full}%
\begin{slim}%
For any saturated set $N_\infty$ that does not contain the empty clause,
there exists a model of $N_\infty$.
\end{slim}

\begin{full}%
Moreover, to prepare the eventual extension the model of the variable-irreducible instances
to all instances, for each level,
we show a property that roughly states the following:
For any set of closures $N_0$
that contains all closures
of the form $C\closure\rho$ for all $\rho$ whenever it contains a closure $C\closure\theta$,
then the variable-irreducible instances of $N_0$ entail all of $N_0$.
\end{full}%

Finally, in level $\levelH$,
we bring everything together by showing that the constructed model is
also a model 
\begin{full}of the variable-irreducible ground instances\end{full}
of $N_0$\begin{full} and thus of $N_0$ itself\end{full}.
It follows that the calculus $\HInf$ is refutationally complete.

\subsubsection{First-Order Levels}

In this subsubsection, let $\succ$ be an admissible term order for $\PFInf$ (Definition~\ref{def:PF:admissible-term-order}),
and let $\mathit{\slimfull{}{p}fsel}$ be a selection function on 
\slimfull{$\clausesF$}{$\clausesPF$ (Definition~\ref{def:PF:selection-function})}.

The completeness proof for $\PFInf$
relies on constructing a first-order term rewrite system.
For any first-order term rewrite system $R$, there exists
a first-order interpretation, which we also denote $R$, such that
$R \modelsfol s \ceq t$ if and only if $s \leftrightarrow^*_R t$.
Formally, this can be implemented by
a first-order interpretation whose universe for each type $\tau$ consists of
the $R$-equivalence classes of ground terms of type $\tau$.

\begin{defi}[$\Rbasic_N$]
\label{def:RbasicN}
Let $N$ be a set of ground first-order \slimfull{clauses}{closures}
with $\bot \not\in \slimfull{N}{\mapT{N}}$.
By well-founded induction, we define
term rewrite systems $\Rbasic_e$ and $\Deltabasic_e$
for all ground clauses and ground terms $e \in \termsF \cup \clausesF$
and finally a term rewrite system $R_N$.
As our well-founded order on $\termsF \cup \clausesF$,
we employ our term and clause order $\succ$.
To compare terms with clauses, we define a term $s$ to be larger than a clause $C$ if and only if
$s$ is larger than every term in $C$.
Formally, this can be defined using the clause order by Bachmair and Ganzinger \cite[\Section~2.4]{bachmair-ganzinger-1994}
and encoding a term $s$ as the multiset $\{\{\{s\}\}\}$.

\begin{enumerate}[label=($\Delta$\arabic*)]
  \item\label{cond:deltac-logical-boolean} \textsc{Logical Boolean rewrites:} Given a term $s$, let $\Delta_s = \{s \rewrite t\}$ if
  \begin{itemize}
    \item 
  $(s, t)$ is one of the following:
  \begin{align*}
    &(\inot \ifalse, \itrue)                     && (\itrue \iand \itrue, \itrue)   && (\itrue \ior \itrue, \itrue)     && (\itrue \iimplies \itrue, \itrue) && (u \ieq^\tau u, \itrue) \\
    &(\inot \itrue, \ifalse)                     && (\itrue \iand \ifalse, \ifalse) && (\itrue \ior \ifalse, \itrue)    && (\itrue \iimplies \ifalse, \ifalse) && (u \ieq^\tau v, \ifalse) \text{ with } u \ne v \\
    &                                            && (\ifalse \iand \itrue, \ifalse) && (\ifalse \ior \itrue, \itrue)    && (\ifalse \iimplies \itrue, \itrue) && (u \ineq^\tau u, \ifalse) \\
    &                                            && (\ifalse \iand \ifalse, \ifalse) && (\ifalse \ior \ifalse, \ifalse) && (\ifalse \iimplies \ifalse, \itrue) && (u \ineq^\tau v, \itrue) \text{ with } u \ne v
    \end{align*}
    \item $s$ is irreducible \wrt\ $\Rbasic_s$.
  \end{itemize}

  \item\label{cond:deltac-backstop-boolean} \textsc{Backstop Boolean rewrites:} Given a clause $C$, let $\Delta_C = \{s \rewrite \ifalse\}$ if
  \begin{itemize}
    \item $C = s \ceq \ifalse$;
    \item $s\notin\{\ifalse,\itrue\}$;
    \item $s$ is irreducible \wrt\ $\Rbasic_C$.
  \end{itemize}

  \item\label{cond:deltac-function} \textsc{Function rewrites:}
  Given a clause $C$, let $\Delta_C = \{ \mapF{u}\rewrite \mapF{w} \}$ if
  \begin{itemize}
    \item $C = \mapF{u} \ceq \mapF{w}$ for functional terms $u$ and $w$;
    \item 
    $\mapF{u} \succ \mapF{w}$
    \item
    $\mapF{u\>\diff^{\tau,\upsilon}_{s,t}} \leftrightrewrite^*_{R_C} \mapF{v\>\diff^{\tau,\upsilon}_{s,t}}$ for all $s,t$;
    \item $\mapF{u}$ is irreducible \wrt\ $\Rbasic_C$.
  \end{itemize}

  \item\label{cond:deltac-produced} \textsc{Produced rewrites:}
  Given a clause $C$, let $\Deltabasic_C = \{s \rewrite t\}$ if
  \begin{enumerate}[label=(CC\arabic*), leftmargin=3.2em]
    \begin{full}
    \item there exists a closure $C_0\closure\theta \in N$ such that $C = C_0\theta$; \label{cond:deltac-instance}
    \item $C_0\closure\theta$ is variable-irreducible \wrt\ $R_C$; \label{cond:deltac-varirred}
    \end{full}
    \item $C = C' \llor s \ceq t$ for some clause $C'$ and terms $s$ and $t$;\label{cond:deltac-split-clause}
    \item $s$ is nonfunctional; \label{cond:deltac-no-fun}
    \item the root of $s$ is not a logical symbol;\label{cond:deltac-logical-symbol}
    \item if $t$ is Boolean, then $t = \itrue$ \label{cond:deltac-boolean}
    \item $s \succ t$; \label{cond:deltac-order}
    \item $s \eq t$ is maximal in $C$;\label{cond:deltac-maximal}
    \item there are no selected literals in $\slimfull{C}{C_0\closure\theta}$;\label{cond:deltac-no-sel}
    \item $s$ is irreducible by $\Rbasic_C$;  \label{cond:deltac-irreducible}
    \item $\Rbasic_C \not\modelsfol C$; \label{cond:deltac-false}
    \item $\Rbasic_C \cup \{s \rewrite t\} \not\modelsfol C'$.\label{cond:deltac-rest-false}
  \end{enumerate}
  In this case, we say that \slimfull{$C$}{$C_0\closure\theta$} produces $s \rewrite t$
  and that \slimfull{$C$}{$C_0\closure\theta$} is \emph{productive}.
  \item
  For all other terms and clauses $e$, Let $\Deltabasic_e = \emptyset$.
\end{enumerate}

Let $\Rbasic_e = \bigcup_{f\prec e} \Deltabasic_f$.
Let $\Rbasic_N = \bigcup_{e\in\termsF \cup\clausesF} \Deltabasic_e$.
\end{defi}

\begin{lem}\label{lem:RN-oriented}
  The rewrite systems $\Rbasic_C$ and $\Rbasic_N$
  do not have critical pairs and are oriented by $\succ$.
\end{lem}
\begin{proof}
  It is easy to check that all rules in $\Rbasic_C$ and $\Rbasic_N$ are oriented by $\succ$,
  using \ref{cond:PF:order:t-f-minimal}.

  To show the absence of critical pairs, suppose that there exists a critical pair $s \rewrite t$ and $s' \rewrite t'$ in $\Rbasic_N$,
  originating from $\Delta_e$ and $\Delta_{e'}$ respectively, for some $e,e' \in \termsF \cup \clausesF$.
  Without loss, we assume $e \succ e'$.
  Inspecting the rules of Definition~\ref{def:RbasicN},
  it follows that $s \succeq s'$.
  By the subterm property \ref{cond:PF:order:subterm},
  $s$ cannot be a proper subterm of $s'$.
  So for the rules to be a critical pair, $s'$ must be a subterm of $s$.
  But then $s$ is not irreducible by $\Delta_{e'}\subseteq \Rbasic_{e}$, contradicting
  the irreducibility conditions of Definition~\ref{def:RbasicN}.
\end{proof}

\begin{lem}\label{lem:RN-bool-normal}
  The normal form of any ground Boolean term \wrt\ $\Rbasic_N$ is $\itrue$ or $\ifalse$.
\end{lem}
\begin{proof}
  Inspecting the rules of Definition~\ref{def:RbasicN},
  in particular \ref{cond:deltac-logical-symbol},
  we see that $\itrue$ and $\ifalse$ are irreducible \wrt\ $\Rbasic_N$.

  It remains to show that any ground Boolean term $s$ reduces to $\itrue$ or $\ifalse$.
  We prove the claim by induction on $s$ \wrt\ $\succ$.
  If $s = \itrue$ or $s = \ifalse$, we are done.
  Otherwise, consider the rule \ref{cond:deltac-backstop-boolean}
  for $C = s \ceq \ifalse$.
  Either $s$ is reducible by $\Rbasic_C$ or
  \ref{cond:deltac-backstop-boolean} triggers, making $s$ reducible by $\Delta_C$.
  In both cases, $s$ is reducible by $\Rbasic_N$.
  Let $s'$ be the result of reducing $s$ by $\Rbasic_N$.
  By Lemma~\ref{lem:RN-oriented}, $s \succ s'$.
  By the induction hypothesis, $s'$ reduces to $\itrue$ or $\ifalse$.
  Therefore, $s$ reduces to $\itrue$ or $\ifalse$.
\end{proof}

\begin{lem} \label{lem:preservation-of-truth}
  For all ground clauses $C$, if $\Rbasic_C \modelsfol C$, then $\Rbasic_N \modelsfol C$.
\end{lem}
\begin{proof}
  We assume that $\Rbasic_C \modelsfol C$.
  Then we have $\Rbasic_C \modelsfol L$ for some literal $L$ of $C$.
  It suffices to show that $R_N \modelsfol L$.

  If $L = t \ceq t'$ is a positive literal,
  then $t\leftrightrewrite_{\Rbasic_C}^* t'$.
  Since $\Rbasic_C\subseteq \Rbasic_N$,
  this implies $t\leftrightrewrite_{\Rbasic_N}^* t'$.
  Thus, $R_N \modelsfol L$.

  If $L = t \cneq t'$ is a negative literal,
  then $t\not\leftrightrewrite_{\Rbasic_C}^* t'$.
  By Lemma~\ref{lem:RN-oriented},
  this means that $t$ and $t'$ have different normal forms \wrt\ $\Rbasic_C$.
  Without loss of generality, let $t \succ t'$. Let $s \ceqneq s'$ be the maximal literal in $C$ with $s\succeq s'$.
  We have $s \succ t$ if $s \ceqneq s'$ is positive and
  $s \succeq t$ if $s \ceqneq s'$ is negative.
  Hence, inspecting Definition~\ref{def:RbasicN},
  we see that the left-hand sides of rules in $\bigcup_{e\succeq C} \Delta_e$
  are larger than $t$.
  Since only rules with a left-hand side smaller or equal to $t$ can be involved in normalizing $t$ and $t'$
  and $\Rbasic_C \cup \bigcup_{e\succeq C} \Delta_e = \Rbasic_N$,
  it follows that $t$ and $t'$ have different normal forms \wrt\ $\Rbasic_N$.
  Therefore, $t\not\leftrightrewrite_{\Rbasic_N}^* t'$ and $R_N \modelsfol L$.
  \qedhere
\end{proof}

\begin{lem} \label{lem:producing-rest-false}
If a%
\begin{full}
closure $C_0 = C'_0\vee s_0\approx t_0 \closure \theta\in\clausesPF$
produces $s_0\theta\rewrite t_0\theta$,
then $\Rbasic_{N}\not\modelsfol C'_0\theta$.
\end{full}%
\begin{slim}
clause $C = C' \llor s \ceq t \in\clausesF$
produces $s\rewrite t$,
then $\Rbasic_{N}\not\modelsfol C'$.
\end{slim}
\end{lem}
\begin{proof}
  \begin{full}Let $C = C_0\theta$, $C' = C'_0\theta$, $s = s_0\theta$, and $t = t_0\theta$.\end{full}
  By \ref{cond:deltac-order} and \ref{cond:deltac-maximal}, all terms in $C$ are smaller or equal to $s$.
  By \ref{cond:deltac-rest-false}, we have $R_{C}\cup\left\{ s\rewrite t\right\} \not\models C'$.
  The other rules
  $R_{N}\setminus (R_{C}\cup\left\{ s\rewrite t\right\})$
  do not play any role in the truth of $C$ because
  their left-hand
  sides are greater than $s$,
  as we can see by inspecting the rules of Definition~\ref{def:RbasicN},
  in particular the irreducibility conditions,
  and because 
  $R_N$ is confluent and terminating (Lemma~\ref{lem:RN-oriented}).
  So,
  $R_{C}\cup\left\{ s\rewrite t\right\} \not\modelsfol C'$
  implies
  $R_{N}\not\modelsfol C'$.
  \qedhere
\end{proof}

\begin{full}
\begin{lem}\label{lem:productive-variable-irreducible}
  If $C\closure\theta \in \clausesPF$ is productive,
  then it is variable-irreducible \wrt\ $R_N$.
\end{lem}
\begin{proof}
  Let $s \rewrite t$ be the rule produced by $C\closure\theta$.
  By \ref{cond:deltac-varirred},
  $C\closure\theta$ is variable-irreducible \wrt\ $R_C$.
  Let $s' \rewrite t' \in \Rbasic_N \setminus \Rbasic_C$.
  Then $s' \ceq t' \in \Delta_e$ for some $e \in \termsF \cup \clausesF$
  that is larger than $C\theta$.
  So if $e$ is a term, then $s' = e \succ s$
  and thus $s' \ceq t' \succ s \ceq t$.
  If $e$ is a clause,
  then its maximal literal
  (which is $s' \ceq t'$ by \ref{cond:deltac-backstop-boolean}, \ref{cond:deltac-function}, and \ref{cond:deltac-maximal}) is at least as large as $C$'s
  maximal literal (which is $s \ceq t$ by \ref{cond:deltac-maximal}).
  So in either case, $s' \ceq t' \succeq s \ceq t$.
  Since $s \ceq t$ is the maximal literal of $C\theta$,
  $s' \ceq t'$ is at least as large as each literal of $C\theta$.
  So the rule $s' \ceq t'$ has no effect on the variable-irreducibility of $C\closure\theta$
  by Definition~\ref{def:PF:irred}.
  Therefore, $C\closure\theta$ is variable-irreducible \wrt\ $R_N$.
\end{proof}
\end{full}

\begin{lem}\label{lem:RN-argcong-diff}
  Let $u$ and $w$ be higher-order ground terms of type $\tau \fun \upsilon$.
  If $\mapF{u} \leftrightrewrite^*_{R_N} \mapF{w}$,
  then $\mapF{u\>\diff^{\tau,\upsilon}_{s,t}} \leftrightrewrite^*_{R_N} \mapF{w\>\diff^{\tau,\upsilon}_{s,t}}$
  for all $s,t$.
\end{lem}
\begin{proof}
By induction over each rewrite step in $\mapF{u} \leftrightrewrite^*_{R_N} \mapF{w}$,
it suffices to show the following claim:
If $\mapF{u} \rewrite_{R_N} \mapF{w}$,
then $\mapF{u\>\diff^{\tau,\upsilon}_{s,t}} \leftrightrewrite^*_{R_N} \mapF{w\>\diff^{\tau,\upsilon}_{s,t}}$
for all $s,t$.
Here, it is crucial that $s$ and $t$ are not necessarily equal to $u$ and $w$.

\begin{full}
If the rewrite position is in a proper subterm of $\mapF{u}$,
by definition of $\mapFonly$,
the rewrite position
corresponds to 
a proper yellow subterm of $u$.
Yellow subterms of a functional term remain when applying the term to an argument.
So the same rewrite step can be applied to obtain
$\mapF{u\>\diff^{\tau,\upsilon}_{s,t}} \rewrite_{R_N} \mapF{w\>\diff^{\tau,\upsilon}_{s,t}}$.

For a rewrite in the root of the term,
the rewrite rule must originate from \ref{cond:deltac-function}
because the terms are functional.
\end{full}
\begin{slim}
By definition of $\mapFonly$, since $u$ is functional,
$\mapF{u} = \cst{fun}_u$. So $\mapF{u}$ has no proper subterms,
and thus the rewrite step must happen at the root of $\mapF{u}$.
Inspecting the definition of $R_N$, we observe that the rewrite rule must originate from \ref{cond:deltac-function}.
\end{slim}
One of the conditions of \ref{cond:deltac-function} then yields the claim.
\end{proof}

\begin{lem}\label{lem:RN-ext}
Let $u$ and $w$ be higher-order ground terms of type $\tau \fun \upsilon$.
If $\mapF{u\>\diff^{\tau,\upsilon}_{s,t}} \leftrightarrow^*_{R_N} \mapF{w\>\diff^{\tau,\upsilon}_{s,t}}$
for all $s,t$, then $\mapF{u} \leftrightarrow^*_{R_N} \mapF{w}$.
\end{lem}
\begin{proof}
Let $\mapF{u'} = \mapF{u}\downarrow_{R_N}$ and $\mapF{w'} = \mapF{w}\downarrow_{R_N}$.
By applying Lemma~\ref{lem:RN-argcong-diff}
to $\mapF{u} \leftrightarrow^*_{R_N} \mapF{u'}$ and to $\mapF{w'} \leftrightarrow^*_{R_N} \mapF{w}$,
we have $\mapF{u'\>\diff^{\tau,\upsilon}_{s,t}} \leftrightarrow^*_{R_N} \mapF{w'\>\diff^{\tau,\upsilon}_{s,t}}$
for all $s,t$.

We want to show that $\mapF{u} \leftrightarrow^*_{R_N} \mapF{w}$---i.e., that $\mapF{u'} = \mapF{w'}$.
To derive a contradiction, we assume that $\mapF{u'} \neq \mapF{w'}$.
Without loss of generality, we may assume that $\mapF{u'}\succ \mapF{w'}$.
Then, using \ref{cond:PF:order:ext}, all conditions of \ref{cond:deltac-function} are satisfied for the rule $\mapF{u'}\rewrite \mapF{w'}$,
contradicting the fact that $\mapF{u'}$ is a normal form.
\end{proof}

\begin{lem}\label{lem:PF:R-olammodels}
$R_N$ is a $\modelsolam$-interpretation.
\end{lem}
\begin{proof}
We must prove all conditions listed in Section~\ref{ssec:redundancy}.
\begin{itemize}
  \item By Lemma~\ref{lem:RN-bool-normal}, the Boolean type has exactly two elements,
  namely the interpretations of $\itrue$ and $\ifalse$.
  The rule~\ref{cond:deltac-logical-boolean} ensures that the symbols
  $\inot$, $\iand$, $\ior$, $\iimplies$, $\ieq^\tau$, $\ineq^\tau$ are interpreted as the corresponding logical operations.
  Note that $R_s$ never contains any rules rewriting $s$
  because $s$ is smaller than any clause containing $s$.
  So $s$ can be reducible \wrt\ $R_s$ only when one of its proper subterms is reducible.
  Since every term has a normal form, adding rules only for the irreducible terms is sufficient.
  
  \item By 
  Lemma~\ref{lem:PG:mapI-mapF-fol}, we have
   $\mapF{\mapI{u}\>\diff^{\tau,\upsilon}_{s,t}} = \mapF{u\>\diff\typeargs{\tau,\upsilon}(s,t)}$
   for all $u,s,t \in \TT_\mathrm{ground}(\SigmaH)$.
  Since $\mapIonly$ is a bijection on ground terms,
  Lemma~\ref{lem:RN-ext} proves the extensionality condition in Section~\ref{ssec:redundancy}.
  \item 
  The argument congruence condition in Section~\ref{ssec:redundancy}
  follows from Lemma~\ref{lem:RN-argcong-diff} in the same way.\qedhere
\end{itemize}
\end{proof}

\begin{full}
\begin{lem} \label{lem:inferences-preserve-varirred}
If the premises of an inference are variable-irreducible \wrt~a ground rewrite system $R$
with $R \subseteq {\succ}$
and the inference is not a $\PFExt$ or $\PFDiff$ inference,
then the conclusion is also variable-irreducible \wrt~$R$.
\end{lem}
\begin{proof}
Let $C_0\closure\theta$ be the conclusion of the inference.
By Definition~\ref{def:PF:irred},
we have to show that for all literals $L\closure\theta$ of $C_0\closure\theta$
and all variables $x$ of $L$,
$x\theta$ is is irreducible \wrt\ all rules $l \rewrite r \in R$ with $L\theta \succ l \eq r$
and all Boolean subterms of $x\theta$ are either $\itrue$ or $\ifalse$.
Since the premises of the infererence are variable-irreducible,
this is evident for all literals that occur also in one of the premises.
The Boolean subterm condition is satisfied for all inferences because no inference
other than $\PFExt$ and $\PFDiff$ introduces
a variable that is not present in the premises.
It remains to check the the irreducibility condition for newly introduced literals in $C_0\closure\theta$.

For $\PFSup$ inferences, the only newly introduced literal has the form
$L[t']\closure\theta$, where $L[u]\closure\theta$
is a literal of the second premise and $(t \ceq t')\closure\theta$
is a literal in the first premise with $t\theta = u\theta$ and $t\theta \succ t'\theta$.
Let $x$ be a variable in $L[t']$.
If $x$ occurs in $t'$, then $x\theta$ is irreducible
\wrt\ all rules $l \rewrite r \in R$ with $t\theta \ceq t'\theta \succ l \eq r$
since the first premise is variable-irreducible \wrt~$R$.
On the other hand,
$x\theta$ must also be irreducible \wrt\ all rules $l \rewrite r \in R$
with $t\theta \ceq t'\theta \preceq l \eq r$,
since then $l \succeq t\theta \succ t'\theta \succeq x\theta$.
Therefore, $x\theta$ is irreducible \wrt\ all rules $l \rewrite r \in R$.
If $x$ occurs in $L[t']$ but not in $t'$,
then it occurs in $L[u]$,
and because the second premise is variable-irreducible \wrt~$R$,
$x\theta$ is irreducible
\wrt\ all rules $l \rewrite r \in R$ with $L[u]\theta \succ l \eq r$.
Since $L[t']\theta \prec L[u]\theta$,
$x\theta$ is also irreducible
\wrt\ all rules $l \rewrite r \in R$ with $L[t']\theta \succ l \eq r$.

For all other inferences, it is easy to verify that
whenever a variable $x$ occurs in a newly introduced literal $L\closure\theta$
in the conclusion,
then $x$ occurs also in a literal $L'\closure\theta$ in the premise
with $L\theta \preceq L'\theta$,
so the premise's variable-irreducibility implies
the conclusion's variable-irreducibility.
\qedhere
\end{proof}
\end{full}

We employ a variant of Bachmair and Ganzinger's
framework of reducing counterexamples \cite[Sect.~4.2]{bachmair-ganzinger-2001-resolution}.
Let $N \subseteq \clausesPF$ with $\bot \not\in\slimfull{N}{\mapT{N}}$.
A \slimfull{clause $C_0 \in\clausesF$}{closure $C_0\closure\theta \in \clausesPF$} 
is called a \emph{counterexample} if 
\begin{full}it is variable-irreducible w.r.t.~$\Rbasic_N$ and\end{full}
$\Rbasic_N \not\modelsfol C_0\slimfull{}{\closure\theta}$.
An inference \emph{reduces} a counterexample $C_0\slimfull{}{\closure\theta}$ if its main premise is $C_0\slimfull{}{\closure\theta}$,
its side premises are in $N$ and true in $\Rbasic_N$, and
its conclusion is a counterexample smaller than $C_0\slimfull{}{\closure\theta}$.
An inference system has the \emph{reduction property for counterexamples}
if for all $N \subseteq \clausesPF$ and all counterexamples $C_0\slimfull{}{\closure\theta} \in N$, there exists an inference from $N$
that reduces $C_0\slimfull{}{\closure\theta}$.

\begin{lem}
  \label{lem:reducible-s}
  Let $\slimfull{C}{C_0\closure\theta} \in N$ be a counterexample.
  Let $\slimfull{L}{L_0}$ be a literal in $\slimfull{C}{C_0\closure\theta}$
  that is eligible and negative or strictly eligible and positive.
  \begin{full}Let $C = C_0\theta$ and $L = L_0\theta$.\end{full}
  We assume that the larger side of $L$
  is reducible by a rule $s \rewrite s' \in \Rbasic_C$.
  Then the inference system $\PFInf$ reduces the counterexample $\slimfull{C}{C_0\closure\theta}$.
\end{lem}
\begin{proof}
  Let $p$ be the position of $C$ that is
  located at the larger side of $L$ and reducible by $s \rewrite s'$.%
\begin{full}

  First, we claim that $p$ is not at or below a variable position of $C_0$
  and thus $s = s_0\theta$ for a subterm $s_0$ of $C_0$ that is not a variable.

  To see this, assume for a contradiction
  that $p$ is at or below a variable position of $C_0$, i.e.,
  there is a variable $x$ in $C$ such that $s$ is a subterm of $x\theta$.

  If $s$ were Boolean, then $s$ must be $\itrue$ or $\ifalse$ by variable-irreducibility,
  contradicting the fact that $s \rewrite s' \in \Rbasic_C$ because Definition~\ref{def:RbasicN}
  does not produce rules rewriting $\itrue$ or $\ifalse$.
  
  So $s$ is not a Boolean term. Then, by the rules of Definition~\ref{def:RbasicN},
  $s \rewrite s' \in \Rbasic_C$ implies that $C \succeq s \ceq s'$.
  By variable-irreducibility, since $s$ is a subterm of $x\theta$,
  $L \preceq s \ceq s'$, where $L$ is the literal of $C$ containing position $p$.
  Since $L$ contains $s$, it follows that $L$ must be positive
  and its larger side must be $s$.
  Since $p$ is eligible and $s$ is not a Boolean term, $L$ must be the maximal literal of $C$.
  So, since $C \succeq s \ceq s'$ and $L \preceq s \ceq s'$,
  we have $L = s \ceq s'$.
  But this contradicts this lemma's assumption that $\Rbasic_{N}\not\modelsfol C$.

  This concludes the proof of our claim that $s = s_0\theta$ for a subterm $s_0$ of $C_0$ that is not a variable.

  \smallskip

  If the subterm $s$ at position $p$ is not a green subterm of $C$, then
  it must be contained in functional green subterm $t$ of $C$.
  Let $q$ be the position of $t$ in $C$.
  Since $p$ is eligible in $C_0\closure\theta$, $q$ is also eligible in $C_0\closure\theta$.
  Let $t'$ the normal form of $t$ \wrt\ $\Rbasic_N$.
  Let $u = \mapFonly^{-1}(t)$ and $w = \mapFonly^{-1}(t')$.
  Then $\mapF{u} \leftrightrewrite^*_{R_N} \mapF{w}$.
  By Lemma~\ref{lem:RN-argcong-diff},
  $\mapF{u\>\diff^{\tau,\upsilon}_{u,w}} \leftrightrewrite^*_{R_N} \mapF{w\>\diff^{\tau,\upsilon}_{u,w}}$.
  Since $t$ contains $s$, it is reducible \wrt\ $\Rbasic_N$,
  and thus $\mapF{u} = t \succ t' = \mapF{w}$.
  Thus, we can apply \PFExt{} to reduce the counterexample.
  To fulfill the conditions of \PFExt{} on $w$,
  we must replace the nonfunctional yellow subterms of $w$
  by fresh variables and choose $\rho$ accordingly.
  Given the above properties of $R_N$, the conclusion of this inference is equivalent to the premise.
  It is also smaller than the premise by \ref{cond:PF:order:ext}
  and because $\mapF{u} \succ \mapF{w}$.

  It remains to show that the conclusion is variable-irreducible.
  First, consider one of the fresh variables introduced to replace the
  nonfunctional yellow subterms of $w$.
  Since yellow subterms in $w$ correspond to subterms in $t'$,
  and $t'$ is a normal form,
  $x\rho$ is irreducible \wrt\ $R_N$.
  Next, consider a variable $x$
  that occurs in a literal $L\closure\theta$
  in the conclusion but is not one of the fresh variables.
  Then $x$ occurs also in a literal $L'\closure\theta$ in the premise $C_0\closure\theta$
  with $L\theta \preceq L'\theta$.
  Since $C_0\closure\theta$ is variable-irreducible,
  this shows that the conclusion is variable-irreducible.

  \smallskip

  Otherwise, $s$ is a green subterm of $C$.
\end{full}
  \slimfull{We}{Then we} make a case distinction on which case of Definition~\ref{def:RbasicN} the rule $s \rewrite s'$ originates from:
  \begin{itemize}
    \item 
    \ref{cond:deltac-logical-boolean}
      Then the root of $s$ is a logical symbol and $s\notin\{\itrue,\ifalse\}$.
      By Lemma~\ref{lem:RN-bool-normal},
      $R_N$ reduces $s$ to $\itrue$ or to $\ifalse$.
      \begin{full}
      By our claim above, $s_0$ is not a variable and
      since $s = s_0\theta$, $s_0$ has a logical symbol at its root.
      \end{full}
      \begin{itemize}
      \item First consider the case where the position $p$ in $C$ is in a literal of the form $s \ceq \itrue$
      or $s \ceq \ifalse$.
      Then \slimfull{\FClausify}{\PFClausify{}} is applicable to $\slimfull{C}{C_0\closure\theta}$
      and the conclusion of this inference
      is smaller than it.
      Moreover, the conclusion is equivalent to $\slimfull{C}{C_0\closure\theta}$
      by Lemma~\ref{lem:PF:R-olammodels}%
\begin{full}
      and variable-irreducible by Lemma~\ref{lem:inferences-preserve-varirred}
\end{full}.
      \item Otherwise, we apply either \PFBoolHoist{}
      (if $\slimfull{s}{s_0}$ reduces to $\ifalse$)
      or \PFLoobHoist{} (if $\slimfull{s}{s_0}$ reduces to $\itrue$).
      In both cases, the conclusion of the inference is smaller than $\slimfull{C}{C_0\closure\theta}$.
      Moreover, the conclusion is equivalent to $\slimfull{C}{C_0\closure\theta}$
      by Lemma~\ref{lem:PF:R-olammodels}%
\begin{full}
      and variable-irreducible by Lemma~\ref{lem:inferences-preserve-varirred}
\end{full}.
      \end{itemize}
    \item \ref{cond:deltac-backstop-boolean}
    Then $R_N$ reduces $s$ to $\ifalse$ and $s\notin\{\itrue,\ifalse\}$.
    Due to the presence of the rule $s\rewrite \ifalse$ in $R_C$, $C$ must be larger than $s \ceq \ifalse$.
    So, since $p$ is eligible in $C$, this position cannot be in a literal of the form $s \ceq \itrue$.
    It cannot be in a literal of the form $s \ceq \ifalse$ either because $s \ceq \ifalse$ is true in $R_N$.
    So we can apply \slimfull{\FBoolHoist}{\PFBoolHoist{}} to reduce the counterexample,
    again using Lemma~\ref{lem:PF:R-olammodels}\begin{full} and Lemma~\ref{lem:inferences-preserve-varirred}\end{full}.
    \item 
    \ref{cond:deltac-function} 
    Then $s$ is functional and reducible \wrt\ $\Rbasic_N$.
    \begin{full}Then we can proceed as in the $\PFExt$ case above, using $s$ in the role of $t$.\end{full}%
    \begin{slim}
      Consider the normal form $s\downarrow_{\Rbasic_N}$ of $s$ \wrt\ $\Rbasic_N$.
      Let $u = \mapFonly^{-1}(s)$ and $w = \mapFonly^{-1}(s\downarrow_{\Rbasic_N})$.
      Then $\mapF{u} \leftrightrewrite^*_{R_N} \mapF{w}$.
      By Lemma~\ref{lem:RN-argcong-diff},
      $\mapF{u\>\diff^{\tau,\upsilon}_{u,w}} \leftrightrewrite^*_{R_N} \mapF{w\>\diff^{\tau,\upsilon}_{u,w}}$.
      Since $\Rbasic_N$ is oriented by $\succ$,
      we have $\mapF{u} = s \succ s\downarrow_{\Rbasic_N} = \mapF{w}$.
      Thus, we can apply \slimfull{\FExt}{\PFExt{}} to reduce the
      counterexample\slimfull{, using Remark~\ref{rem:FInf-SigmaI-SigmaH-equivalence}}{}.
      Given the above properties of $R_N$, the conclusion of this inference is equivalent to the premise.
      It is also smaller than the premise by \ref{cond:PF:order:ext}
      and because $\mapF{u} \succ \mapF{w}$.
    \end{slim}
    \item 
    \ref{cond:deltac-produced}
  Then some 
  \begin{slim}clause $D \llor s\ceq s'$\end{slim}%
  \begin{full}closure $D_0 \llor t \ceq t'\closure\>\rho$ with $(D_0 \llor t \ceq t')\rho = D \vee s\ceq s'$\end{full}
  smaller than $C$ produces the rule
  $s\rewrite s'$. 
  We claim that the counterexample $C$ is reduced by the
  inference
  \begin{full}
  \[
  \namedinference{\PFSup}{D_0 \llor t\ceq t'\closure\rho \quad \greensubterm{C_0}{s_0} \closure \theta}
  {D_0\llor \greensubterm{C_0}{t'}\closure(\rho\cup\theta)}
  \]
  \end{full}%
  \begin{slim}
    \[
    \namedinference{\FSup}{D \llor s\ceq s'\quad \subterm{C}{s}}
    {D\llor \subterm{C}{s'}}
    \]
  \end{slim}
  This superposition is a valid inference:
  \begin{itemize}
  \begin{full}
  \item $t\rho = s = s_0\theta$.
  \item By our claim above, $s_0$ is not a variable.
  \end{full}
  \item $\slimfull{s}{s_0}$ is nonfunctional by \ref{cond:deltac-no-fun}.
  \item We have $s\succ s'$ by \ref{cond:deltac-order}.
  \item $D\llor s\ceq s' \prec C[s]$ because $D\llor s\ceq s'$ produces a rule in $R_C$.
  \item The position $p$ of $s$ in $\slimfull{C}{C_0\closure\theta}$ is eligible by assumption of this lemma.
  \item The literal $\slimfull{s\ceq s'}{t\ceq t'}$ is eligible in
  $\slimfull{D\llor s\ceq s'}{(D_0 \llor t \ceq t')\closure\>\rho}$ by \ref{cond:deltac-maximal} and \ref{cond:deltac-no-sel}.
  It is strictly eligible because if $s\ceq s'$ also occurred as a literal in $D$,
  we would have $R_{D \llor s\ceq s'} \cup \{s \rewrite s'\} \modelsfol D$,
  in contradiction to \ref{cond:deltac-rest-false}.
  \item If $\slimfull{s'}{t'\rho}$ is Boolean, then $\slimfull{s'}{t'\rho} = \itrue$ by \ref{cond:deltac-boolean}.
  \end{itemize}
  As $\slimfull{D\llor s\ceq s'}{D_0 \llor t \ceq t'\closure\>\rho}$ is productive, $R_{N}\not\modelsfol D$ by Lemma \ref{lem:producing-rest-false}.
  Hence $D\llor C\left[s'\right]$ is
  equivalent to $C\left[s'\right]$, which is equivalent
  to $C\left[s\right]$ \wrt\ $R_{N}$.
  \begin{full}
  Moreover, $(D_0 \llor t\ceq t')\closure\rho$
  is variable-irreducible by Lemma~\ref{lem:productive-variable-irreducible}.
  So $D\llor C\left[s'\right]$ is variable-irreducible by Lemma~\ref{lem:inferences-preserve-varirred}.
  \end{full}
  It remains to show that the new counterexample $D\llor C\left[s'\right]$ is strictly smaller than $C$.
  Using \ref{cond:PF:order:comp-with-contexts}, $C[s']\prec C$ because $s'\prec s$ and  $D\prec C$ because $D\vee s\approx s'\prec C$.
  Thus, the inference reduces the counterexample $C$.\qedhere
  \end{itemize}
\end{proof}

\begin{lem} \label{lem:reduction-prop}
  The inference system $\PFInf$ has the reduction property for counterexamples.
\end{lem}
\begin{proof}
  Let $\slimfull{C}{C_0\closure\theta} \in N$ be a counterexample---i.e.,\ a \slimfull{clause}{closure in $\irred_{\Rbasic_{N}}(N)$}
  that is false in $R_{N}$.
  We must show that there is an inference from $N$ that reduces
  $\slimfull{C}{C_0\closure\theta}$; i.e., 
  the inference has main premise $\slimfull{C}{C_0\closure\theta}$,
  side premises in $N$ that are true in $R_{N}$,
  and a conclusion that is a smaller counterexample
  than $\slimfull{C}{C_0\closure\theta}$.
  \begin{full}For all claims of a reducing inference in this proof, we use Lemma~\ref{lem:inferences-preserve-varirred}
  to show that the conclusion is variable-irreducible.\end{full}
  
  Let $\slimfull{L}{L_0}$ be an eligible literal in $\slimfull{C}{C_0\closure\theta}$.
  \begin{full}Let $C = C_0\theta$.\end{full}
  We proceed by a case distinction:

  \medskip
  \noindent\textsc{Case 1:}\enskip $\slimfull{L}{L_0\theta}$ is of the form $s \cneq s'$.
    \begin{itemize}
      \item Case 1.1: $s = s'$. Then \slimfull{\FEqRes}{\PFEqRes{}} reduces $C$.
      \item Case 1.2: $s \ne s'$. Without loss of generality, $s \succ s'$.
      Since $R_{N}\not\modelsfol C$, we have $R_C \not\modelsfol C$ by Lemma~\ref{lem:preservation-of-truth}.
      Therefore, $R_C\not\modelsfol s\not\approx s'$
      and $R_C\modelsfol s\approx s'$. Thus, $s$ must be reducible
      by $R_C$ because $s\succ s'$. Therefore,
      we can apply Lemma~\ref{lem:reducible-s}.
    \end{itemize}
  \noindent\textsc{Case 2:}\enskip $\slimfull{L}{L_0\theta}$ is of the form $s \ceq s'$. Since $R_N \not\modelsfol C$, we can assume without loss of generality that $s \succ s'$.
    \begin{itemize}
      \item Case 2.1: $\slimfull{L}{L_0}$ is eligible, but not strictly eligible.
      Then $\slimfull{L}{L_0\theta}$ occurs more than once in $C$.
      So we can apply \slimfull{\FEqFact}{\PFEqFact{}} to reduce the counterexample.
      \item Case 2.2: $\slimfull{L}{L_0}$ is strictly eligible and $s$ is reducible by $R_C$.
      Then we apply Lemma \ref{lem:reducible-s}.
      \item Case 2.3: $\slimfull{L}{L_0}$ is strictly eligible and $s = \ifalse$.
       Then, since $s \succ s'$, we have $s' = \itrue$ by \ref{cond:PF:order:t-f-minimal}. So, \PFFalseElim{} reduces the counterexample.
      \item Case 2.4: $\slimfull{L}{L_0}$ is strictly eligible and $s$ is functional.
      Then we apply \slimfull{\FArgCong}{\PFArgCong{}} to reduce the counterexample.
      The conclusion is smaller than the premise by \ref{cond:PF:order:ext}.
      By Lemma~\ref{lem:RN-ext},
      there must be at least one choice of $u$ and $w$ in the \slimfull{\FArgCong}{\PFArgCong{}} rule such that
      the conclusion is a counterexample.
      \item Case 2.5: $\slimfull{L}{L_0}$ is strictly eligible and $s \ne \ifalse$ is nonfunctional and not reducible by $R_C$.
      Since $R_N \not\modelsfol C$, $\slimfull{C}{C_0\closure\theta}$ cannot be productive.
      So at least one of the conditions of \ref{cond:deltac-produced} of Definition~\ref{def:RbasicN} is violated.
      \begin{full}\ref{cond:deltac-instance}, \ref{cond:deltac-varirred}, \end{full}\ref{cond:deltac-split-clause}, \ref{cond:deltac-no-fun}, \ref{cond:deltac-order},
      \ref{cond:deltac-irreducible},
      and \ref{cond:deltac-false} are clearly satisfied.

      For  \ref{cond:deltac-logical-symbol}, \ref{cond:deltac-boolean}, \ref{cond:deltac-maximal}, and \ref{cond:deltac-no-sel}, we argue as follows:
      \begin{itemize}
        \item\ref{cond:deltac-logical-symbol}:
          If $s$ were headed by a logical symbol, then
          one of the cases of \ref{cond:deltac-logical-boolean} applies.
          The condition in \ref{cond:deltac-logical-boolean} that
          any Boolean arguments of $s$ must be $\itrue$ or $\ifalse$
          is fulfilled by Lemma~\ref{lem:RN-bool-normal} and the fact that
          the rules applicable to subterms of $s$ in $R_N$
          are already contained in $R_s$.
          So \ref{cond:deltac-logical-boolean}
          adds a rewrite rule for $s$ to $R_C$, contradicting irreducibility of~$s$.

        \item\ref{cond:deltac-boolean}:
        If $s'$ were a Boolean other than $\itrue$, 
        since $s \succ s'$, we would have $s \ne \itrue,\ifalse$ by \ref{cond:PF:order:t-f-minimal}.
        Moreover, $s' \succeq \ifalse$, and thus $C \succeq s \ceq \ifalse$.
        Since $s$ is not reducible by $R_C$, is is also irreducible by $R_{s \ceq \ifalse} \subseteq R_C$.
        So \ref{cond:deltac-backstop-boolean} triggers and sets $\Delta_{s \ceq \ifalse} = \{s \rewrite \ifalse\}$.
        Since $s$ is not reducible by $R_C$, we must have $C = s \ceq \ifalse$.
        But then $C$ istrue in $R_N$, a contradiction.

        \item\ref{cond:deltac-maximal}: By \ref{cond:deltac-boolean}, $\slimfull{L}{L_0}$ cannot be selected and thus eligibility implies maximality.
        
        \item\ref{cond:deltac-no-sel}: By \ref{cond:deltac-boolean}, $\slimfull{L}{L_0}$ cannot be selected. If another literal was selected, $L_0$ would not be eligible.
      \end{itemize}
      So \ref{cond:deltac-rest-false} must be violated.
      Then $R_{C}\cup\{ s\rewrite s'\} \modelsfol C'$,
      where $C'$ is the subclause of $C$ with $\slimfull{L}{L_0\theta}$ removed.
      However, $R_{C}\not\modelsfol C$, and
      therefore, $R_{C}\not\modelsfol C'$.
      Thus, we must have $C'=C''\llor r\ceq t$ for some terms $r$ and $t$, where $R_{C}\cup\{ s\rewrite s'\} \modelsfol r\ceq t$
      and $R_{C}\not\modelsfol r\ceq t$.
      So $r\neq t$ and without loss of generality we assume $r\succ t$. Moreover $s\rewrite s'$ must
      participate in the normalization of $r$ or $t$ by $R_{C}\cup\left\{ s\rewrite s'\right\} $.
      Since $s \ceq s'$ is maximal in $C$ by \ref{cond:deltac-maximal},
      $r \preceq s$.
      So the rule $s\rewrite s'$ can be
      used only as the first step in the normalization of $r$. Hence $r=s$ and
      $R_{C}\modelsfol s'\ceq t$.
      Then \slimfull{\FEqFact}{\PFEqFact{}} reduces the counterexample.\qedhere
    \end{itemize}
\end{proof}

Using Lemma~\ref{lem:reduction-prop} and the same ideas as for Theorem 4.9 of
Bachmair and Ganzinger's framework~\cite{bachmair-ganzinger-2001-resolution},
we obtain the following theorem:

\begin{thm} \label{thm:PF:model-for-varirred-closures}
Let $N$ be a set of closures that is saturated up to redundancy
\wrt\ \slimfull{$\FInf$ and $\FRedI$}{$\PFInf$ and $\PFRedI$},
and 
\begin{full}$N$ does not contain $\bot\closure\theta$ for any $\theta$.\end{full}%
\begin{slim}$\bot \not\in N$.\end{slim}
Then $\Rbasic_N \modelsolam \slimfull{N}{\irred_{\Rbasic_N}(N)}$.
\end{thm}
\begin{proof}
By Lemma~\ref{lem:PF:R-olammodels},
it suffices to show that $\Rbasic_N \modelsfol \slimfull{N}{\irred_{\Rbasic_N}(N)}$.
For a proof by contradiction,
we assume that  $\Rbasic_N \not\modelsfol \slimfull{N}{\irred_{\Rbasic_N}(N)}$.
Then $N$ contains a minimal counterexample, i.e.,
a \slimfull{clause $C$}{closure $C_0\closure \theta$} \begin{full}that is variable-irreducible \wrt\ $\Rbasic_N$\end{full}
with $\Rbasic_N \not\modelsfol \slimfull{C}{C_0\closure \theta}$.
Since $\PFInf$ has the reduction property for counterexamples by Lemma~\ref{lem:reduction-prop},
there exists an inference that reduces $\slimfull{C}{C_0\closure \theta}$---i.e., an inference $\iota$ with main premise $\slimfull{C}{C_0\closure \theta}$,
side premises in $N$ that are true in $\Rbasic_N$,
and a conclusion $\concl(\iota)$ that is smaller than $\slimfull{C}{C_0\closure\theta,}$
\begin{full}variable-irreducible \wrt\ $\Rbasic_N$,\end{full} and false in $\Rbasic_N$.
By saturation up to redundancy,
$\iota \in \slimfull{\FRedI}{\PFRedI}$.
By Definition~\ref{def:PF:RedI},
we have 
\begin{slim}$\{E \in N \mid E \prec C\} \modelsolam \concl(\iota)$.\end{slim}
\begin{full}$\Rbasic_N \cup \{E \in \irred_{\Rbasic_N}(N) \mid E \prec C_0\closure\theta\}\modelsolam \concl(\iota)$.\end{full}
By minimality of the counterexample $\slimfull{C}{C_0\closure\theta}$,
the 
\begin{full}closures $\{E \in \irred_{\Rbasic_N}(N) \mid E \prec C_0\closure\theta\}$\end{full}
\begin{slim}clauses $\{E \in N \mid E \prec C\}$\end{slim}
must be true in $\Rbasic_N$,
and it follows that $\concl(\iota)$ is true in
$\Rbasic_N$, a contradiction.
\end{proof}

\begin{full}
\begin{lem}\label{lem:PF:irred-entails-all}
Let $R$ be a confluent term rewrite system oriented by $\succ$
whose only Boolean normal forms are $\itrue$ and $\ifalse$.
Let $N\subseteq\clausesPF$
such that for every $C\closure\theta \in N$
and every grounding substitution $\rho$ that
coincides with $\theta$ on all variables not occurring in $C$, 
we have $C\closure \rho\in N$.
Then $R \cup \irred_R(N) \modelsolam N$.
\end{lem}
\begin{proof}
Let $C\closure\theta \in N$.
We must show that $R \cup \irred_R(N) \modelsolam C\closure\theta$.
We define a substitution $\theta'$
by $x\theta' = (x\theta)\!\!\downarrow_{R}$
for variables $x$ occurring in $C$ and $x\theta' = x\theta$ for all other variables.
Then $R \cup \{ C\closure\theta'\} \modelsolam  C\closure\theta$. 
Moreover, $\theta'$ is grounding and coincides with $\theta$ on all variables not occurring in $C$.
By the assumption of this lemma, we have $C\closure\theta' \in N$.
Finally, we observe that the closure $C\closure\theta'$ 
is variable-irreducible \wrt{} $R$---i.e.,
$C\closure\theta' \in \irred_R(N)$.
It follows that $R \cup \irred_R(N) \modelsolam  C\closure\theta$.  
\end{proof}
\end{full}

\begin{slim}\subsubsection{Indexed Ground Higher-Order Level}\mbox{}\end{slim}%
\begin{full}\subsubsection{Indexed Partly Substituted Ground Higher-Order Level}\mbox{}\end{full}%

In this subsubsection, let $\succ$ be an admissible term order for $\IPGInf$ (Definition~\ref{def:IPG:admissible-term-order}),
let $\mathit{i\slimfull{}{p}gsel}$ be a selection function on $\clausesIPG$, and
let $N \subseteq \clausesIPG$ such that $N$ is saturated up to redundancy \wrt{} $\IPGInf$ and 
\begin{full}$\bot\closure\theta\not\in N$ for all $\theta$.\end{full}%
\begin{slim}$\bot \not\in N$.\end{slim}
We write $R$ for the term rewrite system $R_{\mapF{N}}$ constructed 
in the previous subsubsection
\wrt\ $\succ_\mapFonly$ and $\mapF{\mathit{i\slimfull{}{p}gsel}}$.
We write $\eqR{t}{s}$ for $\mapF{t} \leftrightrewrite^*_R \mapF{s}$, where $t, s \in \TT_\mathrm{ground}(\SigmaI)$.

Our goal in this subsubsection is to use $R$
to define a higher-order interpretation
that is a model of $N$.
To obtain a valid higher-order
interpretation, we need to show that
$\eqR{s\theta}{s\theta'}$ 
whenever $\eqR{x\theta}{x\theta'}$ for all $x$ in $s$.

\begin{lem}[Argument congruence]\label{lem:argcong}
  Let $\eqR{s}{s'}$ for $s, s' \in \TT_\mathrm{ground}(\SigmaI)$.
  Let $u \in \TT_\mathrm{ground}(\SigmaI)$.
  Then $\eqR{s\>u}{s'\>u}$.
\end{lem}
\begin{proof}
\begin{full}
Let $t,t',v$ be terms and $\theta$ a grounding substitution such that
$t\theta = s$, $t'\theta = s'$, $v\theta = u$, and
the nonfunctional yellow subterms of $t,t',v$ are different variables.
Let $\rho$ the substitution 
resulting from $R$-normalizing all values of $\theta$ (via $\mapFonly$).
Then there exists the %
inference
\begin{align*}
  \namedinference{\IPGDiff}{}
  {t\>\diff^{\tau,\upsilon}_{t\rho,t'\rho} \noteq t'\>\diff^{\tau,\upsilon}_{t\rho,t'\rho} \llor t\>v \eq t'\>v \closure\rho}
\end{align*}
which we call $\iota$.
By construction of $\rho$, its conclusion is variable-irreducible.

\end{full}%
\begin{slim}
Consider the following inference $\iota$:
\begin{align*}
  \namedinference{\IGDiff}{}
  {s\>\diff^{\tau,\upsilon}_{s,s'} \noteq s'\>\diff^{\tau,\upsilon}_{s,s'} \llor s\>u \eq s'\>u}
\end{align*}%
\end{slim}%
Since $N$ is saturated,
$\iota$ is redundant and thus
$\slimfull{\mapF{N}}{R\cup\mapF{\irred_R(N)}} \models \mapF{\concl(\iota)}$.
Hence
$\RfN\models\mapF{\concl(\iota)}$ by
Theorem~\ref{thm:PF:model-for-varirred-closures}
and
Lemma~\ref{lem:IPG:saturated}.

\begin{full}
We have $\eqR{t\rho} {t\theta}= \eqR{s}{s'} = \eqR{t'\theta}{t'\rho}$.
\end{full}
By
Lemma~\ref{lem:RN-argcong-diff},
\begin{full}$R \models \mapF{t\rho\>\diff^{\tau,\upsilon}_{t\rho,t'\rho}} \ceq \mapF{t'\rho\>\diff^{\tau,\upsilon}_{t\rho,t'\rho}}$.\end{full}%
\begin{slim}$R \models \mapF{s\>\diff^{\tau,\upsilon}_{s,s'}} \ceq \mapF{s'\>\diff^{\tau,\upsilon}_{s,s'}}$.\end{slim}
\slimfull{It}{Using Lemma~\ref{lem:IPG:mapF-subst}, it} follows that 
$R \models \mapF{\slimfull{s\>u\eq s'\>u}{(t\>v \eq t'\>v)\rho}}$.
\begin{slim}Thus, $\eqR{s\>u}{s'\>u}$.\end{slim}%
\begin{full}
Since applying a functional term to an argument preserves nonfunctional yellow subterms of both
the functional term and its argument,
we have $s\>u = \eqR{(t\>v)\theta}{(t\>v)\rho}$ and
$s'\>u = \eqR{(t'\>v)\theta}{(t'\>v)\rho}$.
So $R \models \mapF{s\>u \eq s'\>u}$ and thus $\eqR{s\>u}{s'\>u}$.
\end{full}
\end{proof}

The following lemma and its proof are essentially identical to
Lemma~54 of
Bentkamp et al.~\cite{bentkamp-et-al-2023-hosup-journal}.
We have adapted the proof to use De Bruijn indices,
and we have removed the notion of term-ground
and replaced it by preprocessing term variables,
which arguably would have been more elegant in the 
original proof as well.

\begin{lem} \label{lem:subst-congruence}
	Let $s\in\TT(\SigmaI)$, and let $\theta$, $\theta'$ be grounding
	substitutions such that $\eqR{x\theta}{x\theta'}$ for all variables~$x$
	and $\alpha\theta = \alpha\theta'$ for all type variables $\alpha$.
	Then $\eqR{s\theta}{s\theta'}$.
\end{lem}
\begin{proof}
  In this proof, we work directly on $\lambda$-terms. To prove the
  lemma, it suffices to prove it
  for any $\lambda$-term $s\in\TT^\lambda(\SigmaI)$.
  Here, for $t_1, t_2\in\TT^\lambda_\mathrm{ground}(\SigmaI)$, the notation
  $\eqR{t_1}{t_2}$ is to be read as $\eqR{\bnf{t_1}}\bnf{{t_2}}$
  because $\flooronly$ is defined only on $\beta$-normal $\lambda$-terms.

  Without loss of generality, we may assume that $s$ contains no type variables.
  If $s$ does contain type variables, we can instead use the term
  $s_0$ resulting from instantiating each type variable $\alpha$ in $s$ with $\alpha\theta$.
  If the result holds for the term $s_0$, which does not contain type variables, then
  $\eqR{s_0\theta}{s_0\theta'}$,
  and thus the result also holds for $s$ because $s\theta = s_0\theta$
  and $s\theta' = s_0\theta'$.

	\newcommand{\choice}{\oplus}%
  \medskip\noindent
	\textsc{Definition}\enskip  %
  We extend the syntax of $\lambda$-terms with a new
  polymorphic function symbol $\choice\oftype\forallty{\alpha}\alpha\fun\alpha\fun\alpha$.
  We will omit its type argument. It is equipped with two reduction rules:
  $\choice\>t\>s \rewrite t$ and $\choice\>t\>s \rewrite s$.
  A \emph{$\beta\choice$-reduction step}
  is either a rewrite step following one of these rules or a $\beta$-reduction step.

  \medskip\noindent
  The computability path order $\succ_\cst{CPO}$ \cite{blanqui-et-al-2015} guarantees that
  \begin{itemize}
    \item $\choice\>t\>s \succ_\cst{CPO} s$ by applying rule $@\rhd$;
    \item $\choice\>t\>s \succ_\cst{CPO} t$ by applying rule $@\rhd$ twice;
    \item $(\lambda\>t)\>s \succ_\cst{CPO} t\dbsubst{s}$ by applying rule $@\beta$.
 \end{itemize}
  Since this order is moreover monotone, it decreases with $\beta\choice$-reduction steps.
  The order is also well founded; thus, $\beta\choice$-reductions terminate.
  And since the $\beta\choice$-reduction steps describe a finitely branching term rewriting
  system, by K\H{o}nig's lemma \cite{koenigs-lemma-1927}, there exists a maximal number of
  $\beta\choice$-reduction steps
  from each $\lambda$-term.

	\newcommand{\choicesubst}{\sigma}%
	\newcommand{\livesize}{\mathscr{S}}%
  \medskip\noindent
	\textsc{Definition}\enskip %
  We introduce an auxiliary function $\livesize$
	that essentially measures the size of a $\lambda$-term but assigns a size of $1$ to
  ground $\lambda$-terms.
\[\livesize(s) =
\begin{cases}
1
  & \text{if $s$ is ground or if $s$ is a variable} \\
1 + \livesize(t)
  & \text{if $s$ is not ground and has the form $\lambda\>t$} \\
\livesize(t) + \livesize(u)
  & \text{if $s$ is not ground and has the form $t\>u$}
\end{cases}\]

    We prove $\eqR{s\theta}{s\theta'}$ by well-founded
    induction on $s$, $\theta$, and $\theta'$ using
    the left-to-right lexicographic order on the triple
	$\bigl(n_1(s), n_2(s), n_3(s)\bigr)\in\mathbb{N}^4$, where
	\begin{itemize}
		\item \label{mea:bi-red} $n_1(s)$ is the maximal number of
            $\beta\choice$-reduction steps starting from $s\choicesubst$, where
            $\choicesubst$ is the substitution mapping each variable $x$
            to  $\choice\>x\theta\>x\theta'$;
        \item $n_2(s)$ is the number of variables occurring more than once in $s$;
		\item $n_3(s) = \livesize(s)$.
	\end{itemize}

	\medskip\noindent
	\textsc{Case 1:}\enskip
	The $\lambda$-term $s$ is ground. Then the lemma is trivial.

	\medskip\noindent
	\textsc{Case 2:}\enskip
	The $\lambda$-term $s$ contains $k \geq 2$ variables. Then we can apply the induction
	hypothesis twice and use the transitivity of $\eqR{}{}$ as follows. Let $x$
	be one of the variables in $s$. Let $\rho = \{x \mapsto x\theta\}$ the
	substitution that maps $x$ to $x\theta$ and ignores all other variables. Let
	$\rho' = \theta'[x\mapsto x]$.

	We want to invoke the induction hypothesis on $s\rho$ and $s\rho'$.
  This is justified because $s\choicesubst$ $\choice$-reduces to $s\rho\choicesubst$ and to
  $s\rho'\choicesubst$, for $\choicesubst$ as given in the definition of $n_1$.
  These $\choice$-reductions have at least one step because $x$ occurs in $s$ and $k \geq 2$.
  Hence, $n_1(s)>n_1(s\rho)$ and $n_1(s)>n_1(s\rho')$.

	This application of the induction hypothesis gives us
	$\eqR{s\rho\theta}{s\rho\theta'}$ and $\eqR{s\rho'\theta}{s\rho'\theta'}$.
  Since $s\rho\theta = s\theta$ and $s\rho'\theta' = s\theta'$,
  this is equivalent to
  $\eqR{s\theta}{s\rho\theta'}$ and $\eqR{s\rho'\theta}{s\theta'}$.
  Since moreover $s\rho\theta' = s\rho'\theta$, we have $\eqR{s\theta}{s\theta'}$ by
  transitivity of $\eqR{}{}$.
  The following illustration visualizes the above argument:
	\[
	\begin{tikzpicture}
		\matrix (m) [matrix of math nodes,row sep=1.5em,column sep=0.1em,minimum width=1em]
		{
		   		   & s\rho &    &    &    & s\rho' &          \\
		   s\theta & \underset{\scriptscriptstyle\text{IH}}{\eqR{}{}} & s\rho\theta' & = & s\rho'\theta & \underset{\scriptscriptstyle\text{IH}}{\eqR{}{}} & s\theta' \\};
		\draw[-stealth] (m-1-2) edge node [left] {$\theta$\,} (m-2-1);
		\draw[-stealth] (m-1-2) edge node [right] {\,\,$\theta'$} (m-2-3);
		\draw[-stealth] (m-1-6) edge node [left] {$\theta$\,} (m-2-5);
		\draw[-stealth] (m-1-6) edge node [right] {\,\,$\theta'$} (m-2-7);
	\end{tikzpicture}
	\]

  \vskip-\baselineskip %
  \medskip\noindent
	\textsc{Case 3:}\enskip
	The $\lambda$-term $s$ contains a variable that occurs more than once. Then we rename
	variable occurrences apart by replacing each occurrence of each variable $x$ by
	a fresh variable $x_i$, for which we define $x_i\theta = x\theta$ and
	$x_i\theta' = x\theta'$. Let $s'$ be the resulting $\lambda$-term.
  Since $s\choicesubst
    = s'\choicesubst$ for $\choicesubst$ as given in the definition of $n_1$, we have $n_1(s)=n_1(s')$. All variables occur only once in $s'$.
    Hence, $n_2(s)>0=n_2(s')$. Therefore, we can invoke the induction
	hypothesis on $s'$ to obtain $\eqR{s'\theta}{s'\theta'}$. Since $s\theta =
	s'\theta$ and $s\theta' = s'\theta'$, it follows that
	$\eqR{s\theta}{s\theta'}$.

	\medskip\noindent
	\textsc{Case 4:}\enskip
	The $\lambda$-term $s$ contains only one variable $x$, which occurs exactly once.

	\medskip\noindent
	\textsc{Case 4.1:}\enskip
	The $\lambda$-term $s$ is of the form $\cst{f}\typeargs{\tuple{\tau}}\>\tuple{t}$
	for some symbol~$\cst{f}$, some types $\tuple{\tau}$, and some
    $\lambda$-terms~$\tuple{t}$. Then let $u$ be the $\lambda$-term in
    $\tuple{t}$ that contains $x$.
    We want to apply the induction hypothesis to $u$, which can be
	justified as follows.
  For $\choicesubst$ as given in the definition of $n_1$, consider  the longest sequence of
	$\beta\choice$-reductions from $u\choicesubst$. This sequence can be
	replicated inside $s\choicesubst=(\cst{f}\typeargs{\tuple{\tau}}\>\tuple{t})\choicesubst$.
	Therefore,
	the longest sequence of $\beta\choice$-reductions from $s\choicesubst$ is at
	least as long---i.e., $n_1(s)\geq n_1(u)$. Since both $s$ and $u$ have
	only one variable occurrence, we have $n_2(s) = 0 =
	n_2(u)$. But $n_3(s) > n_3(u)$ because $u$ is a nonground subterm of
	$s$.

	Applying the induction hypothesis gives us $\eqR{u\theta}{u\theta'}$. By
	definition of $\flooronly$, we have
  $\floor{(\cst{f}\typeargs{\tuple{\tau}}\>\tuple{t})\theta}
  = \cst{f}_m^{\smash{\tuple{\tau}}}\>\floor{\tuple{t}\theta}$
  and analogously for $\theta'$,
	where $m$ is the length of $\tuple{t}$.
	By congruence of $\eq$ in first-order logic, it follows that $\eqR{s\theta}{s\theta'}$.

	\medskip\noindent
	\textsc{Case 4.2:}\enskip
	The $\lambda$-term $s$ is of the form $x\>\tuple{t}$ for some $\lambda$-terms $\tuple{t}$. Then we
	observe that, by assumption, $\eqR{x\theta}{x\theta'}$.
  Since $x$ occurs only once,
	$\tuple{t}$ are ground.
  Then $\eqR{x\theta\>\tuple{t}}{x\theta'\>\tuple{t}}$ by applying
	Lemma~\ref{lem:argcong} repeatedly.
  Hence
  $s\theta = x\theta\>\tuple{t}$ and $s\theta = x\theta'\>\tuple{t}$,
  and it follows that $\eqR{s\theta}{s\theta'}$.

	\medskip\noindent
	\textsc{Case 4.3:}\enskip
	The $\lambda$-term $s$ is of the form $\lambda\>u$ for some $\lambda$-term $u$. Then we observe that to prove
	$\eqR{s\theta}{s\theta'}$, by Lemma~\ref{lem:RN-ext}, it suffices to show that
	$\eqR{s\theta\>\diff_{s\theta,s\theta'}}{s\theta' \>\diff_{s\theta,s\theta'}}$.
  Via $\beta$-conversion, this is equivalent to
	$\eqR{v\theta}{v\theta'}$, where $v = u\dbsubst{\diff_{s\theta,s\theta'}}$.
	To prove $\eqR{v\theta}{v\theta'}$, we apply the
	induction hypothesis on $v$.

	It remains to show that the induction hypothesis applies on $v$.
	For $\choicesubst$ as given in the definition of $n_1$, consider the longest sequence of $\beta\choice$-reductions from
	$v\choicesubst$. Since $\diff_{s\theta,s\theta'}$ is not a $\lambda$-abstraction, substituting it for $\DB{0}$
	will not cause additional $\beta\choice$-reductions. Hence, the same
	sequence of $\beta\choice$-reductions can be applied inside $s\choicesubst =
    (\lambda\>u)\choicesubst$, proving that $n_1(s) \geq n_1(v)$.
    Since both $s$ and $v$ have	only one variable occurrence,
	$n_2(s) = 0 = n_2(v)$. But $n_3(s) = \livesize(s) = 1 + \livesize(u)$
	because $s$ is nonground. Moreover,
	$\livesize(u)=\livesize(v)=n_3(v)$. Hence, $n_3(s)
	> n_3(v)$, which justifies the application of the induction hypothesis.

	\medskip\noindent
	\textsc{Case 4.4:}\enskip
	The $\lambda$-term $s$ is of the form $(\lambda\>u)\>t_0\>\tuple{t}$ for some $\lambda$-terms $u$,
	$t_0$, and $\tuple{t}$.
	We apply the induction hypothesis on $s' = u\dbsubst{t_0}\>\tuple{t}$,
	justified as follows. 
  For $\choicesubst$ as given in the definition of $n_1$, consider  the longest sequence of $\beta\choice$-reductions from
	$s'\choicesubst$. Prepending the reduction $s\choicesubst \rewrite_\beta
	s'\choicesubst$ to it gives us a longer sequence from $s\choicesubst$. Hence,
	$n_1(s) > n_1(s')$.
	The induction hypothesis gives us $\eqR{s'\theta}{s'\theta'}$. Since
	$\eqR{}{}$ is invariant under $\beta$-reductions, it follows that
	$\eqR{s\theta}{s\theta'}$.
	\qedhere
\end{proof}

Using the term rewrite system $R$,
we define a higher-order interpretation 
\begin{full}$\mathscr{I}^\levelIPG = (\mathscr{U}^\levelIPG, \mathscr{J}^\levelIPG_\text{ty}, \mathscr{J}^\levelIPG, \mathscr{L}^\levelIPG)$.\end{full}%
\begin{slim}$\mathscr{I}^\levelIPG = (\mathscr{U}^\levelIPG, \mathscr{J}^\levelIPG_\text{ty}, \mathscr{J}^\levelIPG, \mathscr{L}^\levelIPG)$.\end{slim}
The construction proceeds as in 
the completeness proof of the original $\lambda$-superposition calculus~\cite{bentkamp-et-al-2023-hosup-journal}.
Let $(\mathscr{U}, \mathscr{J}) = R$; i.e., $\mathscr{U}_\tau$ is the universe for the first-order type $\tau$, and $\mathscr{J}$ is the interpretation function.
Since the higher-order and first-order type signatures are identical, we can identify ground higher-order and first-order types.
We will define a domain $\mathscr{D}_\tau$ for each ground type $\tau$ and then let $\mathscr{U}^{\levelIPG}$ be the set of all these domains $\mathscr{D}_\tau$.
We cannot identify the domains $\mathscr{D}_\tau$ with the first-order domains $\mathscr{U}_\tau$
because domains $\mathscr{D}_\tau$ for functional types $\tau$ must contain functions.
Instead, we will define suitable domains $\mathscr{D}_\tau$ and a bijection $\mathscr{E}_\tau$ between $\mathscr{U}_\tau$ and $\mathscr{D}_\tau$ for each ground type~$\tau$.

We define $\mathscr{E}_\tau$ and $\mathscr{D}_\tau$ in mutual recursion.
To ensure well definedness, we must show that $\mathscr{E}_\tau$ is bijective.
We start with nonfunctional types $\tau$:
Let $\mathscr{D}_\tau = \mathscr{U}_\tau$, and let $\mathscr{E}_\tau : \mathscr{U}_\tau \to \mathscr{D}_\tau$ be the identity.
Clearly, the identity is bijective.
For functional types, we define
\begin{align*}
  &\mathscr{D}_{\tau \to \upsilon} = \{ \varphi : \mathscr{D}_\tau \to \mathscr{D}_\upsilon \mid \exists s : \tau \to \upsilon.\ \forall u : \tau.\ \varphi\left(\mathscr{E}_\tau(\interpret{\mapF{u}}{R}{})\right) = \mathscr{E}_\upsilon\left(\interpret{\mapF{s\> u}}{R}{}\right) \} \\[1.5\jot]
  &\mathscr{E}_{\tau \to \upsilon} : \mathscr{U}_{\tau \to \upsilon} \to \mathscr{D}_{\tau \to \upsilon} \\
  &\mathscr{E}_{\tau \to \upsilon}(\interpret{\mapF{s}}{R}{})\left(\mathscr{E}_\tau(\interpret{\mapF{u}}{R}{})\right) = \mathscr{E}_\upsilon(\interpret{\mapF{s\> u}}{R}{})
\end{align*}

To verify that this equation is a valid definition of $\mathscr{E}_{\tau \to \upsilon}$, we must show that
\begin{itemize}
\item every element of $\mathscr{U}_{\tau \to \upsilon}$ is of the form $\interpret{\mapF{s}}{R}{}$ for some $s \in \TT_\mathrm{ground}(\SigmaI)$;
\item every element of $\mathscr{D}_\tau$ is of the form $\mathscr{E}_\tau\left(\interpret{\mapF{u}}{R}{}\right)$ for some $u \in \TT_\mathrm{ground}(\SigmaI)$;
\item the definition does not depend on the choice of such $s$ and $u$; and
\item $\mathscr{E}_{\tau \to \upsilon}(\interpret{\mapF{s}}{R}{}) \in \mathscr{D}_{\tau \to \upsilon}$ for all $s \in \TT_\mathrm{ground}(\SigmaI)$.
\end{itemize}

The first claim holds because $R$ is term-generated and $\mapFonly$ is a bijection.
The second claim holds because $R$ is term-generated and $\mapFonly$ and $\mathscr{E}_\tau$ are bijections.
To prove the third claim, we assume that there are other terms $t\in \TT_\mathrm{ground}(\SigmaI)$ and $v \in \TT_\mathrm{ground}(\SigmaI)$ such that $\interpret{\mapF{s}}{R}{} = \interpret{\mapF{t}}{R}{}$ and $\mathscr{E}_\tau\left(\interpret{\mapF{u}}{R}{}\right) = \mathscr{E}_\tau\left(\interpret{\mapF{v}}{R}{}\right)$.
Since $\mathscr{E}_\tau$ is bijective, we have $\interpret{\mapF{u}}{R}{} = \interpret{\mapF{v}}{R}{}$---i.e., $u \sim v$.
The terms $s, t, u, v$ are in $\TT_\mathrm{ground}(\SigmaI)$, allowing us to apply Lemma \ref{lem:subst-congruence}
to the term $x\> y$ and the substitutions $\{x \mapsto s{,}\; y \mapsto u\}$ and $\{x \mapsto t{,}\; y \mapsto v\}$.
Thus, we obtain $s\> u \sim t\> v$---i.e., $\interpret{\mapF{s\> u}}{R}{} = \interpret{\mapF{t\> v}}{R}{}$---indicating that the definition of $\mathscr{E}_{\tau \to \upsilon}$ above does not depend on the choice of $s$ and $u$.
The fourth claim is obvious from the definition of $\mathscr{D}_{\tau \to \upsilon}$ and the third claim.

It remains to show that $\mathscr{E}_{\tau \to \upsilon}$ is bijective.
For injectivity, we fix two terms $s,t \in \TT_\mathrm{ground}(\SigmaI)$ such that for all $u \in \TT_\mathrm{ground}(\SigmaI)$, we have $\interpret{\mapF{s\> u}}{R}{} = \interpret{\mapF{t\> u}}{R}{}$.
By Lemma~\ref{lem:RN-ext}, it follows that $\interpret{\mapF{s}}{R}{} = \interpret{\mapF{t}}{R}{}$, which shows that $\mathscr{E}_{\tau \to \upsilon}$ is injective.
For surjectivity, we fix an element $\varphi \in \DD_{\tau \to \upsilon}$.
By definition of $\DD_{\tau \to \upsilon}$, there exists a term $s$ such that $\varphi\left(\mathscr{E}_\tau(\interpret{\mapF{u}}{R}{})\right) = \mathscr{E}_\upsilon\left(\interpret{\mapF{s\> u}}{R}{}\right)$ for all $u$.
Hence, $\EE_{\tau \to \upsilon}(\interpret{\mapF{s}}{R}{}) = \varphi$, proving surjectivity and therefore bijectivity of $\EE_{\tau \to \upsilon}$.
Below, we will usually write $\EE$ instead of $\EE_\tau$ since the type $\tau$ is determined by $\EE_\tau$'s first argument.

We define the higher-order universe as $\UU^{\levelIPG} = \{ \DD_\tau \mid \tau \text{ ground}\}$.
In particular, by Lemma~\ref{lem:PF:R-olammodels}, this implies that $\DD_\omicron = \{0,1\} \in \UU^{\levelIPG}$ as needed, where $0$ is identified with $[\ifalse]$ and $1$ with $[\itrue]$.
Moreover, we define $\IIty^{\levelIPG}(\kappa)(\DD_{\tuple{\tau}}) = \DD_{\kappa(\tuple{\tau})}$ for all $\kappa \in \Sigma_\ty$,
completing the type interpretation of $\IIIty^{\levelIPG} = (\UU^{\levelIPG}, \IIty^{\levelIPG})$ and ensuring that $\IIty^{\levelIPG}(\omicron) = \DD_\omicron = \{0,1\}$.

We define the interpretation function $\II^{\levelIPG}$ for symbols $\cst{f} : \forallty{\tuple{\alpha}_m} \tau$
by $\II^{\levelIPG}(\cst{f}, \DD_{\tuple{\upsilon}_{m}}) = \EE(\interpret{\mapF{\cst{f}\typeargs{\tuple{\upsilon}_{m}}}}{R}{})$.

We must show that this definition indeed fulfills the requirements of an interpretation function.
By definition, we have
\ref{item:interpretation:true} $\II^{\levelIPG}(\itrue) = \EE(\interpretR{\itrue}) = \interpretR{\itrue} = 1$
and \ref{item:interpretation:false} $\II^{\levelIPG}(\ifalse) = \EE(\interpretR{\ifalse}) = \interpretR{\ifalse} = 0$.

Let $a,b \in \{0,1\}$, $u_0 = \ifalse$, and $u_1 = \itrue$. Then, by Lemma~\ref{lem:PF:R-olammodels},

\[
\begin{array}{crl}
\text{\ref{item:interpretation:and}} &
\II^{\levelIPG}(\iand)(a,b) & = \EE(\interpretR{\mapF{\iand}})(\interpretR{\mapF{u_a}}, \interpretR{\mapF{u_b}}) \\
& &= \EE(\interpretR{\mapF{u_a \iand u_b}}) = \min\{a,b\} \\[\jot]
\text{\ref{item:interpretation:or}} &
\II^{\levelIPG}(\ior)(a,b) & = \EE(\interpretR{\mapF{u_a \ior u_b}}) = \max\{a,b\} \\[\jot]
\text{\ref{item:interpretation:not}} &
\II^{\levelIPG}(\inot)(a) & = \EE(\interpretR{\mapF{\inot}})(\interpretR{\mapF{u_a}}) \\
& & = \EE(\interpretR{\mapF{\inot u_a}}) = \interpretR{\mapF{\inot u_a}} = 1 - a \\[\jot]
\text{\ref{item:interpretation:implies}} &
\II^{\levelIPG}(\iimplies)(a,b) & = \EE(\interpretR{\mapF{u_a \iimplies u_b}}) = \max\{1-a,b\}
\end{array}
\]

\ref{item:interpretation:eq} 
Let $\DD_\tau \in \UU^{\levelIPG}$ and $a',b' \in \DD_\tau$.
Since $\EE$ is bijective and $R$ is term-generated, there exist ground terms $u$ and $v$ such that $\EE(\interpretRmapF{u}) = a'$ and $\EE(\interpretRmapF{v}) = b'$.
Then
\[
\II^{\levelIPG}(\ieq, \DD_\tau)(a',b') = \EE(\interpretRmapF{{\ieq}\typeargs{\tau}})(\EE(\interpretRmapF{u}), \EE(\interpretRmapF{v})) = \EE(\interpretRmapF{u\mathrel{{\ieq}\typeargs{\tau}} v})
\]
which is $1$ if $a' = b'$ and $0$ otherwise by Lemma~\ref{lem:PF:R-olammodels}.
\ref{item:interpretation:neq} Similarly $\II^{\levelIPG}(\ineq, \DD_\tau)(a',b') = 0$ if $a' = b'$ and $1$ otherwise.
This concludes the proof that $\II^{\levelIPG}$ is an interpretation function.

Finally, we need to define the designation function $\LL^{\levelIPG}$, which takes a valuation $\xi$ and a $\lambda$-expression as arguments.
Given a valuation $\xi$, we choose a grounding substitution $\theta$ such that $\DD_{\alpha\theta} = \xity(\alpha)$ and $\EE(\interpretRmapF{x\theta}) = \xite(x)$ for all type variables $\alpha$ and all variables $x$.
Such a substitution can be constructed as follows: We can fulfill the first equation in a unique way because there is a one-to-one correspondence between ground types and domains.
Since $\EE^{-1}(\xite(x))$ is an element of a first-order universe and $R$ is term-generated, there exists a ground term $s$ such that $\interpret{s}{R}{\xi} = \EE^{-1}(\xite(x))$.
Choosing one such $s$ and defining $x\theta = \mapFonly^{-1}(s)$ gives us a grounding substitution $\theta$ with the desired property.

Let $\LL^{\levelIPG}(\xi, \lambda \> t) = \EE(\interpretRmapF{(\lambda \> t)\theta})$.
We need to show that our definition does not depend on the choice of $\theta$.
We assume that there exists another substitution $\theta'$ with the properties $\DD_{\alpha\theta'} = \xity(\alpha)$ for all $\alpha$ and $\EE(\interpretRmapF{x\theta'}) = \xite(x)$ for all $x$.
Then we have $\alpha\theta = \alpha\theta'$ for all $\alpha$ due to the one-to-one correspondence between domains and ground types.
We have $\interpretRmapF{x\theta} = \interpretRmapF{x\theta'}$ for all $x$ because $\EE$ is injective.
By Lemma \ref{lem:subst-congruence} it follows that $\interpretRmapF{(\lambda \> t)\theta} = \interpretRmapF{(\lambda \> t)\theta'}$, which proves that $\LL^{\levelIPG}$ is well defined.
This concludes the definition of the interpretation $\III^{\levelIPG} = (\UU^{\levelIPG}, \IIty^{\levelIPG}, \II^{\levelIPG}, \LL^{\levelIPG})$.
It remains to show that $\III^{\levelIPG}$ is proper.

The higher-order interpretation $\III^{\levelIPG}$ relates to the first-order
interpretation $\RfN$ as follows:

\begin{lem}
  \label{lem:IPG:ho-fo-correspondence}
  Given a ground $\lambda$-term $t\in\TT^\lambda_\mathrm{ground}(\SigmaI)$, we have
  \[
  \interpret{t}{\III^{\levelIPG}}{} = \EE(\interpretRmapF{\bnf{t}})
  \]
  \end{lem}

\begin{proof}
  The proof is adapted from the proof of Lemma 40 in Bentkamp et al.~\cite{bentkamp-et-al-2021-lamsup-journal}.
	We proceed by induction on $t$.
	If $t$ is of the form $\cst{f}\typeargs{\tuple{\tau}}$, then
	\begin{align*}
	\interpret{t}{\III^{\levelIPG}}{} &= \II^{\levelIPG}(\cst{f},\DD_{\tuple{\tau}})\\
	&=\EE(\interpretfo{\floor{\cst{f}\typeargs{\tuple{\tau}} }}{})
	=\EE(\interpretfo{\floor{\bnf{t}}}{})
	\end{align*}
  If $t$ is an application $t = t_1\>t_2$, where $t_1$ is of type $\tau\fun\upsilon$, then
	\begin{align*}
	\interpret{t_1\>t_2}{\III^{\levelIPG}}{}
	&= \interpret{t_1}{\III^{\levelIPG}}{} (\interpret{t_2}{\III^{\levelIPG}}{}) \\
	&\overset{\!\scriptscriptstyle\text{IH}\!}{=}
	\EE_{\tau\fun\upsilon}(\interpretfo{\floor{\bnf{t_1}}}{}) (\EE_\tau(\interpretfo{\floor{\bnf{t_2}}}{}))\\
	&\overset{\kern-10mm\text{Def }\EE\kern-10mm}{=}
	\enskip\EE_\upsilon(\interpretfo{\floor{\bnf{(t_1\>t_2)}}}{})
	\end{align*}
	If $t$ is a $\lambda$-expression, then
	\begin{align*}
	\interpret{\lambda\>u}{\III^{\levelIPG}}{\xi}
	&= \LL^{\levelIPG}(\xi, \lambda\>u) \\
	& = \EE(\interpret{\floor{\bnf{(\lambda\>u)\theta}}}{R}{}) \\
	& = \EE(\interpret{\floor{\bnf{(\lambda\>u)}}}{R}{})
  \end{align*}
  where $\theta$
  is a substitution as required by the definition of $\LL^{\levelIPG}$.
\end{proof}

We need to show that the interpretation $\III^{\levelIPG}$ is
proper. In the proof, we will need the following lemma, which is
very similar to the substitution lemma (Lemma~\ref{lem:subst-lemma}),
but we
must prove it here for our particular interpretation $\III^{\levelIPG}$ because we have not
shown that $\III^{\levelIPG}$ is proper yet.

\begin{lem}
  \label{lem:IPG:interpretation-substitution}
  Let $\rho$ be a grounding substitution, $t$ be a $\lambda$-term, and $\xi$ be a valuation.
  Moreover, we define a valuation $\xi'$ by $\vphantom{(_{(_(}}\xity'(\alpha) = \interpret{\alpha\rho}{\IIIty^\levelIPG}{\xity}$ for all type variables $\alpha$ and $\xite'(x) = \interpret{x\rho}{\III^\levelIPG}{\xi}$ for all term variables $x$.
  We then have
  \[
  \interpret{t\rho}{\III^{\levelIPG}}{\xi} = \interpret{t}{\III^{\levelIPG}}{\xi'}
  \]
  \end{lem}

\begin{proof}
  The proof is adapted from the proof of Lemma 41 in Bentkamp et al.~\cite{bentkamp-et-al-2021-lamsup-journal}.
  We proceed by induction on the structure of $\tau$ and $t$.
  The proof is identical to that of Lemma~\ref{lem:subst-lemma},
  except for the last case, which uses properness of the interpretation, a
  property we cannot assume here.
  However, here, we have the assumption that $\rho$ is a grounding substitution.
  Therefore, if $t$ is a $\lambda$-expression, we argue as follows:
\begin{align*}
\interpret{(\lambda\>u)\rho}{\III^{\levelIPG}}{\xi}
&=\interpret{\lambda\>u\rho}{\III^{\levelIPG}}{\xi}&&\text{}\\
&= \LL^{\levelIPG}(\xi,\lambda\>u\rho)
&&\text{ by the definition of the term denotation}\\
&= \EE(\interpret{\floor{\bnf{(\lambda\>u)\rho\theta}}}{R}{})
&&\text{ for some $\theta$ by the definition of $\LL^{\levelIPG}$}\\
&= \EE(\interpret{\floor{\bnf{(\lambda\>u)\rho}}}{R}{})
&&\text{ because $(\lambda\>u)\rho$ is ground}\\
&\overset{\smash{*}}{=} \LL^{\levelIPG}(\xi',\lambda\>u)
&&\text{ by the definition of $\LL^{\levelIPG}$ and Lemma~\ref{lem:IPG:ho-fo-correspondence}}\\
&= \interpret{\lambda\>u}{\III^{\levelIPG}}{\xi'}
&&\text{ by the definition of the term denotation}
\end{align*}
The step labeled with $*$ is justified as follows:
We have
$\LL^{\levelIPG}(\xi',\lambda\>u) = \EE(\interpretfo{\floor{\bnf{(\lambda\>u)\theta'}}}{})$
by the definition of $\LL^{\levelIPG}$,
if $\theta'$ is a substitution such that
$\smash{\dho_{\alpha\theta'}}=\xity'(\alpha)$ for all $\alpha$ and
$\EE(\interpretfo{\floor{\bnf{x\theta'}}}{}) = \xite'(x)$
for all $x$.
By the definition of $\xi'$ and by Lemma~\ref{lem:IPG:ho-fo-correspondence},
$\rho$ is such a substitution.
Hence,
$\LL^{\levelIPG}(\xi',\lambda\>u) = \EE(\interpretfo{\floor{\bnf{(\lambda\>u)\rho}}}{})$.
\qedhere
\end{proof}

\begin{lem}
\label{lem:IPG:interpretation-proper}
The interpretation $\III^{\levelIPG}$ is proper.
\end{lem}
\begin{proof}
We need to show that $\interpret{\lambda \> t}{\III^{\levelIPG}}{(\xity,\xite)}(a) = \interpret{t\dbsubst{x}}{\III^\levelIPG}{(\xity,\xite[x\mapsto a])}$, where $x$ is a fresh variable.
\[
\begin{array}{rll}
\interpret{\lambda \> t}{\III^{\levelIPG}}{(\xity,\xite)}(a) & = \LL^{\levelIPG}((\xity,\xite), \lambda\> t)(a)
& \text{by the definition of term denotation}  \\[\jot]
& = \EE(\interpretRmapF{\bnf{(\lambda\> t)\theta}})(a)
& \text{by the definition of $\LL^{\levelIPG}$ for some $\theta$} \\
&& \text{such that $\EE(\interpretRmapF{z\theta}) = \xite(z)$ for} \\
&& \text{all $z$ and $\DD_{\alpha\theta} = \xity(\alpha)$ for all $\alpha$} \\[\jot]
& = \EE(\interpretRmapF{\bnf{((\lambda\> t)\theta\> s)}})
& \text{by the definition of $\EE$} \\
&& \text{where $\EE(\interpretRmapF{s}) = a$} \\[\jot]
& = \EE(\interpretRmapF{\bnf{t\dbsubst{x}(\theta[x \mapsto s])}})
& \text{by $\beta$-reduction} \\
&& \text{where $x$ is fresh} \\[\jot]
& = \interpret{t\dbsubst{x}(\theta[x \mapsto s])}{\III^{\levelIPG}}{}
& \text{by Lemma \ref{lem:IPG:ho-fo-correspondence}} \\[\jot]
& = \interpret{t\dbsubst{x}}{\III^{\levelIPG}}{(\xity,\xite[x\mapsto a])}
& \text{by Lemma \ref{lem:IPG:interpretation-substitution}}
\end{array}
\]
\end{proof}

\begin{lem}\label{lem:IPG:term-generated}
$\III^{\levelIPG}$ is term-generated; i.e.,
for all $\DD \in \UU^{\levelIPG}$ and all $a \in \DD$,
there exists a ground type $\tau$ such that
$\interpret{\tau}{\IIIty^{\levelIPG}}{} = \DD$ and
a ground term $t$ such that
$\interpret{t}{\III^{\levelIPG}}{} = a$.
\end{lem}
\begin{proof}
In the construction above, it is clear that there is a one-to-one
correspondence between ground types and domains, which yields a suitable
ground type $\tau$. 

Since $R$ is term-generated, there must be a ground term $s\in\termsPF$ such
that $\interpret{s}{R}{} = \EE^{-1}(a)$.
Let $t = \mapFonly^{-1}(s)$.
Then, by Lemma~\ref{lem:IPG:ho-fo-correspondence},
$\interpret{t}{\III^{\levelIPG}}{} = \EE(\interpret{s}{R}{}) = a$.
\end{proof}

\begin{lem}\label{thm:IPG:model-mapF-iff}
Given $C\slimfull{}{\closure\theta} \in \clausesIPG$,
we have $\III^\levelIPG \models C\slimfull{}{\closure\theta}$ if and only if
$R \models \mapF{C\slimfull{}{\closure\theta}}$.
\end{lem}
\begin{proof}
By Lemma~\ref{lem:IPG:ho-fo-correspondence},
we have 
\[
\interpret{t}{\III^\levelIPG}{} = \EE(\interpretRmapF{\bnf{t}})
\]
for any $t \in \TT_\mathrm{ground}(\SigmaI)$.
Since $\EE$ is a bijection, 
it follows that a ground literal 
$s\slimfull{}{\theta} \doteq t\slimfull{}{\theta}$ in a clause $C\slimfull{}{\closure\theta} \in \clausesIPG$
is true in $\III^\levelIPG$ if and only if 
$\mapF{s\slimfull{}{\theta} \doteq t\slimfull{}{\theta}}$ is true in $R$.
So any closure $C\slimfull{}{\closure\theta} \in \clausesIPG$ is 
true in $\III^\levelIPG$ if and only if
$\mapF{C\slimfull{}{\closure\theta}}$ is true in $R$.
\end{proof}

\begin{thm} \label{thm:IPG:model-for-varirred-closures}
  Let $N\subseteq \clausesIPG$ be saturated up to redundancy \wrt\ $\IPGRedI$,
  and 
  \begin{full}$N$ does not contain a closure of the form $\bot\closure\theta$ for any $\theta$.\end{full}%
  \begin{slim}$\bot \not\in N$.\end{slim}
  Then $\III^\levelIPG \models \slimfull{N}{\irred_{R}({N})}$\slimfull{.}{, where $R = \Rbasic_{\mapF{N}}$.}
\end{thm}
\begin{proof}
By Lemma~\ref{thm:IPG:model-mapF-iff}, it suffices to show that $R$ is a model of 
$\slimfull{N}{\irred_{R}(\mapF{N})}$.
We apply Theorem~\ref{thm:PF:model-for-varirred-closures}.
Lemma~\ref{lem:IPG:saturated} shows the condition of saturation up to redundancy.
\end{proof}

\begin{full}
\begin{lem}\label{lem:IPG:irred-entails-all}
  Let $R$ be a confluent term rewrite system on $\termsPF$ oriented by $\succ_\mapFonly$
  whose only Boolean normal forms are $\itrue$ and $\ifalse$.
  Let $N\subseteq\clausesIPG$
  such that for every $C\closure\theta \in N$
  and every grounding substitution $\rho$ that
  coincides with $\theta$ on all variables not occurring in $C$, 
  we have $C\closure \rho\in N$.
  Then $R \cup \mapF{\irred_R(N)} \modelsolam \mapF{N}$.
\end{lem}
\begin{proof}
We apply Lemma~\ref{lem:PF:irred-entails-all}.
The required condition on $\mapF{N}$
can be derived from this lemma's condition on $N$
and the fact that $\mapFonly$ is a bijection (Lemma~\ref{lem:IPG:F-bijection}).
\end{proof}
\end{full}

\begin{full}\subsubsection{Partly Substituted Ground Higher-Order Level}\mbox{}\end{full}%
\begin{slim}\subsubsection{Ground Higher-Order Level}\mbox{}\end{slim}%
In this subsubsection,
let $\succ$ be an admissible term order for $\PGInf$ (Definition~\ref{def:PG:admissible-term-order}),
and let $\mathit{\slimfull{}{p}gsel}$ be a selection function on $\clausesPG$\slimfull{}{ (Definition~\ref{def:PF:selection-function})}.

It is inconvenient to construct a model of $N_0$ for the $\levelPG$ level
because $\mapIonly$ converts parameters into subscripts.
For example, in the model constructed in the previous subsubsection,
it can happen that $\cst{a} \ceq \cst{b}$ holds,
but $\cst{f}_\cst{a} \ceq \cst{f}_\cst{b}$ does not hold,
where $\cst{a}$ and $\cst{b}$ are constants
and $\cst{f}_\cst{a}$ and $\cst{f}_\cst{b}$ are constants originating
from a constant $\cst{f}$ with a parameter.
For this reason, our completeness result for the $\levelPG$ level
only constructs a model of $\slimfull{\mapI{N}}{\irred_R(\mapI{N})} \subseteq \clausesIPG$ instead of $\slimfull{N}{\irred_R(N)} \subseteq \clausesPG$.
We will overcome this flaw when we lift the result to the $\levelH$ level
where the initial clause set can be assumed not to contain any constants with parameters.

\begin{thm} \label{thm:PG:model-for-varirred-closures}
  Let $N\subseteq \clausesPG$ be saturated up to redundancy \wrt\ $\PGRedI$,
  and 
  \begin{full}$N$ does not contain a closure of the form $\bot\closure\theta$ for any $\theta$.\end{full}%
  \begin{slim}$\bot \not\in N$.\end{slim}
  Then $\III^\levelIPG \models \slimfull{\mapI{N}}{\mapI{\irred_{R}(N)}}$\slimfull{}{, where $R = \Rbasic_{\mapF{\mapI{N}}}$}.
\end{thm}
\begin{proof}
This follows from
Theorem~\ref{thm:IPG:model-for-varirred-closures}
and Lemma~\ref{lem:PG:saturated}.
\end{proof}

\begin{full}
\begin{lem}\label{lem:PG:irred-entails-all}
  Let $R$ be a confluent term rewrite system on $\termsPF$ oriented by $\succ_{\mapIonly\mapFonly}$
  whose only Boolean normal forms are $\itrue$ and $\ifalse$.
  Let $N\subseteq\clausesPG$ be a clause set without parameters
  such that for every $C\closure\theta \in N$
  and every grounding substitution $\rho$ that
  coincides with $\theta$ on all variables not occurring in $C$, 
  we have $C\closure \rho\in N$.
  Then $R \cup \mapF{\mapI{\irred_R(N)}} \modelsolam \mapF{\mapI{N}}$.
\end{lem}
\begin{proof}
We apply Lemma~\ref{lem:IPG:irred-entails-all}.
The required condition on $\mapI{N}$
can be derived from this lemma's condition on $N$ as follows.
We must show that
for every $C\closure\theta \in \mapI{N}$ and every grounding substitution $\rho$ that
coincides with $\theta$ on all variables not occurring in $C$,
we have $C\closure \rho\in \mapI{N}$.
The closure $C\closure\theta \in \mapI{N}$ must be of the form $\mapI{C'\closure\theta'}$
with $C'\closure\theta' \in N$.
Define $\rho'$ as $x\rho' =\mapIonly^{-1}(x\rho)$ for all $x$.
By this lemma's condition on $N$,
it follows that $C'\closure \rho'\in N$.
and thus $C\closure \rho = \mapI{C'\closure \rho'} \in \mapI{N}$.
Here, it is crucial that $N$ does not contain parameters because only this guarantees
that $C = \mapIonly_{\theta'}(C') = \mapIonly_{\rho'}(C')$.
\end{proof}
\end{full}

\subsubsection{Full Higher-Order Level}

In this subsubsection, let $\succ$ be an admissible term order
(Definition~\ref{def:admissible-term-order}),
\begin{slim}which by Lemma~\ref{lem:G:admissible-term-order} is also an admissible term order for $\PGInf$,\end{slim}%
\begin{full}extended to be an admissible term order for $\PGInf$ as in Section~\ref{ssec:H:redundancy},\end{full}
and let
$\mathit{hsel}$ be a selection function \slimfull{on $\clausesH$}{} (Definition~\ref{def:lit-sel}).

\begin{defi}
  A \emph{derivation} is a finite or infinite sequence of sets $(N_i)_{i\geq 0}$
  such that $N_i \setminus N_{i+1} \subseteq \HRedC(N_{i+1})$ for all $i$.
  A derivation is called \emph{fair} if 
  all $\HInf$-inferences from clauses in $\bigcup_i\bigcap_{j\geq i} N_j$
  are contained in $\bigcup_i \HRedI(N_i)$.
\end{defi}

\begin{lem}\label{lem:H:red-criteria-properties}
  The redundancy criteria $\HRedC$ and $\HRedI$ fulfill the following properties,
  as stated by Waldmann et al.~\cite{waldmann-et-al-saturation-journal}:
  \begin{enumerate}
    \item[(R2)] if $N \subseteq N'$, then $\HRedC(N) \subseteq \HRedC(N')$ and $\HRedI(N) \subseteq \HRedI(N')$;
    \item[(R3)] if $N' \subseteq \HRedC(N)$, then $\HRedC(N)\subseteq \HRedC(N \setminus N')$
     and $\HRedI(N)\subseteq \HRedI(N \setminus N')$;
    \item[(R4)] if $\iota \in \HInf$ and $\concl(\iota) \in N$, then $\iota \in HRedI(N)$.
  \end{enumerate}
\end{lem}
\begin{proof}

  (R2): This is obvious by definition of clause and inference redundancy.
  
  (R3) for clauses:

  Define $\blacktriangleright$ as a relation on sets of 
  \begin{full}closures\end{full}
  \begin{slim}pairs of a clause $C \in \clausesH$ and a grounding substitution $\theta$, written\end{slim}
  $C\closure\theta$,
  \begin{full}where $C \in \clausesH$\end{full}
  as
  \[C \closure \theta \blacktriangleright D \closure \rho \text{\quad iff\quad 
  $C\theta \succ D\rho$ or ($C\theta = D\rho$ and $C \sqsupset D$)}\]

  Clearly, for all $C \in \clausesH$ and all $N \subseteq \clausesH$,
  we have $C\in\HRedC(N)$ 
  if and only if for
  \begin{full}
  all  confluent term rewrite systems $R$ on $\termsF$ oriented by $\succ$
  whose only Boolean normal forms are $\itrue$ and $\ifalse$
  and
  \end{full}
  all \slimfull{grounding substitutions $\theta$}{$C \closure \theta \in \irred_R(\mapG{C})$}, we have
  \[\slimfull{}{R \cup} \{\mapF{E\zeta} \mid \slimfull{E \in N,\ \zeta\text{ grounding, and }}{E\closure\zeta \in \irred_R(\mapG{N}) \text{ and }} E\closure\zeta \blacktriangleleft C\closure\theta\}\modelsolam \mapF{C\theta}\]

  Now we are ready to prove (R3).
  Let $C\in\HRedC(N)$. We must show that $C \in \HRedC(N\setminus N')$.
  \begin{full}
  Let $R$ be a confluent term rewrite system on $\termsPF$ oriented by $\succ$
  whose only Boolean normal forms are $\itrue$ and $\ifalse$.
  Let $C\closure\theta \in \irred_R(\fipg{C})$.
  \end{full}%
  \begin{slim}Let $\theta$ be a grounding substitution.\end{slim}
  We must show that
  \[\slimfull{}{R \cup} \{\mapF{E\zeta} \mid \slimfull{E \in N\setminus N',\ \zeta\text{ grounding, and }}{E\closure\zeta \in \irred_R(\mapG{N\setminus N'}) \text{ and }} E\closure\zeta \blacktriangleleft C\closure\theta\}\modelsolam \mapF{C\theta}\]
  Since $C\in\HRedC(N)$, we know that
  \[\slimfull{}{R \cup} \{\mapF{E\zeta} \mid \slimfull{E \in N,\ \zeta\text{ grounding, and }}{E\closure\zeta \in \irred_R(\mapG{N}) \text{ and }} E\closure\zeta \blacktriangleleft C\closure\theta\}\modelsolam \mapF{C\theta}\]
  So it suffices to show that 
  \begin{align*}
  & \slimfull{}{R \cup} \{\mapF{E\zeta} \mid \slimfull{E \in N\setminus N',\ \zeta\text{ grounding, and }}{E\closure\zeta \in \irred_R(\mapG{N\setminus N'}) \text{ and }} E\closure\zeta \blacktriangleleft C\closure\theta\}\\[-\jot]
  {\modelsolam}\; & \slimfull{}{R \cup} \{\mapF{E\zeta} \mid \slimfull{E \in N,\ \zeta\text{ grounding, and }}{E\closure\zeta \in \irred_R(\mapG{N}) \text{ and }} E\closure\zeta \blacktriangleleft C\closure\theta\}
  \end{align*}
  Let \slimfull{$E_0\in N$ and $\zeta_0$ grounding}{$E_0\closure\zeta_0 \in \irred_R(\mapG{N})$}
  with $E_0\closure\zeta_0 \blacktriangleleft C\closure\theta$.
  We will show by well-founded induction on $E_0\closure\zeta_0$ \wrt\ $\blacktriangleleft$  that
  \[\slimfull{}{R \cup} \{\mapF{E\zeta} \mid \slimfull{E \in N\setminus N',\ \zeta\text{ grounding, and }}{E\closure\zeta \in \irred_R(\mapG{N\setminus N'}) \text{ and }} E\closure\zeta \blacktriangleleft C\closure\theta\}\modelsolam
  \mapF{E_0\zeta_0} \tag{$*$}\]
  Our induction hypothesis states:
  \begin{align*}
    & \slimfull{}{R \cup} \{\mapF{E\zeta} \mid \slimfull{E \in N\setminus N',\ \zeta\text{ grounding, and }}{E\closure\zeta \in \irred_R(\mapG{N\setminus N'}) \text{ and }} E\closure\zeta \blacktriangleleft C\closure\theta\}\\[-\jot]
    {\modelsolam}\; & \{ \mapF{E\zeta} \mid \slimfull{E \in N,\ \zeta\text{ grounding, and }}{E\closure\zeta \in \irred_R(\mapG{N}) \text{ and }} E\closure\zeta \blacktriangleleft E_0\closure\zeta_0\}
  \end{align*}
If 
\begin{full}$E_0\closure\zeta_0 \in \irred_R(\mapG{N\setminus N'})$,\end{full}%
\begin{slim}$E_0 \in N\setminus N'$,\end{slim}
the claim ($*$) is obvious.
So we may assume that
\begin{full}$E_0\closure\zeta_0 \in \irred_R(\mapG{N'})$.\end{full}%
\begin{slim}$E_0 \in N'$.\end{slim}
The assumption of (R3) states $N' \subseteq \HRedC(N)$, and thus
we have
\[\slimfull{}{R \cup} \{\mapF{E\zeta} \mid \slimfull{E \in N,\ \zeta\text{ grounding, and }}{E\closure\zeta \in \irred_R(\mapG{N}) \text{ and }} E\closure\zeta \blacktriangleleft E_0\closure\zeta_0\}\modelsolam \mapF{E_0\zeta_0}\]
By the induction hypothesis, this implies ($*$).

(R3) for inferences:

Inspecting this definition of $\HRedI$ (Definition~\ref{def:H:RedI}), we observe that
to show that $\HRedI(N)\subseteq \HRedI(N \setminus N')$,
it suffices to prove that
\begin{gather*}
\slimfull{}{R \cup }\{E \in \slimfull{\mapF{\mapG{N\setminus N'}}}{\irred_R(\fipg{N\setminus N'})} \mid E \prec_{\slimfull{}{\mapIonly}\mapFonly} \mapF{C_m\theta_m}\}\\
\modelsolam\\
\slimfull{}{R \cup }\{E \in \slimfull{\mapF{\mapG{N}}}{\irred_R(\fipg{N})} \mid E \prec_{\slimfull{}{\mapIonly}\mapFonly} \mapF{C_m\theta_m}\}\\
\end{gather*}
(possibly without the condition $E \prec_{\slimfull{}{\mapIonly}\mapFonly} \mapF{C_m\theta_m}$ for $\Diff$ inferences),
where $C_m$\slimfull{ and $\theta_m$}{, $\theta_m$, and $R$} are given in the definition of $\HRedI$.
We can equivalently write this as
\begin{align*}
  & \slimfull{}{R \cup }\{\mapF{E\zeta} \mid \slimfull{E \in N\setminus N',\ \zeta\text{ grounding, and }}{E\closure\zeta \in \irred_R(\mapG{N\setminus N'}) \text{ and }} E\zeta \prec C_m\theta_m\}\\[-\jot]
  \modelsolam\; & \{\mapF{E\zeta} \mid \slimfull{E \in N,\ \zeta\text{ grounding, and }}{E\closure\zeta \in \irred_R(\mapG{N}) \text{ and }} E\zeta \prec C_m\theta_m\}
\end{align*}

Let \slimfull{$E_0\in N$ and $\zeta_0$ grounding}{$E_0\closure\zeta_0 \in \irred_R(\mapG{N})$} with $E_0\zeta_0 \prec C_m\theta_m$.
We must show that 
\[  \slimfull{}{R \cup }\{\mapF{E\zeta} \mid \slimfull{E \in N\setminus N',\ \zeta\text{ grounding, and }}{E\closure\zeta \in \irred_R(\mapG{N\setminus N'}) \text{ and }} E\zeta \prec C_m\theta_m\}
\modelsolam \mapF{E_0\zeta_0} \tag{$\dagger$}\]

If \begin{full}$E_0\closure\zeta_0 \in \irred_R(\mapG{N\setminus N'})$,\end{full}%
\begin{slim}$E_0 \in N\setminus N'$,\end{slim}
the claim ($\dagger$) is obvious.
So we may assume that
\begin{full}$E_0\closure\zeta_0 \in \irred_R(\mapG{N'})$.\end{full}%
\begin{slim}$E_0 \in N'$.\end{slim}
The assumption of (R3) states $N' \subseteq \HRedC(N)$, and thus
$N' \subseteq\HRedC(N \setminus N')$ by (R3) for clauses.
So we have
\[ \slimfull{}{R \cup }\{\mapF{E\zeta} \mid \slimfull{E \in N\setminus N',\ \zeta\text{ grounding, and }}{E\closure\zeta \in \irred_R(\mapG{N\setminus N'}) \text{ and }} E\closure\zeta \blacktriangleleft E_0\closure\zeta_0\}\modelsolam \mapF{E_0\zeta_0}\]
This implies ($\dagger$) because 
for any $E\closure\zeta$ with $E\closure\zeta \blacktriangleleft E_0\closure\zeta_0$,
we have $E\zeta \preceq E_0\zeta_0 \prec C_m\theta_m$.

(R4)
Let $\iota \in \HInf$ with $\concl(\iota) \in N$.
We must show that $\iota \in \HRedI(N)$.
Let $C_1\slimfull{}{\constraint{S_1}}$, \dots, $C_m\slimfull{}{\constraint{S_m}}$ be $\iota$'s premises and
$C_{m+1}\slimfull{}{\constraint{S_{m+1}}}$ its conclusion.
Let $\theta_1, \dots, \theta_{m+1}$ be a tuple of substitutions
for which $\iota$ is rooted in $\FInf{}$ (Definition~\ref{def:fol-inferences}).
\begin{full}
Let $R$ be a confluent term rewrite systems $R$ oriented by $\succ_{\mapIonly\mapFonly}$
whose only Boolean normal forms are $\itrue$ and $\ifalse$
such that $C_{m+1}\closure\theta_{m+1}$ is variable-irreducible. 
\end{full}
According to the definition of $\HRedI$ (Definition~\ref{def:H:RedI}),
we must show that
\begin{full}
\[R \cup O \modelsolam \mapF{C_{m+1}\theta_{m+1}}\]
where 
$O = \irred_R(\fipg{N})$
if $\iota$ is a $\Diff$ inference and
$O = \{E \in \irred_R(\fipg{N}) \mid E \prec_{\mapIonly\mapFonly} \mapF{C_m\theta_m}\}$
if $\iota$ is some other inference.
\end{full}%
\begin{slim}
\begin{itemize}
  \item $\mapF{\mapG{N}}\modelsolam \mapF{C_{m+1}\theta_{m+1}}$ if $\iota$ is a $\Diff$ inference; and
  \item $\{E \in \mapF{\mapG{N}} \mid E \prec_{\mapFonly} \mapF{C_m\theta_m}\}\modelsolam \mapF{C_{m+1}\theta_{m+1}}$ if $\iota$ is some other inference.
\end{itemize}
\end{slim}

Since $\concl(\iota) \in N$ and $\concl(\iota) = C_{m+1}\slimfull{}{\constraint{S_{m+1}}}$, we
have $C_{m+1}\slimfull{}{\constraint{S_{m+1}}} \in N$.
\begin{full}
Thus, by Lemma~\ref{lem:fo-eq-tfip},
$\mapF{C_{m+1}\theta_{m+1}} \in \fipg{N}$.
Since $C_{m+1}\closure\theta_{m+1}$ is variable-irreducible,
we have 
$\mapF{C_{m+1}\theta_{m+1}} \in \irred_R(\fipg{N})$.
\end{full}%
\begin{slim}
Thus, $\mapF{C_{m+1}\theta_{m+1}} \in \mapF{\mapG{N}}$.
\end{slim}
This completes the proof for
$\Diff$ inferences
because
$\mapF{C_{m+1}\theta_{m+1}} \modelsolam \mapF{C_{m+1}\theta_{m+1}}$.
For the other inferences, it remains to prove that
$\mapF{C_{m+1}\theta_{m+1}} \prec_{\slimfull{}{\mapIonly}\mapFonly} \mapF{C_m\theta_m}$.

By Definition~\ref{def:fol-inferences},
$\mapF{C_m\theta_m}$ is the main premise and
$\mapF{C_{m+1}\theta_{m+1}}$ is the conclusion of
an $\FInf$ inference.
We will show for each $\FInf$ rule that the conclusion is smaller than the main premise.

\begin{slim}
By Lemma~\ref{lem:PG:mapI-mapF-fol},
$\succ_{\mapIonly\mapFonly} = \succ_{\mapFonly}$.
By Lemmas~\ref{lem:PG:succ-transfer}
and~\ref{lem:IPG:succ-transfer}, it follows that $\succ_\mapFonly$ is admissible for $\FInf$.
\end{slim}

For $\FSup$,
we must argue that $\subterm{C}{t} \succ_{\slimfull{}{\mapIonly}\mapFonly} D'\llor\subterm{C}{t'}$.
Since the literal $t \ceq t'$ is strictly eligible in $D$
and if $t'$ is Boolean, then $t' = \itrue$,
the literal $t \ceq t'$ is strictly maximal in $D$.
Since the position of $t$ is eligible in $C[t]$,
it must either occur in a negative literal, 
in a literal of the form $t \ceq \ifalse$, or in a maximal literal in $C[t]$.
If the position of $t$ is in a negative literal or in a literal of the form $t \ceq \ifalse$, then
that literal is larger than $t \ceq t'$
because if $t'$ is Boolean, then $t' = \itrue$.
Thus, the literal in which $t$ occurs in $C[t]$ is larger than $D'$ because $t \ceq t'$ is
strictly maximal in $D$.
If the position of $t$ is in a  maximal literal of $C[t]$,
then that literal is larger than or equal to $t \ceq t'$
because $D \prec_{\slimfull{}{\mapIonly}\mapFonly} \subterm{C}{t}$,
and thus it is larger than $D'$ as well.
In $\subterm{C}{t'}$, this literal is replaced by a smaller literal
because $t \succ_{\slimfull{}{\mapIonly}\mapFonly} t'$.
So $\subterm{C}{t} \succ_{\slimfull{}{\mapIonly}\mapFonly} D'\llor\subterm{C}{t'}$.

For $\FEqRes$, clearly, $C' \llor u \cneq u \succ_{\slimfull{}{\mapIonly}\mapFonly} C'$.

For $\FEqFact$, we have $u \ceq v \succeq_{\slimfull{}{\mapIonly}\mapFonly} u \ceq v'$ and thus $v \succeq_{\slimfull{}{\mapIonly}\mapFonly} v'$.
Since $u \succ_{\slimfull{}{\mapIonly}\mapFonly} v$, we have $u \ceq v \succ_{\slimfull{}{\mapIonly}\mapFonly} v \cneq v'$ and thus
the premise is larger than the conclusion.

For $\FClausify$, it is easy to see that for any of the listed
values of $s$, $t$, and $D$, we have $s \ceq t \succ_{\slimfull{}{\mapIonly}\mapFonly} D$,
using \ref{cond:PF:order:subterm} and \ref{cond:PF:order:t-f-minimal}.
Thus the premise is larger than the conclusion.

For $\FBoolHoist$ and $\FLoobHoist$,
we have $u \succ_{\slimfull{}{\mapIonly}\mapFonly} \ifalse$ and $u \succ_{\slimfull{}{\mapIonly}\mapFonly} \itrue$
by \ref{cond:PF:order:t-f-minimal}
because $u \ne \ifalse$ and $u \ne \itrue$.
Moreover, the occurrence of $u$ in $\subterm{C}{u}$
is required not to be in a literal of the form
$u \ceq \ifalse$ or $u \ceq \itrue$, 
and thus, by \ref{cond:PF:order:t-f-minimal},
it must be in a literal larger than these.
It follows that the premise is larger than the conclusion.

For $\FFalseElim$, clearly, $C' \llor \ifalse \ceq \itrue \succ_{\slimfull{}{\mapIonly}\mapFonly} C'$.

For $\FArgCong$, the premise is larger than the conclusion by \ref{cond:PF:order:ext}.

For $\FExt$, we use the condition that $u \succ_{\slimfull{}{\mapIonly}\mapFonly} w$ and \ref{cond:PF:order:subterm}
to show that $\subterm{C}{\mapF{w}}$ is smaller than the premise.
We use $u \succ_{\slimfull{}{\mapIonly}\mapFonly} w$ and \ref{cond:PF:order:ext} to show that
$\mapF{u\>\diff\typeargs{\tau,\upsilon}(u,w)} \cneq \mapF{w\>\diff\typeargs{\tau,\upsilon}(u,w)}$
is smaller than the premise.
\end{proof}

\begin{full}
\begin{lem}\label{lem:G:irred-entails-all}
  Let $R$ be a confluent term rewrite system on $\termsPF$ oriented by $\succ_{\mapIonly\mapFonly}$
  whose only Boolean normal forms are $\itrue$ and $\ifalse$.
  Let $N\subseteq\clausesG$ be a clause set without parameters
  such that for every $C\closure\theta \in N$
  and every grounding substitution $\rho$,
  we have $C\closure \rho\in N$.
  Then $R \cup \mapF{\mapI{\mapP{\irred_R(N)}}} \modelsolam \mapF{\mapI{\mapP{N}}}$.
\end{lem}
\begin{proof}
We apply Lemma~\ref{lem:PG:irred-entails-all}.
The required condition on $\mapI{N}$
can be derived from this lemma's condition on $N$ as follows.
We must show that
for every $C\closure\theta \in \mapP{N}$ and
every grounding substitution $\rho$ that
coincides with $\theta$ on all variables not occurring in $C$,
we have $C\closure \rho\in \mapP{N}$.
The closure $C\closure\theta \in \mapP{N}$ must be of the form $C'\mapp{\theta'}\closure\mapq{\theta'}$
with $C'\closure\theta' \in N$.
Define $\rho' = \mapp{\theta'}\rho$.
By this lemma's condition on $N$,
it follows that $C'\closure \rho'\in N$.
By Lemma~\ref{lem:G:mapp-mapp},
$\mapp{\rho'} = \mapp{\theta'}$.
We have
$y\rho = y\theta = y\mapq{\theta'}$ 
for all variables $y$ not occurring in $C$
and in particular
for all $y$ not introduced by $\mapp{\theta'}$.
Thus, by Lemma~\ref{lem:G:mapq-mapp},
$\mapq{\rho'} = \rho$.
So, $C\closure \rho = \mapP{C'\closure \rho'} \in \mapP{N}$.
\end{proof}
\end{full}

\begin{thm}\label{thm:H:completeness}
  Given a fair derivation 
  $(N_i)_{i\geq 0}$, where 
  \begin{enumerate}[label=\arabic*.,ref=\arabic*]
    \item \label{thm:H:completeness:unsat} $N_0$ does not have a term-generated model,
    \item \label{thm:H:completeness:no-params} $N_0$ does not contain parameters, \slimfull{}{and}
    \begin{full}\item \label{thm:H:completeness:no-constr} $N_0$ does not contain constraints,\end{full}
  \end{enumerate}
  we have $\bot\slimfull{}{\constraint{S}} \in N_i$ for \slimfull{}{some satisfiable constraints $S$ and} some index $i$.
\end{thm}
\begin{proof}
By Lemma~9 of Waldmann et al.~\cite{waldmann-et-al-saturation-journal},
using Lemma~\ref{lem:H:red-criteria-properties},
the limit $N_\infty = \bigcup_i\bigcap_{j\geq i} N_j$
is saturated up to redundancy \wrt\ $\HInf$ and $\HRedI$.
By Lemma~\ref{lem:H:saturation},
$\slimfull{\mapG{N_\infty}}{\mapPG{N_\infty}}$ is saturated up to redundancy
\wrt\ $\PGInf$ and $\PGRedI$.

For a proof by contradiction,
assume that for
\begin{full}all $S$ and\end{full}
all $i$,
$\bot\slimfull{}{\constraint{S}} \not\in N_i$.
Then $N_\infty$ does not contain \slimfull{$\bot$}{such a clause $\bot\constraint{S}$} either,
and thus $\slimfull{\mapG{N_\infty}}{\mapPG{N_\infty}}$ does not contain \slimfull{$\bot$}{a clause of the form $\bot\closure\theta$ for any $\theta$}.
By 
Lemma~\ref{thm:PG:model-for-varirred-closures},
$\III^\levelIPG \models \slimfull{\mapI{\mapG{N_\infty}}}{\irred_{R}(\mapI{\mapPG{N_\infty}})}$%
\begin{full}, where $R = \Rbasic_{\fipg{N_\infty}}$\end{full}.

By Lemma~8 of Waldmann et al.~\cite{waldmann-et-al-saturation-journal},
using Lemma~\ref{lem:H:red-criteria-properties},
$N_0 \subseteq N_\infty \cup \HRedC(N_\infty)$.
Thus, 
\begin{full} $R \cup \irred_R(\fipg{N_\infty}) \modelsolam \irred_R(\fipg{N_0})$.\end{full}%
\begin{slim} $\mapF{\mapG{N_\infty}} \modelsolam \mapF{\mapG{N_0}}$.\end{slim}%
\begin{full}
By Lemma~\ref{lem:G:irred-entails-all} and
conditions \ref{thm:H:completeness:no-params}~and~\ref{thm:H:completeness:no-constr} from the present
theorem,
$R \cup \irred_R(\fipg{N_0}) \modelsolam \fipg{N_0}$
and thus $R \cup \irred_R(\fipg{N_\infty}) \modelsolam \fipg{N_0}$.
\end{full}
Since 
$\III^\levelIPG \models \slimfull{\mapI{\mapG{N_\infty}}}{\irred_{R}(\mapI{\mapPG{N_\infty}})}$,
by Lemma~\ref{thm:IPG:model-mapF-iff}\slimfull{ and Lemma~\ref{lem:PG:mapI-mapF-fol}}{},
it follows that
$\III^\levelIPG \models \mapI{\slimfull{\mapG{N_0}}{\mapP{\mapG{N_0}}}}$.

\begin{full}
If we applied each closure's substitution to its clause in the sets
$\mapI{\mapP{\mapG{N_0}}}$ and $\mapI{\gnd(N_0)}$,
the two sets would be identical.
So, since 
$\III^\levelIPG \models \mapI{\mapP{\mapG{N_0}}}$,
we have $\III^\levelIPG \models \mapI{\gnd(N_0)}$.
\end{full}

Now $\III^\levelIPG$ can be shown to be a model of $N_0$ as follows.
Let $C \in N_0$. Let $\xi$ be a valuation. Since $\III^\levelIPG$ is term-generated by Lemma~\ref{lem:IPG:term-generated},
there exists a substitution $\theta$ such that $\interpret{\alpha\theta}{\IIIty^\levelIPG}{} = \xity(\alpha)$
for all type variables $\alpha$ in $C$ and $\interpret{x\theta}{\III^\levelIPG}{} = \xite(x)$ for all term variables $x$ in $C$.
Since $C$ does not contain parameters by condition~\ref{thm:H:completeness:no-params} of this theorem, 
$C\theta \in \mapI{\slimfull{\mapG{N_0}}{\gnd(N_0)}}$. Thus we have $\III^\levelIPG \models C\theta$.
By Lemma~\ref{lem:subst-lemma-clause}, it follows that $C$ is true \wrt\ $\xi$ and $\III^\levelIPG$.
Since $\xi$ and $C \in N_0$ were arbitrary, we have $\III^\levelIPG \models N_0$.
This contradicts condition~\ref{thm:H:completeness:unsat} of the present theorem.

\end{proof}

\begin{lem}\label{lem:H:soundmodels-implies-models}
  Let $N$ be a clause set that does not contain $\cst{diff}$.
  If $N$ has a term-generated model, then $N$ has a $\cst{diff}$-aware model.
\end{lem}
\begin{proof}
Let $\III = (\IIIty, \II, \LL)$ be a model of $N$.
We assume that the signature of $\III$ does not contain $\cst{diff}$.
We extend it into a $\cst{diff}$-aware model $\III'=(\IIIty', \II', \LL')$ as follows.

We define $\II'(\cst{diff}, \DD_1, \DD_2, a, b)$
to be an element $e \in \DD_1$ such that
$a(e) \ne b(e)$ if such an element exists and an arbitrary element of $\DD_1$ otherwise.
This ensures that $\III'$ is $\diff$-aware (Definition~\ref{def:H:diff-aware}).

To define $\LL'$, let $\xi$ be a valuation and $t$ be a $\lambda$-abstraction.
We replace each occurrence of $\diff\typeargs{\tau,\upsilon}(u,w)$ in $t$
with a ground term $s$ that does not contain $\cst{diff}$ such that
$\interpret{s}{\III}{} = \II'(\cst{diff}, \interpret{\tau}{\IIIty}{\xity},  \interpret{\upsilon}{\IIIty}{\xity}, \interpretaxi{u},  \interpretaxi{w})$.
Such a term $s$ exists because $\III$ is term-generated.
We start replacing the innermost occurrences of $\diff$ and proceed outward
to ensure that the parameters of a replaced $\cst{diff}$ do not contain $\cst{diff}$ themselves.
Let $t'$ be the result of this replacement.
Then we define $\LL'(\xi, t) = \LL(\xi, t')$.
This ensures that $\III'$ is a proper interpretation.

Since $N$ does not contain $\cst{diff}$
and $\III$ is a model of $N$,
it follows that $\III'$ is a model of $N$ as well.
\end{proof}

\begin{cor}\label{cor:H:completeness}
  Given a fair derivation 
  $(N_i)_{i\geq 0}$, where 
  \begin{enumerate}[label=\arabic*.,ref=\arabic*]
    \item \label{col:H:completeness:unsat} $N_0 \soundmodels \bot$,\slimfull{ and}{}
    \item \label{col:H:completeness:no-params} $N_0$ does not contain parameters,\slimfull{}{ and}
    \begin{full}\item \label{col:H:completeness:no-constr} $N_0$ does not contain constraints,\end{full}
  \end{enumerate}
  we have $\bot\slimfull{}{\constraint{S}} \in N_i$ for \slimfull{}{some satisfiable constraints $S$ and} some index $i$.
\end{cor}
\begin{proof}
  By Theorem~\ref{thm:H:completeness} and Lemma~\ref{lem:H:soundmodels-implies-models}.
\end{proof}

\begin{full}

\section{Conclusion}

We presented the optimistic $\lambda$-superposition calculus. It is inspired by
the original $\lambda$-super\-position calculus of Bentkamp et
al.~\cite{bentkamp-et-al-2023-hosup-journal}, which in turn generalizes the
standard superposition calculus by Bachmair and Ganzinger
\cite{bachmair-ganzinger-1994}.
Our calculus has many advantages over the original $\lambda$-superposition
calculus, including more efficient handling of unification, functional
extensionality, and redundancy. Admittedly, its main disadvantage is its
lengthy refutational completeness proof.

We have some ideas on how to extend the calculus further:

\begin{itemize}
\item We believe that the inference rules that still require full unification
  could be adapted to work with partial unification by adding annotations to constrained clauses.
  The annotations would indicate which variables and which constraints stem from
  rules with the $\infname{Fluid}$- prefix.
  A modification of the map $\mapponly$ used in our proof could ensure that these variables
  do not carry the guarantee of being variable-irreducible that currently all variables carry.
  As a result, the proof of Lemma~\ref{lem:H:saturation} would
  no longer require full unification,
  but additional $\FluidSup$ inferences would be required into the variables
  marked by the annotations.

\item We conjecture that the $\Diff$ axiom is not necessary for
  refutational completeness although our proof
  currently requires it.
  Our proof uses it in Lemma~\ref{lem:argcong} to show that the constructed model
  is a valid higher-order model in the sense that
  equality of functions
  implies equality of their values on all arguments.
  We suspect that one can construct a
  model with this property using saturation \wrt\ $\ArgCong$ alone,
  but the model construction must be different from the one used in the present proof.

\item
One of the most explosive rules of the calculus is $\FluidSup$.
Bhayat and Suda \cite{bhayat-suda-2024} propose 
a modification of inference rules 
that delays flex-rigid pairs and flex-flex pairs by adding
them as negative literals to the conclusion.
They suggest that this modification in conjunction
with additional inference rules for the unification of
flex-rigid pairs could remove the need for $\FluidSup$.
We conjecture that one could prove refutational completeness of such a calculus
by restructuring Lemma~\ref{lem:reduction-prop}
to apply the modified inference rules
instead of Lemma~\ref{lem:reducible-s} whenever
the only terms reducible by $R_C$ 
correspond to positions below applied variables on level $\levelH$.

\item Similarly,
we conjecture that one could remove the $\Ext$ rule
by following the idea of Bhayat \cite{bhayat-2021-thesis} to
delay unification of functional terms
by adding them as negative literals to the conclusion.
If we immediately apply $\infname{NegExt}$ to these additional literals,
one can possibly prove refutational completeness
by restructuring Lemma~\ref{lem:reduction-prop}
to apply the modified inference rules
instead of Lemma~\ref{lem:reducible-s} whenever
the only terms reducible by $R_C$ are functional terms.
\end{itemize}

\subsubsection*{Acknowledgment}

We thank
Ahmed Bhayat,
Massin Guerdi, and
Martin Desharnais
for suggesting textual improvements.
We thank Nicolas Peltier and Maria Paola Bonacina
who we discussed some early ideas with.

Bentkamp and Blanchette's research has received funding from the 
European Research Council
(ERC, Matryoshka, 713999 and Nekoka, 101083038)
Bentkamp's research has
received funding from a Chinese Academy of Sciences President's International Fellowship
for Postdoctoral Researchers (grant No. 2021PT0015)
and from the program Freiraum 2022 of the Stiftung Innovation in der Hochschullehre
(ADAM: Anticipating the Digital Age of Mathematics, FRFMM-83/2022).
Blanchette's research has received
funding from the Netherlands Organization for Scientific Research (NWO) under the Vidi
program (project No. 016.Vidi.189.037, Lean Forward).
Hetzenberger's research has received funding from the 
European Research Council (ERC, ARTIST, 101002685).

Views and opinions expressed are however those of the authors only and do not necessarily reflect those of the European Union or the European Research Council. Neither the European Union nor the granting authority can be held responsible for them.

We have used artificial intelligence tools for textual editing.
\end{full}

\bibliographystyle{plain}
\bibliography{ms}

\begin{thebibliography}{10}

\bibitem{bachmair-ganzinger-1994}
Leo Bachmair and Harald Ganzinger.
\newblock Rewrite-based equational theorem proving with selection and
  simplification.
\newblock {\em J. Log. Comput.}, 4(3):217--247, 1994.

\bibitem{bachmair-ganzinger-2001-resolution}
Leo Bachmair and Harald Ganzinger.
\newblock Resolution theorem proving.
\newblock In John~Alan Robinson and Andrei Voronkov, editors, {\em Handbook of
  Automated Reasoning}, volume~I, pages 19--99. Elsevier and {MIT} Press, 2001.

\bibitem{bachmair-et-al-1992}
Leo Bachmair, Harald Ganzinger, Christopher Lynch, and Wayne Snyder.
\newblock Basic paramodulation and superposition.
\newblock In Deepak Kapur, editor, {\em CADE-11}, volume 607 of {\em LNCS},
  pages 462--476. Springer, 1992.

\bibitem{bentkamp-et-al-optimistic-orders}
Alexander Bentkamp, Jasmin Blanchette, and Matthias Hetzenberger.
\newblock Term orders for optimistic superposition (unpublished manusscript).
\newblock \url{https://nekoka-project.github.io/pubs/optimistic_orders.pdf}.

\bibitem{bentkamp-et-al-2023-hosup-errata}
Alexander Bentkamp, Jasmin Blanchette, Sophie Tourret, and Petar Vukmirovi\'c.
\newblock Errata of ``{S}uperposition for higher-order logic''.
\newblock
  \url{https://matryoshka-project.github.io/pubs/hosup_article_errata.pdf}.

\bibitem{bentkamp-et-al-2023-hosup-journal}
Alexander Bentkamp, Jasmin Blanchette, Sophie Tourret, and Petar
  Vukmirovi{\'c}.
\newblock Superposition for higher-order logic.
\newblock {\em J. Autom. Reason.}, 67(1):10, 2023.

\bibitem{bentkamp-et-al-2021-lamsup-journal-errata}
Alexander Bentkamp, Jasmin Blanchette, Sophie Tourret, Petar Vukmirovic, and
  Uwe Waldmann.
\newblock Errata of ``{S}uperposition with lambdas''.
\newblock
  \url{https://matryoshka-project.github.io/pubs/lamsup_article_errata.pdf}.

\bibitem{bentkamp-et-al-2021-lamsup-journal}
Alexander Bentkamp, Jasmin Blanchette, Sophie Tourret, Petar Vukmirovic, and
  Uwe Waldmann.
\newblock Superposition with lambdas.
\newblock {\em J. Autom. Reason.}, 65(7):893--940, 2021.

\bibitem{benzmuller-2015-leo2}
Christoph Benzm{\"{u}}ller, Nik Sultana, Lawrence~C. Paulson, and Frank Theiss.
\newblock The higher-order prover \textsc{Leo}-{II}.
\newblock {\em J. Autom. Reason.}, 55(4):389--404, 2015.

\bibitem{bhayat-2021-thesis}
Ahmed Bhayat.
\newblock {\em Automated theorem proving in higher-order logic}.
\newblock PhD thesis, University of Manchester, 2021.

\bibitem{bhayat-suda-2024}
Ahmed Bhayat and Martin Suda.
\newblock A higher-order vampire (short paper).
\newblock In {\em {IJCAR} {(1)}}, volume 14739 of {\em LNCS}, pages 75--85.
  Springer, 2024.

\bibitem{blanqui-et-al-2015}
Fr{\'{e}}d{\'{e}}ric Blanqui, Jean-Pierre Jouannaud, and Albert Rubio.
\newblock The computability path ordering.
\newblock {\em Log.\ Meth.\ Comput.\ Sci.}, 11(4), 2015.

\bibitem{chargueraud-2012}
Arthur Chargu{\'{e}}raud.
\newblock The locally nameless representation.
\newblock {\em J. Autom. Reason.}, 49(3):363--408, 2012.

\bibitem{de-bruijn-1972}
N.~G. de~Bruijn.
\newblock Lambda calculus notation with nameless dummies, a tool for automatic
  formula manipulation, with application to the {C}hurch--{R}osser theorem.
\newblock {\em Indag. Math}, 75(5):381--392, 1972.

\bibitem{dershowitz-manna-1979}
Nachum Dershowitz and Zohar Manna.
\newblock Proving termination with multiset orderings.
\newblock {\em Commun. {ACM}}, 22(8):465--476, 1979.

\bibitem{dowek-2001}
Gilles Dowek.
\newblock Higher-order unification and matching.
\newblock In John~Alan Robinson and Andrei Voronkov, editors, {\em Handbook of
  Automated Reasoning}, volume~II, pages 1009--1062. Elsevier and {MIT} Press,
  2001.

\bibitem{fitting-2002}
Melvin Fitting.
\newblock {\em Types, Tableaus, and {G}{\"o}del's God}.
\newblock Kluwer, 2002.

\bibitem{gordon-melham-1993}
M.~J.~C. Gordon and T.~F. Melham, editors.
\newblock {\em Introduction to {HOL}: A Theorem Proving Environment for Higher
  Order Logic}.
\newblock Cambridge University Press, 1993.

\bibitem{huet-1975}
G{\'{e}}rard~P. Huet.
\newblock A unification algorithm for typed lambda-calculus.
\newblock {\em Theor. Comput. Sci.}, 1(1):27--57, 1975.

\bibitem{kaliszyk-et-al-2016}
Cezary Kaliszyk, Geoff Sutcliffe, and Florian Rabe.
\newblock {TH1}: The {TPTP} typed higher-order form with rank-1 polymorphism.
\newblock In Pascal Fontaine, Stephan Schulz, and Josef Urban, editors, {\em
  PAAR-2016}, volume 1635 of {\em {CEUR} Workshop Proceedings}, pages 41--55.
  CEUR-WS.org, 2016.

\bibitem{koenigs-lemma-1927}
Dénes K\H{o}nig.
\newblock {\"U}ber eine {S}chlussweise aus dem {E}ndlichen ins {U}nendliche.
\newblock {\em Acta Sci. Math. (Szeged)}, 3499\slash 2009(3:2--3):121--130,
  1927.

\bibitem{miller-1992}
Dale Miller.
\newblock Unification under a mixed prefix.
\newblock {\em J. Symb. Comput.}, 14(4):321--358, 1992.

\bibitem{nieuwenhuis-rubio-1992}
Robert Nieuwenhuis and Albert Rubio.
\newblock Basic superposition is complete.
\newblock In Bernd Krieg{-}Br{\"{u}}ckner, editor, {\em {ESOP} '92}, volume 582
  of {\em LNCS}, pages 371--389. Springer, 1992.

\bibitem{nieuwenhuis-rubio-1995}
Robert Nieuwenhuis and Albert Rubio.
\newblock Theorem proving with ordering and equality constrained clauses.
\newblock {\em J. Symb. Comput.}, 19(4):321--351, 1995.

\bibitem{nummelin-et-al-2021-boolsup-errata}
Visa Nummelin, Alexander Bentkamp, Sophie Tourret, and Petar Vukmirovi\'c.
\newblock Errata of ``{S}uperposition with first-class booleans and
  inprocessing clausification''.
\newblock \url{https://matryoshka-project.github.io/pubs/boolsup_errata.pdf}.

\bibitem{nummelin-et-al-2021}
Visa Nummelin, Alexander Bentkamp, Sophie Tourret, and Petar Vukmirovi\'c.
\newblock Superposition with first-class {B}ooleans and inprocessing
  clausification.
\newblock In Andr\'e Platzer and Geoff Sutcliffe, editors, {\em CADE-28},
  volume 12699 of {\em LNCS}, pages 378--395. Springer, 2021.

\bibitem{schulz-2002-brainiac}
Stephan Schulz.
\newblock E - a brainiac theorem prover.
\newblock {\em {AI} Commun.}, 15(2-3):111--126, 2002.

\bibitem{sutcliffe-casc-j12}
Geoff Sutcliffe.
\newblock The 12th {IJCAR} automated theorem proving system
  competition---{CASC-J12}.
\newblock {\em The European Journal on Artificial Intelligence}, 38(1):3--20,
  2025.

\bibitem{sutcliffe-desharnais-casc-29}
Geoff Sutcliffe and Martin Desharnais.
\newblock The {CADE-29} automated theorem proving system
  competition---{CASC-29}.
\newblock {\em {AI} Commun.}, 37(4):485--503, 2024.

\bibitem{vukmirovic-et-al-2022-making-hosup-work}
Petar Vukmirovi\'c, Alexander Bentkamp, Jasmin Blanchette, Simon Cruanes, Visa
  Nummelin, and Sophie Tourret.
\newblock Making higher-order superposition work.
\newblock {\em J. Autom. Reason.}, 66(4):541--564, 2022.

\bibitem{vukmirovic-et-al-2021-unif}
Petar Vukmirovi\'c, Alexander Bentkamp, and Visa Nummelin.
\newblock Efficient full higher-order unification.
\newblock {\em Log. Methods Comput. Sci.}, 17(4), 2021.

\bibitem{vukmirovic-et-al-2023-extending}
Petar Vukmirovi\'c, Jasmin Blanchette, and Stephan Schulz.
\newblock Extending a high-performance prover to higher-order logic.
\newblock In {\em {TACAS} 2023}, volume 13994 of {\em LNCS}, pages 111--129.
  Springer, 2023.

\bibitem{waldmann-et-al-saturation-journal}
Uwe Waldmann, Sophie Tourret, Simon Robillard, and Jasmin Blanchette.
\newblock A comprehensive framework for saturation theorem proving.
\newblock {\em J. Autom. Reason.}, 66(4):499--539, 2022.

\end{thebibliography}

\end{document}